\documentclass[letterpaper,twocolumn,10pt,accepted=2025-07-30]{quantumarticle}
\pdfoutput=1

\usepackage[utf8]{inputenc}
\usepackage[margin=1in]{geometry} 
\usepackage{amsmath,amsthm,amssymb,bbm}
\usepackage[colorlinks=true, linkcolor=blue, citecolor=blue, urlcolor=blue]{hyperref}

\usepackage[numbers,sort&compress]{natbib}
\usepackage{placeins}
\usepackage{adjustbox}
\usepackage{graphicx}
\usepackage{caption}
\usepackage{float}
\usepackage{subcaption}
\usepackage{xcolor}
\usepackage{soul}
\newtheorem{theorem}{Theorem}
\newtheorem{definition}{Definition}
\begin{document}

\setstcolor{red}

\title{Magic Boundaries of 3D Color Codes}
\author{Zijian Song}
\affiliation{Department of Physics and Astronomy, University of California, Davis, California 95616, USA}

\author{Guanyu Zhu}
\affiliation{IBM Quantum, T.J. Watson Research Center, Yorktown Heights, NY 10598 USA}

\begin{abstract}
 We investigate boundaries of 3D color codes and provide a systematic classification into 101 distinct boundary types, including two novel classes. The first class consists of 1 boundary and is generated by sweeping the codimension-1 (2D) $T$-domain wall across the system and attaching it to the $X$-boundary that condenses only magnetic fluxes. Since the $T$-domain wall cannot condense on the $X$-boundary, a new {\it magic boundary} is produced, where the boundary stabilizers contain $XS$-stabilizers going beyond the conventional Pauli stabilizer formalism, and hence contains `magic'. Neither electric nor magnetic excitations can condense on such a magic boundary, and only the composite of the magnetic flux and codimension-2 (1D) $S$-domain wall can condense on it, which makes the magic boundary going beyond the classification of the Lagrangian subgroup.  The second class consists of 70 boundaries and is generated by sweeping the $S$-domain wall across a codimension-1 submanifold and attaching it to the boundary. This generates a codimension-2 (1D) {\it nested boundary} at the intersection. We also connect these novel boundaries to their previously discovered counterpart in the $\mathbb{Z}_2^3$ gauge theory, equivalent to three copies of 3D toric codes, where the $S$ and $T$ domain walls correspond to gauged symmetry-protected topological (SPT) defects.  New boundaries are produced whenever the corresponding symmetry of the SPT defect remains unbroken on the boundary.  Applications of the magic boundaries include implementing fault-tolerant non-Clifford logical gates, e.g., in the context of fractal topological codes.

\end{abstract}

\maketitle

{\hypersetup{linkcolor=black}
\tableofcontents}

\section{Introduction}

In recent decades, there has been significant progress on the development of quantum error-correcting codes and fault-tolerant quantum computation. A well-known class of error-correcting codes is the topological stabilizer codes (TSCs). The error syndrome that occurs in TSCs can usually be detected by measuring the stabilizers, and the code distance can be made arbitrarily large by increasing the size of the lattice~\cite{dennis2002topological}. The family of TSCs includes the surface codes~\cite{kitaev2003fault,bravyi1998quantum}, the color codes~\cite{bombin2006topological}, and their generalizations with non-trivial topology, boundary conditions, and topological  defects, etc.~\cite{eastin2009restrictions,fowler2009high,bombin2010topological,bombin2007topological}. From the perspective of quantum matter, a  topological stabilizer code also has a corresponding Abelian topological order, which can be described by an underlying topological quantum field theory (TQFT) or equivalently a $G$-gauge theory with a discrete Abelian gauge group $G$.

Among a large family of topological stabilizer codes, the color code has obtained significant interest due to its capability for fault-tolerant quantum operations, particularly transversal logical gates~\cite{bombin2006topological}. Logical gates from the Clifford group can be implemented transversally in the 2D color code by applying a unitary operation to each physical qubit. Additionally, it is worth noting that the 2D color code boundary can be used to distill magic states through code deformation~\cite{bombin2009quantum}, which can be utilized for universal quantum computing~\cite{gottesman1998heisenberg}. More generally, a $D$-dimensional color code admits transversal logical gates in the $D^\text{th}$ level of the Clifford hierarchy \cite{kubica2015unfolding}. Furthermore, an equivalence is established between $D$-dimensional color codes and $D$ copies of $D$-dimensional toric codes~\cite{kubica2015unfolding}.  From the TQFT perspective, a $D$-dimensional toric code corresponds to a ($D$+1)-dimensional $\mathbb{Z}_2$ gauge theory, while a  $D$-dimensional color code corresponds to a ($D$+1)-dimensional $\mathbb{Z}_2^D$ gauge theory.

It was demonstrated in Ref.~\cite{levin2013protected} that for (2+1)D Abelian topological order, each gapped boundary can be associated with a corresponding {\it Lagrangian subgroup}. This subgroup constitutes the abelian subset of the topological excitations in the bulk that can condense (annihilate) on the boundary. However, an exotic boundary in three copies of 3D toric codes equivalent to a (3+1)D $\mathbb{Z}_2^3$ gauge theory, which will be dubbed as the `\textit{magic boundary}', was constructed in Ref.~\cite{zhu2022topological} that goes beyond the classification of the Lagrangian subgroup even though the corresponding (3+1)$D$ topological order is still Abelian. Given the equivalence between a 3D color code and three copies of 3D toric codes ($\mathbb{Z}_2^3$ gauge theory), a natural question arises: what are the microscopic descriptions and properties of this boundary in the context of the 3D color code?

\subsection{Main results and ideas}

In this paper, we explicitly construct an exact microscopic lattice model of this exotic boundary based on the 3D color code and further scrutinize its properties. We refer to this boundary as the {\it magic boundary} since it goes beyond the conventional Pauli stabilizer formalism, or refer to it as the {\it $XS$-boundary} since it has $XS$-stabilizers on it. This boundary is realized by first creating an $X$-boundary of the 3D color code, followed by applying a transversal-$T$ gate, which is effectively sweeping the corresponding $T$-domain wall across the system. Given the duality between the $T$-gate in the 3D color code and the $\mathrm{CCZ}$-gate in three copies of 3D toric codes, the magic boundary realized in our work is also dual to the boundary explored in Ref.~\cite{zhu2022topological}.

This magic boundary displays exotic properties that extend beyond the conventional classification by the Lagrangian subgroup. In this context, none of the topological charges or fluxes from the bulk can condense on this boundary. This is attributed to the fact that the $T$-domain wall can not condense on the $X$-boundary; instead, it attaches to it. The interplay between the $X$-boundary and the $T$-domain wall imparts this boundary with its exotic character. We further demonstrate that altering the condensations on this exotic boundary causes its unique properties to vanish, rendering it equivalent to an elementary boundary found in the 3D color code.

To complete the classification of boundaries, we also consider the boundaries attached with codimension-2 domain walls, namely, the $S$-domain walls. This boundary is dubbed as the `{\it nested boundary}'. There are three independent $S$-domain walls in the 3D color code. We explicitly show that, for the $X$-boundary, no $S$-domain wall can condense on it; instead, it attaches to it. Consequently, the condensation set at the intersection differs from other parts of the boundary which leads to a nested boundary. We further demonstrate that for the boundaries preserving $\mathbb{Z}_2 \times \mathbb{Z}_2$ symmetries, the corresponding $S$-domain walls cannot condense on them. Conversely, for the boundaries where the corresponding $\mathbb{Z}_2 \times \mathbb{Z}_2$ symmetries are spontaneously broken, the corresponding $S$-domain walls can condense on them.

From the TQFT perspective, the transversal $T$ and $S$ gates correspond to two types of emergent symmetries of the $\mathbb{Z}_2^3$ gauge theory. Applying these symmetries transversally within a specific region $\mathcal{R}$ generates gapped domain walls along the boundaries. As pointed out in Ref.~\cite{yoshida2015topological, yoshida2016topological, yoshida2017gapped, barkeshli2023codimension, barkeshli2022higher, zhu2023non}, those two types of domain walls are gauged symmetry-protected topological (SPT) defects, which can be generated by decorating a (3+1)D invertible phases with a (2+1)D $\mathbb{Z}_2^3$ SPT (3-cocycle of type-III) or (1+1)D $\mathbb{Z}_2^2$ SPT (2-cocycle of type-II) \cite{propitius1995topological,chen2013symmetry}, respectively, and then gauging the global $\mathbb{Z}_2^3$ symmetry. Applying transversal $T$ or $S$ gate is effectively sweeping the corresponding domain walls across the entire system or codimension-1 submanifold. 

In the presence of boundaries, the $e$-boundary (boundary with $Z$-type condensations) spontaneously breaks the $\mathbb{Z}_2$ symmetry, whereas the $m$-boundary (boundary with $X$-type condensations) preserves the symmetry. Consequently, when the gauged SPT defects, such as $T$ or $S$ domain walls, touch the $e$-boundary, the corresponding global symmetry is also spontaneously broken, which trivializes the gauged SPT defects, and hence allows them to condense on the $e$-boundary. In contrast, when the gauged SPT defects touch the $m$-boundary, the global symmetry is preserved. As a result, they cannot condense on the $m$-boundary but instead attach to it and produce a new boundary.

Based on these insights, we systematically classify and explicitly construct 101 distinct types of boundaries for the 3D color code, including 1 magic boundary and 70 nested boundaries. While permuting all the color labels could yield additional boundary types, these permutations result in boundaries that are physically equivalent.  This work also serves as a classification of the gapped boundaries of the 3D color code based on the current understanding of emergent symmetries, symmetry defects, and gapped boundaries in the corresponding $\mathbb{Z}^3_2$ gauge theory.

\subsection{Outline}

The paper is organized as follows. In Section~\ref{sec:3D_toric_code} and \ref{sec:3D_color_code}, we review the definition of the 3D toric code, the 3D color code, and their corresponding topological excitations. In Section~\ref{sec:lagrangian_subgroup}, we review the definition of the Lagrangian subgroup, and In Section~\ref{sec:unfold_without_boundary}, we introduce the unfolding unitary transformation that maps between the toric code and the color code without boundaries. In Section~\ref{sec:boundaries}, we introduce the unfolding unitaries on the systems with boundaries, and further define the $X$-boundary, the $Z$-boundary and the other boundaries with a mixture of $X$ and $Z$ types condensations. In Section~\ref{sec:magic}, we construct the magic boundary from the 3D color code with $X$-boundaries, and show the exotic properties of this boundary explicitly in Subsection~\ref{sec:magicboundary}. We further discuss the nested boundaries in Subsection~\ref{sec:nested}, which are boundaries with codimension-2 (1D) $S$-domain wall defects. In Section~\ref{sec:discussion}, we draw the conclusion and have a discussion of our work, followed by an outlook towards future directions. A detailed description of the
unfolding unitary transformation is discussed in Appendix~\ref{sec:unfolding}. Gauging the 3D Ising model in the presence of boundaries is discussed in Appendix~\ref{sec:gauging}. A minimal model that has two $Z$-boundaries and four $X$-boundaries is constructed explicitly in Appendix~\ref{sec:minimal}. 

\section{Backgrounds} \label{sec:review}

In this section, we briefly review the construction of $3$-dimensional toric codes~\cite{kitaev2003fault} and $3$-dimensional topological color codes~\cite{bombin2006topological}. Then we introduce the concept of the Lagrangian subgroup and its application in the classification of gapped boundaries~\cite{kapustin2011topological,levin2013protected}. In the last subsection, we briefly review the unfolding technique to transform the 3D color code to three copies of the 3D toric code without the presence of boundaries~\cite{kubica2015unfolding}.

\subsection{The 3D toric code} \label{sec:3D_toric_code}

Consider a $3$-dimensional square lattice $\mathcal{L}(\mathcal{V}, \mathcal{E}, \mathcal{F}, \mathcal{C})$. $\mathcal{V}$ is the set of vertices $v \in \mathcal{V}$, $\mathcal{E}$ is the set of edges $e \in \mathcal{E}$, $\mathcal{F}$ is the set of faces $f \in \mathcal{F}$, and $\mathcal{C}$ is the set of cells $c \in \mathcal{C}$.

The Hamiltonian of the $3$-dimensional toric code is given by
\begin{align}
    H_{TC} = -\sum_{v} \prod_{v \in \partial e} X(e) - \sum_{f} \prod_{e \in \partial f} Z(e), \label{eq:toric_code}
\end{align}
where $X(e)$ and $Z(e)$ represent for the Pauli $X$- and $Z$-stabilizers applying on the edges labelled by $e$. Pictorially, these two kinds of stabilizers can be drawn as follows:
\begin{align}
        \adjincludegraphics[width=6cm,valign=c]{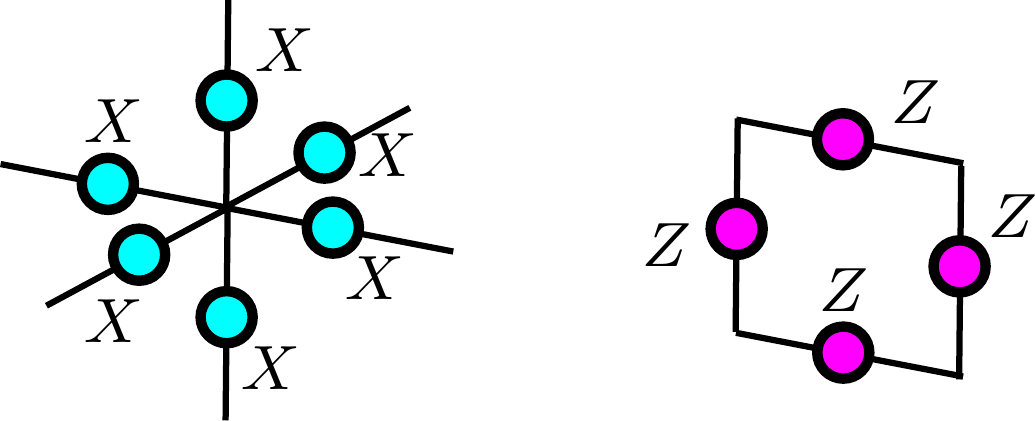}.
\end{align}
One can check that all the Hamiltonian terms are mutually commuting.

In the 3D toric code, there exist two types of topological excitations: electric excitations, referred to as $e$-particles, and magnetic excitations, referred to as $m$-flux loops. These excitations arise from violations of the $X$-stabilizers and $Z$-stabilizers, respectively. We provide examples of these excitations below.
\begin{align}
        \adjincludegraphics[width=6cm,valign=c]{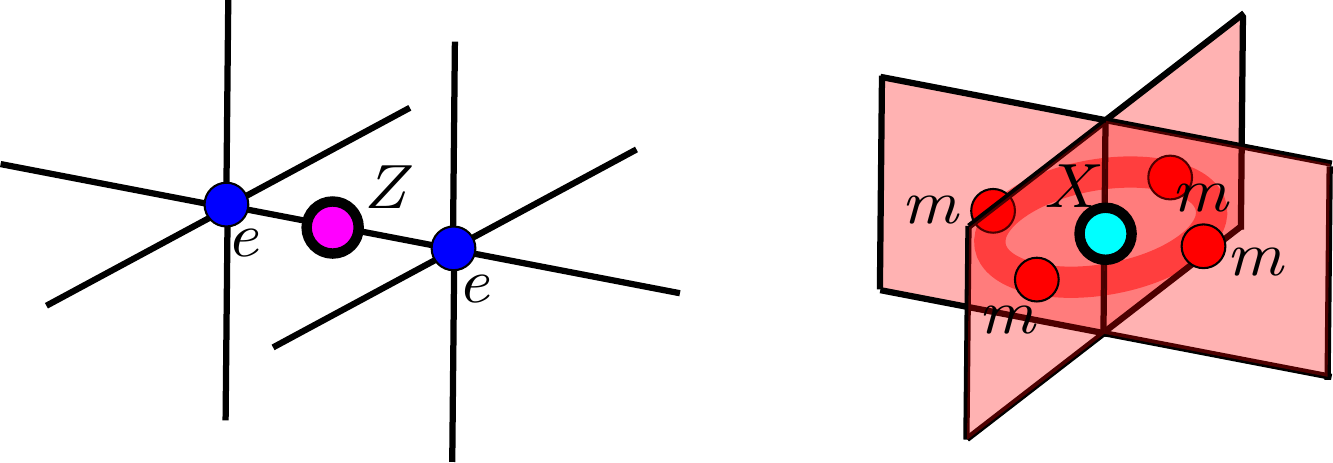}.
\end{align}
Electric charges always appear at the end points of strings of Pauli $Z$-operators and flux loops always appear at the boundaries of Pauli $X$-membranes.

In this paper we consider two types of elementary boundaries of the 3D toric code: the rough (e) boundary and the smooth (m) boundary. $Z$-strings can terminate on the rough boundary without generating any electric excitations at the intersections, while $X$-membranes can terminate on the smooth boundary without generating any flux loops at the intersections. It is equivalent to say that the $e$-particle can {\it condense} on the rough boundary, while the $m$-flux loop can {\it condense} on the smooth boundary. From a mathematical perspective, these condensations form the \textit{Lagrangian subgroup} on the boundary, as detailed in~\cite{levin2013protected}. We have a brief review of the Lagrangian subgroup in Section~\ref{sec:lagrangian_subgroup}.

\begin{figure}
    \centering
    \includegraphics[width = 7cm]{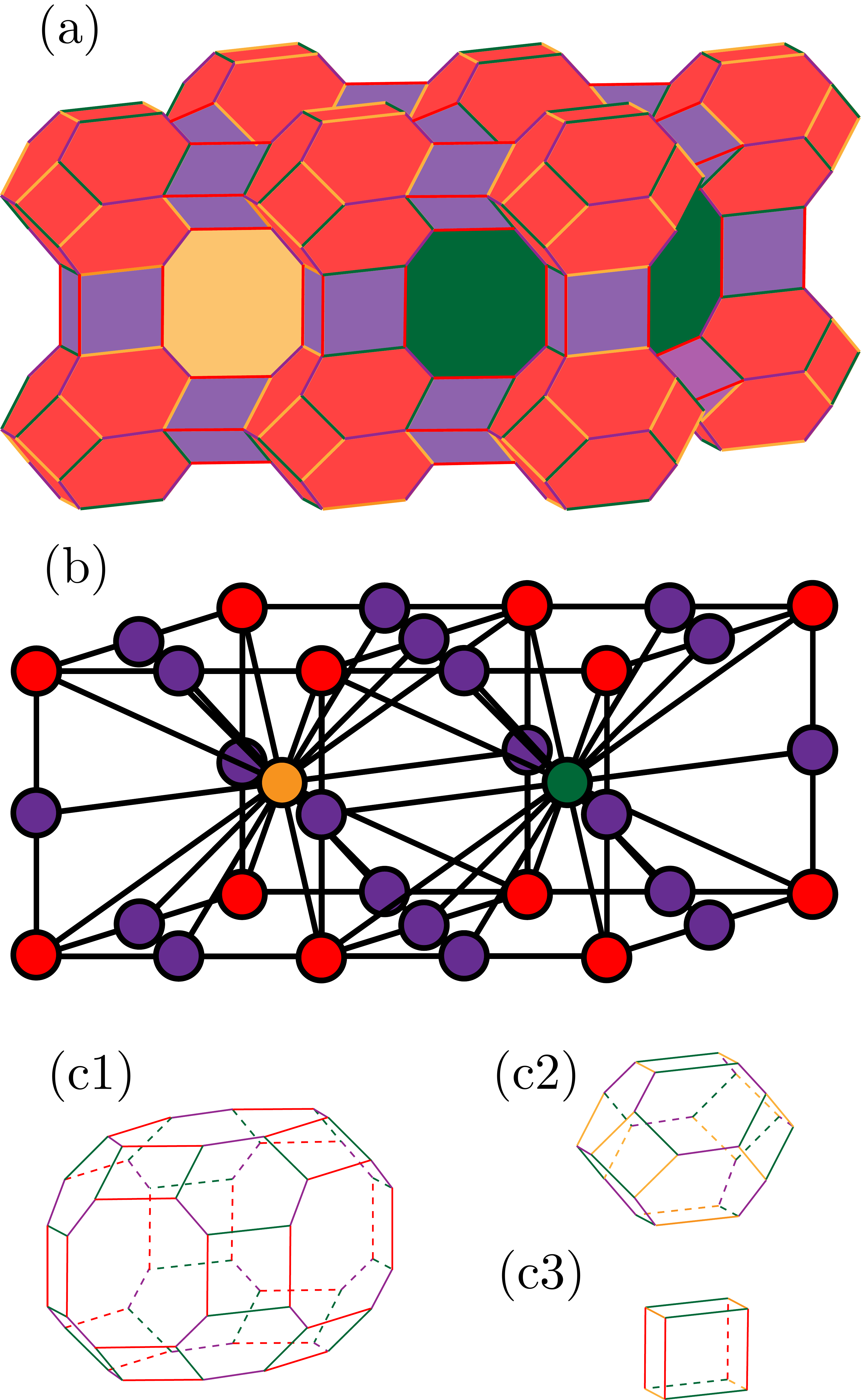}
    \caption{(a) The lattice of the 3D color code, and (b) its dual lattice. We define the color of edges to be the color of the cells they connect. The lower three polyhedrons are (c1) the truncated cuboctahedron, (c2) the truncated octahedron and (c3) the cube, serving as the building blocks of the 3D color code lattice discussed throughout this paper.}
    \label{fig:3DCC}
\end{figure}

\subsection{The 3D color code} \label{sec:3D_color_code}

The $3$-dimensional topological color code is defined on a lattice $\mathcal{L}(\mathcal{V}, \mathcal{E}, \mathcal{F}, \mathcal{C})$ with the following constraints. First, each vertex on the lattice belongs to four different edges. Second, each cells can be assigned one of four colors such that each pair of adjacent cells has different colors.

The Hamiltonian of the $3$-dimensional color code can be written as
\begin{align}
    H_{CC} = -\sum_{c} \prod_{v \subset c} X(v) - \sum_{f} \prod_{v \subset f} Z(v).
\end{align}

In general, one can choose arbitrary 4-colorable and 4-valent cellulations for the $3$D color code. For the sake of convenience, we select a particular cellulation for the three-dimensional topological color code, as depicted in Fig.~\ref{fig:3DCC}. It's important to note that the results we obtain in this paper apply generally and are not contingent on the specific cellulation chosen.

In the lattice shown in Fig.~\ref{fig:3DCC}, we use yellow, green, purple, and red colors, $\mathbf{c}_i \in \{y, g, p, r\} = \mathbf{C}$, to label the color of cells and the edges that connect cells with the same color, and we use two colors to label the plaquettes between two cells with different colors, namely $\{rp, rg, ry, pg, py, yg\}$.

Following this convention, the $X$ and $Z$ stabilizers associated with purple cells, for instance, are depicted as follows\footnote{Without specific statements, we use magenta-colored circles to represent $Z$-operators and cyan-colored circles to represent $X$-operators throughout the paper.}:
\begin{align}
        \adjincludegraphics[width=5cm,valign=c]{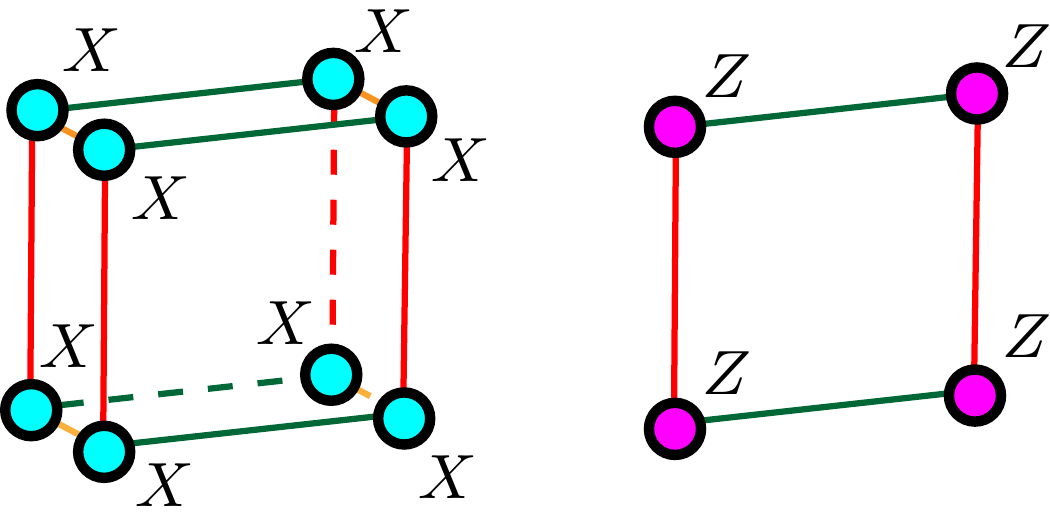}.
\end{align}
Stabilizers for the other cells are defined in a similar manner. It can be verified that the lattice's two constraints ensure all Hamiltonian stabilizers commute with each other.

Similar to the 3D toric code, strings of $Z$-operators generate point-like topological charges (electric charges) at the endpoints. They are denoted by $\{y_{\mathbf{z}}, g_{\mathbf{z}}, p_{\mathbf{z}}, r_{\mathbf{z}}\}$, with the labels representing the colors of their corresponding cells. $Z$-type excitations can be simultaneously created by applying a single $Z$-operator.
\begin{align}        \adjincludegraphics[width=6.5cm,valign=c]{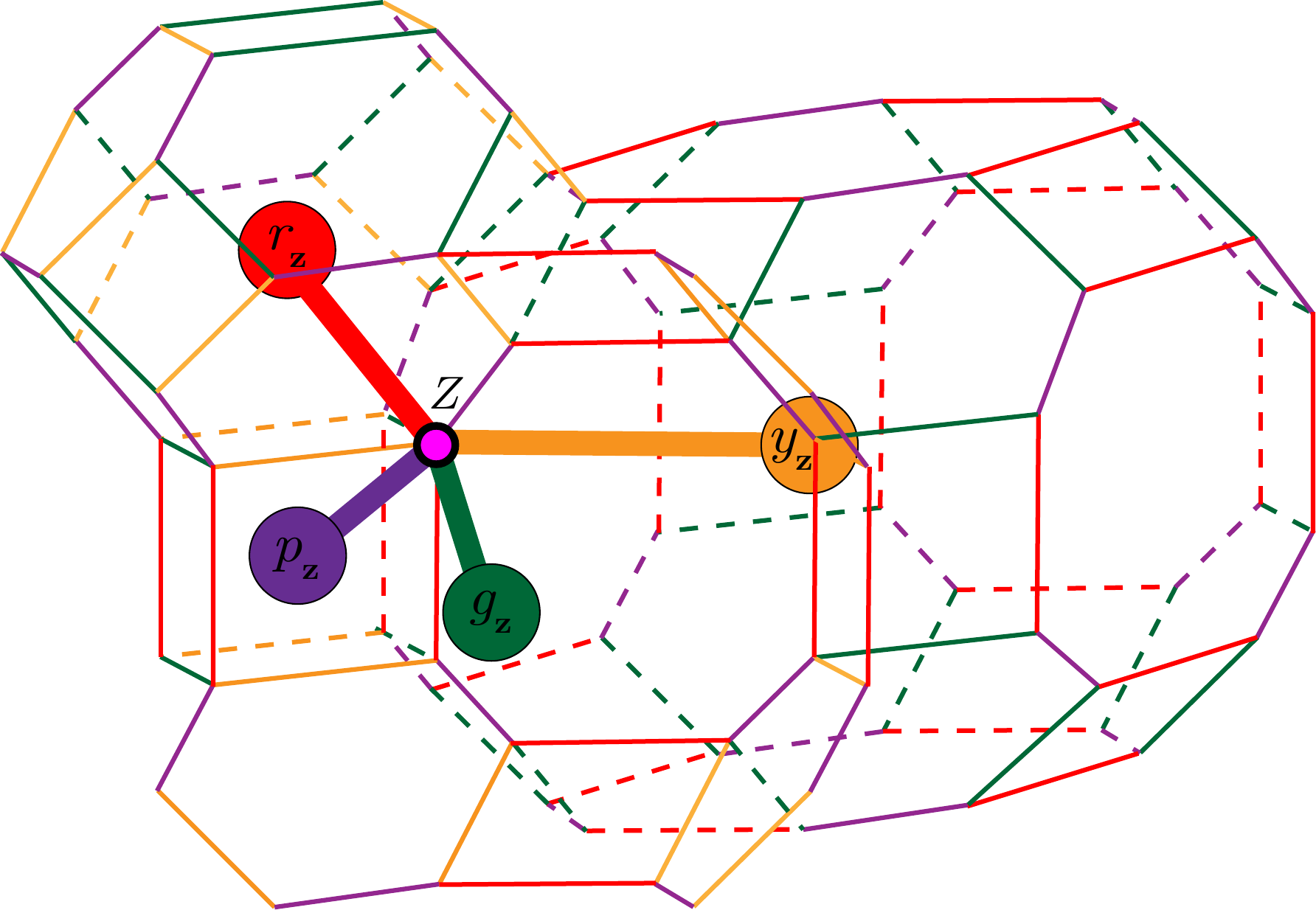}.
\end{align}
Membranes of $X$-operators create flux loops on their boundaries. We label them as $\{rp_{\mathbf{x}}, rg_{\mathbf{x}}, ry_{\mathbf{x}}, pg_{\mathbf{x}}, py_{\mathbf{x}}, yg_{\mathbf{x}}\}$. Here we use two colors to label the plaquette between two cells of different colors. These excitations (flux loops) can be created simultaneously by applying a single $X$-operator.
\begin{align}        \adjincludegraphics[width=6.5cm,valign=c]{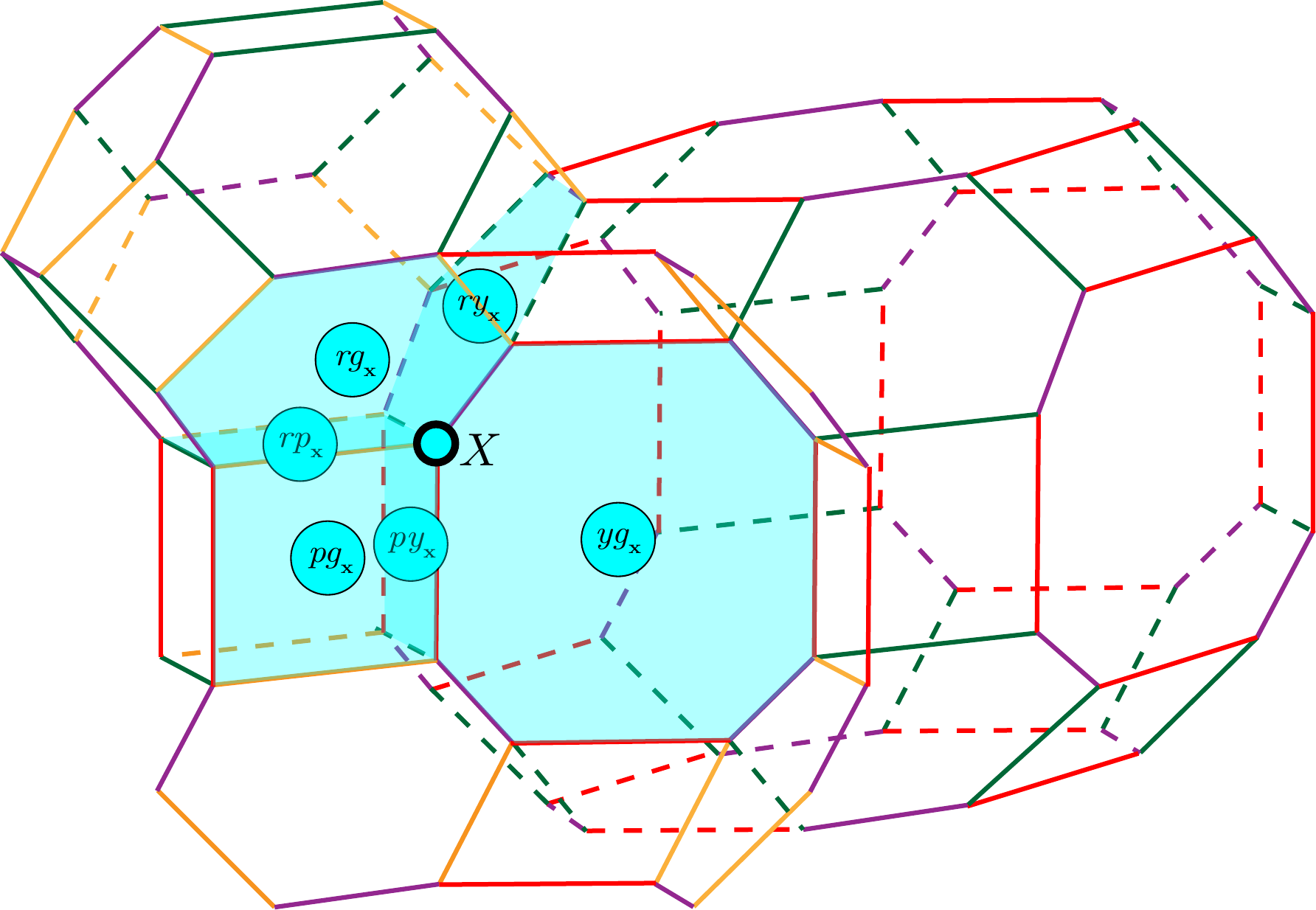}.
\end{align}
To avoid confusion, we use a single color to label these excitations. To create a single type of flux loop, one can apply $X$-operators on the corresponding membrane and the flux loop is located on the boundary of that membrane. This membrane is formed by plaquettes having two particular colors on the edges. For example, the $yg_{\mathbf{x}}$ flux loop is created by the following membrane operator, which is supported on plaquettes bounded by yellow and green edges.
\begin{align}        \adjincludegraphics[width=6.5cm,valign=c]{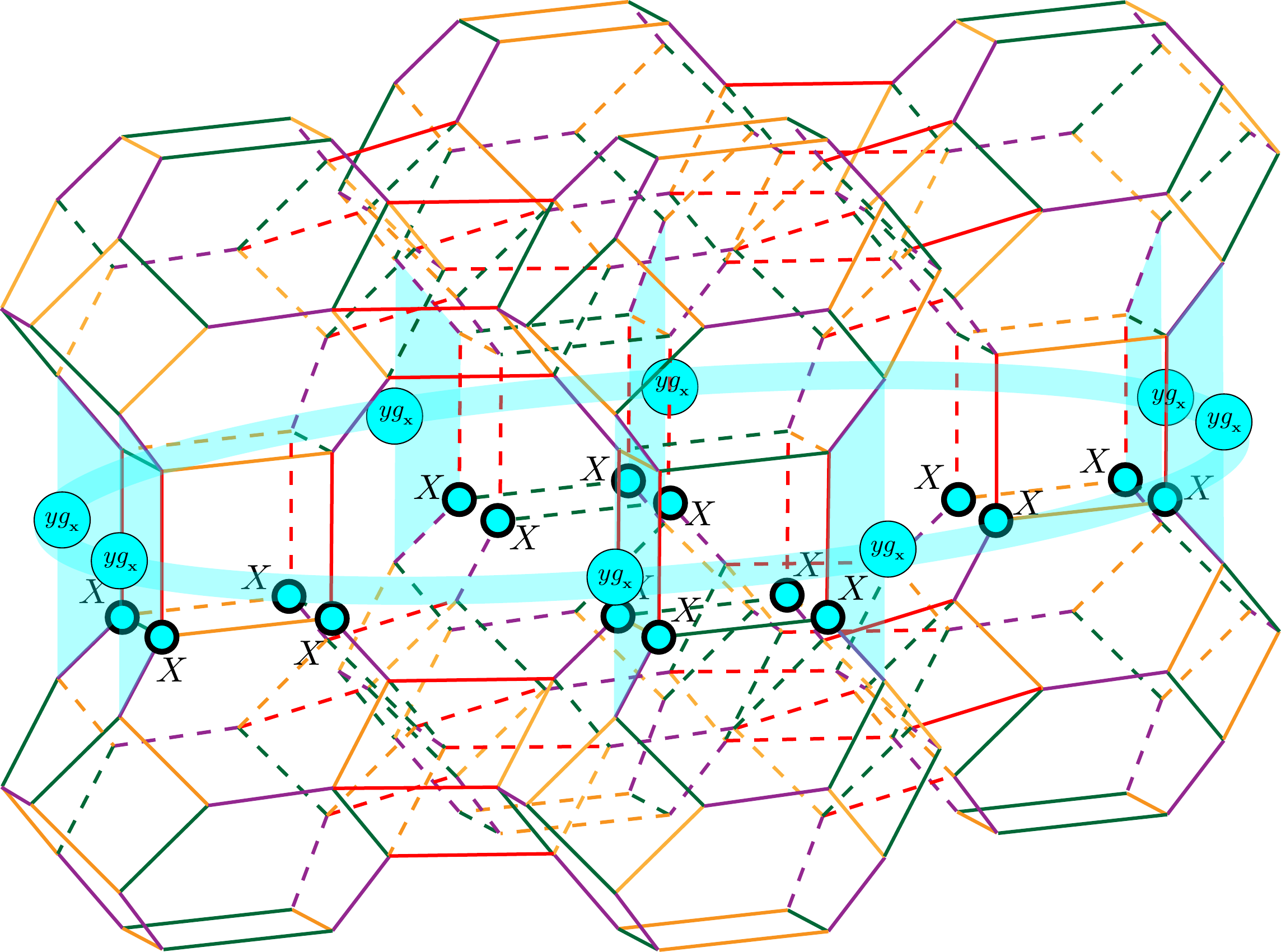}.
\end{align}
A more detailed discussion of the $X$-membrane operators and the flux loops can be found in Section~\ref{sec:paulixboundary}.

There are also excitations created by $Y$-operators, which are the composites of electric charges and magnetic flux loops. 

In the reminder of this subsection, we provide a brief review of the fusion rules of the 3D color code excitations. These excitations are abelian since when two excitations with the same color label and the same Pauli label meet, they annihilate themselves. We have
\begin{align}
    (\mathbf{c}_j \mathbf{c}_k)_{\mathbf{x}} \times (\mathbf{c}_j \mathbf{c}_k)_{\mathbf{x}} = 1, \quad (\mathbf{c}_i)_{\mathbf{z}} \times (\mathbf{c}_i)_{\mathbf{z}} = 1,
\end{align}
in which $\mathbf{c}_i,\mathbf{c}_j,\mathbf{c}_k \in \mathbf{C}$, where $\mathbf{c}_i, \mathbf{c}_j, \mathbf{c}_k$ denote color labels ranging over the set $\mathbf{C}$.

Then let us consider the excitations with the same Pauli label but with different color labels. The fusion rules can be written as,
\begin{align}
    \prod_{j, k} (\mathbf{c}_j  \mathbf{c}_k)_{\mathbf{x}} = \prod_{i} (\mathbf{c}_i)_{\mathbf{z}} =1 \label{eq:color_permutation}
\end{align}
in which we require $j \neq k$ and $j < k$. Eq.~\eqref{eq:color_permutation} indicates that charges with different color labels are not entirely independent; only three among them are independent.

\subsection{Gapped boundaries of topological phases} \label{sec:lagrangian_subgroup}

In the previous subsections, we learned that in the presence of boundaries, some excitations from the bulk of topological phases can condense on these boundaries, while others remain gapped. From the perspective of operators, this is equivalent to saying that the boundaries of certain excitation creation operators can terminate on such boundaries, resulting in no excitations being created at the intersections where the membrane meets the boundary. The set of excitations that can condense on the boundary forms a group, referred to as the {\it Lagrangian subgroup} in finite abelian theories, as discussed in~\cite{kapustin2011topological,levin2013protected}.

The Lagrangian subgroup are defined as follows. Consider the set of topological excitations in the bulk as $\mathsf{A}$. The condensation set $\mathsf{M}$ on the boundary satisfies:
\begin{itemize}
    \item $\mathsf{M}$ is a subset of $\mathsf{A}$.
    \item The excitations in $\mathsf{M}$ have trivial mutual and self statistics.
    \item Any excitations that is not in $\mathsf{M}$ has non-trivial mutual statistics with respect to at least one excitation in $\mathsf{M}$.
\end{itemize}
The set $\mathsf{M}$ that satisfies the above conditions are call a {\it Lagrangian subgroup} of $\mathsf{A}$. 

Using the boundaries of the 3D toric code as an example, we identify the topological excitations in the bulk as $\mathbb{I}$ (vacuum), $e$ (electric charge), $m$ (magnetic flux), and $em$ (composite of $e$ and $m$). There are two different choices for gapped boundaries, specifically, $(\mathbb{I}, e)$ and $(\mathbb{I}, m)$. It can be verified that both sets of condensation meet the previously mentioned conditions. Although $\mathbb{I}$ and $em$ also form a group, $em$ cannot be part of the Lagrangian subgroup due to its non-trivial self statistics.

In the abelian cases, the elements in the Lagrangian subgroup are always irreducible representations of the groups (topological excitations). Therefore, we can always use the generators to represent the condensation set. In the cases of $(\mathbb{I}, e)$ and $(\mathbb{I}, m)$ boundaries, we can simplify the notation and denote them as the $e$ and $m$ boundaries. It can be generalized to the cases when we have several copies of 3D toric codes. For example, consider there are two copies of 3D toric code with $e$-boundaries, the Lagrangian subgroup is $(\mathbb{I}, e_1, e_2, e_1 e_2)$. However, by using the generating set, it simplifies to $(e_1, e_2)$. We use the generating set of condensations to label boundaries throughout this paper unless specified. In the reminder of the paper, we refer to the boundaries that can be classified by the Lagrangian subgroup as the {\it elementary boundaries}.

\subsection{Unfolding the 3D Color Code without Boundaries} \label{sec:unfold_without_boundary}

In this subsection, we provide an overview of the unfolding technique introduced in Ref.~\cite{kubica2015unfolding}, demonstrating that a $D$-dimensional color code is equivalent to $D$ copies of toric codes, modulo local unitaries and entangling/disentangling ancilla qubits. Our discussion primarily concentrates on $3$-dimensional models, in alignment with the focus of this paper. It should be noted that these results are applicable in general $D$ dimension.

Let $CC(\mathcal{L})$ be the stabilizer group of the topological color code defined on a $3$-dimensional lattice $\mathcal{L}$, there exists an  unfolding (disentangling) unitary $U$ such that the color code can be transformed to $3$ copies of $3$-dimensional toric codes $TC(\mathcal{L}'_j)$ defined on lattice $\mathcal{L}'_{j}$, up to stabilizer group $\mathcal{S}$ of the ancilla qubits, which is
    \begin{align}
        U \left[ CC(\mathcal{L}) \otimes \mathcal{S} \right] U^{\dagger} = \otimes_{i=1}^{3} TC \left(\mathcal{L}'_{j}\right). \label{eq:equivalence}
    \end{align}
The mapping between lattice $\mathcal{L}$ and $\mathcal{L}'_j$ can be realized by the {\it shrinking} process (lattice deformation), and the unfolding unitary transformation can be realized by applying local unitaries, shrinking the lattice and disentangling qubits. We briefly summarize this process in the following paragraphs.

Consider the lattice $\mathcal{L}$ associated with a $3$-dimensional topological color code, as we defined in Section~\ref{sec:review}. The unfolding process of color $\mathbf{c}_i, \mathbf{c}_j, \mathbf{c}_k \in \mathbf{C}$, $i \neq j \neq k$ on the lattice and the map between stabilizers can be briefly summarized as follows. First, prepare one ancilla qubit on each edge colored $\mathbf{c}_i$, $\mathbf{c}_j$, or $\mathbf{c}_k$, where $\mathbf{c}_i, \mathbf{c}_j, \mathbf{c}_k \in \mathbf{C}$. Second, apply a local unitary operator $U = \bigotimes_c U_c$ on every cell with color $\mathbf{c}_i, \mathbf{c}_j, \mathbf{c}_k \in \mathbf{C}$ such that the original stabilizers in the 3D color code are mapped to three sets of stabilizers, each supported on one of the three corresponding sets of qubits on edges with different colors. Third, shrink all the $\mathbf{c}_i$-, $\mathbf{c}_j$-, and $\mathbf{c}_k$-colored cells sequentially to points. After this step, we obtain three copies of lattices with different colors.  Fourth, disentangle and discard the qubits on vertices. The outcome of the above process yields three copies of 3D toric codes model defined on the $\mathbf{c}_i$-, $\mathbf{c}_j$-, and $\mathbf{c}_k$-colored lattices, respectively. In a broader sense, for a $D$-dimensional color code, selecting and shrinking $D$ distinct colors can lead to $D$ copies of the $D$-dimensional toric codes.

More concretely, consider the 3D color code defined on the lattice illustrated in Fig.~\ref{fig:3DCC}, we briefly summarize the result after the unfolding unitaries when we select to shrink the green, yellow and purple cells.
\begin{itemize}
    \item By shrinking the green (yellow) cells, the lattice transforms into the following configuration.
    \begin{align*}
        \adjincludegraphics[width=6.5cm,valign=c]{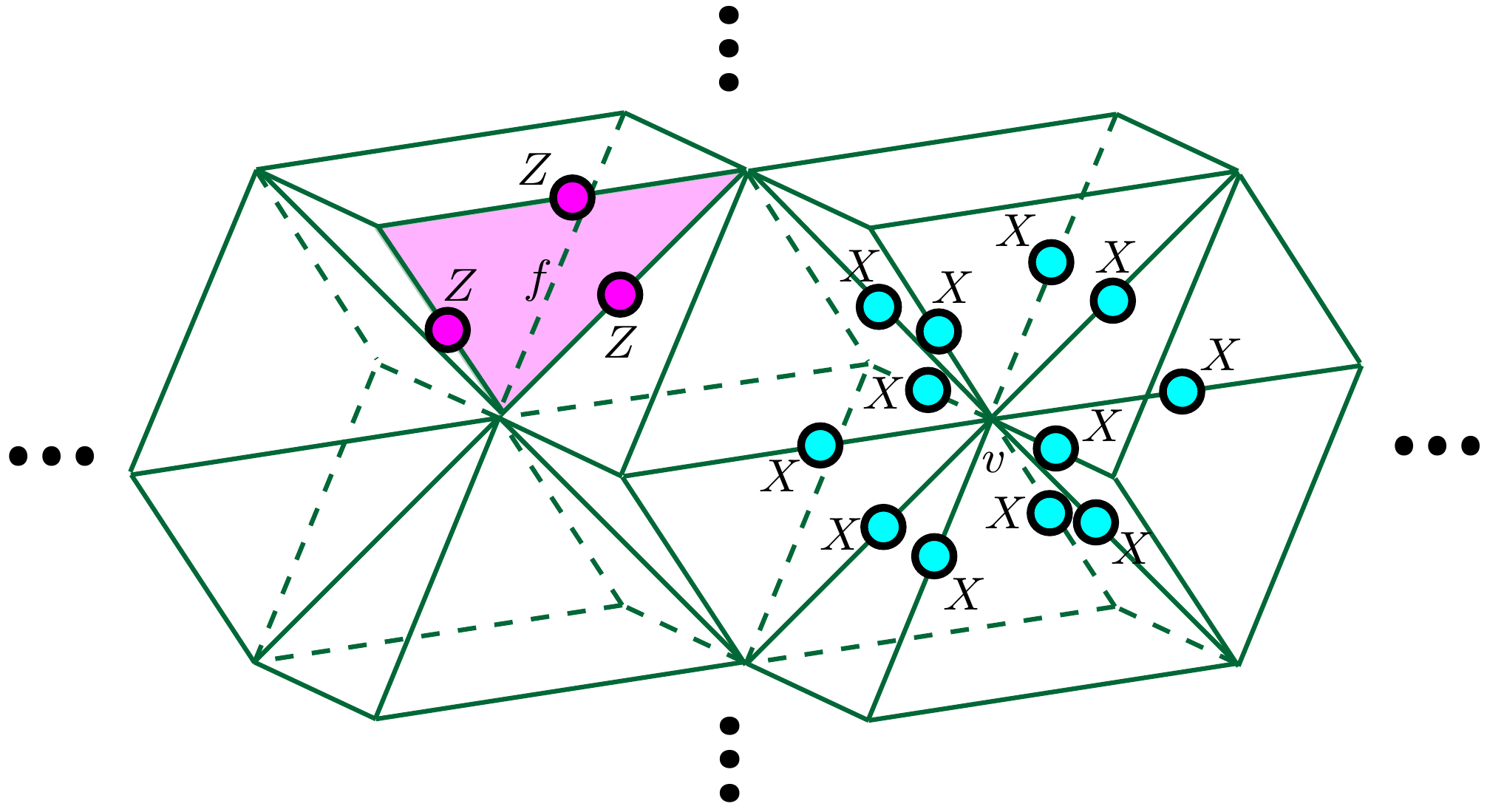}
    \end{align*}
    The ``..." in the figure indicates the lattice is extended along the horizontal and vertical directions. The $X$-type stabilizers are located at vertices, and each $X$-stabilizer consists of the tensor product of $X$-operators on the nearest edges. The $Z$-type stabilizers are given by the tensor product of $Z$-operators along the edges of faces.\footnote{Acute readers may notice that based on the shrinking rule, there are four qubits on each edge of the green (yellow) lattice. However, we can further apply local unitaries to decouple these degrees of freedom and leave only one qubit per edge, which are those of the 3D toric code. Details of this process is discussed in Appendix~\ref{sec:unfolding}}
    \item By shrinking the purple cells, the lattice and stabilizers transform into the following configuration.
    \begin{align*}
        \adjincludegraphics[width=6.5cm,valign=c]{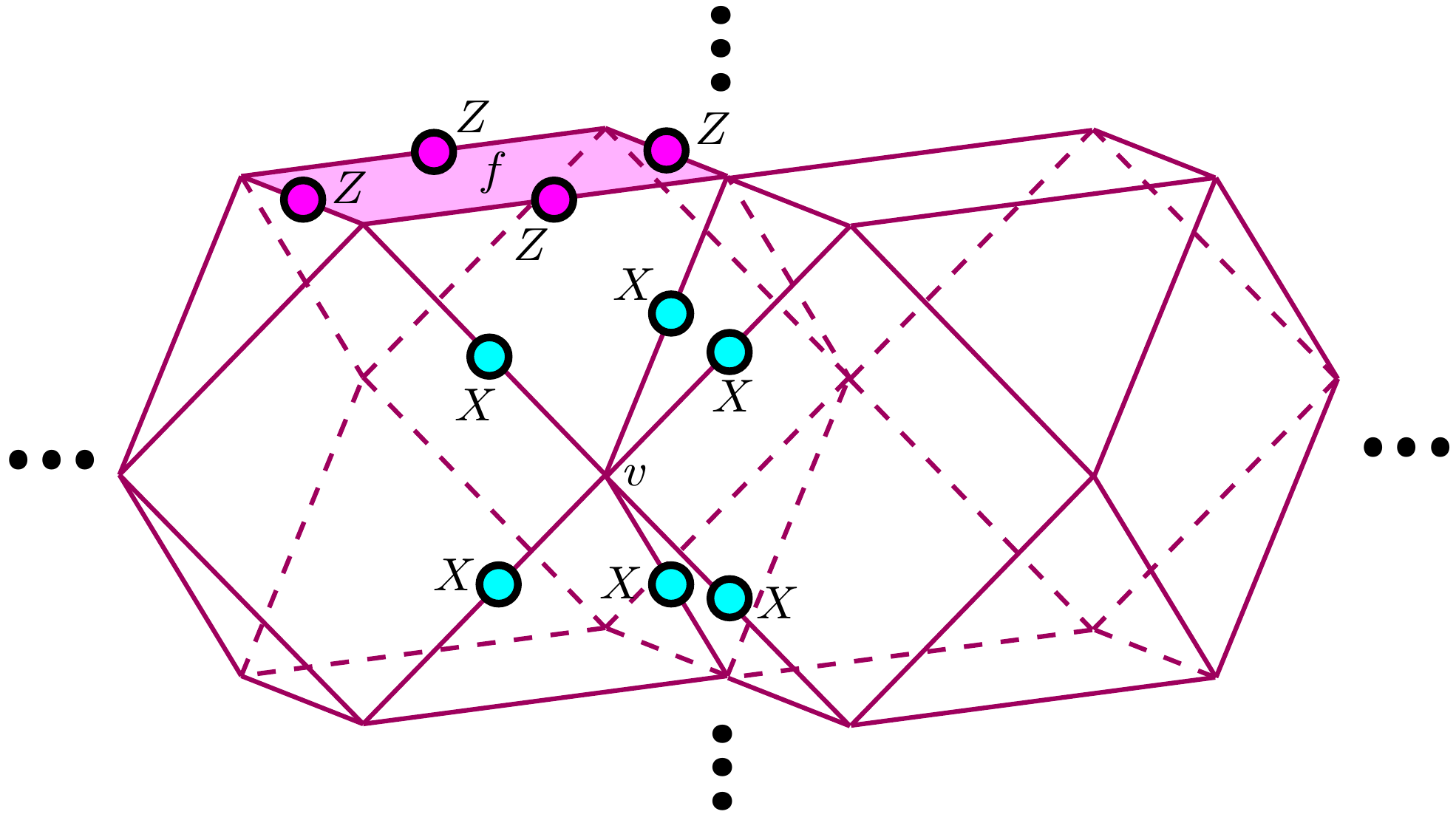}
    \end{align*}
\end{itemize}
An example of the transformations of stabilizers and lattice deformation on the purple lattice is summarized in Fig.~\ref{fig:mapping_1}, and a detailed description of the unfolding unitaries is discussed in Appendix~\ref{sec:unfolding}.

\begin{figure*}
    \centering
    \includegraphics[width = 12cm]{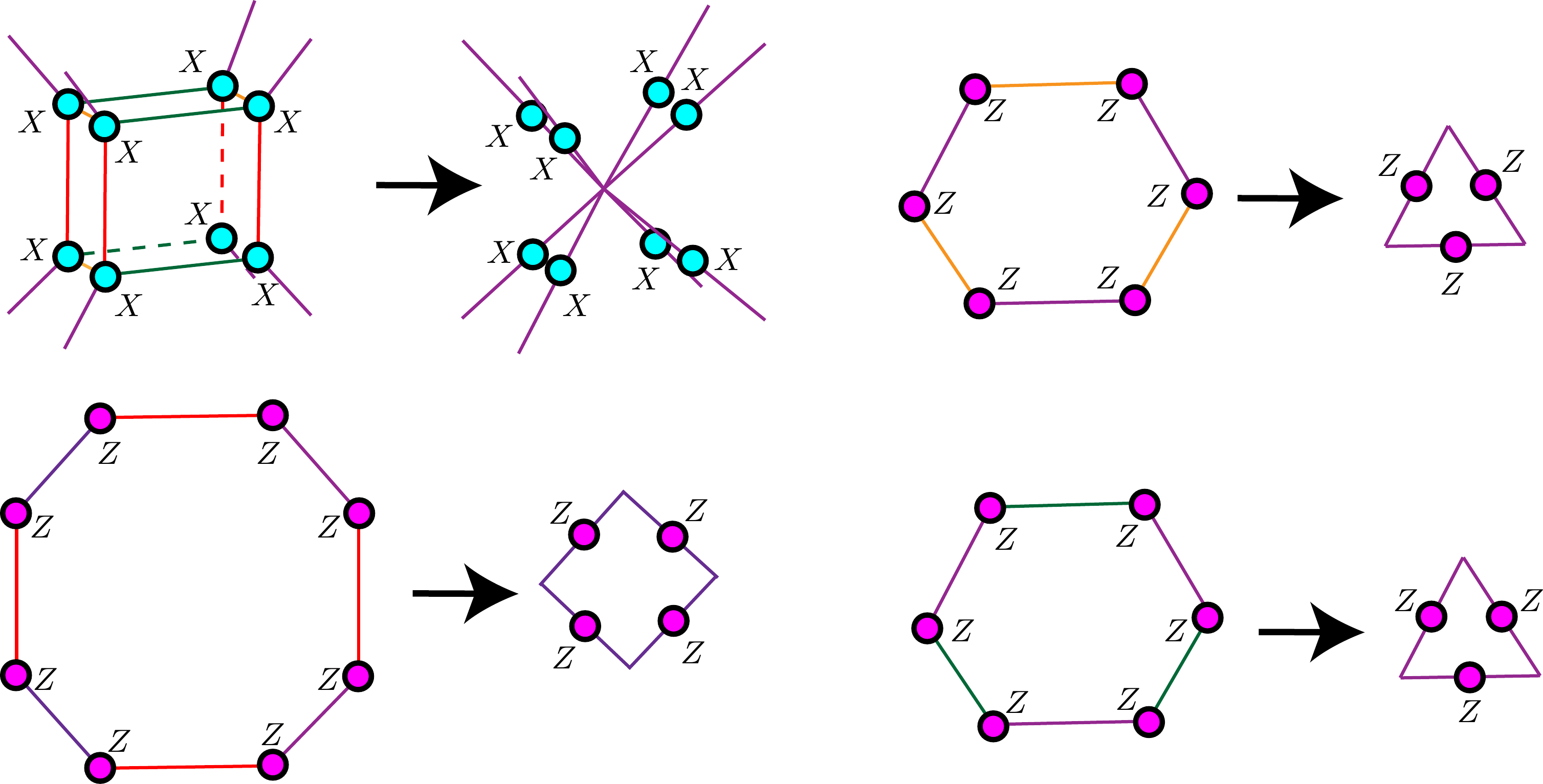}
    \caption{An example of the unfolding unitary transformation. The stabilizers of the 3D color code are on the left-hand side of the arrows, while the stabilizers of the 3D toric code are on the right-hand side. Here we choose the unfolding process for the purple cells as an illustration. However, this map works for every color and every color-code lattice. The map goes as follows: (1) Map every $X$-operator to their nearest purple edges. (2) Map two $Z$-operators at the vertices of a purple edge to a single $Z$-operator on the same edge. (3) Shrink the purple cells to points. The remaining lattice model is a 3D toric code defined on the purple lattice.}
    \label{fig:mapping_1}
\end{figure*}

\section{Unfolding the 3D color code with boundaries} \label{sec:boundaries}

\begin{figure*}
    \centering
    \includegraphics[width = 16cm]{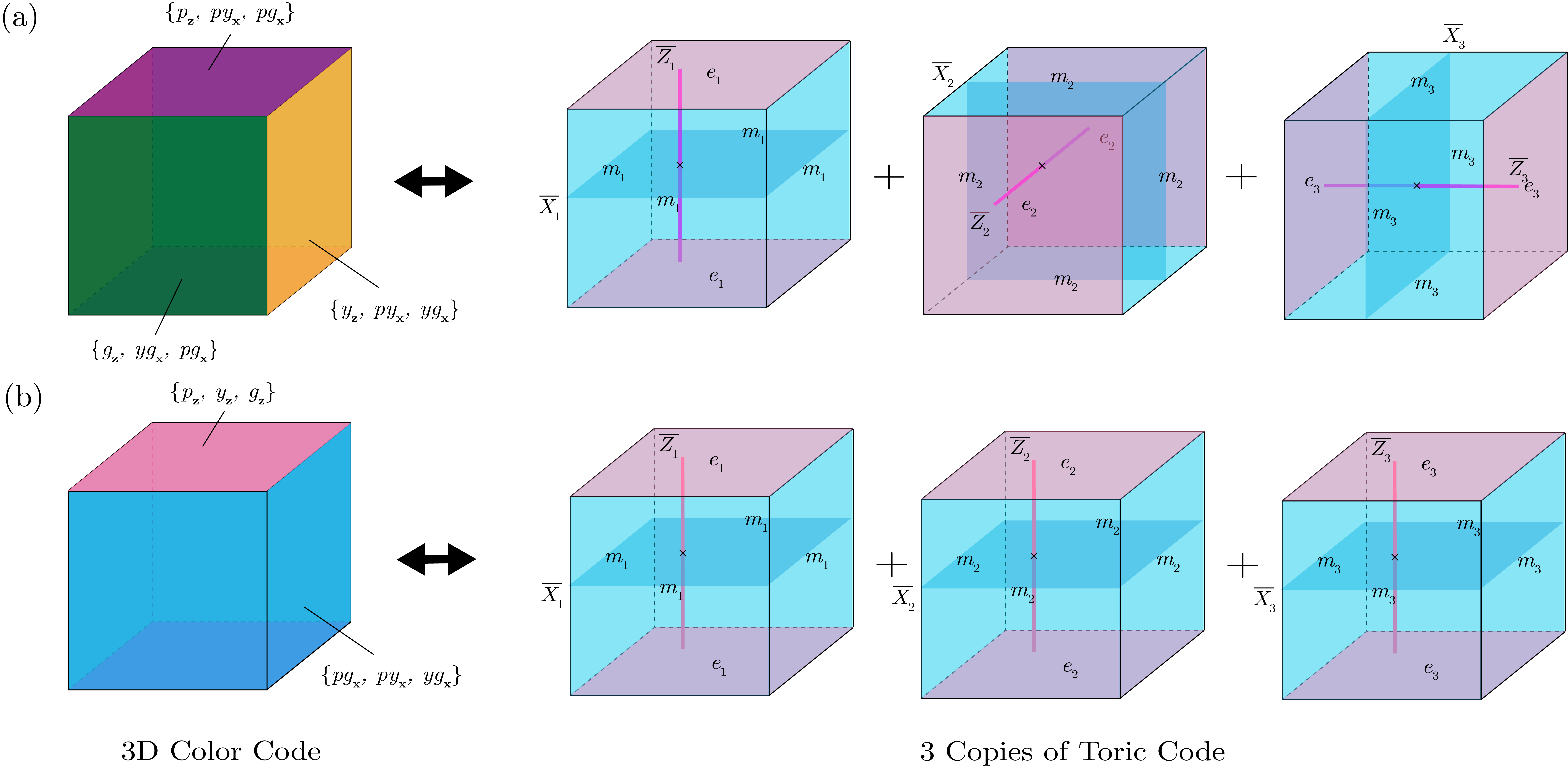}
    \caption{An illustration of the equivalence between the 3D color code with different types of boundaries and three copies of toric codes with different alignments. (a) A 3D color code with color boundaries (each pair of boundaries can condense a distinct type of ($\mathbf{c}_i)_{\mathbf{z}}$-particles.) is equivalent to three copies of toric codes with distinct alignments of boundaries. Therefore, each color boundary corresponds to a $(e_i, m_j, m_k)$-boundary of three copies of 3D toric codes. (b) A 3D color code with two $Z$-boundaries and four $X$-boundaries is equivalent to three copies of toric codes with the same alignment of boundaries. Therefore, the $Z$-boundaries correspond to the $(e_1, e_2, e_3)$-boundaries and the $X$-boundaries correspond to the $(m_1, m_2, m_3)$-boundaries. A miminal model that can realize this boundary condition is explicitly constructed in Appendix~\ref{sec:minimal}}
    \label{fig:magic_boundary_2}
\end{figure*}

In Ref.~\cite{kubica2015unfolding}, the authors show a correspondence between the excitations of two copies of the 2D toric codes and those of the 2D color code. Analogously, we establish a mapping between the excitations of three copies of the toric codes and the 3D color code.\footnote{The map is equivalent under permutations of colors.}

\begin{align}
    y_{\mathbf{z}} \to e_1,\quad pg_{\mathbf{x}} \to m_1, \notag \\
    g_{\mathbf{z}} \to e_2,\quad py_{\mathbf{x}} \to m_2, \label{eq:condensations}\\ 
    p_{\mathbf{z}} \to e_3,\quad yg_{\mathbf{x}} \to m_3. \notag
\end{align}

For the convenience of our discussion, we define $\mathbf{c}_1 = y$, $\mathbf{c}_2 = g$ and $\mathbf{c}_3 = p$. In the remainder of this section, we will demonstrate, following the original proof in Ref.~\cite{kubica2015unfolding}, that the $(e_1, e_2, e_3)$ (all-rough) boundary of three copies of toric codes is equivalent to the $\{y_{\mathbf{z}}, g_{\mathbf{z}}, p_{\mathbf{z}}\}$-boundary ($Z$-boundary) of the 3D color code, and that the $(m_1, m_2, m_3)$ (all-smooth) boundary is equivalent to the $\{pg_{\mathbf{x}}, py_{\mathbf{x}}, yg_{\mathbf{x}}\}$-boundary ($X$-boundary), up to local unitaries and entangling/disentangling ancilla qubits. We illustrate the equivalence between the 3D color code with two $Z$-boundaries and four $X$-boundaries, and three copies of the 3D toric codes with the same boundary alignment in Fig.~\ref{fig:magic_boundary_2}(b). A minimal model capable of realizing this boundary condition is explicitly constructed in Appendix~\ref{sec:minimal}. Moreover, we explicitly construct the boundary Hamiltonians for the other boundaries that can be classified by the Lagrangian subgroup.

\subsection{The $Z$-boundary} \label{sec:paulizboundary}

\begin{definition}
    Consider a 3D color code with colors $\mathbf{C} = \{y, g, p, r\}$. The color of the edges, denoted by $\mathbf{c}_i \in \mathbf{C}$, is defined as the color of the cells they connect. Let us consider a surface that intersects only with edges of a specific color. On one side of this surface, all the degrees of freedom are removed. On the other side, the $X$-stabilizers and  $Z$-stabilizers that intersect with the surface are removed and truncated respectively, while all the other stabilizers remain the same.
    This boundary is referred to as the $Z$-boundary. \label{def:Paulizboundary}
\end{definition}
\begin{theorem}
Consider the surface in Definition~\ref{def:Paulizboundary} intersects with color $\mathbf{c}_4 \in \mathbf{C}$. The condensations on the corresponding $Z$-boundary are given by $\{(\mathbf{c}_1)_{\mathbf{z}}, (\mathbf{c}_2)_{\mathbf{z}}, (\mathbf{c}_3)_{\mathbf{z}}\}$, where $\mathbf{c}_1, \mathbf{c}_2, \mathbf{c}_3, \mathbf{c}_4 \in \mathbf{C}$ and $\mathbf{c}_1 \neq \mathbf{c}_2 \neq \mathbf{c}_3 \neq \mathbf{c}_4$. This boundary is equivalent to the $(e_1, e_2, e_3)$ boundary of three copies of 3D toric codes up to local unitaries and entangling/disentangling ancilla qubits.
\end{theorem}

\begin{proof}
Without loss of generality, we choose $\mathbf{c}_4$ to be the red color, and we use the lattice of 3D color code depicted in Fig.~\ref{fig:3DCC}. To prove the theorem, we first show that the condensations on this boundary are $\{y_{\mathbf{z}}, g_{\mathbf{z}}, p_{\mathbf{z}}\}$. Then we illustrate that, through the application of local unitaries, this boundary is equivalent to the $(e_1, e_2, e_3)$-boundary in three copies of 3D toric codes. This proof is applicable in different colorizations and cell structures.

The lattice of the $Z$-boundary is shown below. Qubits are assigned on the vertices of the lattice.
\begin{align} \adjincludegraphics[width=6cm,valign=c]{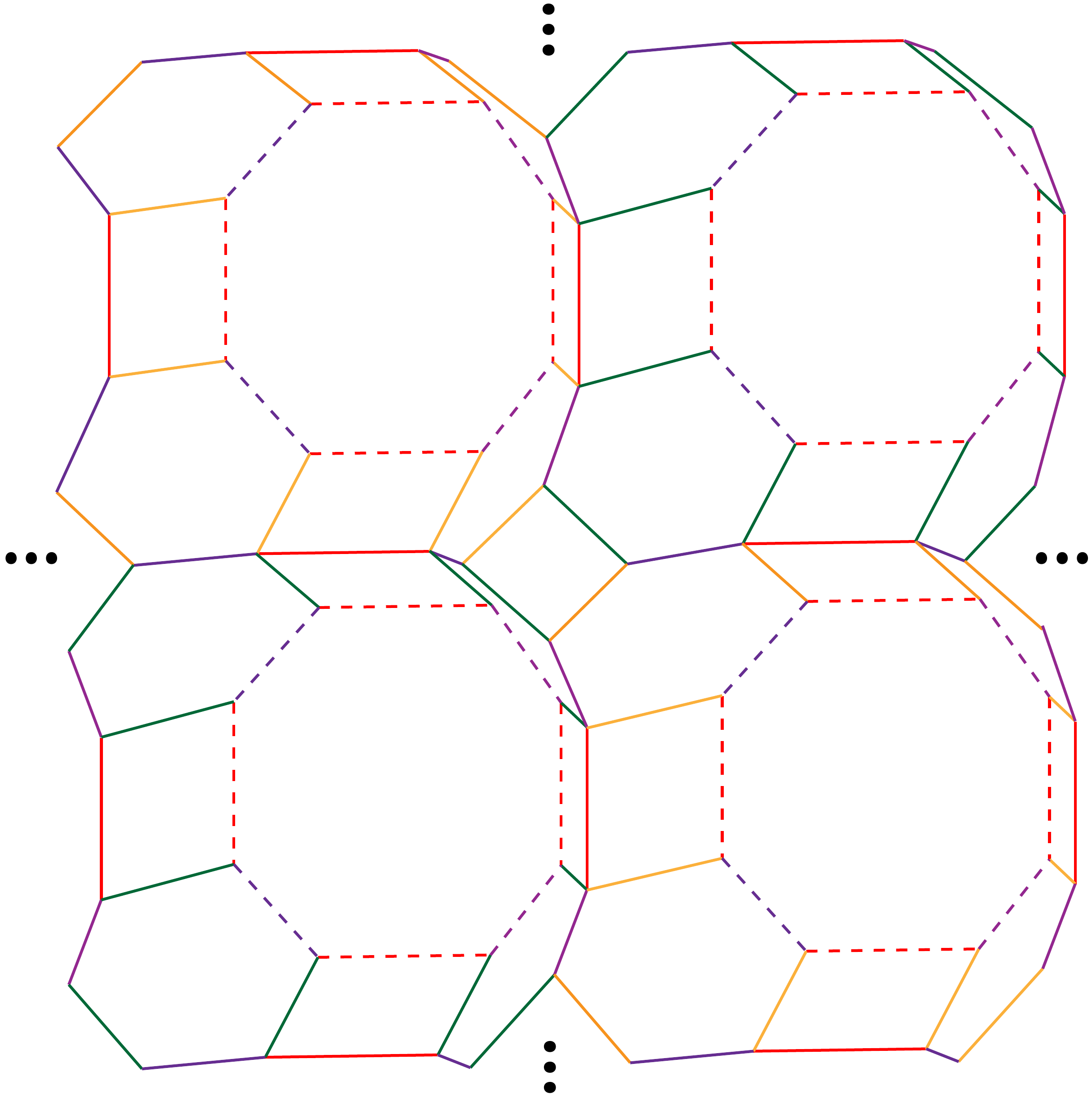} \label{eq:bdry_12}
\end{align}

In addition to the $X$-stabilizers on the cells and $Z$-stabilizers on the plaquettes, there are also truncated $Z$-stabilizers. We define these three types of stabilizers to be $\mathcal{B}^{boundary}_{yg}(Z)$, $\mathcal{B}^{boundary}_{pg}(Z)$, and $\mathcal{B}^{boundary}_{py}(Z)$, respectively and display them below. 
\begin{align*}
    \adjincludegraphics[width=7.5cm,valign=c]{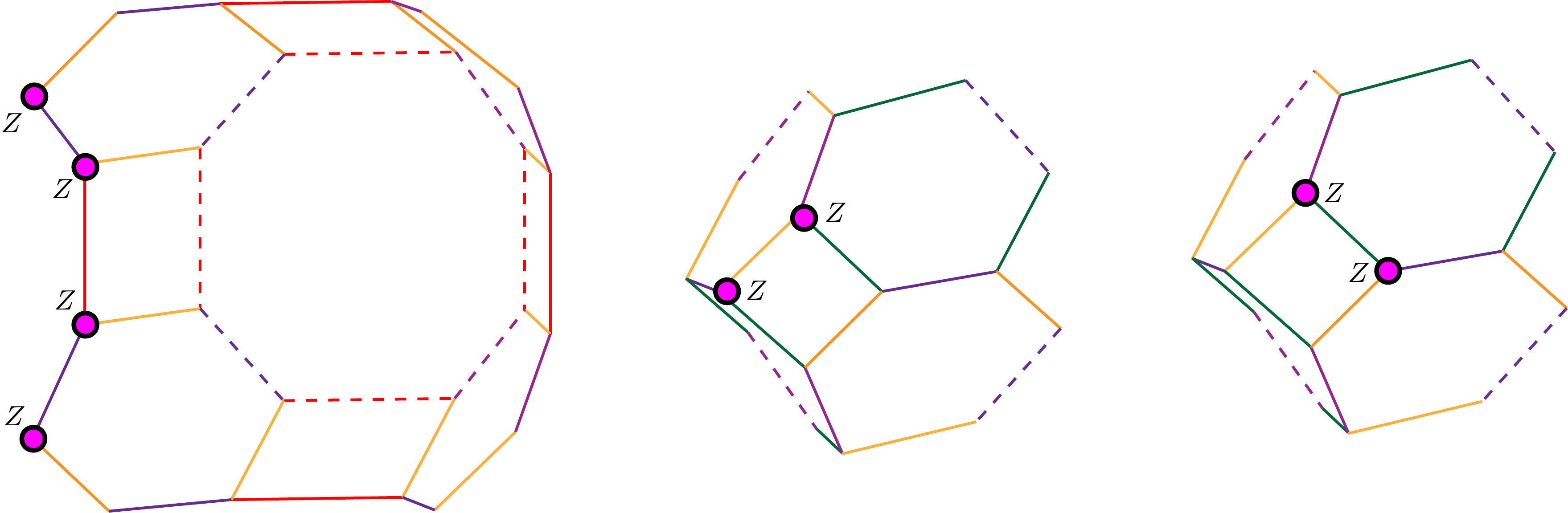}.
\end{align*}
There are no $X$-stabilizers on the truncated cells.

When a pair of $Z$-operators on the same edge with color $\mathbf{c}_i \in \mathbf{C}$ is applied in the bulk, a pair of $(\mathbf{c}_i)_{\mathbf{z}}$ excitations on the connected $\mathbf{c}_i$-colored cells are created. This is due to the fact that a pair of $Z$-operators applying on a $\mathbf{c}_i$-colored edge commutes with every stabilizer term except for the two $X$-stabilizers on the connected $\mathbf{c}_i$-colored cells. 

However, in a 3D color code with two $Z$-boundaries, there exist logical strings of $Z$-operators that terminate on both boundaries and commute with all the stabilizers. Unlike the 3D toric code, the 3D color code has three inequivalent types of logical string operators between two $Z$-boundaries. In the case of a 3D color code with $Z$-boundaries at the top and the bottom, as illustrated in subsequent figures, the string operator for $y_{\mathbf{z}}$ excitations terminating on a $Z$-boundary is given by
\begin{align*}
    \adjincludegraphics[width=5cm,valign=c]{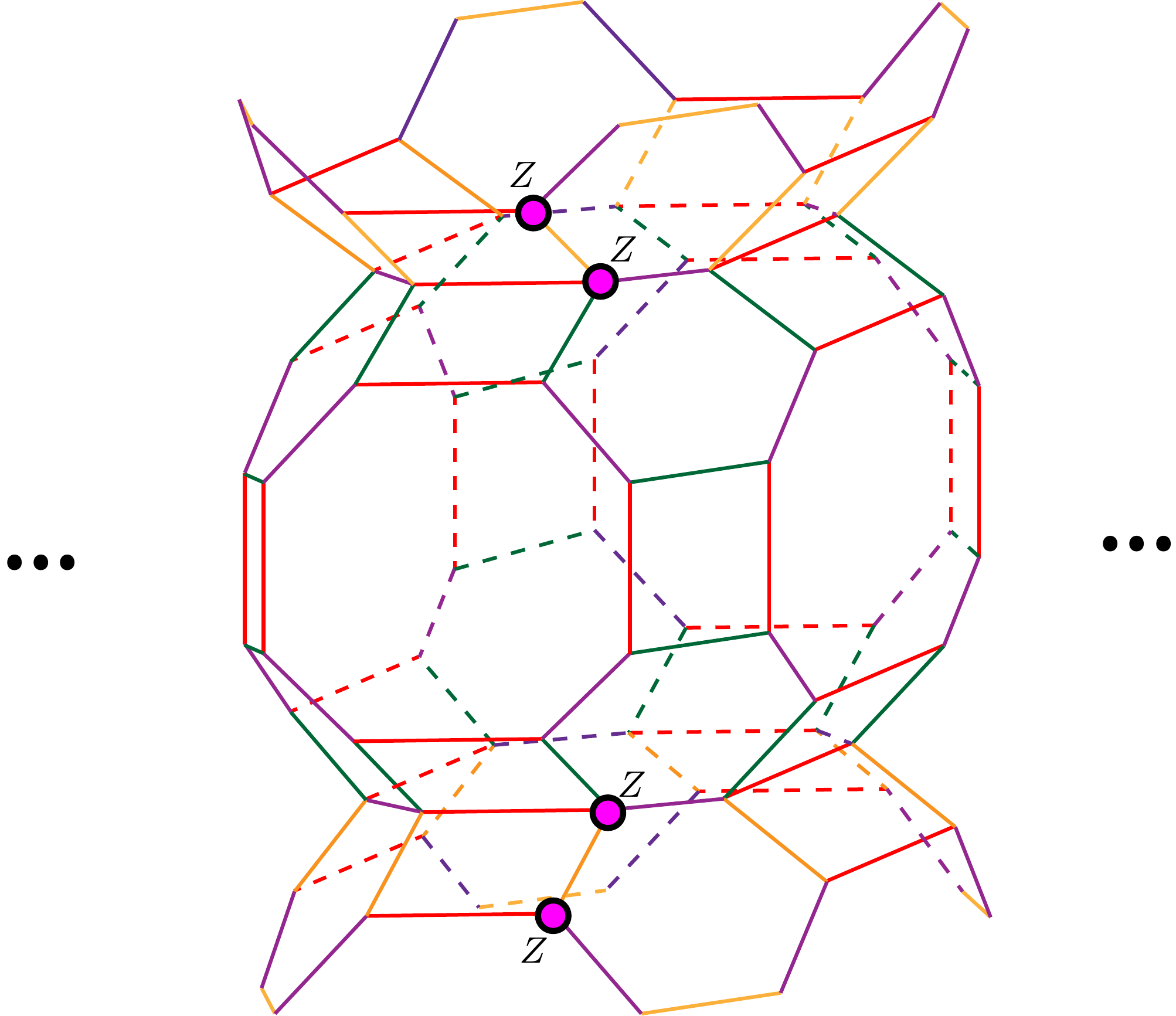}.
\end{align*}
The string operator of $g_{\mathbf{z}}$ excitations that terminates on the $Z$-boundary is given by
\begin{align*}
    \adjincludegraphics[width=5cm,valign=c]{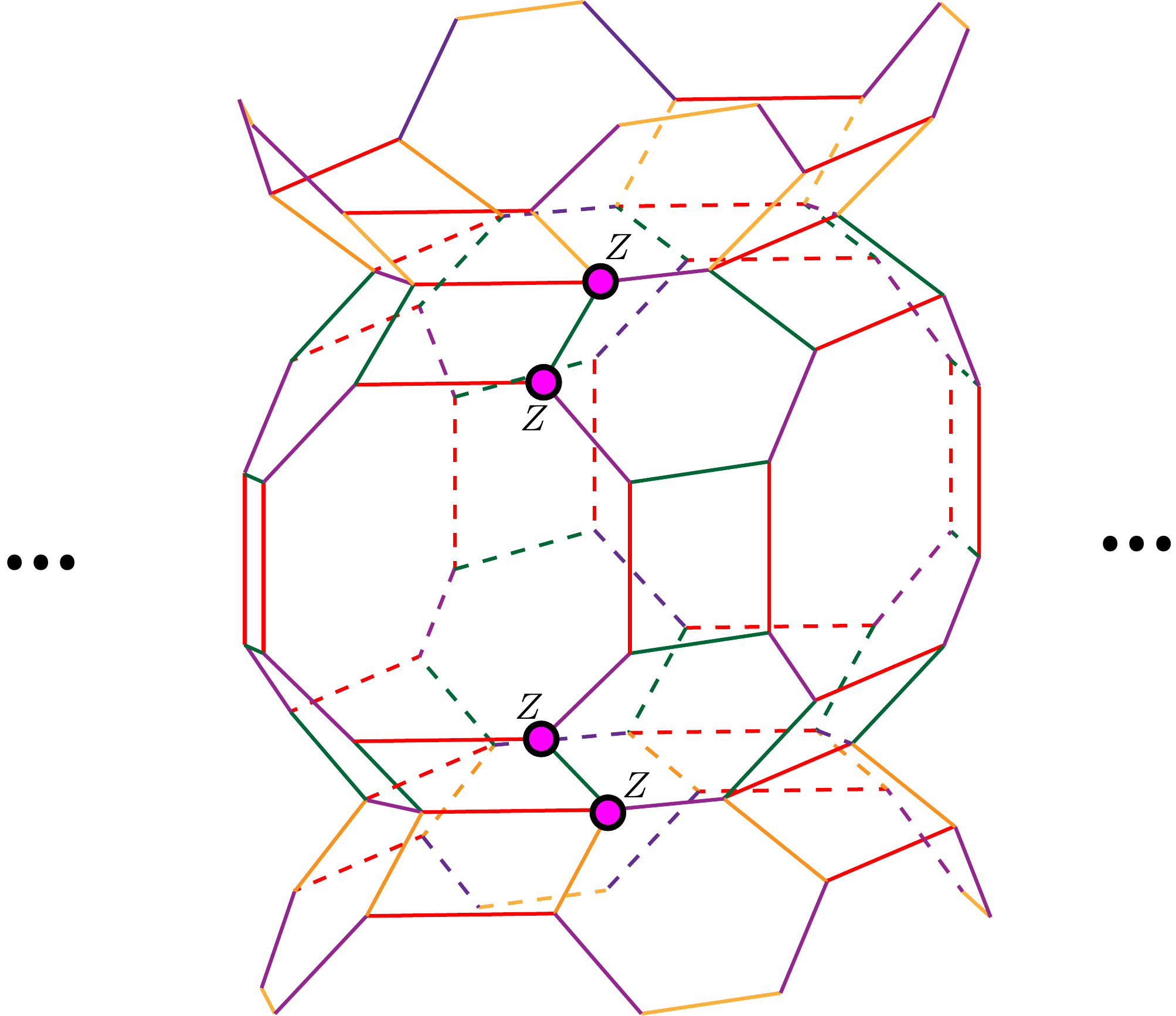}
\end{align*}
The string operator of $p_{\mathbf{z}}$ excitations that terminates on the $Z$-boundary is given by
\begin{align*}
    \adjincludegraphics[width=5cm,valign=c]{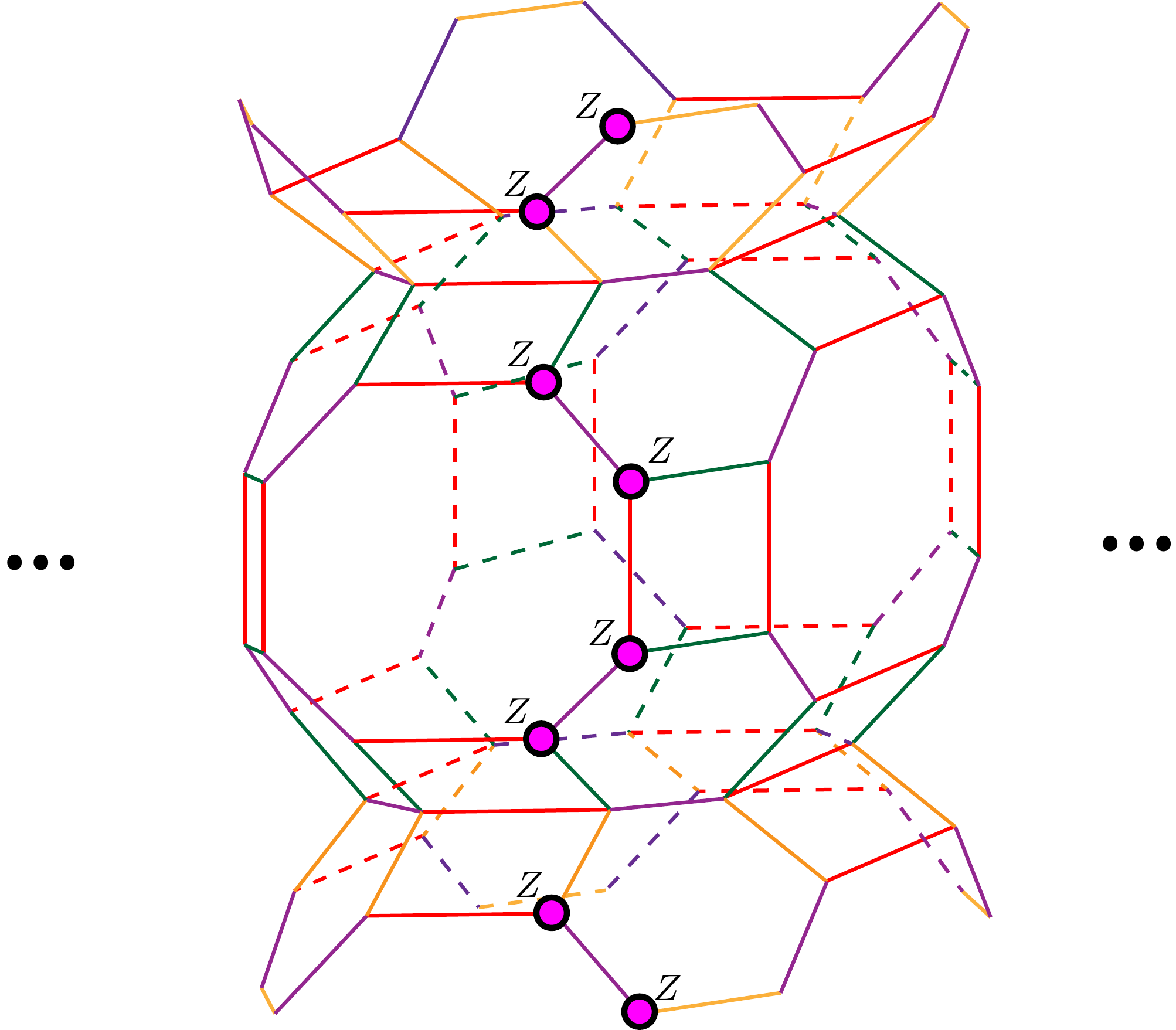}
\end{align*}

It can be verified that these three logical string operators commute with all stabilizers and cannot be decomposed into a product of local stabilizers. Therefore, the condensations on the $Z$-boundary can be written as $\{y_{\mathbf{z}}, g_{\mathbf{z}}, p_{\mathbf{z}}\}$.

Our next step is to show that this boundary, up to local unitary transformations and entangling/disentangling ancilla qubits, is equivalent to the $(e_1, e_2, e_3)$-boundary of three copies of 3D toric codes. We apply the unfolding unitaries introduced in Section~\ref{sec:unfold_without_boundary}. To simplify the presentation of the boundary, we have opted for a different alignments from the one shown for the string operators. In this case, we consider the front and rear sides of the green (and purple) lattice to be the $Z$-boundaries, as opposed to the top and bottom sides as previously shown. The green sublattice is given by
\begin{align*}
    \adjincludegraphics[width=4cm,valign=c]{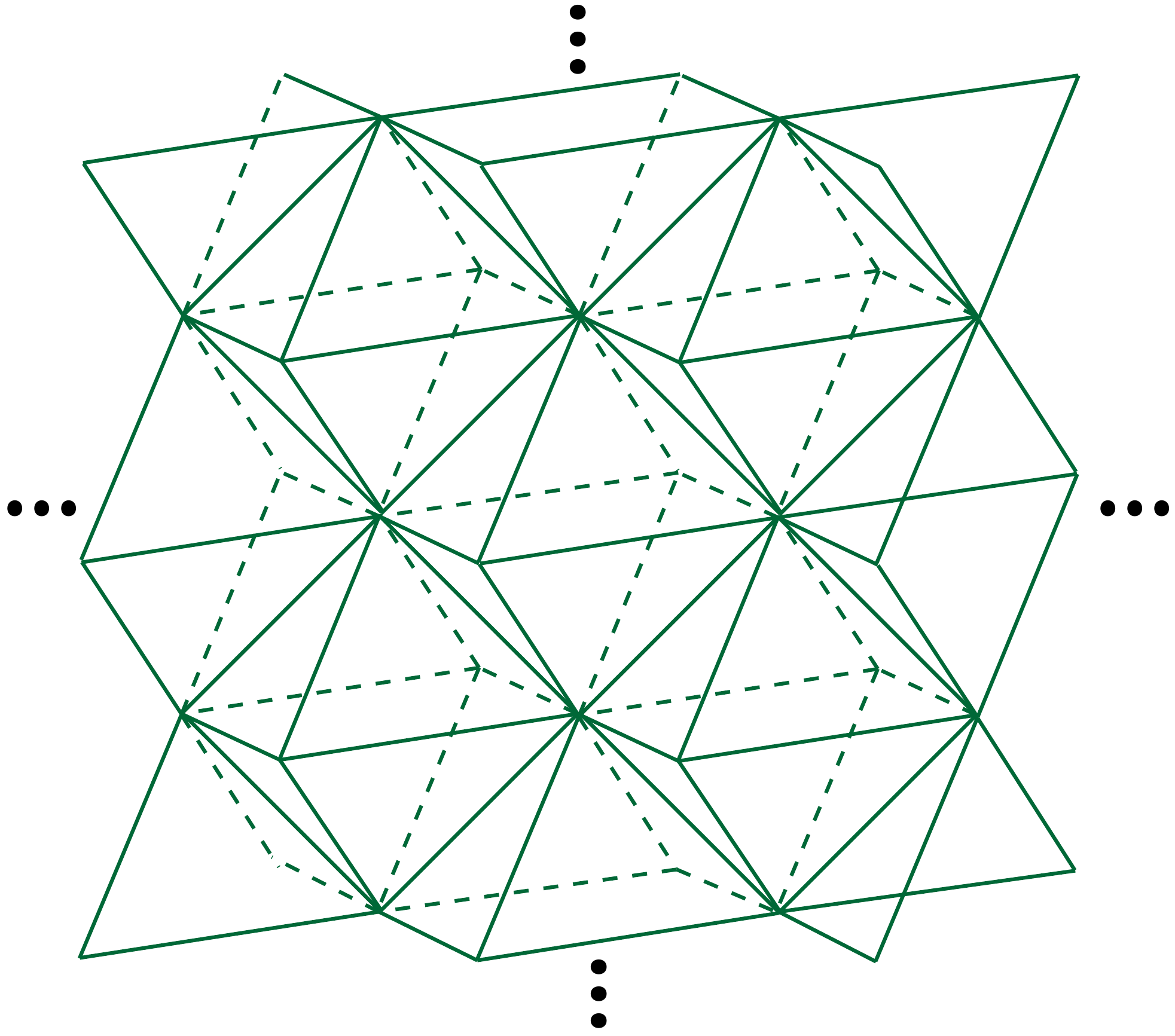}.
\end{align*}
Here, the qubits are located on the edges of the lattice. The stabilizers on this boundary are as follows.
\begin{align*}
    \adjincludegraphics[width=6cm,valign=c]{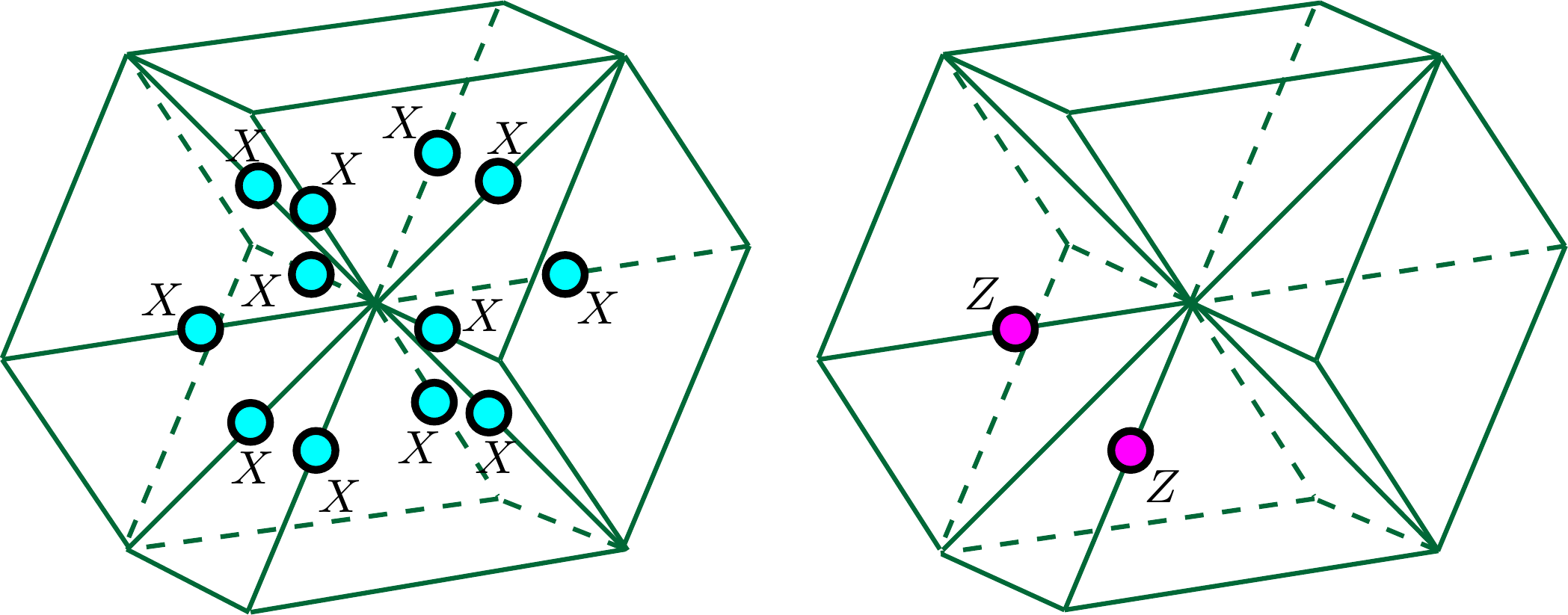}.
\end{align*}
The excitations on this boundary are given by
\begin{align*}
    \adjincludegraphics[width=6cm,valign=c]{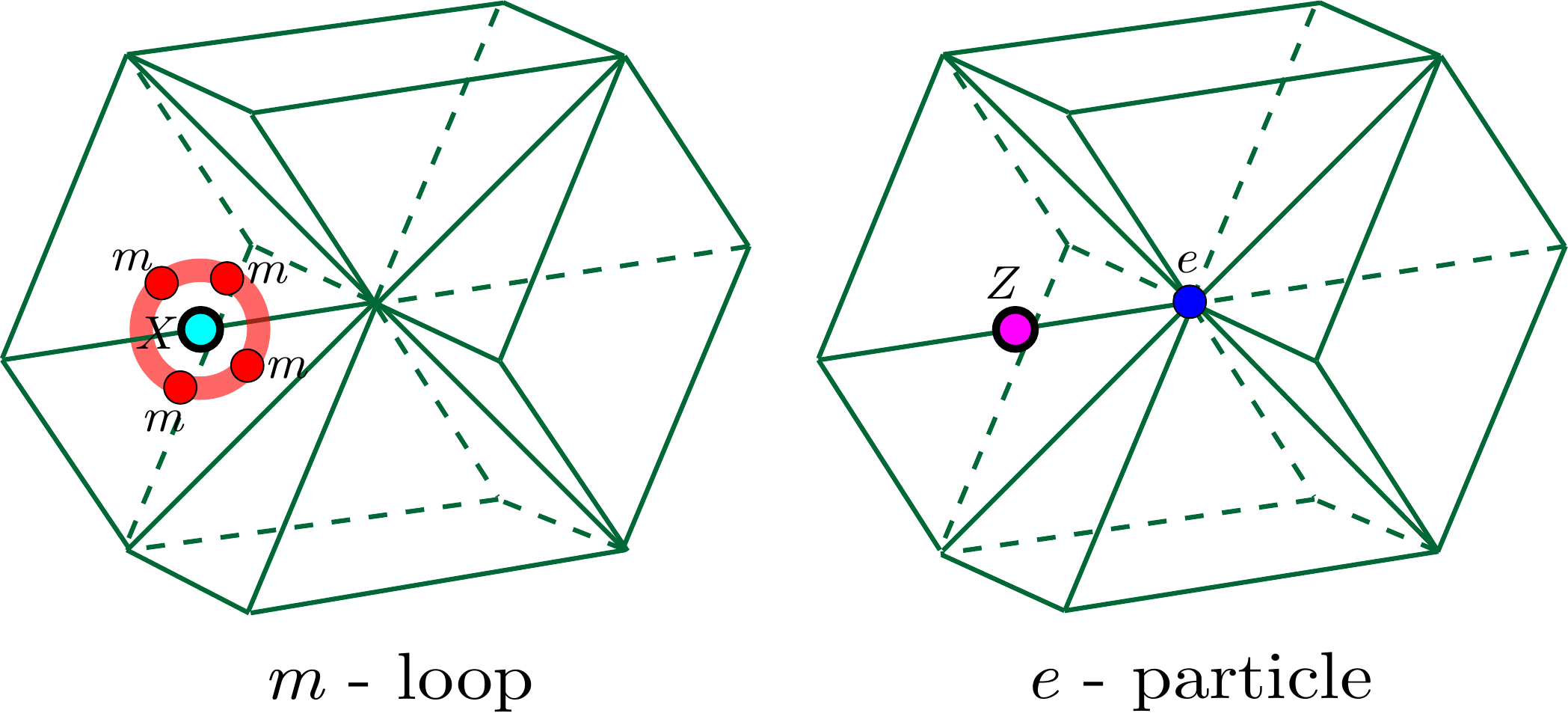}.
\end{align*}
We see that magnetic loop excitations can be created on the boundary, while electric charges can condense on this boundary. Therefore, this boundary is a rough ($e$)-boundary for the green lattice, and similar arguments can be applied to the yellow lattice due to the symmetry of the lattice.

Similarly, we apply the unfolding unitaries to the purple sublattice. The $X$-stabilizers are intact. The $Z$-stabilizers on this boundary are as follows.
\begin{align*}
    \adjincludegraphics[width=4cm,valign=c]{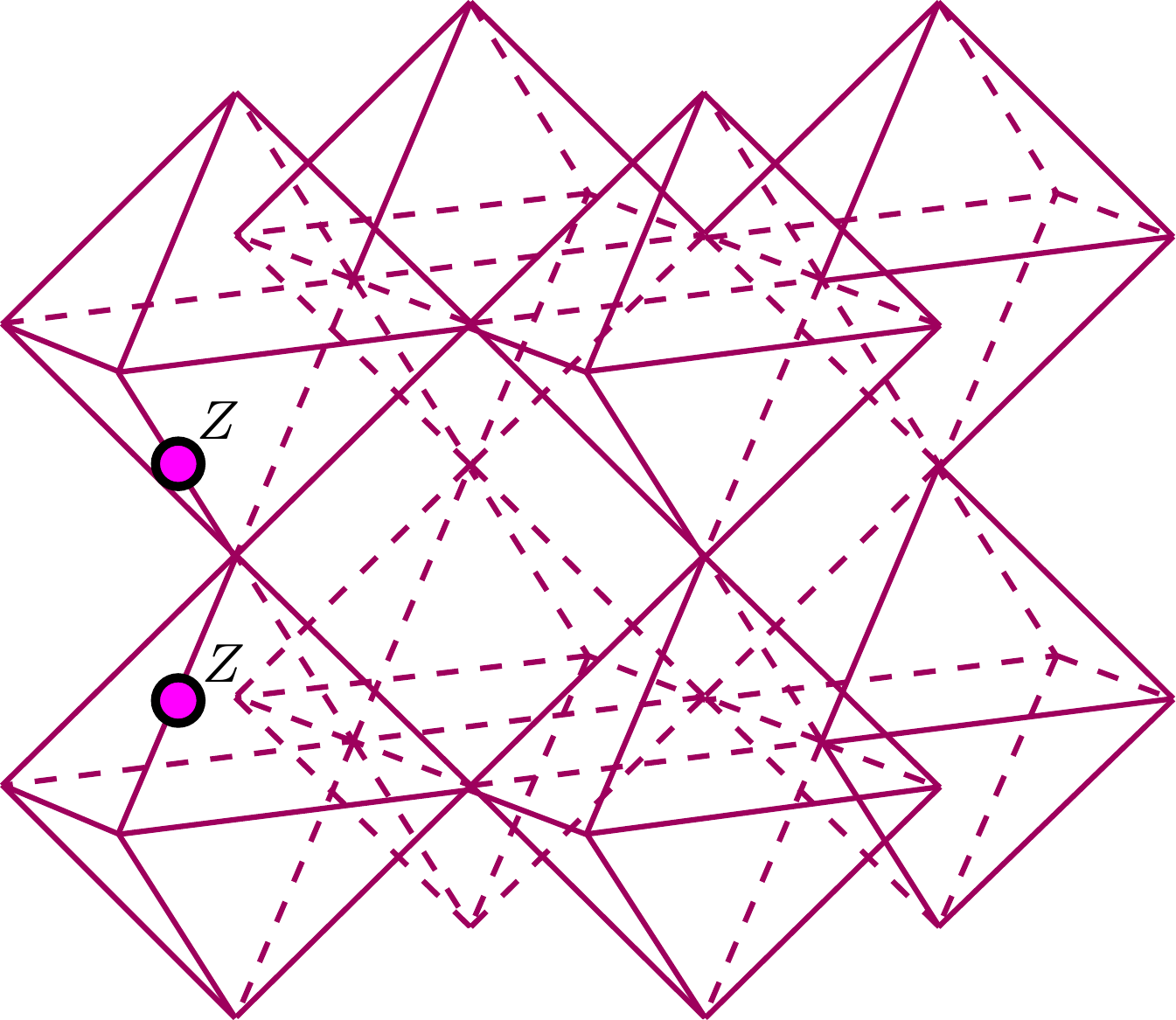}.
\end{align*}
The excitations on this boundary are displayed below.
\begin{align*}
    \adjincludegraphics[width=7.5cm,valign=c]{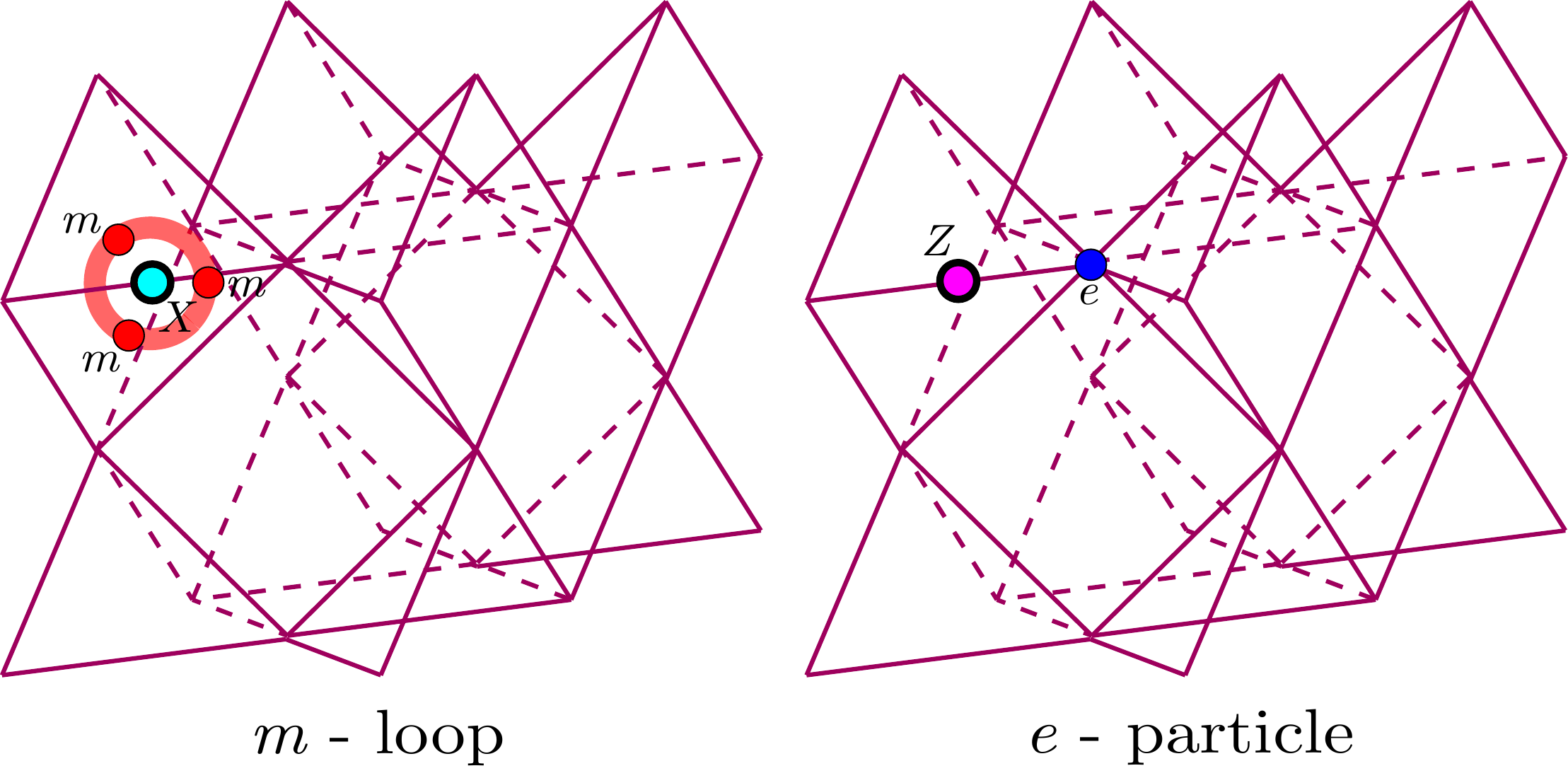}
\end{align*}
We see that the electric charges can condense on this boundary, while the magnetic condensations cannot. 

In conclusion, we show that the $Z$-boundary of the 3D color code is equivalent to the rough boundaries of three copies of 3D toric codes up to local unitary transformations and ancilla qubits. Explicitly, we show the following map.
\begin{align}
    \{y_{\mathbf{z}}, g_{\mathbf{z}}, p_{\mathbf{z}}\} \longleftrightarrow (e_1, e_2, e_3).
\end{align}
The arrow represents the equivalence up to local unitary transformation and entangling/disentangling ancilla qubits. We illustrate this equivalence in Fig.~\ref{fig:magic_boundary_2}(b).
\end{proof}

\subsection{The $X$-Boundary} \label{sec:paulixboundary}

\begin{definition}
    Consider a 3D color code with colors $\mathbf{C} = \{y, g, p, r\}$. The color of the edges, denoted by $\mathbf{c}_i \in \mathbf{C}$, is defined as the color of the cells they connect. Let us consider a surface that intersects only with edges of a specific color. On one side of this surface, all the degrees of freedom are removed. On the other side, the $Z$-stabilizers and  $X$-stabilizers that intersect with the surface are removed and truncated respectively, while all the other stabilizers remain the same. This boundary is referred to as the $X$-boundary. \label{def:Paulixboundary}
\end{definition}
\begin{theorem}
Consider the surface in Definition~\ref{def:Paulixboundary} intersects with color $\mathbf{c}_4 \in \mathbf{C}$. The condensations on the corresponding $X$-boundary are given by $\{(\mathbf{c}_2 \mathbf{c}_3)_{\mathbf{x}}, (\mathbf{c}_1 \mathbf{c}_3)_{\mathbf{x}}, (\mathbf{c}_1 \mathbf{c}_2)_{\mathbf{x}}\}$, where $\mathbf{c}_1, \mathbf{c}_2, \mathbf{c}_3, \mathbf{c}_4 \in \mathbf{C}$ and $\mathbf{c}_1 \neq \mathbf{c}_2 \neq \mathbf{c}_3 \neq \mathbf{c}_4$. This boundary is equivalent to the $(m_1, m_2, m_3)$ boundary of three copies of 3D toric codes up to local unitaries and entangling/disentangling ancilla qubits.
\end{theorem}

\begin{proof}
Without loss of generality, we choose $\mathbf{c}_4$ to be the red color, and we use the lattice of 3D color code in Fig.~\ref{fig:3DCC}. To prove the theorem, we first show that the condensation on this boundary is $\{pg_{\mathbf{x}}, py_{\mathbf{x}}, yg_{\mathbf{x}}\}$. Then we show that by applying local unitaries, it is equivalent to the $(m_1, m_2, m_3)$-boundary of three copies of 3D toric codes. The proof also works in other colorizations and other cell structures.

The lattice of this boundary is identical to the lattice shown in Eq.~\eqref{eq:bdry_12}. However, the stabilizers on it are different.

The $Z$-stabilizers are located on every plaquette. Since the cells of three different colors are truncated, and the $X$-stabilizers are associated with cells. Therefore, there are three types of truncated $X$-stabilizers on the $X$-boundary. We denote them as $\mathcal{A}^{boundary}_{g}(X)$, $\mathcal{A}^{boundary}_{y}(X)$, and $\mathcal{A}^{boundary}_{p}(X)$.

The truncated $X$-stabilizer on green cells, denoted as $\mathcal{A}^{boundary}_{g}(X)$, is displayed as follows.
\begin{align*}
    \adjincludegraphics[width=4cm,valign=c]{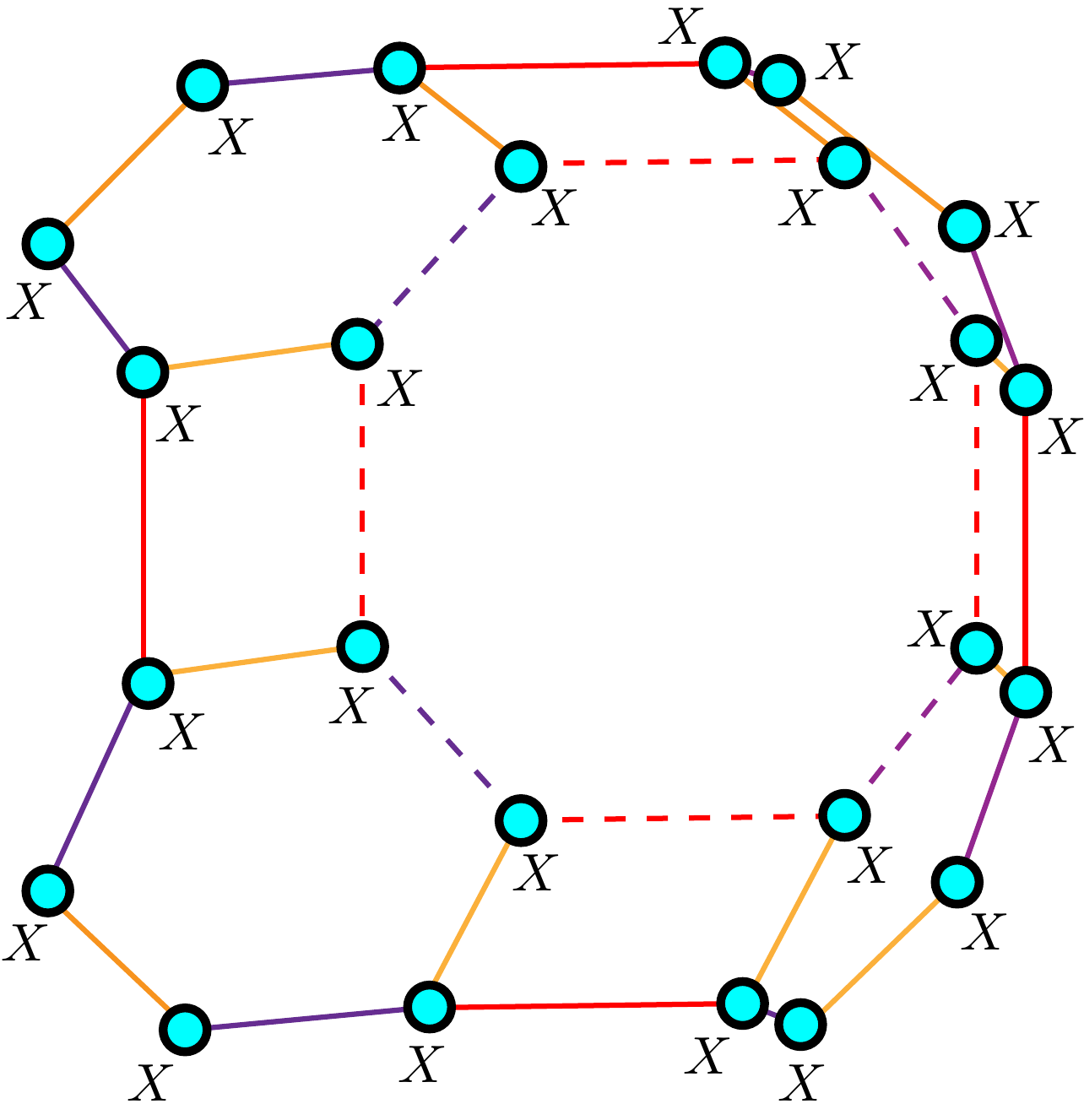}
\end{align*}
The truncated $X$-stabilizer on yellow cells, denoted as $\mathcal{A}^{boundary}_{y}(X)$, has the same structure as those on green cells. Therefore, one can obtain $\mathcal{A}^{boundary}_{y}(X)$ by translations of $\mathcal{A}^{boundary}_{g}(X)$.

The truncated $X$-stabilizer on purple cells, denoted as $\mathcal{A}^{boundary}_{p}(X)$, is displayed as follows.
\begin{align*}
    \adjincludegraphics[width=3cm,valign=c]{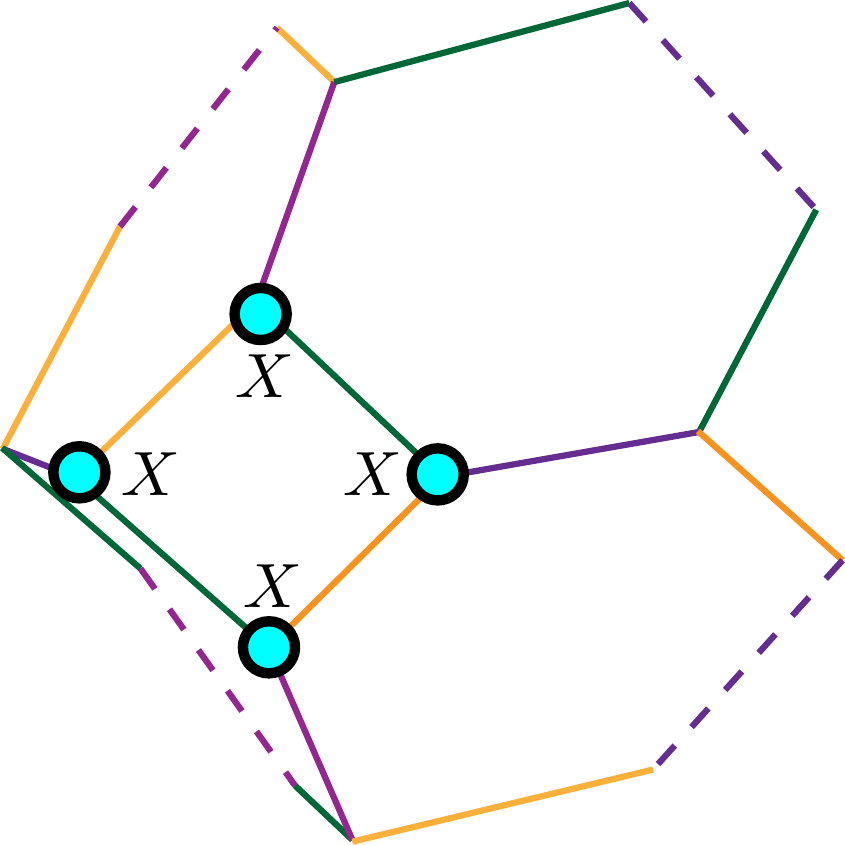}.
\end{align*}
Since we chose to cut the red edges at the beginning, all the remaining red cells are intact. Therefore, there are no truncated $X$-stabilizers of red cells on this boundary.

Similar to the case with $Z$-boundaries, three independent types of membrane operators can terminate on this boundary and commute with every stabilizer. However, the supports of the membrane operators can differ depending on the specific cellulation. In our example, we display the configurations of these (logical) membrane operators in Fig.~\ref{fig:gate_10}. We define the $X$-membrane operator on the $\mathbf{c}_i \mathbf{c}_j$-colored membrane as $\overline{X}_{\mathbf{c}_i \mathbf{c}_j}$ if it is also a logical operator, and define it as $\widetilde{X}_{\mathbf{c}_i \mathbf{c}_j}$ if it is not.
\begin{figure}
    \centering
    \includegraphics[width = 7cm]{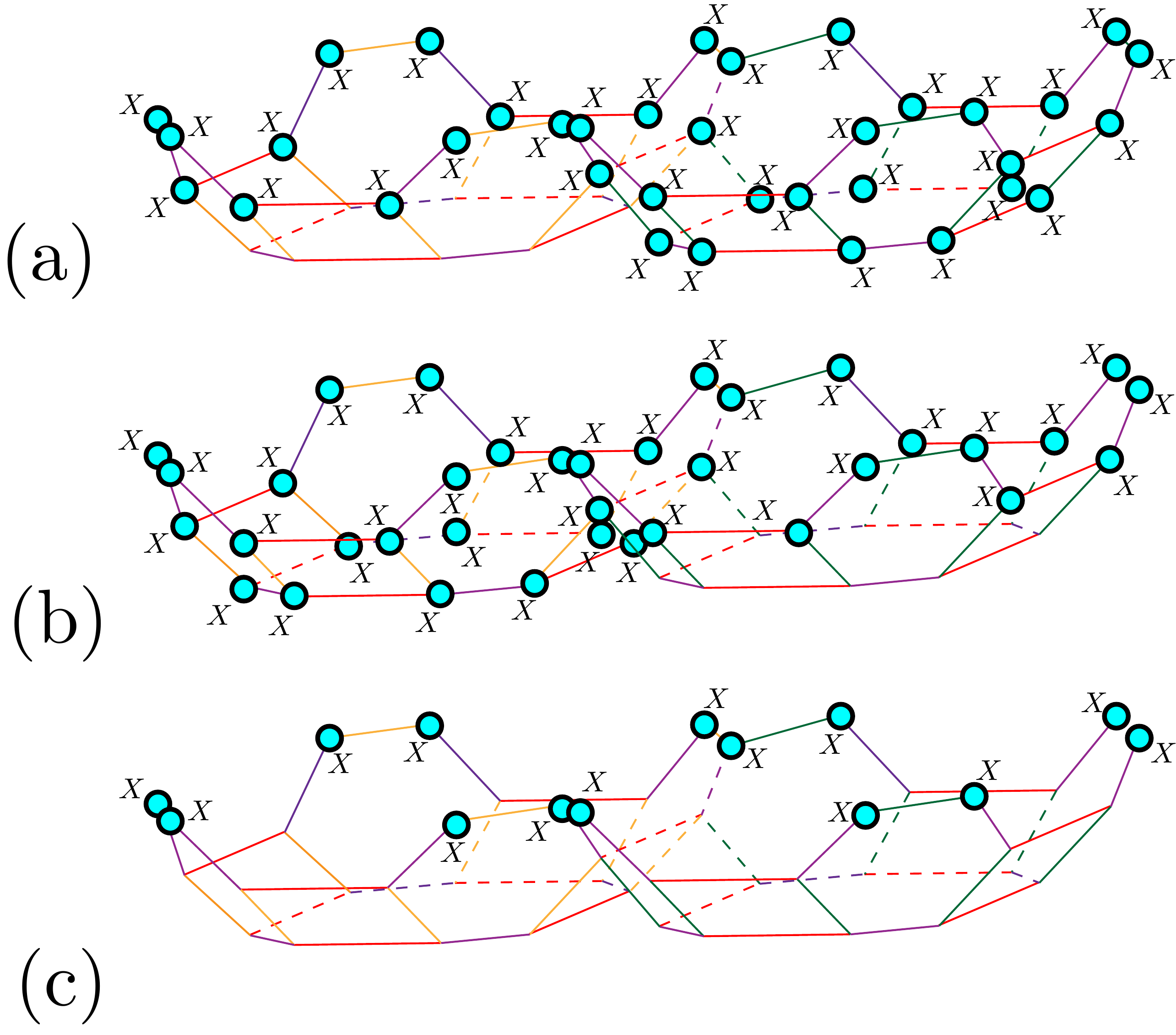}
    \caption{Three types of $X$-membrane operators. (a) The purple-green membrane operator, $\widetilde{X}_{pg}$ (or $\overline{X}_{pg}$ if it is a logical operator). (b) The purple-yellow membrane operator, $\widetilde{X}_{py}$ (or $\overline{X}_{py}$). (c) The yellow-green membrane operator, $\widetilde{X}_{yg}$ (or $\overline{X}_{yg}$).}
    \label{fig:gate_10} 
\end{figure}

The endpoints of the membrane that generates $pg_{\mathbf{x}}$ on this boundary are given by
\begin{align}
    \adjincludegraphics[width=7cm,valign=c]{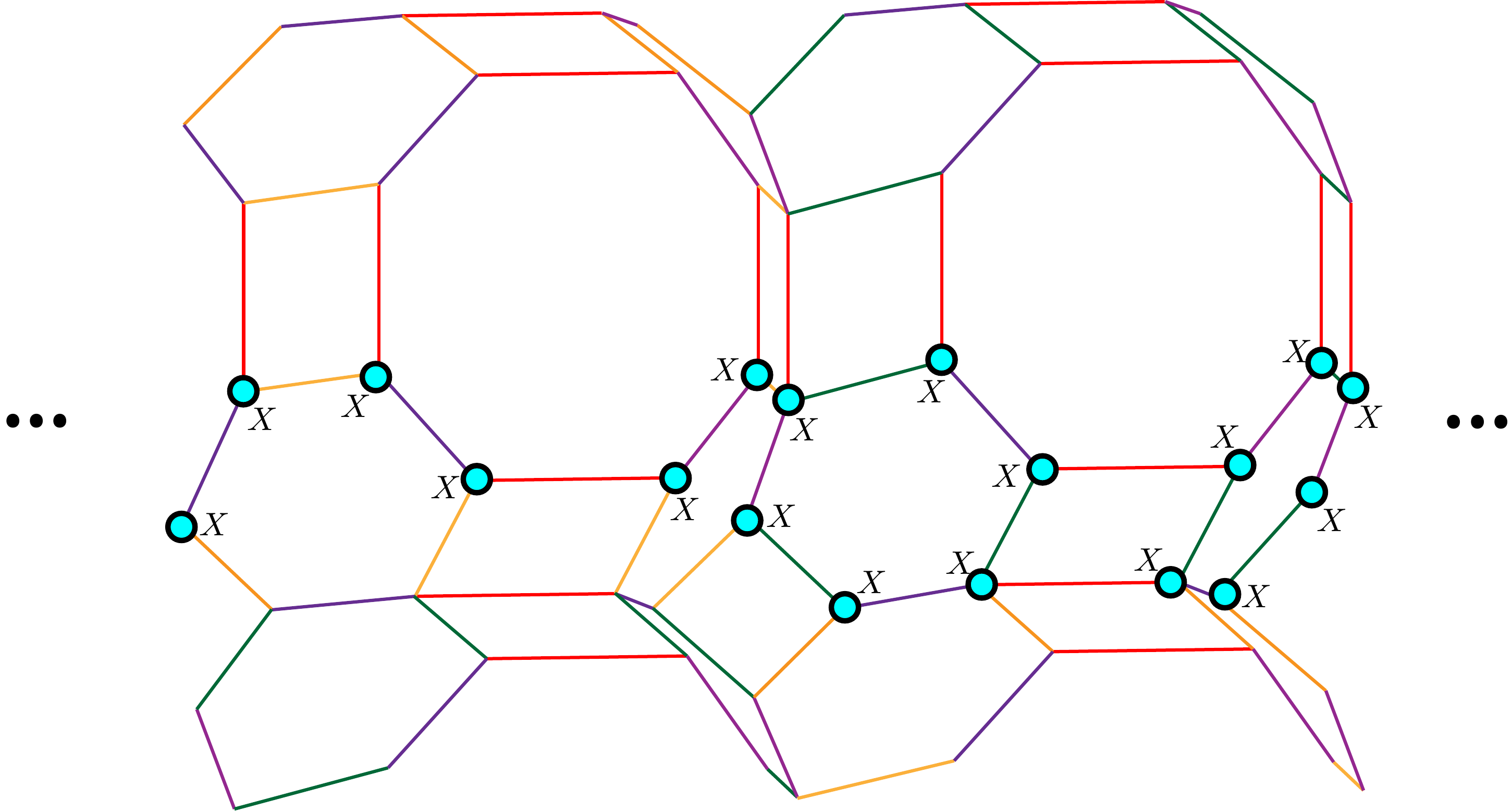}. \label{eq:bdry_27}
\end{align}
Here we show the intersection between the membrane operator and the $X$-boundary. The endpoints of the membrane that creates $py_{\mathbf{x}}$ are given by
\begin{align}
    \adjincludegraphics[width=7cm,valign=c]{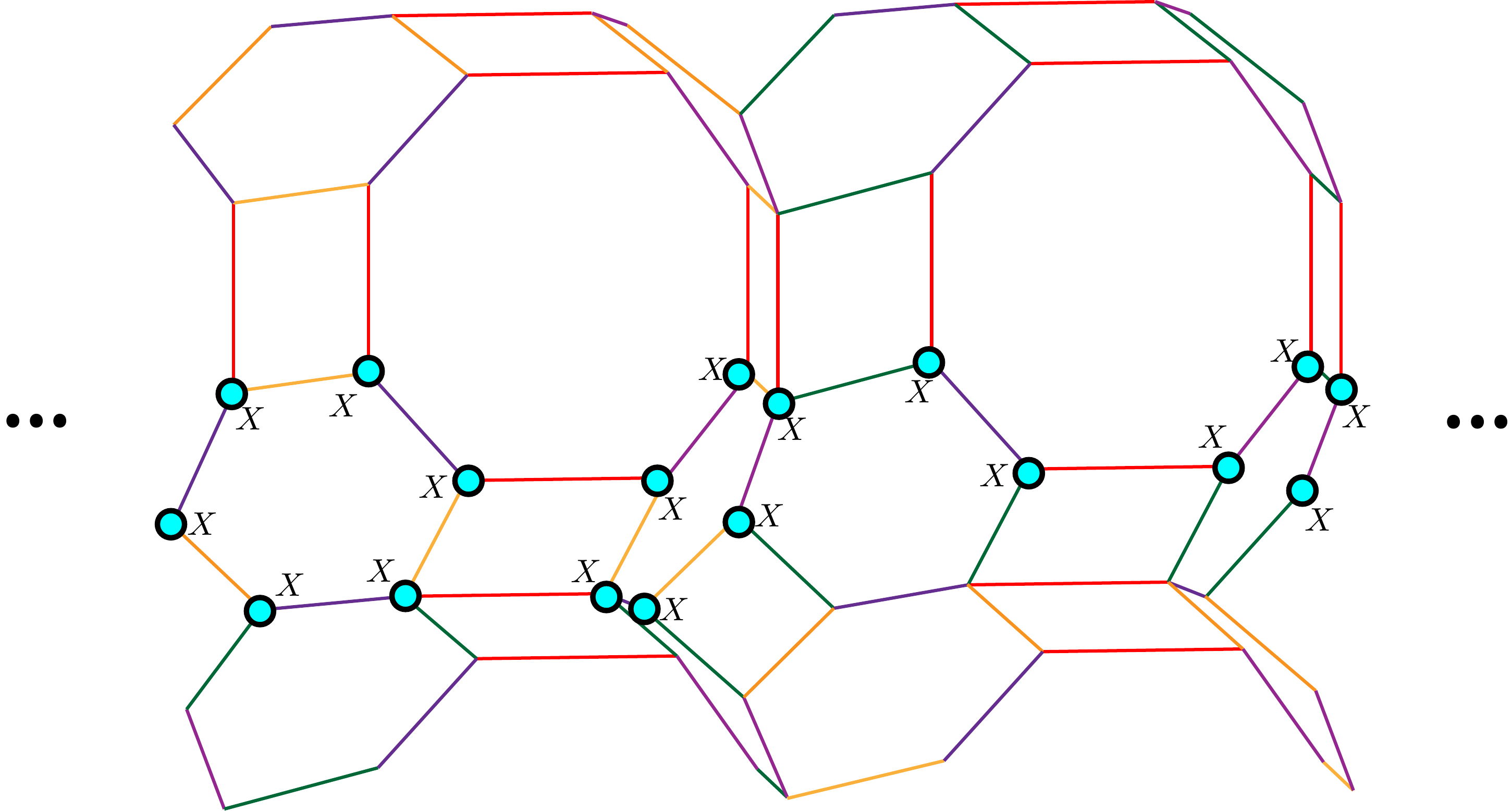}. \label{eq:bdry_28}
\end{align}
The endpoints of the membrane that creates $yg_{\mathbf{x}}$ are given by
\begin{align}
    \adjincludegraphics[width=7cm,valign=c]{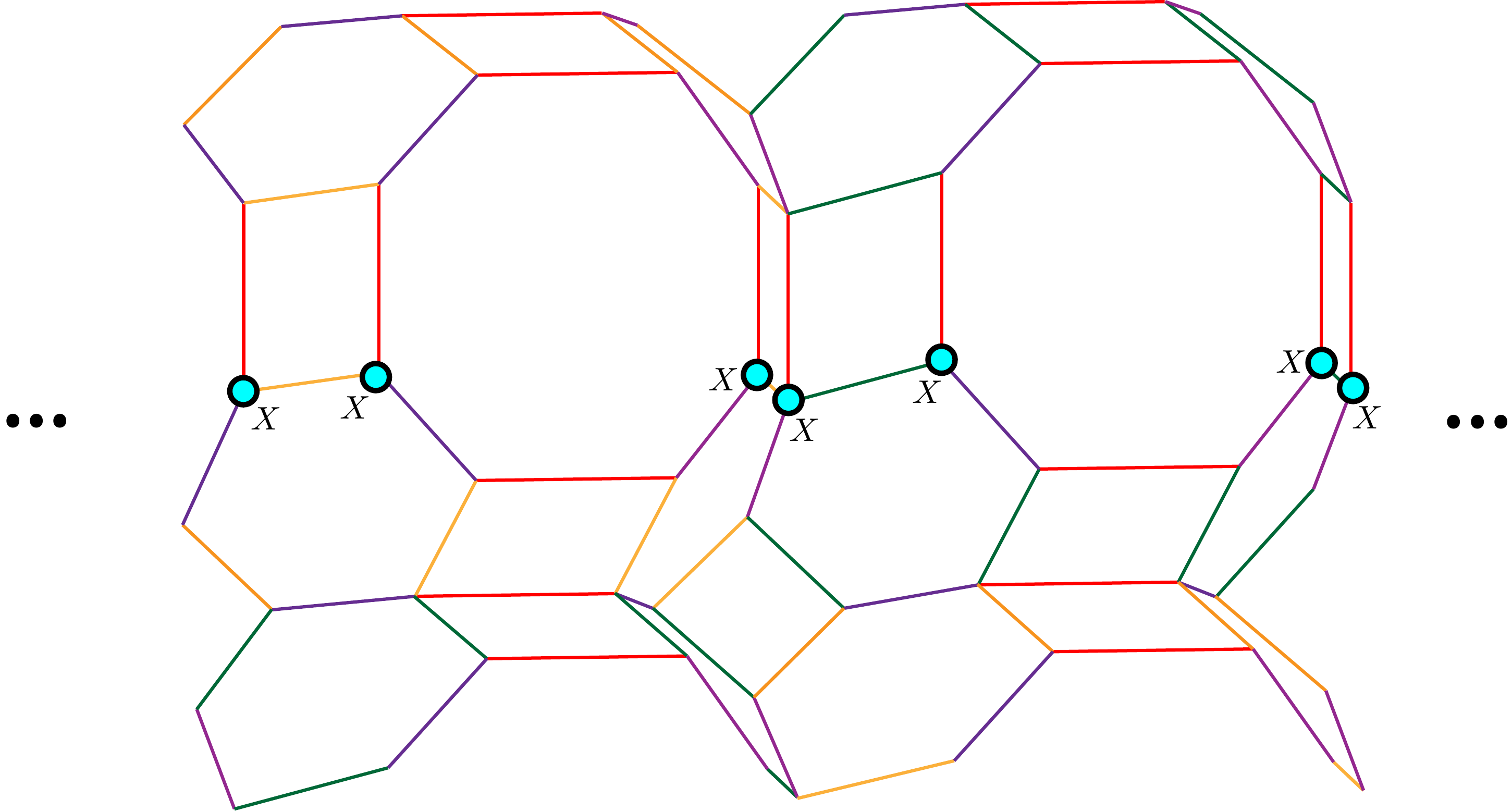}. \label{eq:bdry_29}
\end{align}

One can verify that all these logical membrane operators commute with the stabilizers in the bulk and on the $X$-boundary. Therefore, the condensations on the $X$-boundary are $\{pg_{\mathbf{x}}, py_{\mathbf{x}}, yg_{\mathbf{x}}\}$.

Our next step is to demonstrate this $X$-boundary is equivalent to the $(m_1,m_2,m_3)$-boundary of three copies of 3D toric codes. By applying the unfolding unitaries, the green lattice we get is depicted as follows.
\begin{align*}
    \adjincludegraphics[width=4.5cm,valign=c]{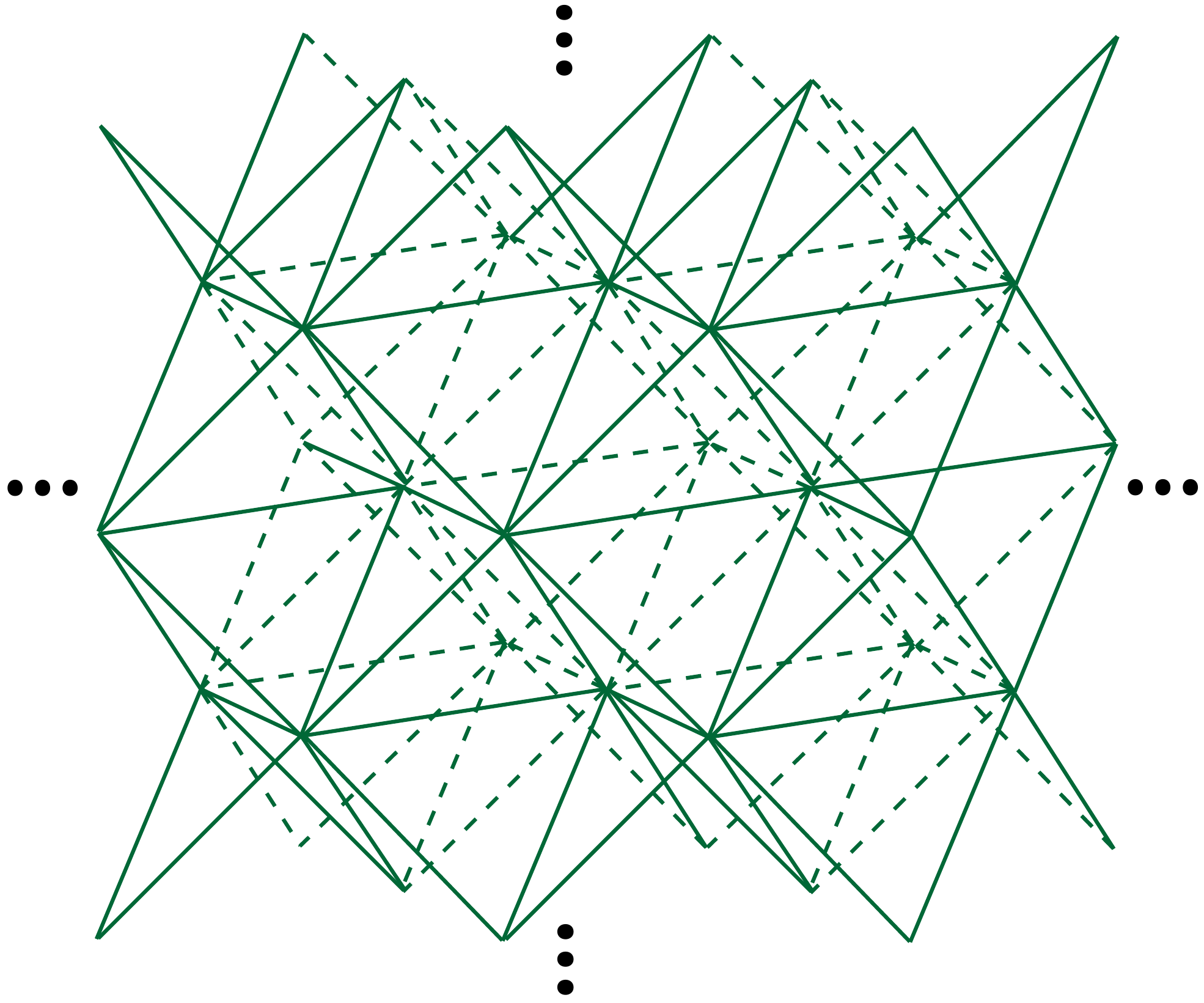}.
\end{align*}
For the same reason, we choose the front and rear boundary of the above lattice to correspond to the original $X$-boundary. The qubits are located on the edges of the lattice. On this boundary, there are truncated $X$-stabilizers on every vertex and intact $Z$ stabilizers on every plaquette. The $X$ and $Z$-stabilizers on this boundary are given by
\begin{align*}
    \adjincludegraphics[width=6cm,valign=c]{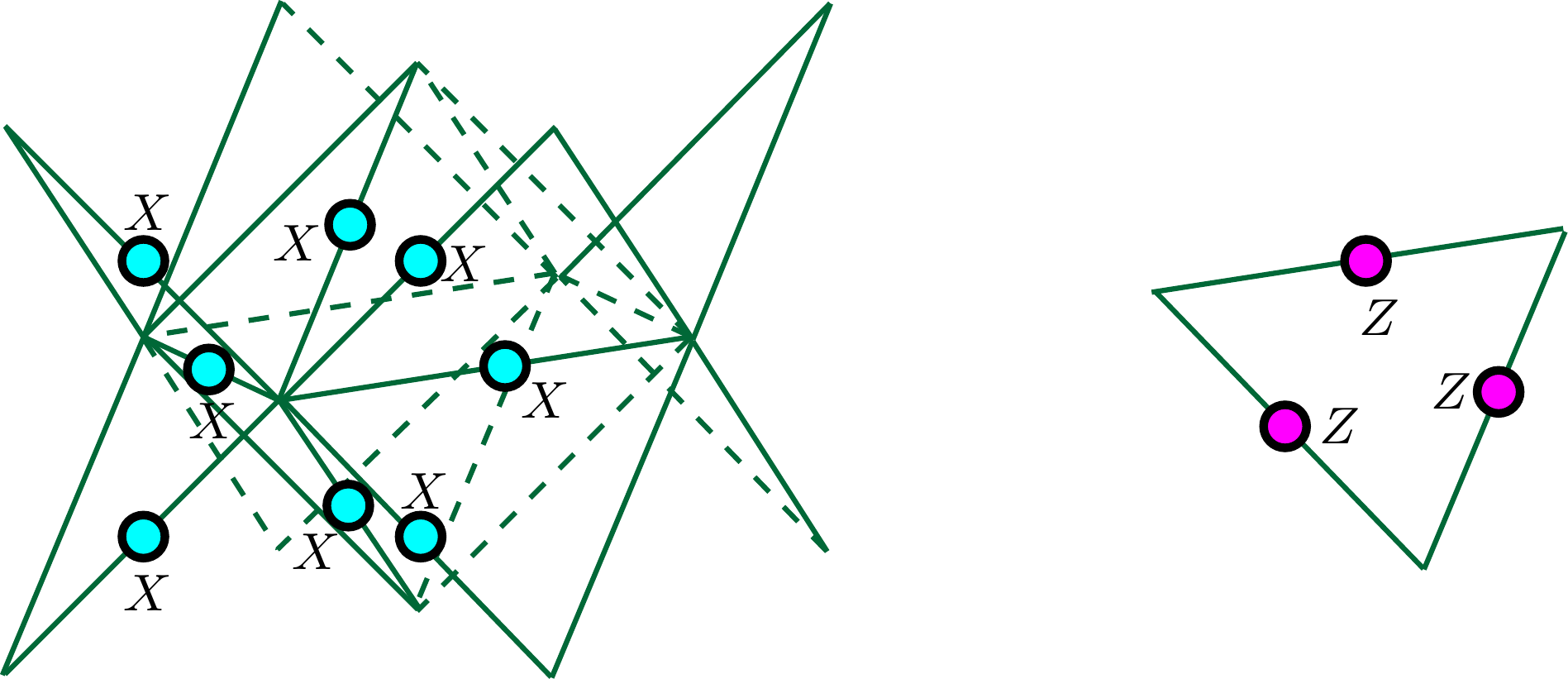}.
\end{align*}

The excitations on this boundary are shown below.
\begin{align*}
    \adjincludegraphics[width=6cm,valign=c]{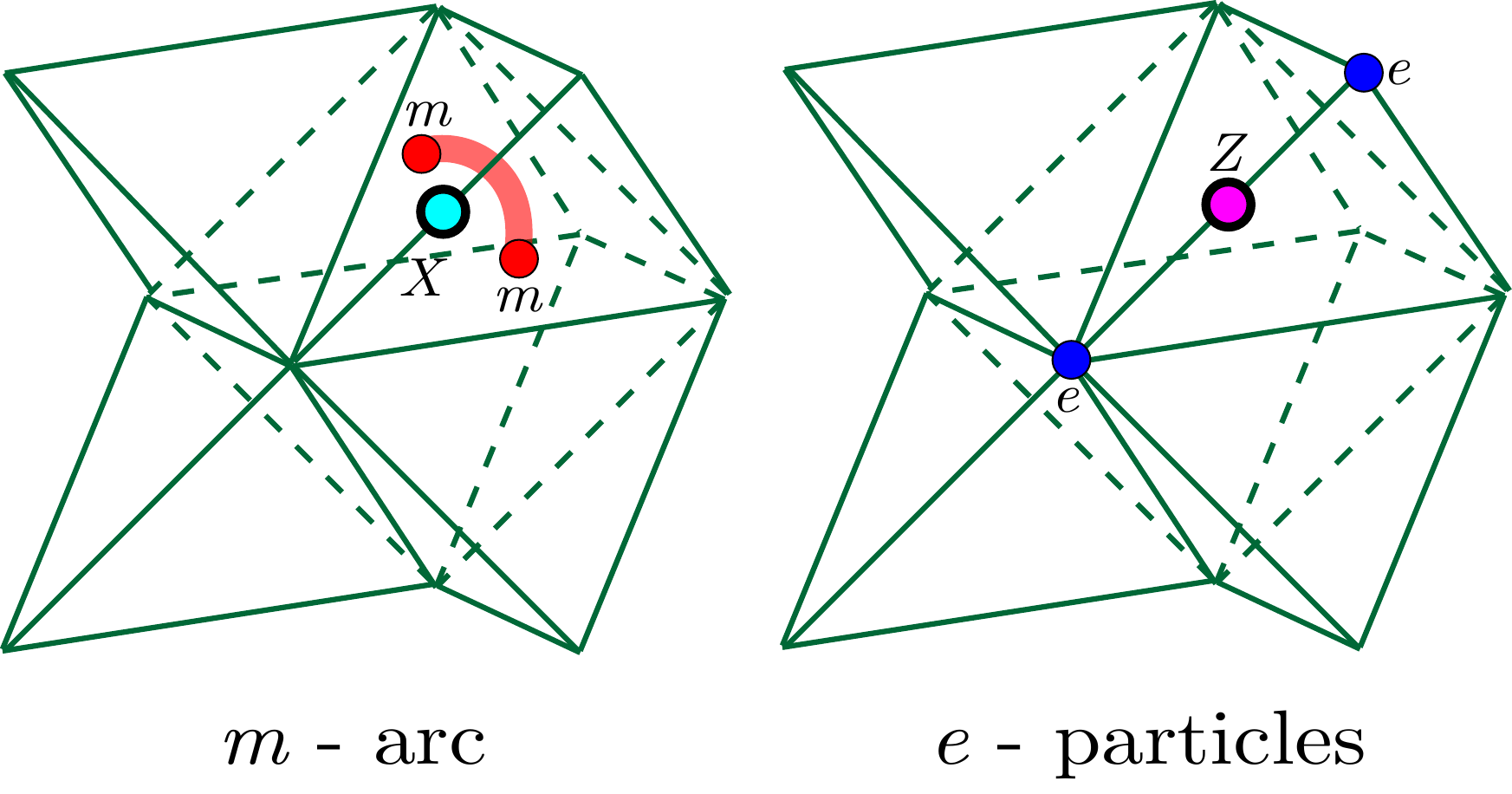}.
\end{align*}
We see that the magnetic loop becomes an arc and condenses on this boundary, while pairs of electric charges can be created on this boundary. This suggests that magnetic fluxes can condense on this boundary, while electric charges cannot. Therefore, this boundary of the green lattice is a smooth boundary.

The yellow lattice has the same structure as the green one. Therefore, all the arguments for the green lattice hold for the yellow lattice. After the unfolding unitaries, the boundary of the yellow lattice is a smooth boundary as well.

Finally, we apply the unfolding unitaries to the purple lattice. The lattice is shown below.
\begin{align*}
    \adjincludegraphics[width=4.5cm,valign=c]{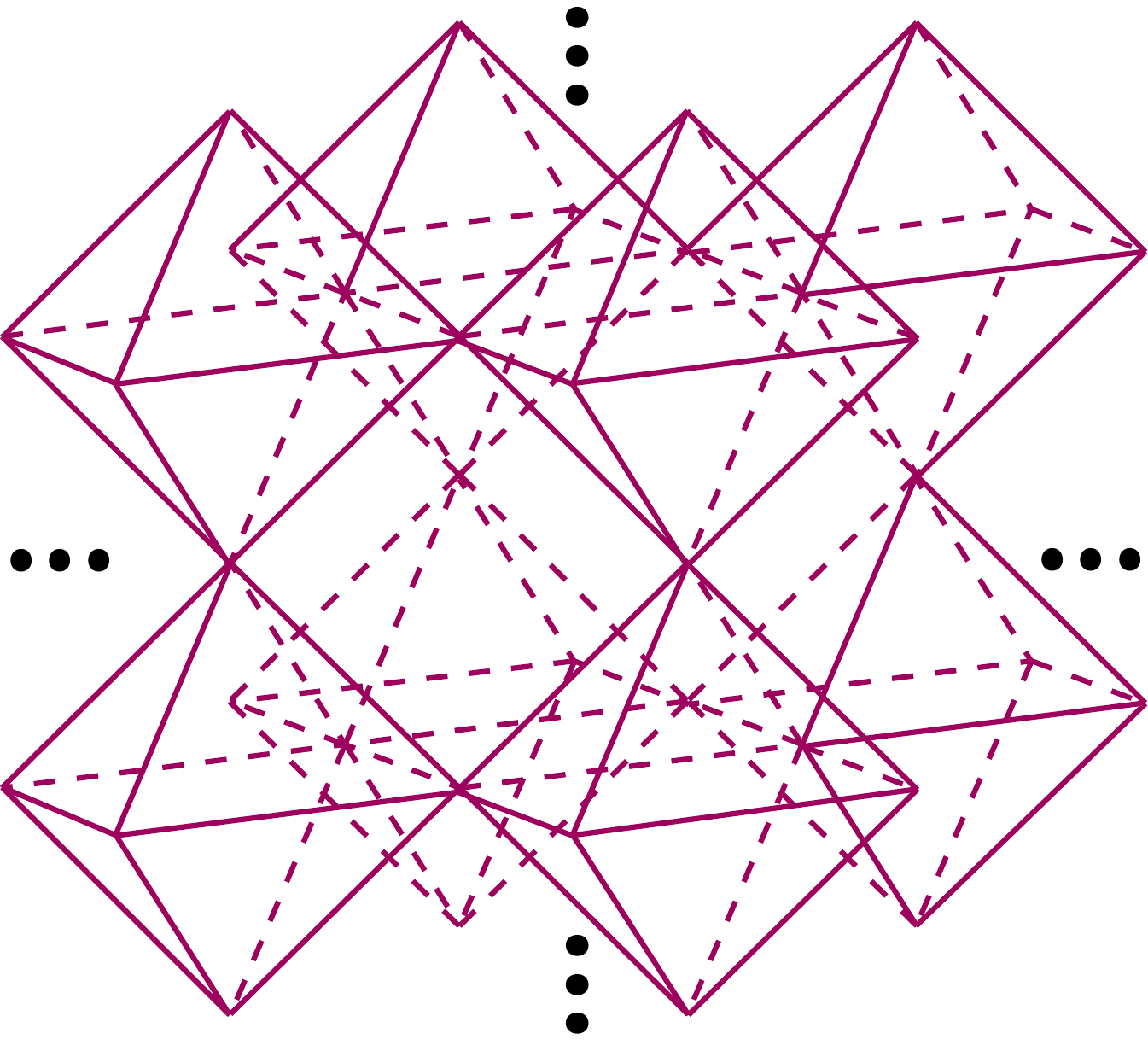}. \label{eq:bdry_18}
\end{align*}
The truncated $X$-stabilizer and intact $Z$-stabilizer on this boundary are depicted as follows.
\begin{align*}
    \adjincludegraphics[width=7cm,valign=c]{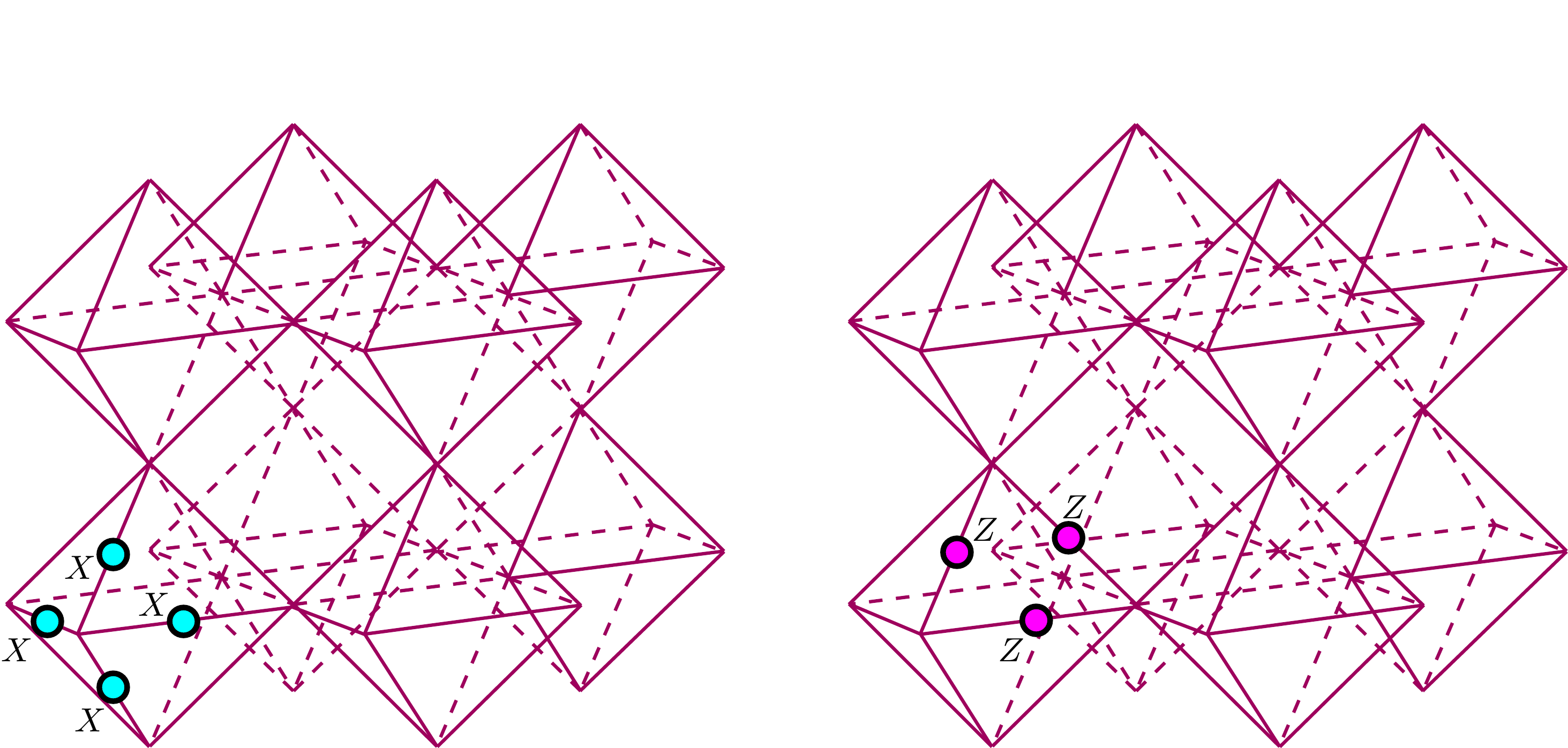}.
\end{align*}
The excitations on this boundary are shown below
\begin{align*}
    \adjincludegraphics[width=7cm,valign=c]{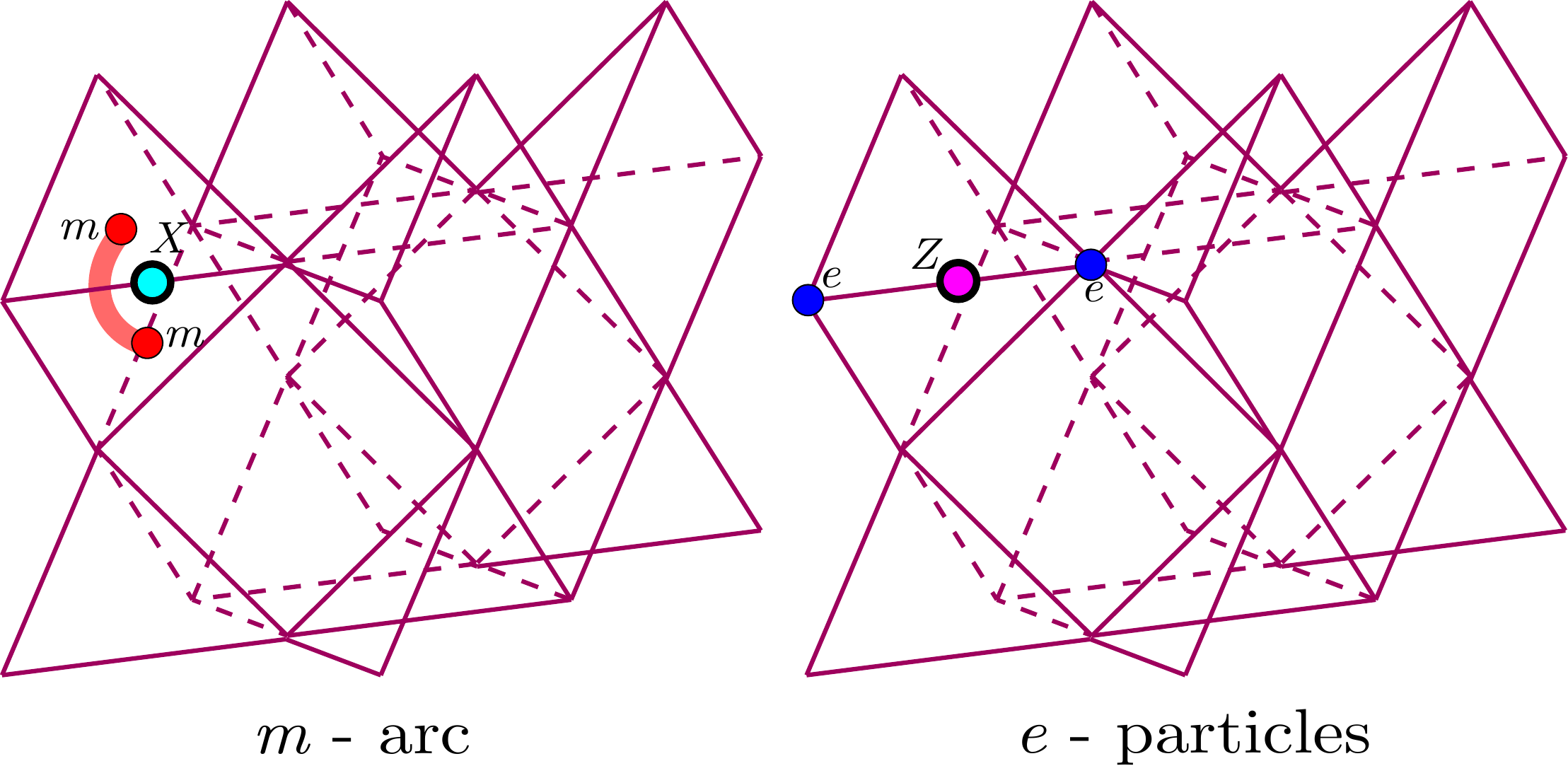}
\end{align*}
We see that the magnetic flux loop becomes an arc and condenses on this boundary, while pairs of electric charges can be created on this boundary.

In conclusion, we show that the $X$-boundary of 3D color code is equivalent to the smooth boundaries of three copies of 3D toric codes up to local unitary transformations and ancilla qubits. Explicitly we show the following map. 
\begin{align}
    \{pg_\mathbf{x}, py_\mathbf{x}, yg_\mathbf{x}\} \longleftrightarrow (m_1, m_2, m_3).
\end{align}
The arrow represents the equivalence up to local unitary transformation and entangling/disentangling ancilla qubits. We illustrate this equivalence in Fig.~\ref{fig:magic_boundary_2}(b).
\end{proof}

\subsection{Other boundaries}

For the completeness of our discussion, we briefly introduce the other boundaries of the 3D color code that allow different condensation patterns compared to the $X$- or $Z$-boundaries.

\subsubsection{The $\{(\mathbf{c}_i)_{\mathbf{z}}, (\mathbf{c}_i\mathbf{c}_k)_{\mathbf{x}}, (\mathbf{c}_i\mathbf{c}_j)_{\mathbf{x}}\}$-boundaries} \label{sec:zxx}

These boundaries are also referred to as the $\mathbf{c}_i$-color boundaries, indicating their capacity to condense $\mathbf{c}_i$-colored electric charges. Conventionally, the 3D color code is defined with three distinct color boundaries and a folded boundary, in a tetrahedron geometry. By applying unfolding unitaries, it is equivalent to three copies of the toric codes with different alignments, as depicted in Fig.~\ref{fig:magic_boundary_2}(a). We use the $\{y_{\mathbf{z}}, py_{\mathbf{x}}, yg_{\mathbf{x}}\}$-boundary as an example to show how to construct this type of boundaries. The boundary stabilizers, $\mathcal{B}^{boundary}_{pg}(Z)$, $\mathcal{A}^{boundary}_{g}(X)$, and $\mathcal{A}^{boundary}_{p}(X)$ are given by
\begin{align}
    \adjincludegraphics[width=6cm,valign=c]{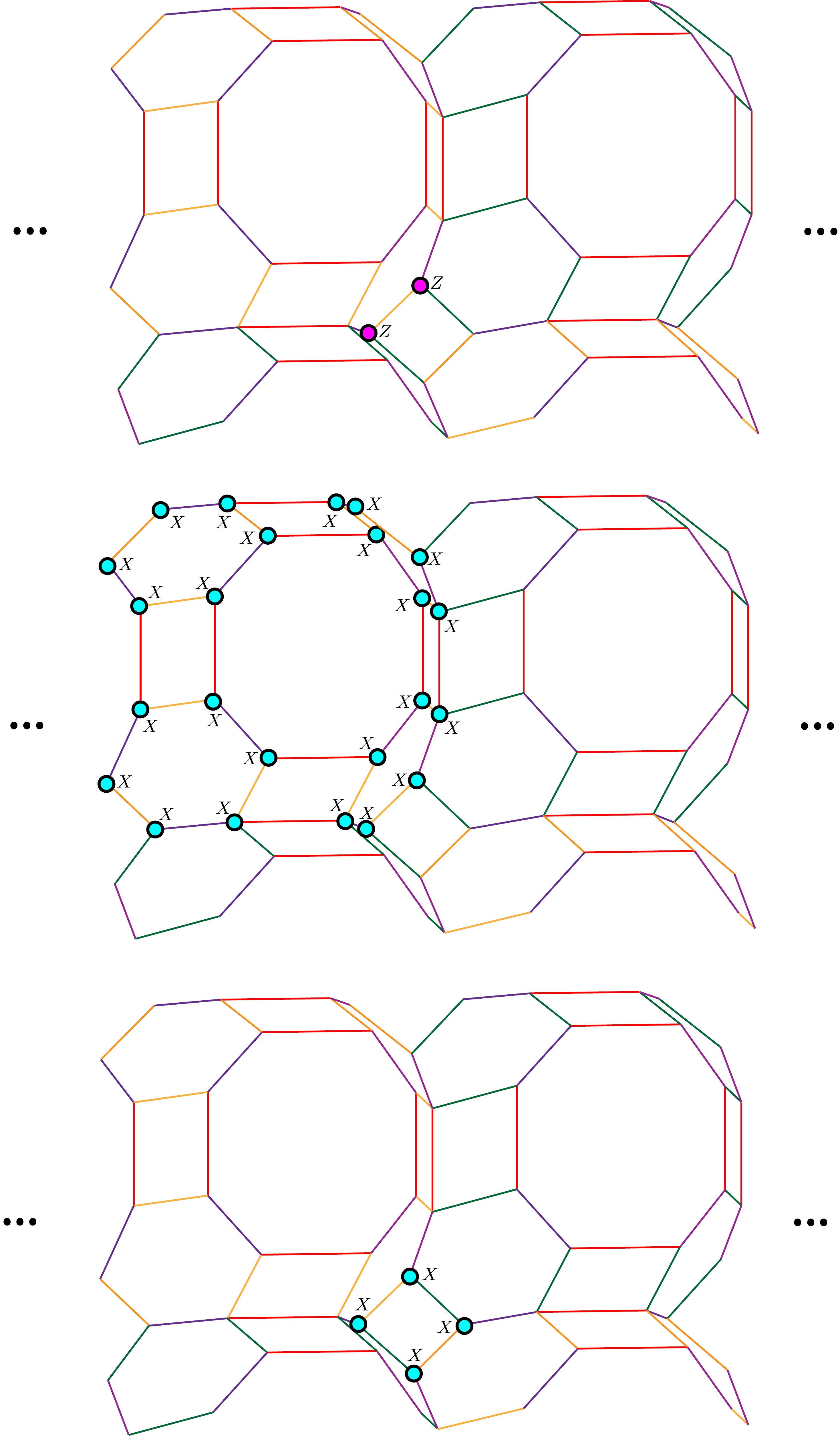}. \label{eq:gate_24}
\end{align}

This boundary is equivalent to the $(e_1, m_2, m_3)$-boundary of three copies of 3D toric codes  after we apply the unfolding unitary transformation.
\begin{align}
    \{y_{\mathbf{z}}, py_{\mathbf{x}}, yg_{\mathbf{x}}\} \longleftrightarrow (e_1, m_2, m_3).
\end{align}
The total number of distinct types of color boundaries is 3.

\subsubsection{The $\{(\mathbf{c}_i)_{\mathbf{z}}, (\mathbf{c}_j)_{\mathbf{z}}, (\mathbf{c}_i\mathbf{c}_j)_{\mathbf{x}}\}$-boundaries} \label{sec:zzx}

This type of boundaries can condense two types of electric charges and one type of flux loops. We use $\{y_{\mathbf{z}}, g_{\mathbf{z}}, yg_{\mathbf{x}}\}$-boundary as an example to illustrate how to construct this type of boundaries. The boundary stabilizers, $\mathcal{B}^{boundary}_{pg}(Z)$, $\mathcal{B}^{boundary}_{py}(Z)$, and $\mathcal{A}^{boundary}_{p}(X)$, are given by
\begin{align}
\adjincludegraphics[width=7cm,valign=c]{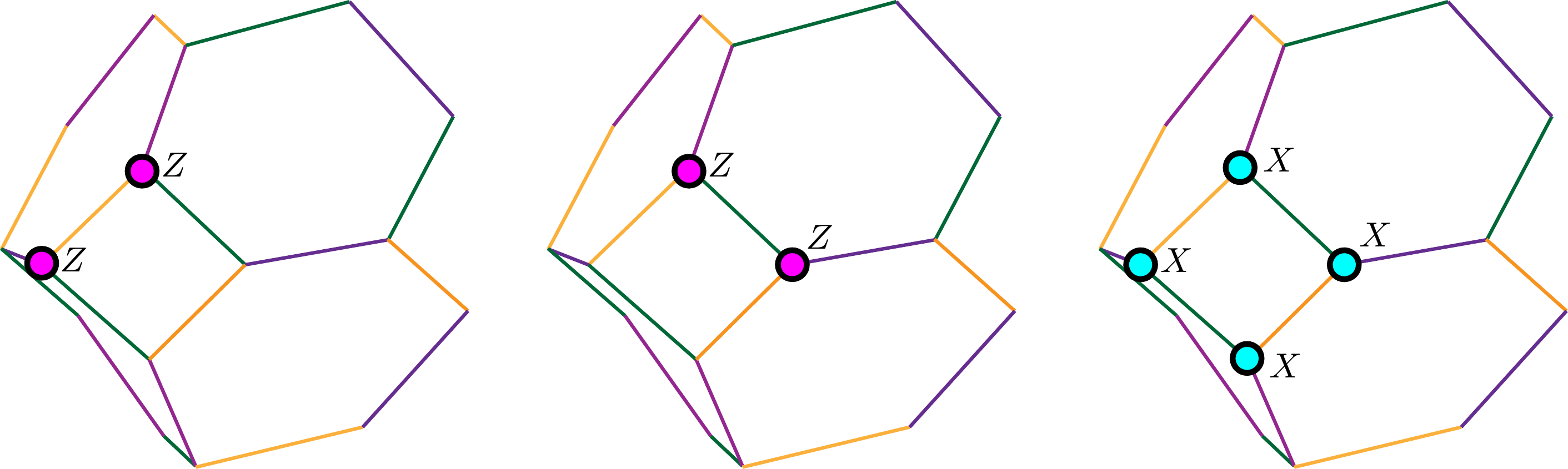}, \label{eq:apdx_9}
\end{align}
where each two-body $Z$-type stabilizer is supported on an edge of a square on the boundary, and each four-body $X$-type stabilizer is supported on a full square on the boundary. We remark that the geometry of the boundary lattice is the same as that of the previously discussed $\{(\mathbf{c}_i)_{\mathbf{z}}, (\mathbf{c}_i\mathbf{c}_k)_{\mathbf{x}}, (\mathbf{c}_i\mathbf{c}_j)_{\mathbf{x}}\}$-boundaries. However, by introducing the second stabilizer in Eq. (21), the second stabilizer in Eq. (19) is no longer preserved.

This boundary is equivalent to the $(e_1, e_2, m_3)$-boundary after we apply the unfolding unitary transformation.
\begin{align}
    \{y_{\mathbf{z}}, g_{\mathbf{z}}, yg_{\mathbf{x}}\} \longleftrightarrow (e_1, e_2, m_3)
\end{align}
The total number of distinct types of $\{y_{\mathbf{z}}, g_{\mathbf{z}}, yg_{\mathbf{x}}\}$-boundaries is 3.

\subsubsection{The folded boundaries} \label{sec:folded_boundary}

Consider there are two copies of 3D toric codes, we can {\it fold} one copy of it, and attach the folded boundary to either the $e$-boundary, or the $m$-boundary of the second copy. We refer to this type of boundaries as the {\it folded boundaries}, implying they are created by folding toric codes. Since one copy of 3D toric code is folded along a codimension-1 (2D) submanifold in the bulk, the condensation set on the folded interface is $(e_i e_j, m_i m_j)$. Together with the boundary of the second copy, one can obtain two different types of boundaries, namely, the $(e_i e_j, m_i m_j, e_k)$-boundary and the $(e_i e_j, m_i m_j, m_k)$-boundary. According to our previous discussion, this boundary correspond to the 3D color-code boundary with condensations $\{(\mathbf{c}_i)_{\mathbf{z}} (\mathbf{c}_j)_{\mathbf{z}}, (\mathbf{c}_j \mathbf{c}_k)_{\mathbf{x}} (\mathbf{c}_i \mathbf{c}_k)_{\mathbf{x}}, (\mathbf{c}_k)_{\mathbf{z}}\}$ and $\{(\mathbf{c}_i)_{\mathbf{z}} (\mathbf{c}_j)_{\mathbf{z}}, (\mathbf{c}_j \mathbf{c}_k)_{\mathbf{x}} (\mathbf{c}_i \mathbf{c}_k)_{\mathbf{x}}, (\mathbf{c}_i \mathbf{c}_j)_{\mathbf{x}}\}$, respectively. For the simplicity of notations, we refer to the former boundary as the fold$|e$-boundary and the latter as the fold$|m$-boundary in the reminder of this paper, as these boundaries are composites of the fold of one copy of 3D toric code and the $e$-boundary or $m$-boundary of another copy, respectively.

Without loss of generality, we use explicity examples to show how to construct these two types of boundaries. For the fold$|e$-boundary, we consider the example when the condensation set is $\{y_{\mathbf{z}} g_{\mathbf{z}}, pg_{\mathbf{x}} py_{\mathbf{x}}, p_{\mathbf{z}}\}$. Just like the cases for 3D toric code, as we illustrated in Appendix~\ref{sec:gauging_folded}, we introduce ancilla qubits on the boundary to realize the condensation property\footnote{Technically speaking, adding ancilla qubits is not necessary; one can multiply $X$-stabilizers to eliminate the need for ancilla qubits. However, ancilla qubits can serve as degrees of freedom that are purely defined on the boundary. By applying local operators on the boundary, typically single $X$ or $Z$ operators, one can observe how the condensations emerge from the boundary and become gapped in the bulk.}. We introduce one ancilla qubit on each horizontal trancated yellow-green faces, as denoted by the white circles in the following figure.
\begin{align}
\adjincludegraphics[width=5cm,valign=c]{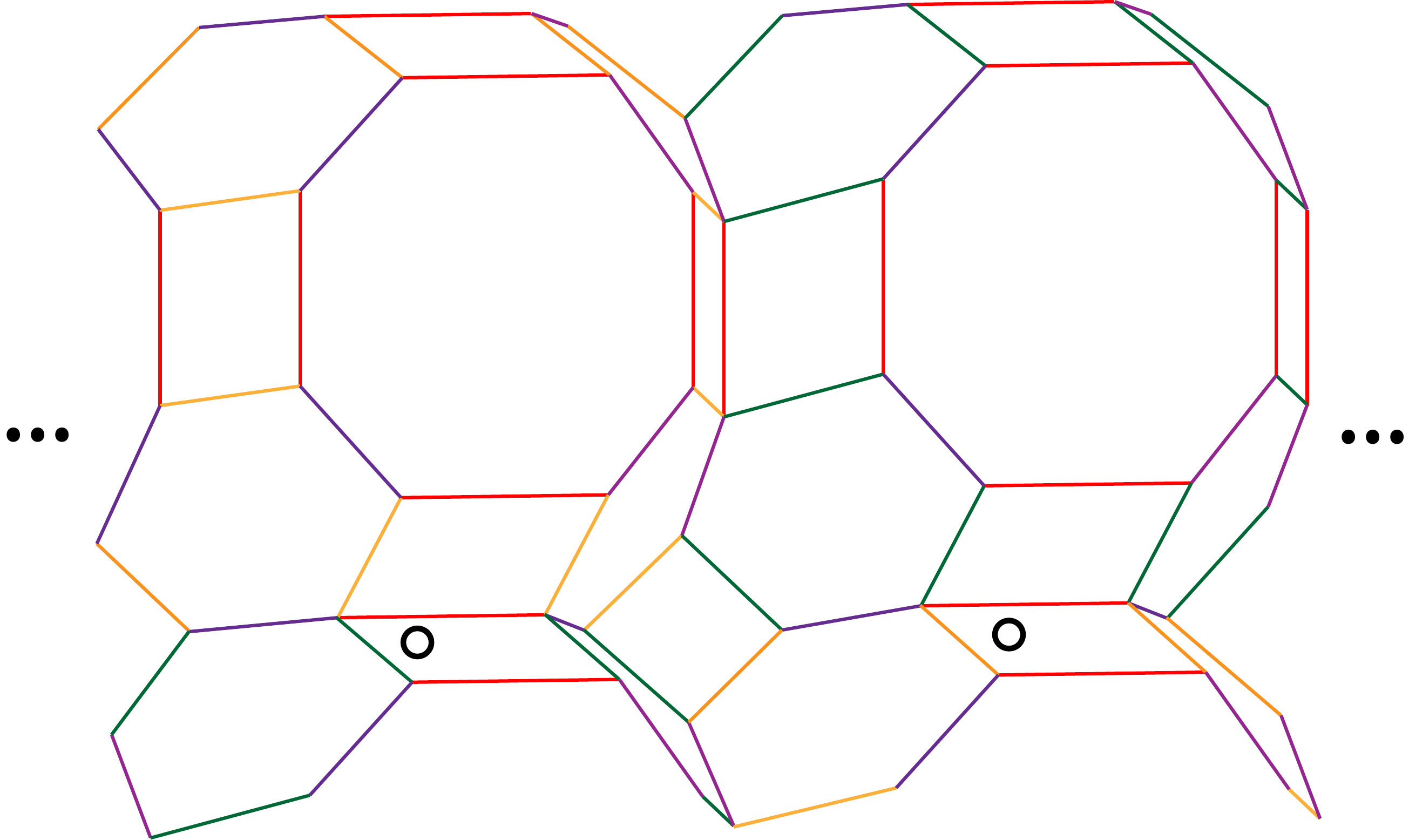}.
\end{align}
We define the $Z$-stabilizers on the boundary as
\begin{align}
\adjincludegraphics[width=5cm,valign=c]{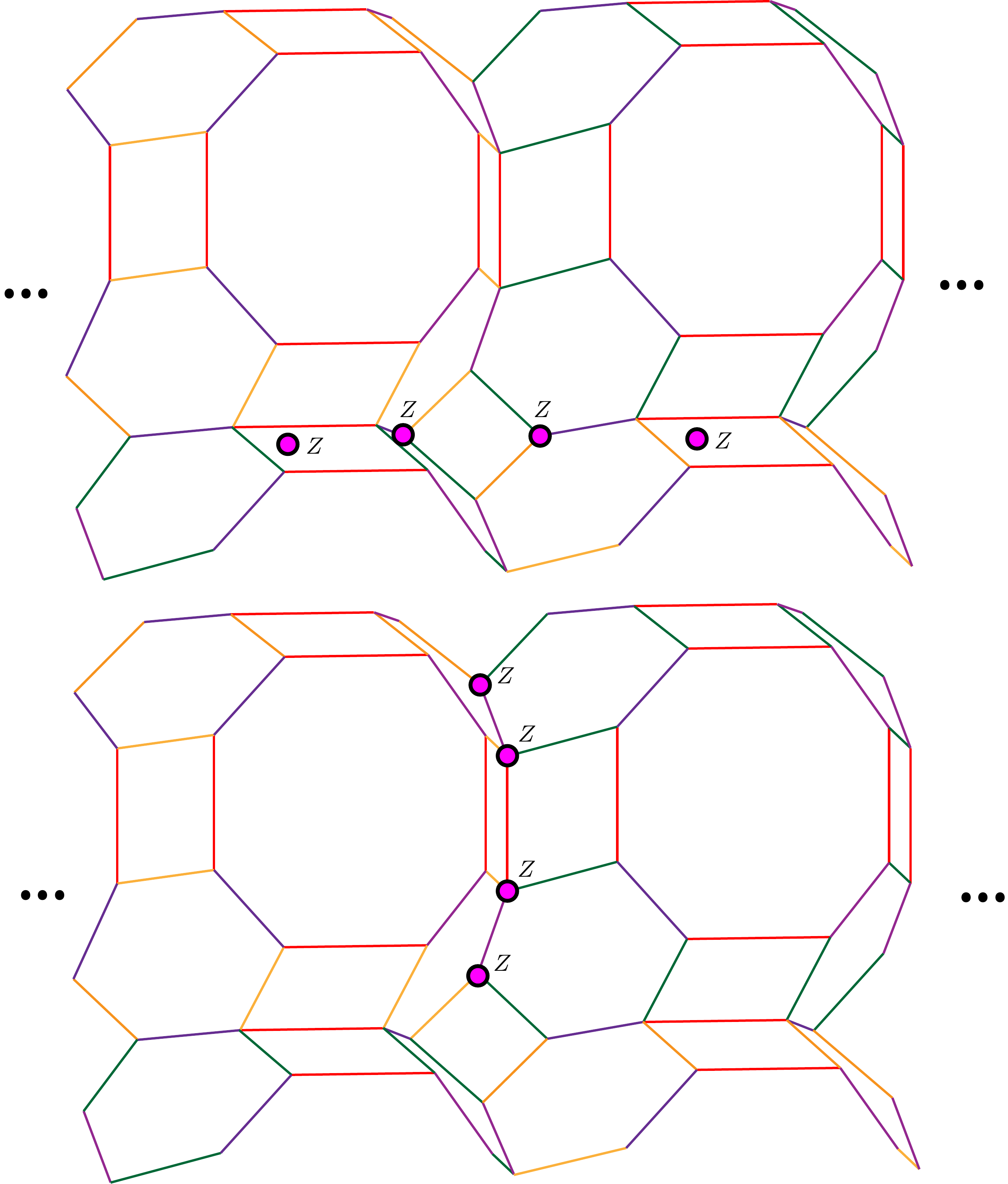}.
\end{align}
We then define the $X$-stabilizers as
\begin{align}
\adjincludegraphics[width=6cm,valign=c]{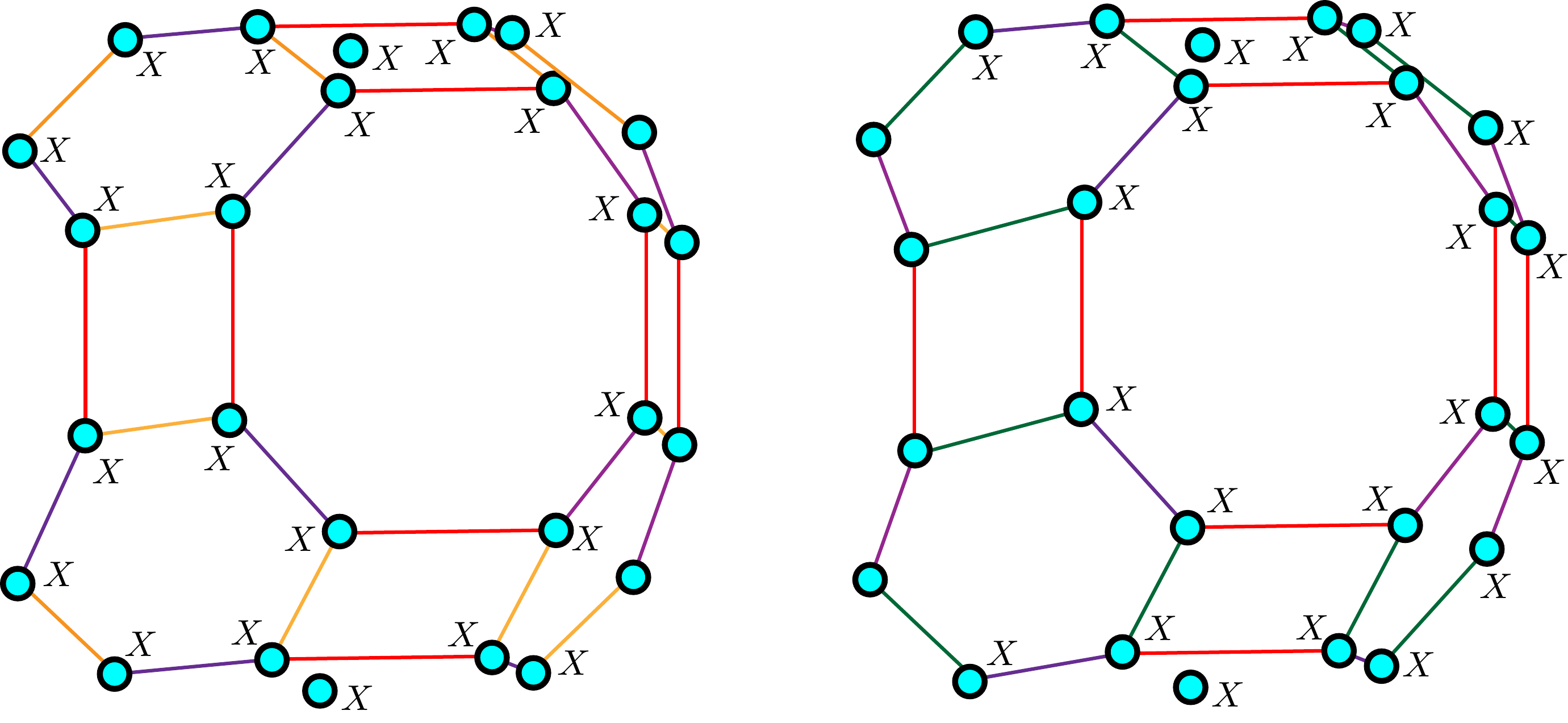}.
\end{align}
Notice that both stabilizers are coupled with the ancilla qubits on the top and the bottom. The first $Z$-stabilizer is given by the product of two $Z$-stabilizers from each copy. In principle, any local product is a generator in the stabilizer group. Therefore, the above figure is just one example that satisfies the mutual commutation condition. The physical meanings of these short $Z$-strings and their relation with the condensation of eletric charges can be stated as follows: In the bulk, the short $Z$-strings can generate pairs of electric charges at their endpoints. However, we add these terms on the boundary and only keep the terms that commute with them in the Hamiltonian. That is equivalent to say, the electric charges at the endpoints are condensed on the boundary, or identified with the vacuum since the corresponding operators commute with every other stabilizers on the boundary. Therefore, any local products of the corresponding truncated $Z$-terms can be regarded as a $Z$-stabilizer on the boundary, as long as they are coupled with the ancilla qubits to make sure every terms are mutually commuting. 

One can check the stabilizers above are mutually commuting, and commute with the string or membrane operators that generate corresponding types of excitations which can condense on it. The total number of distinct types of fold$|e$-boundaries is 3.

We then consider the $\{y_{\mathbf{z}} g_{\mathbf{z}}, pg_{\mathbf{x}} py_{\mathbf{x}}, yg_{\mathbf{x}}\}$-boundary as an example for the fold$|m$-boundary. The stabilizers on this boundary are similar to those of the previous type, except $\mathcal{B}^{boundary}_{yg}(Z)$ is substituted with $\mathcal{A}^{boundary}_{p}(X)$, which is
\begin{align}
\adjincludegraphics[width=5cm,valign=c]{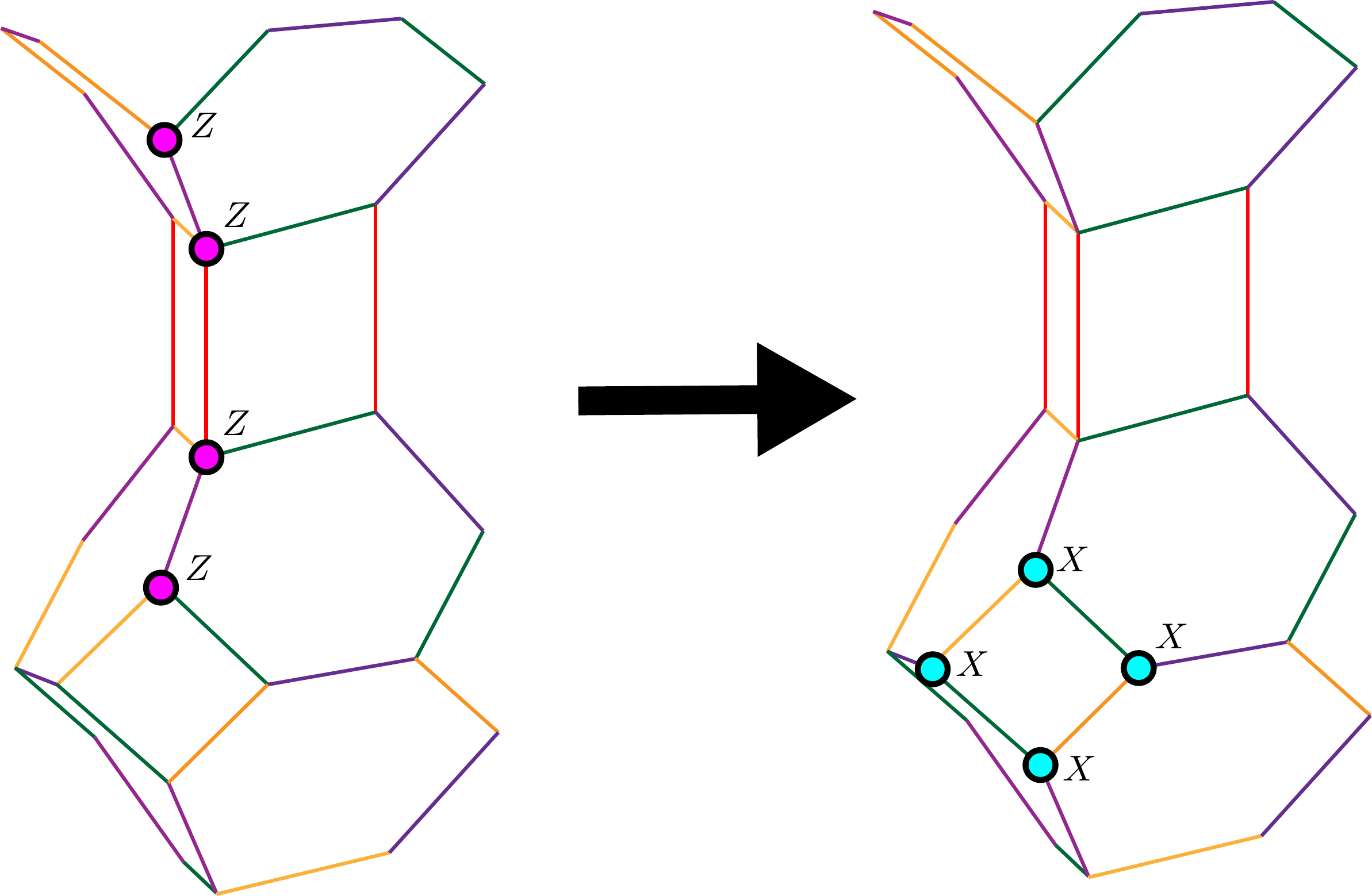}.
\end{align}
Therefore, the generating set of stabilizers on this boundary consists of the first stabilizer in Eq.~(24), the stabilizers in Eq.~(25), and the rightmost stabilizer in Eq.~(26).
The total number of distinct types of fold$|m$-boundaries is 3.

\subsubsection{The $pg_{\mathbf{x}} py_{\mathbf{x}} yg_{\mathbf{x}}$-boundary} \label{sec:m1m2m3}

The condensation set for the $pg_{\mathbf{x}} py_{\mathbf{x}} yg_{\mathbf{x}}$-boundary is given by $\{y_{\mathbf{z}} g_{\mathbf{z}}, g_{\mathbf{z}} p_{\mathbf{z}}, pg_{\mathbf{x}} py_{\mathbf{x}} yg_{\mathbf{x}}\}$. Similar to the folded-boundary cases, we need to couple ancilla qubits to the system to make the boundary stabilizers be mutually commuting. In general, we need to have one ancilla qubit for every pair of $\mathcal{A}^{boundary}_{y}(X)$, $\mathcal{A}^{boundary}_{g}(X)$ and $\mathcal{A}^{boundary}_{p}(X)$ stabilizers, as suggested in the 3D toric code model example in Appendix~\ref{sec:gauging_mmm}. However, due to the asymmetry of the lattice with respect to different colors, the numbers of $X$-stabilizers for each colors don't match. To solve this problem, we need to change the basis of generators for the $\mathcal{A}^{boundary}_{p}(X)$ and $\mathcal{B}^{boundary}_{yg}(Z)$ stabilizers, which are given by 
\begin{align}
\adjincludegraphics[width=6cm,valign=c]{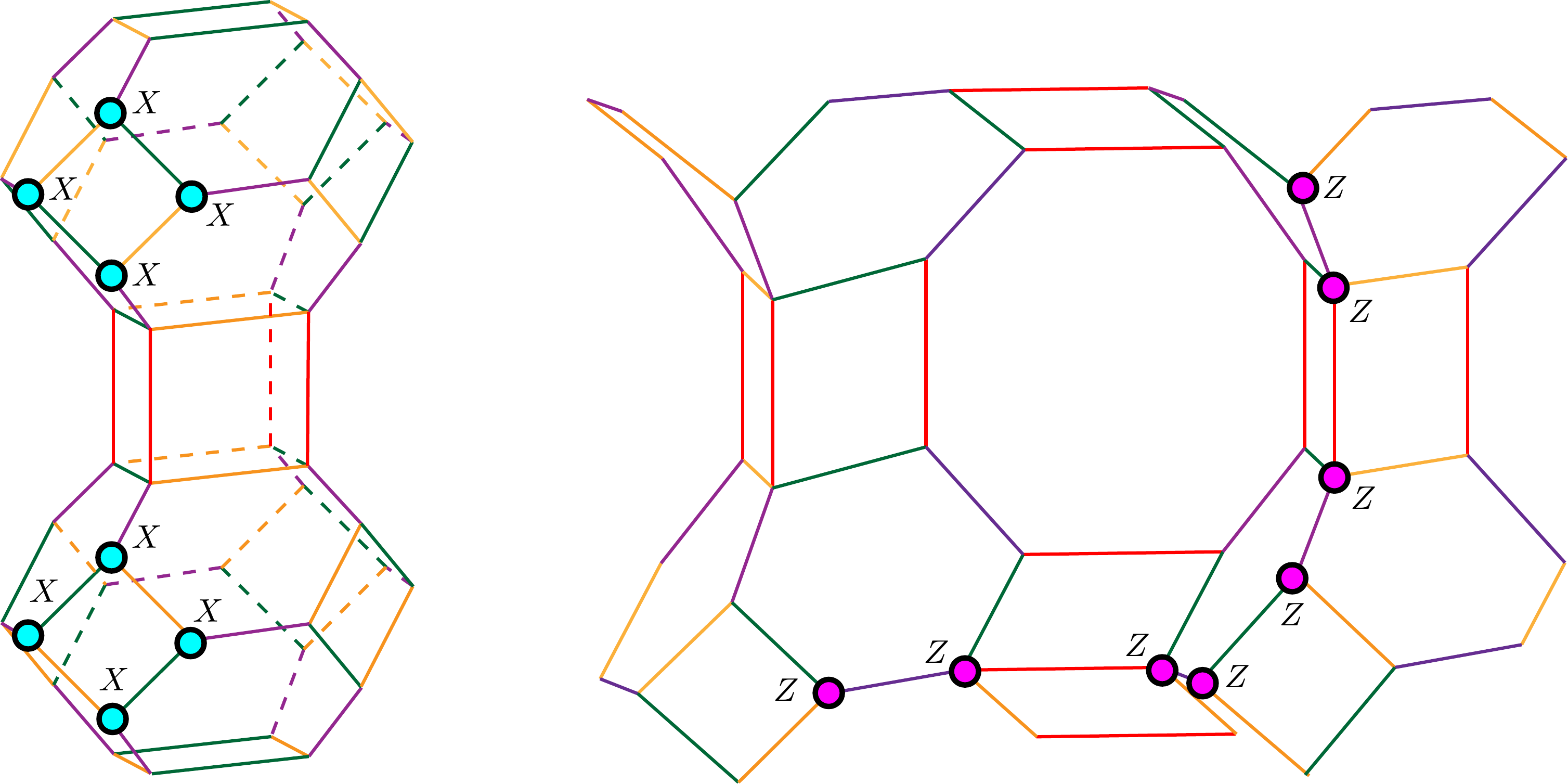}. \label{eq:larger_term}
\end{align}
One can check they anti-commute with each other, while commute with all the other stabilizers. Similar to the cases for the folded boundaries, we introduce ancilla qubits and couple them with the stabilizers. We introduce three ancilla qubits per set of $X$-stabilizers. Each set of $X$-stabilizers contains one $X$-stabilizer from each color. For the explicit forms, one may refer to Eq.~(25) and the stabilizer shown on the left side of Eq.~(27). For example, we add ancilla qubits as follows:
\begin{align}
\adjincludegraphics[width=6cm,valign=c]{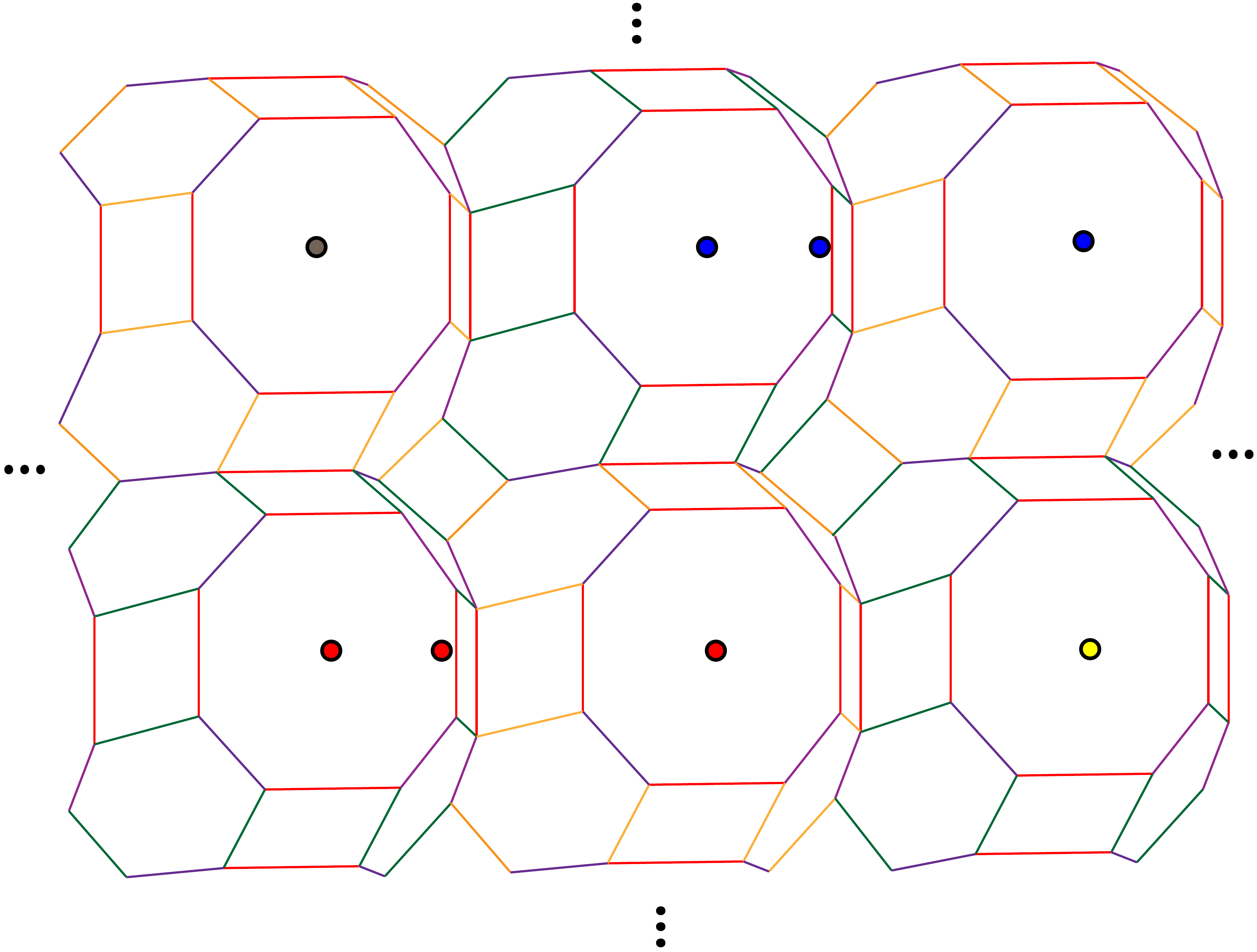}.
\end{align}
Here we use different colors to distinguish ancilla qubits for different sets. The $Z$-stabilizers can be defined as follow:
\begin{align}
\adjincludegraphics[width=7cm,valign=c]{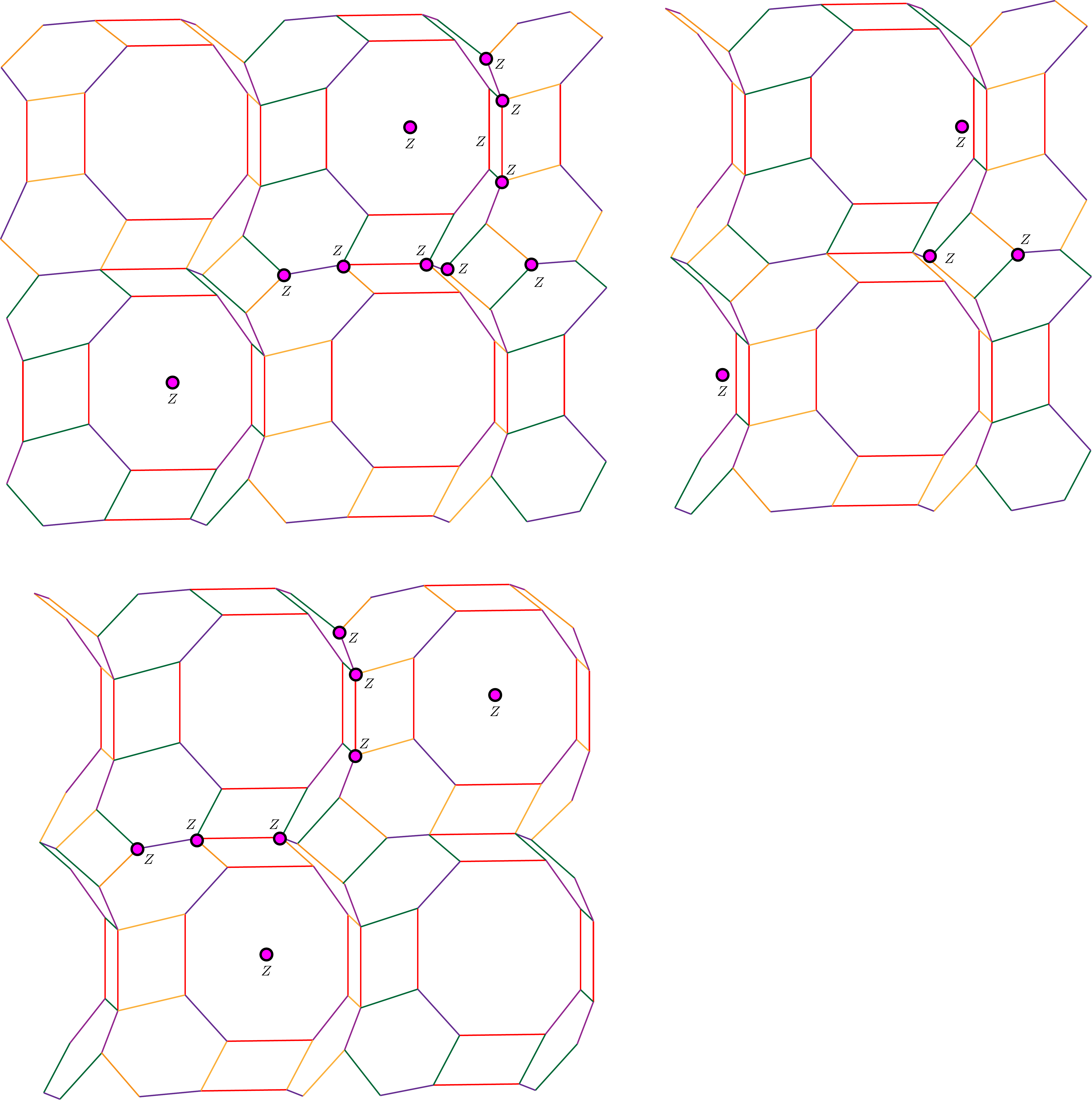}. \label{eq:bdry_51}
\end{align}
Other $Z$-stabilizers can be obtained by translations, and symmetric rotations. However, the translation must be taken in the diagonal directions to make sure the stabilizers are well defined. Similar to the cases of the folded boundaries, the $Z$-stabilizers are given by the product of two $Z$-stabilizers from different copies. As we argued in the previous subsection, any local product is in the stabilizer group. Therefore, the above figure is just one example that satisfies the mutual commutation condition.

We introduce two $X$-stabilizers on each truncated color cell. The $X$-stabilizers on the truncated purple cells are given by
\begin{align}
\adjincludegraphics[width=7cm,valign=c]{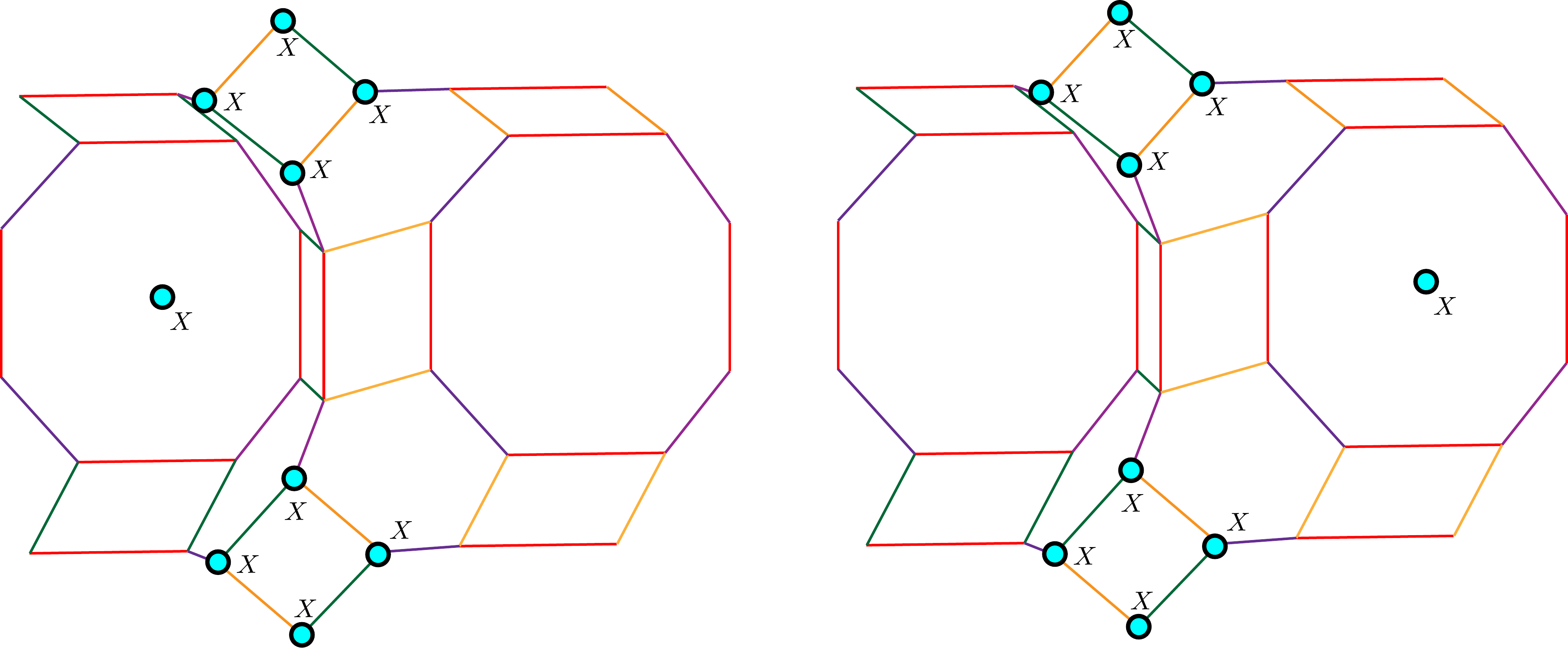}. \label{eq:bdry_41}
\end{align}
The $X$-stabilizers on the yellow cells are given by
\begin{align}
\adjincludegraphics[width=6cm,valign=c]{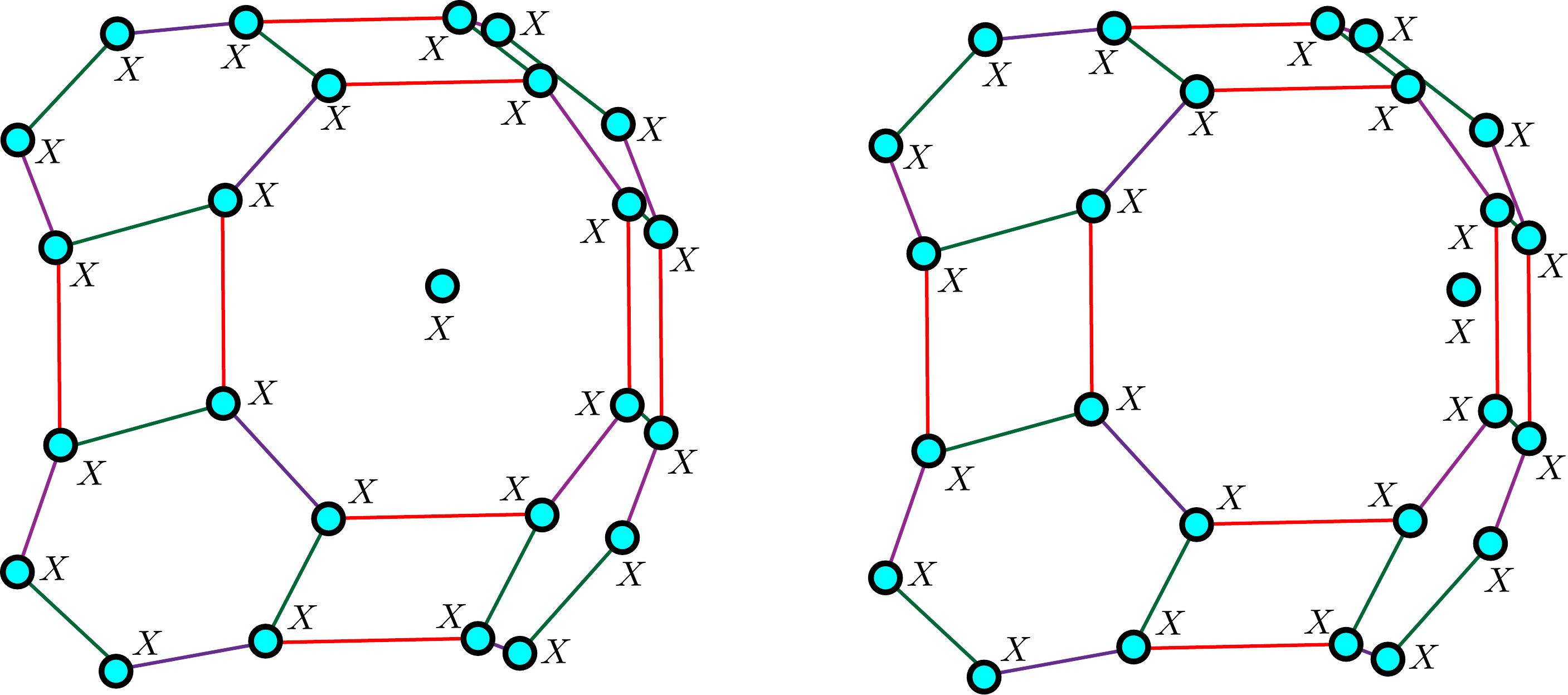}. \label{eq:bdry_42}
\end{align}
The $X$-stabilizers on the truncated green cells are given by
\begin{align}
\adjincludegraphics[width=6cm,valign=c]{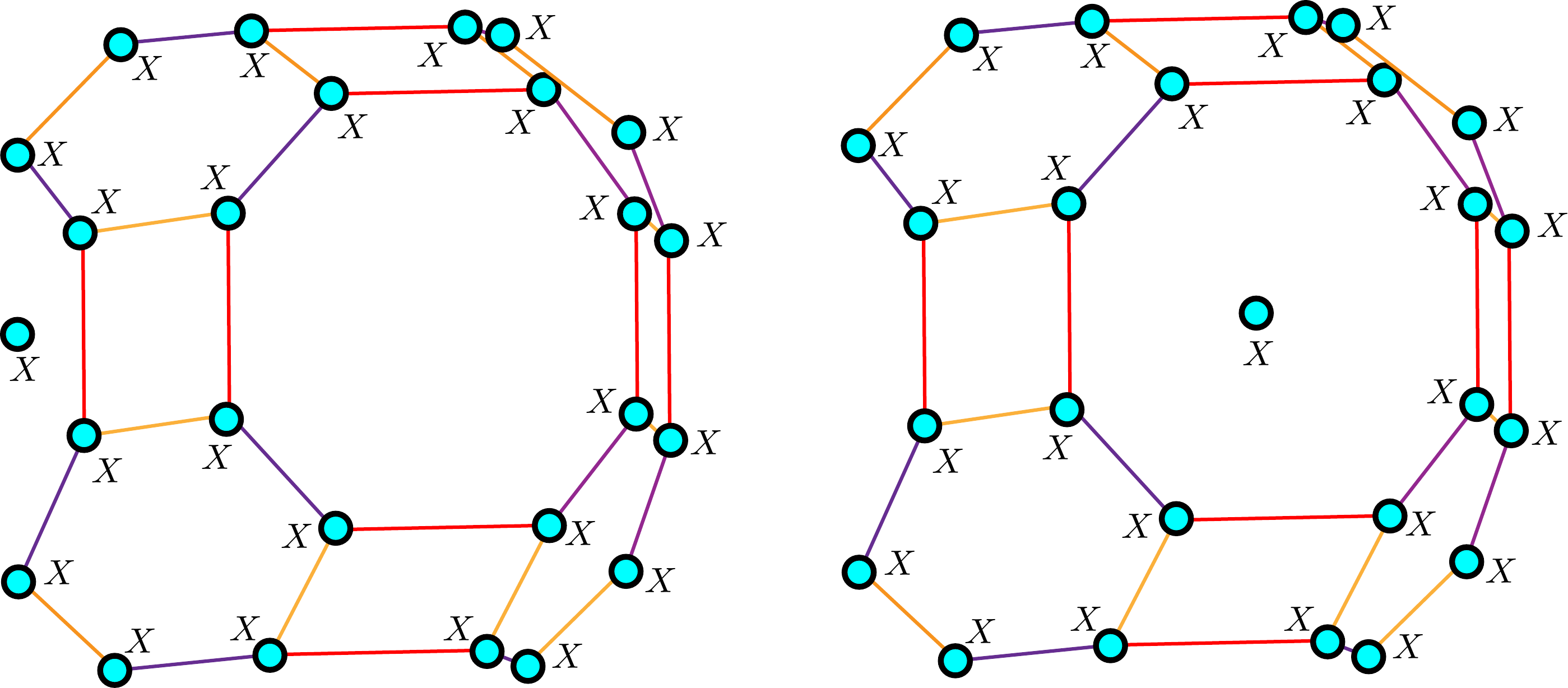}. \label{eq:bdry_43}
\end{align}
One can check those terms in Eq.~(\ref{eq:bdry_51}-\ref{eq:bdry_43}) are mutually commuting, and commute with the string or membrane operators that generate corresponding types of excitations which can condense on it. This boundary is equivalent to the $(e_1 e_2, e_2 e_3, m_1 m_2 m_3)$-boundary after we applying the unfolding unitaries. Namely,
\begin{align}
     \{y_{\mathbf{z}} g_{\mathbf{z}}, g_{\mathbf{z}} p_{\mathbf{z}}, pg_{\mathbf{x}} py_{\mathbf{x}} yg_{\mathbf{x}}\} \longleftrightarrow (e_1 e_2, e_2 e_3, m_1 m_2 m_3).
\end{align}
The total number of distinct types of $pg_{\mathbf{x}} py_{\mathbf{x}} yg_{\mathbf{x}}$-boundaries is 1. A 3D toric code version of this boundary is constructed in Appendix~\ref{sec:gauging_mmm}.

\subsubsection{The $y_{\mathbf{z}} g_{\mathbf{z}} p_{\mathbf{z}}$-boundary} \label{sec:e1e2e3}

The condensation set for the $y_{\mathbf{z}} g_{\mathbf{z}} p_{\mathbf{z}}$-boundary is given by $\{pg_{\mathbf{x}} py_{\mathbf{x}}, py_{\mathbf{x}} yg_{\mathbf{x}}, y_{\mathbf{z}} g_{\mathbf{z}} p_{\mathbf{z}}\}$. To construct the lattice Hamiltonian, we do the same trick as Eq.~\eqref{eq:larger_term} to match the number of stabilizers. However, in this case we couple one ancilla qubit per $X$-stabilizer set, as suggested in the 3D toric code model example in Appendix~\ref{sec:gauging_eee}. The ancilla qubits are placed as follows:
\begin{align}
\adjincludegraphics[width=6cm,valign=c]{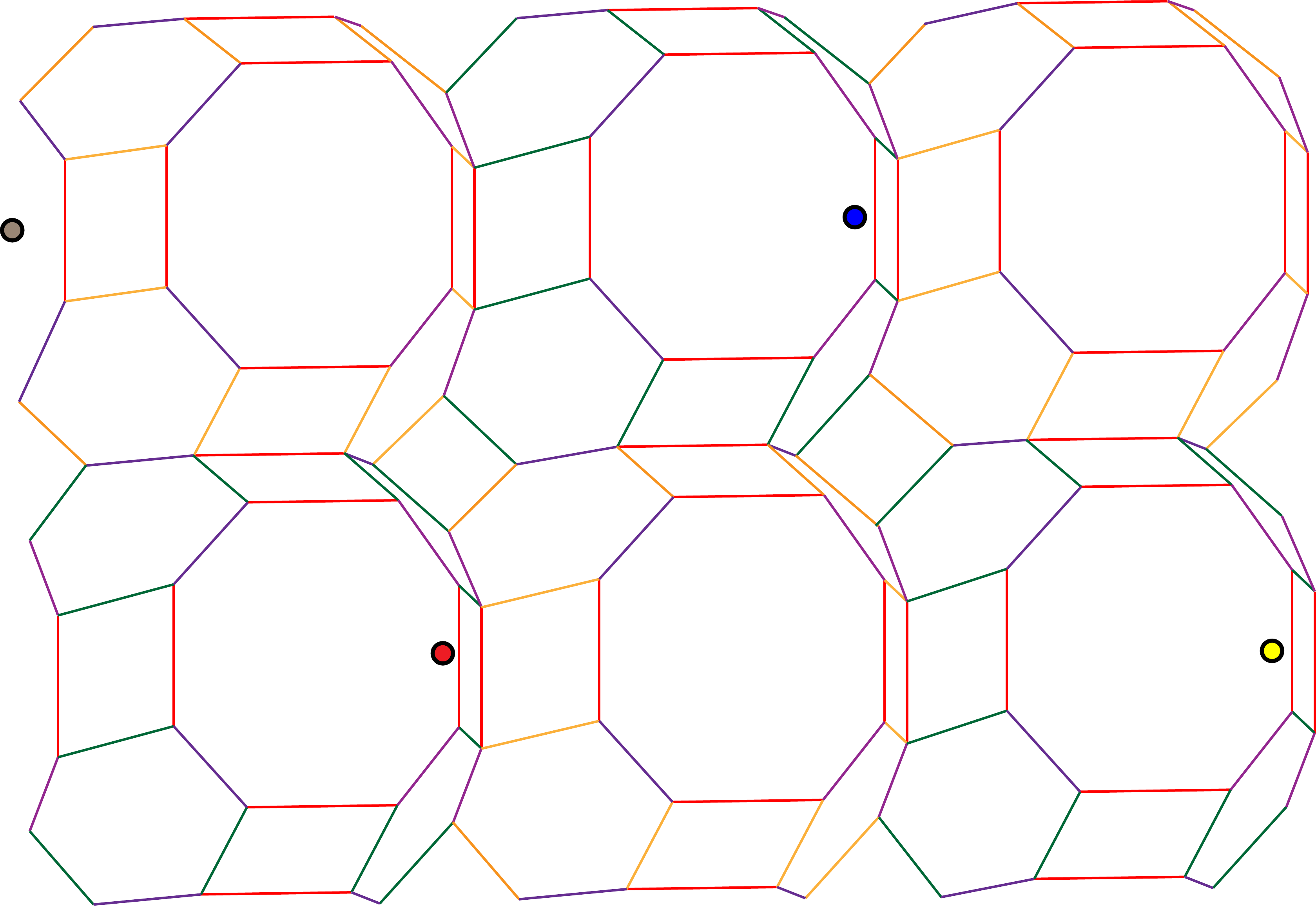}.
\end{align}
Different colors represent they are associated with different sets of $X$-stabilizers. The $Z$-stabilizers are defined as follows:
\begin{align}
\adjincludegraphics[width=3cm,valign=c]{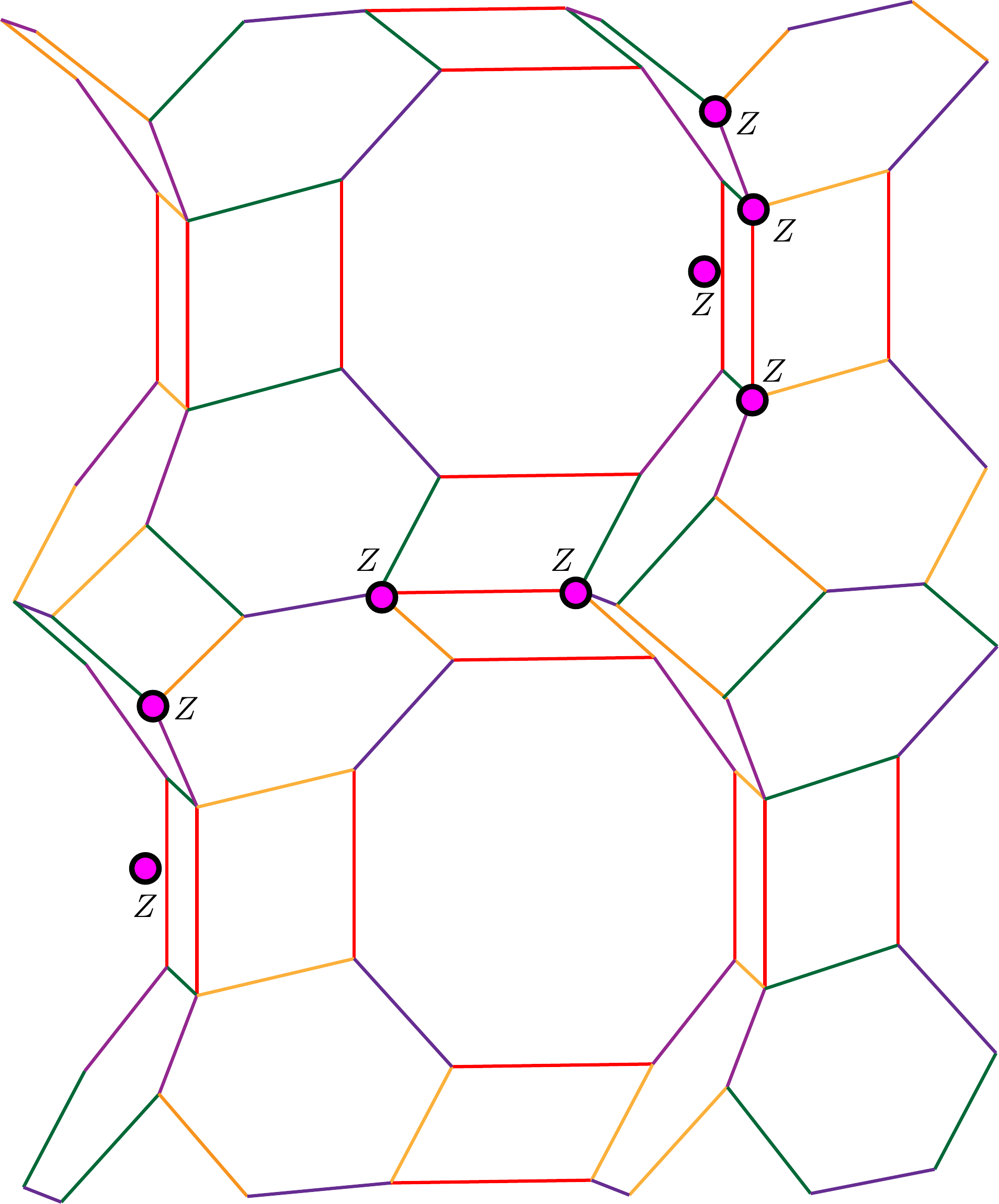},
\end{align}
which is the product of three types of $Z$-stabilizers and the ancilla qubits. Similar to the cases before, the $Z$-stabilizers are given by the product of two $Z$-stabilizers from different copies. In principle, any local product is in the stabilizer group. Therefore, the above figure is just one example that satisfies the mutual commutation condition.

The $X$-stabilizers on the truncated purple cells are given by
\begin{align}
\adjincludegraphics[width=2cm,valign=c]{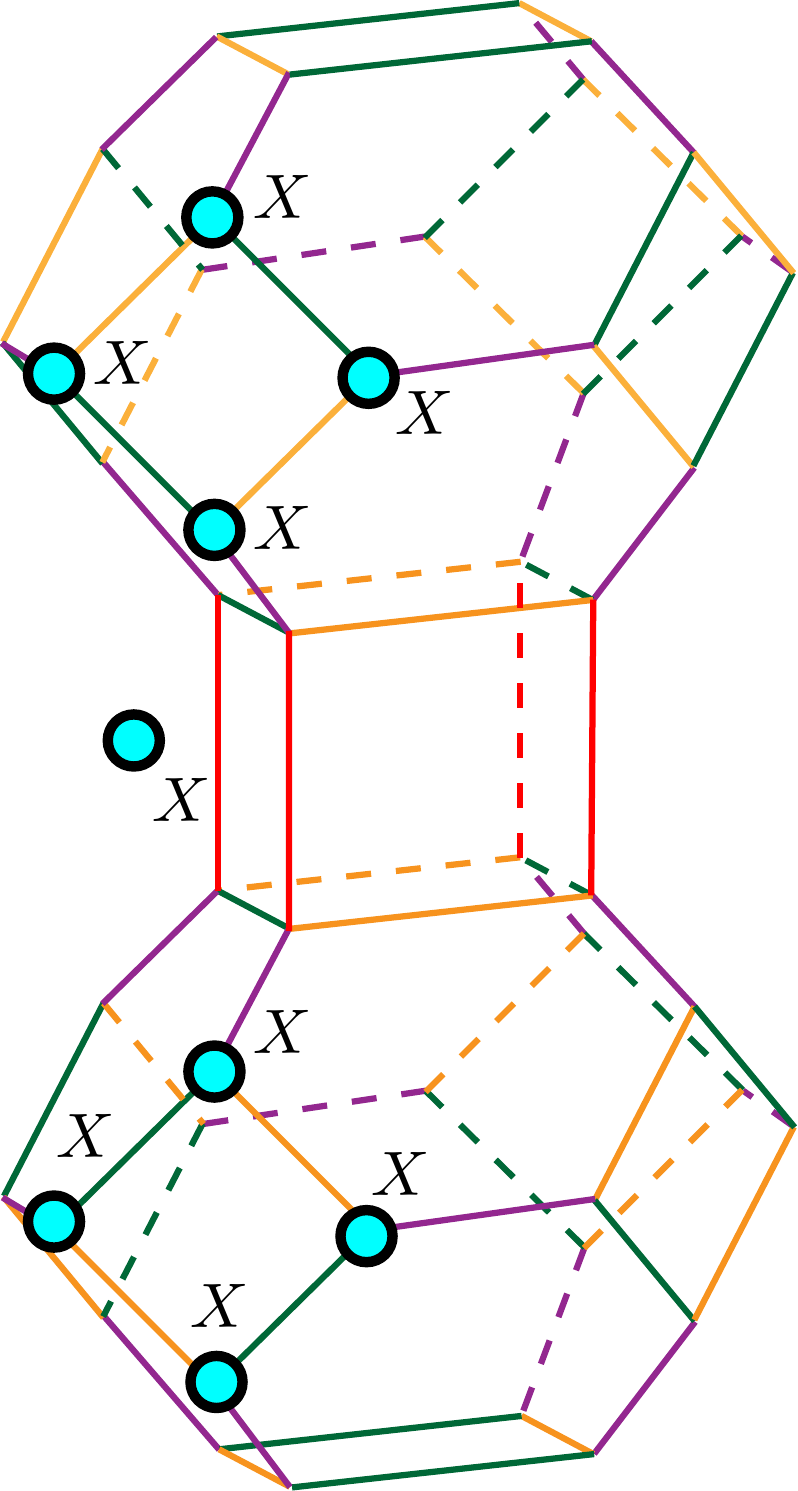}.
\end{align}
The $X$-stabilizers on the yellow and green cells are given by
\begin{align}
\adjincludegraphics[width=6cm,valign=c]{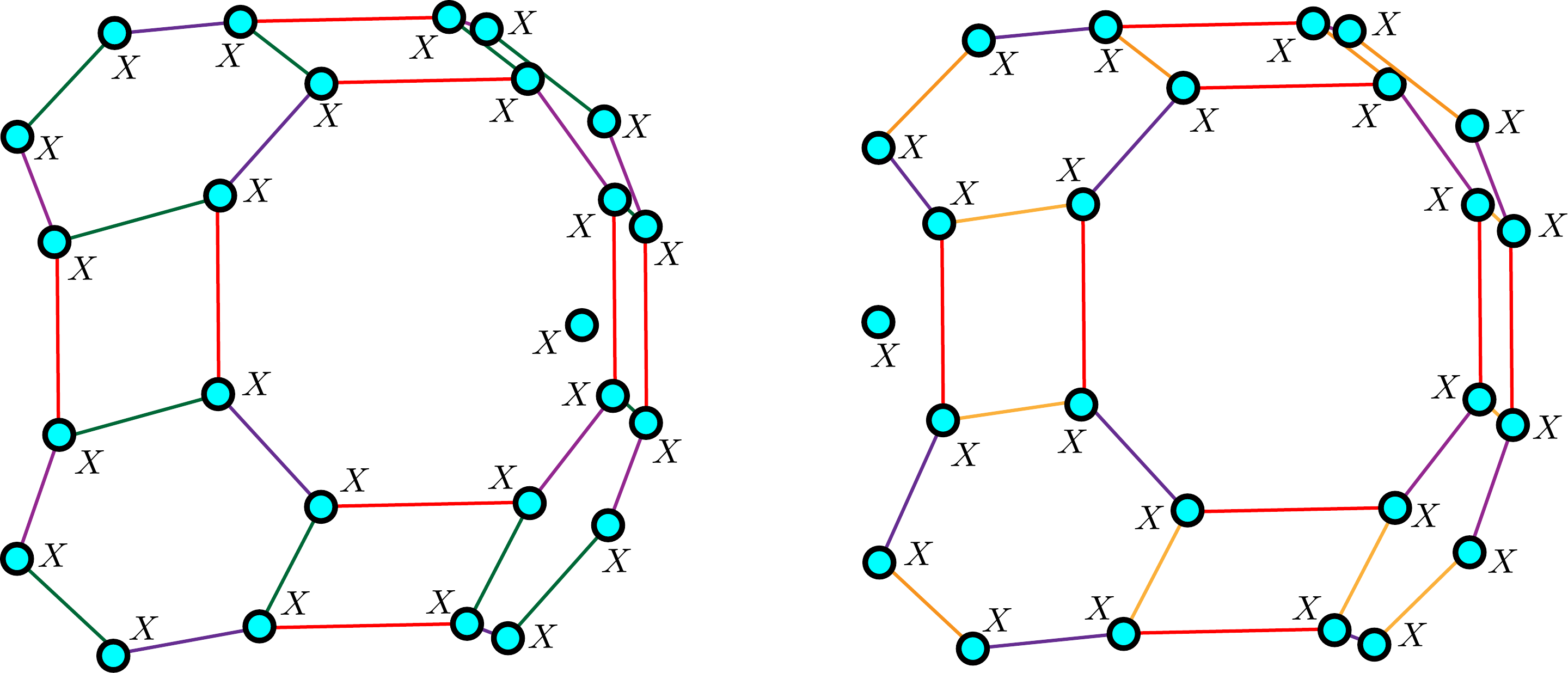}.
\end{align}
One can check those terms are mutually commuting, and commute with the string or membrane operators that generate corresponding types of excitations which can condense on it. This boundary is equivalent to the $(m_1 m_2, m_2 m_3, e_1 e_2 e_3)$-boundary after we applying the unfolding unitaries. Namely,
\begin{align}
      \{pg_{\mathbf{x}} py_{\mathbf{x}}, py_{\mathbf{x}} yg_{\mathbf{x}}, y_{\mathbf{z}} g_{\mathbf{z}} p_{\mathbf{z}}\} \longleftrightarrow (m_1 m_2, m_2 m_3, e_1 e_2 e_3)
\end{align}
The total number of distinct types of $y_{\mathbf{z}} g_{\mathbf{z}} p_{\mathbf{z}}$-boundaries is 1. A 3D toric code version of this boundary is constructed in Appendix~\ref{sec:gauging_eee}.

\subsubsection{Other boundaries with composite condensations} \label{sec:other_boundary_color}

In general, boundaries with composite condensations of $X$ and $Z$ types can also exist, provided they originate from different copies. Therefore, for the completeness of our discussion, we list all other types of boundaries and their counts for the readers' reference. The boundaries and their corresponding counts are as follows:
\begin{equation}
    \begin{aligned}
        &\{pg_{\mathbf{x}} g_{\mathbf{z}} p_{\mathbf{z}}, py_{\mathbf{x}} y_{\mathbf{z}} p_{\mathbf{z}}, yg_{\mathbf{x}} y_{\mathbf{z}} g_{\mathbf{z}}\}, \quad 1, \\
        &\{pg_{\mathbf{x}} py_{\mathbf{x}} y_{\mathbf{z}} g_{\mathbf{z}}, py_{\mathbf{x}} yg_{\mathbf{x}} y_{\mathbf{z}}, y_{\mathbf{z}} g_{\mathbf{z}} p_{\mathbf{z}}\}, \quad 1. \\
        &\{(\mathbf{c}_i)_{\mathbf{z}}, (\mathbf{c}_i \mathbf{c}_k)_{\mathbf{x}} (\mathbf{c}_k)_{\mathbf{z}}, (\mathbf{c}_i \mathbf{c}_j)_{\mathbf{x}}(\mathbf{c}_j)_{\mathbf{z}}\}, \quad 3, \\
        &\{(\mathbf{c}_j \mathbf{c}_k)_{\mathbf{x}} (\mathbf{c}_j)_{\mathbf{z}}, (\mathbf{c}_i \mathbf{c}_k)_{\mathbf{x}} (\mathbf{c}_i)_{\mathbf{z}}, (\mathbf{c}_i \mathbf{c}_j)_{\mathbf{x}}\}, \quad 3, \\
        &\{(\mathbf{c}_j \mathbf{c}_k)_{\mathbf{x}} (\mathbf{c}_j)_{\mathbf{z} }(\mathbf{c}_k)_{\mathbf{z}}, (\mathbf{c}_i \mathbf{c}_k)_{\mathbf{x}} (\mathbf{c}_i)_{\mathbf{z}}, (\mathbf{c}_i \mathbf{c}_j)_{\mathbf{x}}(\mathbf{c}_i)_{\mathbf{z}}\}, \quad 3, \\
        &\{(\mathbf{c}_i)_{\mathbf{z}} (\mathbf{c}_j)_{\mathbf{z}}, (\mathbf{c}_j \mathbf{c}_k)_{\mathbf{x}}(\mathbf{c}_i \mathbf{c}_k)_{\mathbf{x}}(\mathbf{c}_k)_{\mathbf{z}}, (\mathbf{c}_i \mathbf{c}_j)_{\mathbf{x}}(\mathbf{c}_j)_{\mathbf{z}}\}, \quad 3.
    \end{aligned}
\end{equation}

\subsubsection{Summary of the elementary boundary types}

In this section, we study the {\it elementary boundaries} of the 3D color code which can be classified by the Lagrangian subgroup. We explicitly construct the lattice model for the $X$-boundary (1), the $Z$-boundary (1), the color boundaries (3), the $\{(\mathbf{c}_i)_{\mathbf{z}}, (\mathbf{c}_j)_{\mathbf{z}}, (\mathbf{c}_i\mathbf{c}_j)_{\mathbf{x}}\}$-boundaries (3), the folded boundaries (6), the $pg_{\mathbf{x}} py_{\mathbf{x}} yg_{\mathbf{x}}$-boundary (1) and the $y_{\mathbf{z}} g_{\mathbf{z}} p_{\mathbf{z}}$-boundary (1). The number represents the distinct types of certain boundaries. Furthermore, we listed all the other elementary boundaries that are classified by the Lagrangian subgroup, which contains the composite of $Z$ and $X$ type excitations as condensations. The number of distinct types of those boundaries is 14. Therefore, the total number of distinct types of elementray boundaries of the 3D color code that are classified by the Lagrangian subgroup is 30.

\section{Boundaries from attaching domain walls} \label{sec:magic}

It has been studied that for the code space of $D$-dimensional topological stabilizer codes, the logical gate set generated by encoded gates implementable by constant-depth circuits is contained in the $D^\text{th}$ level of the {\it Clifford hierarchy}~\cite{gottesman1999quantum,haah2012logical,bravyi2013classification}. When $D = 3$, the logical gates can be applied are in the third level of Clifford hierarchy, such as the $\overline{T}$ gate~\cite{bombin2015gauge} or the $\overline{\mathrm{CCZ}}$ gate~\cite{kubica2015unfolding,vasmer2019three}.

In Ref.~\cite{zhu2022topological}, the authors showed that for a specific alignment of three copies of the 3D toric codes with the $(m_1, m_2, m_3)$-boundary, a transversal-CCZ gate maps this boundary to an exotic boundary which is different from either the usual $e$ or $m$ type boundaries. Neither $e$ nor $m$ type of topological excitations can condense on this boundary. This suggests that this type of boundary goes beyond the standard classification of gapped boundaries based on the Lagrangian subgroup. Specifically, we have
\begin{align}
   \widetilde{\mathrm{CCZ}}: (m_1, m_2, m_3) \to (m_1 s_{2,3}^{(2)}, m_2 s_{3,1}^{(2)}, m_3 s_{1,2}^{(2)}).
\end{align}
$s_{i,j}^{(2)}$ is a line defect, whose definition and detailed discussion can be found in Ref.~[13].
This boundary is not equivalent to the original all-smooth boundary.

The structure of this section is organized as follows. In subsection~\ref{sec:sweeping}, we review the domain wall sweeping picture for implementing (generalized) global symmetries. In subsection~\ref{sec:TQFT}, we discuss the classification of boundaries in $\mathbb{Z}_2^3$ gauge theories using the language of TQFT, which follows Ref.~\cite{zhu2022topological}. In subsection~\ref{sec:transversal}, we construct the explicit form of transversal gates on the 3D color code defined on our chosen lattice, and analyze their transformation under the application of transversal $T$ and $S$ gates. In subsection~\ref{sec:logical}, we study the action of logical operators on the $T$-conjugated ground state. In subsection~\ref{sec:magicboundary}, we construct the boundary Hamiltonian for the magic boundary and study its exotic properties. In subsection~\ref{sec:elementary}, we construct the boundary Hamiltonians for other types of elementary boundaries. Finally, in subsection~\ref{sec:nested}, we study the nested boundaries, which are boundaries that have codimension 1 condensation defects.

\subsection{Domain wall sweeping}
\label{sec:sweeping}

It has been understood that transversal gates, or more generally constant-depth local circuits, are equivalent to sweeping gapped domain walls across the system~\cite{yoshida2015topological,yoshida2016topological,yoshida2017gapped,webster2018locality, zhu2020quantum, zhu2022topological}. Consider a transversal gate $\widetilde{U}$. For $\widetilde{U}$ to also be a logical gate (denoted by $\overline{U}$), it should satisfy the following condition,
\begin{align}
    \widetilde{U}: \mathcal{H}_C \to \mathcal{H}_C, \label{eq:code_space}
\end{align}
in which $\mathcal{H}_C$ is the code space. The parent Hamiltonian of the system is the following.
\begin{align}
    H = H^{bulk} + H^{boundary},
\end{align}
where $H^{bulk}$ and $H^{boundary}$ are the bulk and boundary Hamiltonians. To satisfy the condition in Eq.~\eqref{eq:code_space}, the transversal gate $\widetilde{U}$ should satisfy the following condition
\begin{align}
    P_C \widetilde{U} H \widetilde{U}^{\dagger} P_C = P_C H P_C,
\end{align}
where $P_C$ is the projector onto the code space $\mathcal{H}_C$. For a translational invariant Hamiltonian, the above condition becomes
\begin{align}
    P_C \widetilde{U} H^{bulk} \widetilde{U}^{\dagger} P_C = P_C H^{bulk} P_C. \label{eq:bulk_invariant}
\end{align}
In this case, we say $\widetilde{U}$ corresponds to a certain type of \textit{emergent symmetries} for the  system, including 0-form symmetries which act on the entire system and higher-form ($q$-form) symmetries which act on a codimension-$q$ submanifold\footnote{Here we only consider the spatial dimensions}.  Note that a codimension-$q$ submanifold of a $d$-dimensional system corresponds to a dimension-$(d-q)$ submanifold.  An important property of a higher-form symmetry operator is that it is equivalent up to local deformation as long as its support, the codimension-$q$ submanifold, belongs to the same homotopy class.

Nevertheless, the existence of boundaries breaks the translational symmetry. To satisfy the condition in Eq.~\eqref{eq:code_space}, we further require
\begin{align}
    P_C \widetilde{U} H^{boundary} \widetilde{U}^{\dagger} P_C = P_C H^{boundary} P_C.
\end{align}
We denote the code space of the boundary of the system as $\mathfrak{B}$, defined as the code space stabilized by the Hamiltonian terms supported entirely on the boundary. To satisfy the condition in Eq.~\eqref{eq:code_space}, the transversal gate $\widetilde{U}$ must keep the boundary code space invariant. We have
\begin{align}
    \widetilde{U}: \mathfrak{B} \to \mathfrak{B}. \label{eq:code_space_boundary}
\end{align}

\begin{figure}
    \centering
    \includegraphics[width = 7cm]{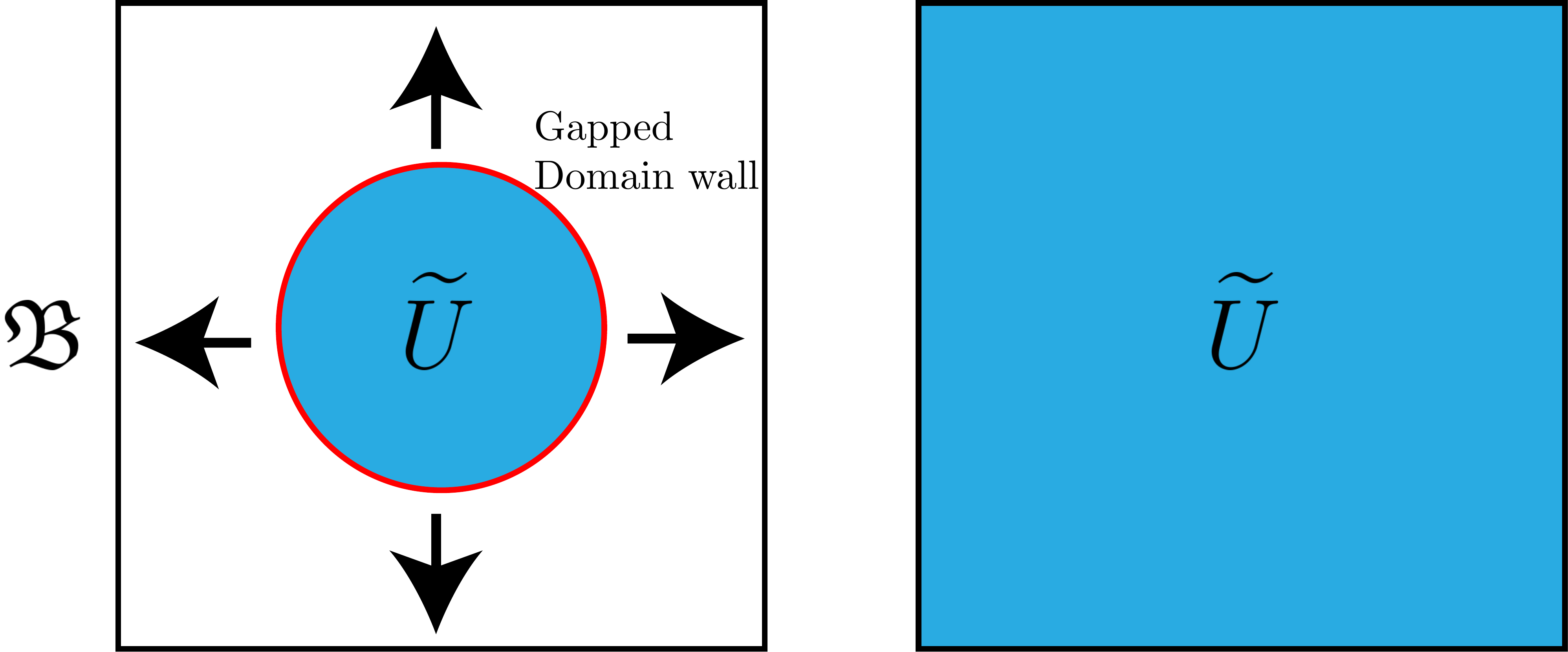}
    \caption{Illustration of domain wall sweeping. Consider a unitary operator acting transversally on a region $\mathcal{R}$ (The blue-colored region in the left figure). A gapped domain wall exists between the region $\mathcal{R}$ and its complement. If the transversal $\widetilde{U}$ operation is a logical gate of the code space, then as we sweep the domain wall across the system, the gapped boundary can condense on the boundary $\mathfrak{B}$. }
    \label{fig:sweeping}
\end{figure}
\begin{figure}
    \centering
    \includegraphics[width = 7cm]{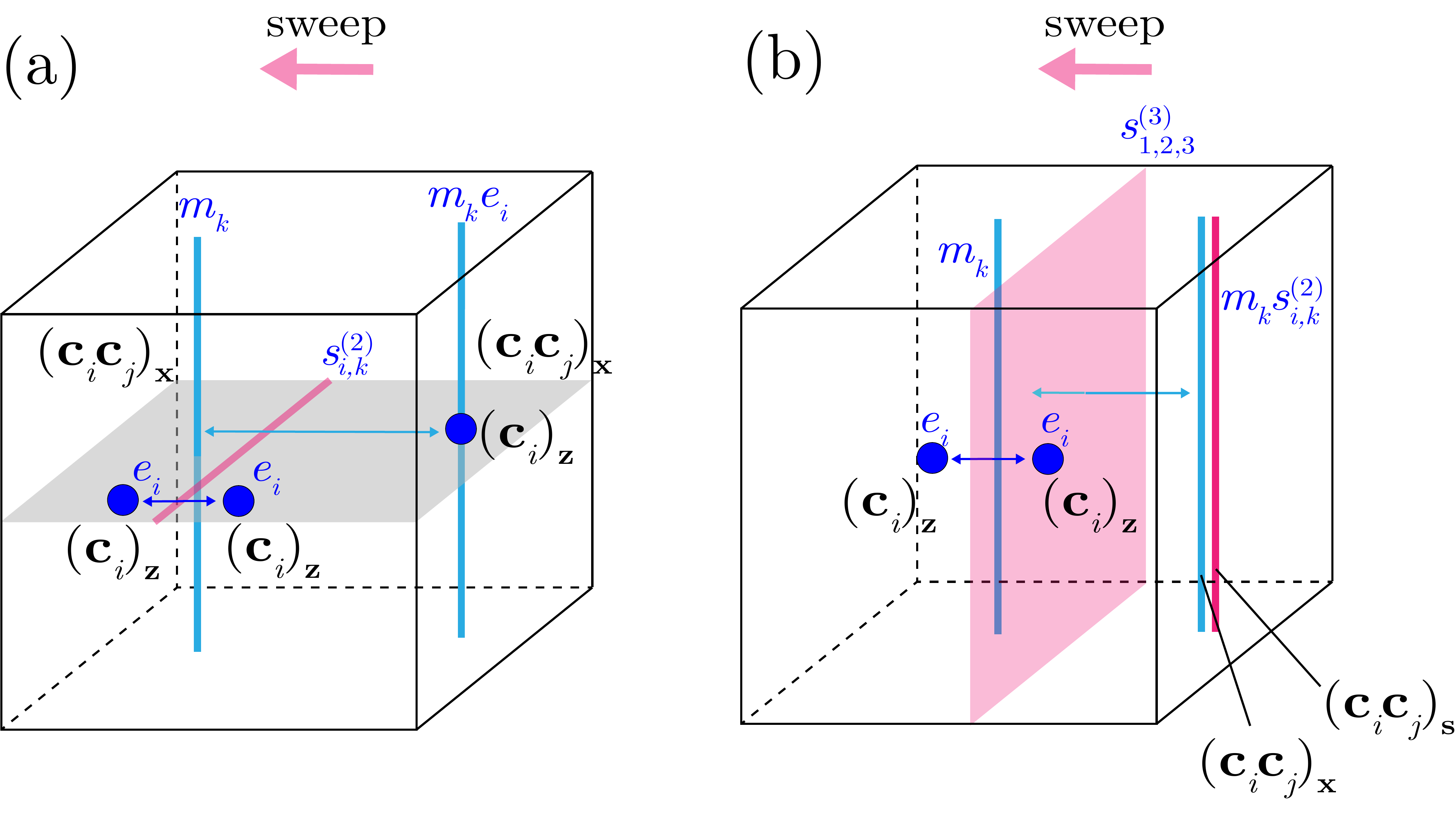}
    \caption{An illustration of how topological excitations of 3D color code map upon traversing the $S$-domain wall. We use two equivalent languages, the 3D color-code language and the 3+1D TQFT/ toric-code language, to illustrate these phenomena, and use black and blue colored fonts, respectively. (a) When sweeping the $S$-domain wall across a non-contractible codimension-1 submanifold labeled by colors $\mathbf{c}_i$ and $\mathbf{c}_j$, as depicted by the gray surface, electric charges remain unchanged, while magnetic flux loops become attached to electric charges as the $S$-domain wall is traversed. (b) When sweeping the $T$-domain wall across the entire system, electric charges remain the same, while the magnetic flux loops get attached to an $S$-domain wall as the $T$-domain wall is traversed. We denote this new object as $(\mathbf{c}_i \mathbf{c}_j)_{\mathbf{xs}}$.}
    \label{fig:two_domain_walls}
\end{figure}

The above statement can be alternatively interpreted by domain-wall sweeping. Consider $\widetilde{U}_{\mathcal{R}}$ to be the transversal unitary operation on region $\mathcal{R}$. As illustrated in Fig.~\ref{fig:sweeping}, one can show that there exists a domain wall along the boundary of the region $\mathcal{R}$. The transversal gate $\widetilde{U}$ is a logical gate on the code space if the corresponding domain wall can condense on the boundaries.

\subsection{An overview of domain walls and boundaries using the TQFT description} \label{sec:TQFT}

In Ref.~\cite{yoshida2015topological, yoshida2017gapped, webster2018locality, zhu2022topological, barkeshli2023codimension, zhu2023non}, it has been shown that there exists two types of emergent symmetries $\widetilde{U}$ and the corresponding domain walls in three copies of 3D toric codes, which is equivalent to a $\mathbb{Z}_2 \times \mathbb{Z}_2 \times \mathbb{Z}_2$ gauge theory in the TQFT language \cite{zhu2022topological, barkeshli2023codimension,  zhu2023non}.  

The first type of domain wall is the codimension-1 (2D) CCZ-domain wall in the context of 3D toric codes.  Sweeping the CCZ-domain wall across the whole system corresponds to applying a transversal-CCZ gate on all three copies of the toric codes, which gives rise to a 0-form emergent symmetry acting on the entire system.  In the TQFT language, we denote such a domain wall by $s_{1,2,3}^{(3)}$, where the susbscript means it couples toric-code copy 1, 2 and 3 while the superscript suggests the corresponding logical gate is in the third-level of the Clifford hierarchy.  As pointed out in Refs.~\cite{yoshida2015topological, yoshida2016topological, yoshida2017gapped, barkeshli2023codimension, barkeshli2022higher, zhu2023non}, $s_{1,2,3}^{(3)}$ is a codimension-1 gauged symmetry-protected topological (SPT) defect generated by decorating a (3+1)D invertible phase with a (2+1)D $\mathbb{Z}_2 \times \mathbb{Z}_2 \times \mathbb{Z}_2$ SPT (also called a 2D cluster state) and then gauging the global $\mathbb{Z}_2 \times \mathbb{Z}_2 \times \mathbb{Z}_2$ symmetry\footnote{Here, (3+1)D and (2+1)D refer to the space-time dimension.}.  

The second type is a codimension-2 (1D) CZ-domain wall in the toric-code context.  Sweeping the CZ-domain wall on a codimension-1 (2D) submanifold corresponds to applying a transversal-CZ gate between two copies of the toric codes, which gives rise to a 1-form emergent symmetry acting only on this submanifold. In the TQFT language, we denote such a domain wall by $s^{(2)}_{i,j}$, where the subscript means it couples toric-code copy $i$ and $j$ while the subscripts suggest that the corresponding logical gate is in the second-level of the Clifford hierarchy.  As shown in Refs.~\cite{yoshida2015topological, yoshida2016topological, yoshida2017gapped, barkeshli2023codimension, barkeshli2022higher}, $s^{(2)}_{i,j}$ is a codimension-2 gauged SPT defect generated by decorating a (3+1)D invertible phase with a (1+1)D $\mathbb{Z}_2 \times \mathbb{Z}_2 $ SPT (also called a 1D cluster state) and then gauging the global $\mathbb{Z}_2 \times \mathbb{Z}_2 $ symmetry. 

Note that both $s_{1,2,3}^{(3)}$ and $s^{(2)}_{i,j}$ are invertible defects with them being their own inverse and satisfies the $\mathbb{Z}_2$ fusion rules: 
\begin{equation}
s^{(3)}_{1,2,3}\times s^{(3)}_{1,2,3} = \mathbb{I}, \quad s^{(2)}_{i,j}\times s^{(2)}_{i,j} = \mathbb{I}.
\end{equation}

When the excitations in the model pass the codimension-2 domain wall  $s^{(2)}_{i,j}$, they are transformed as
\begin{align}\label{eq:s2_transform}
s^{(2)}_{i,j}: m_i \rightarrow m_i e_j, \  m_j \rightarrow e_i m_j,
\  e_i \rightarrow e_i, \  e_j \rightarrow e_j,
\end{align} 
where we have expressed transformation as a mapping of the domain wall $s^{(2)}_{i,j}$ on the excitations.  This is also illustrated in Fig.~\ref{fig:two_domain_walls}(a).  

When the excitations pass the codimension-1 domain wall $s^{(3)}_{1,2,3}$, they are transformed as
\begin{align}\label{eq:s3_transform}
\nonumber s^{(3)}_{1,2,3}:& m_1 \rightarrow m_1 s^{(2)}_{2,3}, m_2 \rightarrow m_2 s^{(2)}_{3,1}, m_3 \rightarrow m_3 s^{(2)}_{1,2}, \\
&  e_1  \rightarrow e_1,   e_2  \rightarrow e_2,  e_3  \rightarrow e_3.
\end{align}
As we can see, both  $s^{(2)}_{i,j}$ and $s^{(3)}_{1,2,3}$ induce no transformation on the $e$-particles. This is due to the fact that both $s^{(2)}_{i,j}$ and $s^{(3)}_{1,2,3}$ can be viewed as a condensation defect of $e$ particles, as pointed out by Ref.~\cite{barkeshli2023codimension}.  On the other hand,  $s^{(2)}_{i,j}$ attaches an additional $e_j$-particles to the $m_i$-flux, while $s^{(3)}_{1,2,3}$ attaches a domain wall  $s^{(2)}_{j,k}$ to the $m_i$-flux.

\subsubsection{The $(e_i, m_j, m_k)$-boundary}

Now we investigate the condensation properties of these domain walls on the boundaries.  In Ref.~\cite{zhu2022topological}, it has been pointed out that both $s^{(2)}_{i,j}$ and $s^{(3)}_{1,2,3}$ can condense on the $(e_i, m_j, m_k)$-boundary. The condensation properties of  $s^{(2)}_{i,j}$ can be shown as
\begin{align}\label{eq:condensation_equivalence}
\nonumber s^{(2)}_{i,j}:(e_i, m_j, m_k) &\rightarrow (e_i, e_i m_j, e_k) \equiv (e_i, m_j, m_k) \\
s^{(2)}_{i,k}:(e_i, m_j, m_k) &\rightarrow (e_i , m_j, e_i m_k) \equiv (e_i, m_j, m_k),
\end{align}
as illustrated in Fig.~\ref{fig:s_domain_wall_2}(b) in the later subsection. Here, the `$\equiv$' on the right side of the above equations essentially corresponds to a different basis choice of the same condensations, which can also be considered as a different choice of the generator set of the same Lagrangian subgroup.  In particular, we have used the fusion rule 
$m_j e_i \times e_i = m_j$.  Since $s^{(2)}_{i,j}$ and $s^{(2)}_{i,k}$ condense on the $(e_i, m_j, m_k)$-boundary, we can introduce the redundant notation for this boundary $(e_i, m_j, m_k)\equiv  (e_i, m_j, m_k, s^{(2)}_{i,j}, s^{(2)}_{i,k})$ by adding the codimension-2 domain walls also into the condensations.  

The condensation properties of  $s^{(3)}_{1,2,3}$  can be shown as
\begin{align}\label{eq:CCZ_condensation}
\nonumber &s^{(3)}_{1,2,3} : \  (e_i, m_j, m_k)\equiv  (e_i, m_j, m_k, s^{(2)}_{i,j}, s^{(2)}_{i,k})  \\
\nonumber 		    \rightarrow & (e_i, m_j s^{(2)}_{i,k}, m_k s^{(2)}_{i,j}, s^{(2)}_{i,j}, s^{(2)}_{i,k}) \equiv (e_i, m_j, m_k, s^{(2)}_{i,j}, s^{(2)}_{i,k})  \\
 		    \equiv &  (e_i, m_j, m_k), 
\end{align}
where the `$\equiv$' in the second line comes from the fusion rule is due to the fact that two invertible domain walls fuse into vacuum $s^{(2)}_{i,k}\times s^{(2)}_{i,k} = \mathbb{I}$ and hence $m_j s^{(2)}_{i,k}$$\times$$s^{(2)}_{i,k} = m_j$.   Therefore, the $s^{(3)}_{1,2,3}$-domain wall also condenses on the $(e_i, m_j, m_k)$-boundary.

With the development of the TQFT description in Ref.~\cite{barkeshli2023codimension}, we can now also understand the above condensation properties, which are deduced from simple algebraic manipulations, by using the perspective of the $\mathbb{Z}_2^3$ gauge theory. One can obtain the $\mathbb{Z}_2$ gauge theory (toric code) by gauging the $\mathbb{Z}_2$ global symmetry of an invertible paramgetic phase of a transversal Ising model  \cite{barkeshli2023codimension, kubica2018ungauging, tantivasadakarn2024long}. The $e$-boundary in the ungauged model (Ising model) spontaneously breaks the $\mathbb{Z}_2$ global  symmetry, while the $m$-boundary preserves this symmetry. This asymmetric boundary properties  comes from the fact that $e$ is the $\mathbb{Z}_2$ symmetry charge in the  $\mathbb{Z}_2$ gauge theory, while $m$ is the symmetry flux. This properties caries over to the $\mathbb{Z}_2^3$ gauge theory, which is just three copies of $\mathbb{Z}_2$ gauge theories (toric codes).   Since  $s^{(2)}_{i,j}$ is a gauged symmetry-protected topological (SPT) defect according to Refs.~\cite{yoshida2015topological, yoshida2016topological, yoshida2017gapped, barkeshli2023codimension, barkeshli2022higher, zhu2023non}, the corresponding $\mathbb{Z}_2 \times \mathbb{Z}_2$ symmetry will be spontaneously broken when they reach a boundary containing a copy of $e_i$- or $e_j$-boundary, which means $s^{(2)}_{i,j}$ is trivialized and hence condenses on the boundary. This is the reason that $s^{(2)}_{i,j}$ and $s^{(2)}_{i,k}$ condense on the $(e_i, m_j, m_k)$-boundary. Similarly, $s^{(3)}_{1,2,3}$ is also a gauged SPT defect, the corresponding $\mathbb{Z}_2 \times \mathbb{Z}_2 \times \mathbb{Z}_2$ symmetry will be spontaneously broken when it reaches a boundary containing any $e_1$, $e_2$ or $e_3$ boundary.   This is the reason why $s^{(3)}_{1,2,3}$ can condense on the $(e_i, m_j, m_k)$-boundary.

In addition, we have the following mapping:
\begin{align}
    s_{j,k}^{(2)}: (e_i, m_j, m_k) \to (e_i, m_j e_k, m_k e_j). \label{eq:sij_2}
\end{align}
Here, a codimension-2 (1D) boundary $(e_i, m_j e_k, m_k e_j)$ is produced at the interface of two codiemsnion-1 (2D) boundaries $(e_i, m_j, m_k)$. We dub this type of boundaries as the {\it nested boundaries}, as illustrated in Fig.~\ref{fig:s_domain_wall_2}(a) in the later subsection.

From the TQFT perspective, the $\mathbb{Z}_2 \times \mathbb{Z}_2$ symmetry of the $j^{\mathrm{th}}$ and the $k^{\mathrm{th}}$ copies is preserved. Therefore, the $s_{j,k}^{(2)}$ domain wall cannot condense on this boundary.

\subsubsection{The $(e_i, e_j, m_k)$- and $(e_1, e_2, e_3)$-boundaries}

Besides the example of $(e_i, m_j, m_k)$ discussed in Ref.~\cite{zhu2022topological}, we also consider the $(e_i, e_j, m_k)$ and the $(e_1, e_2, e_3)$ boundaries in this paper. According to the above discussion of the spontaneous $\mathbb{Z}_2$ symmetry breaking of the $e$-boundary, we can conclude that $s^{(2)}_{i,j}$, $s^{(2)}_{i,k}$, $s^{(2)}_{j,k}$ and $s^{(3)}_{1,2,3}$ will condense on the $(e_i, e_j, m_k)$ and $(e_1, e_2, e_3)$ boundaries. One can check the condensation properties of the $(e_i, e_j, m_k)$ with the algebraic manipulations shown before:
\begin{align}
\nonumber &s^{(2)}_{i,j}:(e_i, e_j, m_k) \rightarrow (e_i, e_j, m_k)\\
\nonumber & s^{(2)}_{i,k}:(e_i, e_j, m_k) \rightarrow (e_i, e_j, e_i m_k) \equiv (e_i, e_j, m_k)\\
\nonumber & s^{(2)}_{j,k}:(e_i, e_j, m_k) \rightarrow  (e_i, e_j, e_j m_k) \equiv (e_i, e_j, m_k)\\
\nonumber &s^{(3)}_{1,2,3} : \  (e_i, e_j, m_k)\equiv  (e_i, e_j, m_k, s^{(2)}_{i,j})  \\
\nonumber 		    \rightarrow & (e_i, e_j, m_k s^{(2)}_{i,j},  s^{(2)}_{i,j}) \equiv  (e_i, e_j, m_k, s^{(2)}_{i,j})  \\
 		    \equiv &  (e_i, e_j, m_k). \label{eq:eem}
\end{align}
We illustrate the condensation property of $s^{(2)}_{i,j}$ in Fig.~\ref{fig:s_domain_wall_2}(c) in the later subsection.

For the $(e_1, e_2, e_3)$-boundaries, the condensation follows obviously from Eq.~\eqref{eq:s2_transform} and Eq.~\eqref{eq:s3_transform} since both $s^{(2)}_{i,j}$ and $s^{(3)}_{1,2,3}$ act trivially on the $e$-particles.

\subsubsection{The folded boundaries} \label{sec:folded_boundary_TQFT}

In this subsection we investigate the condensation properties of these domain walls on the folded boundaries. The condensation properties of $s_{i,j}^{(2)}$ on the $(e_i e_j, m_i m_j, e_k)$-boundary can be shown as
\begin{equation}
    \begin{aligned}
        s_{i,j}^{(2)}: (e_i e_j, m_i m_j, e_k) &\to (e_i e_j, m_i m_j e_i e_j, e_k) \\
        &\equiv (e_i e_j, m_i m_j, e_k) \\
        s_{j,k}^{(2)}: (e_i e_j, m_i m_j, e_k) &\to (e_i e_j, m_i m_j  e_k, e_k) \\
        &\equiv (e_i e_j, m_i m_j, e_k) \\
        s_{i,k}^{(2)}: (e_i e_j, m_i m_j, e_k) &\to (e_i e_j, m_i m_j  e_k, e_k) \\
        &\equiv (e_i e_j, m_i m_j, e_k)
    \end{aligned} \label{eq:S_fold_e}
\end{equation}
Therefore, we conclude that all types of $s_{i,j}^{(2)}$ domain walls can condense on the $(e_i e_j, m_i m_j, e_k)$-boundary.

The condensation properties of $s_{i,j}^{(2)}$ on the $(e_i e_j, m_i m_j, m_k)$-boundary can be shown as
\begin{equation}
    \begin{aligned}
        s_{i,j}^{(2)}: (e_i e_j, m_i m_j, m_k) &\to (e_i e_j, m_i m_j e_i e_j, m_k) \\
    &\equiv (e_i e_j, m_i m_j, m_k) \\
    s_{j,k}^{(2)}: (e_i e_j, m_i m_j, m_k) &\to (e_i e_j, m_i m_j e_k, m_k e_j) \\
    s_{i,k}^{(2)}: (e_i e_j, m_i m_j, m_k) &\to (e_i e_j, m_i m_j e_k, m_k e_i)
    \end{aligned} \label{eq:S_fold_m}
\end{equation}
Therefore, we conclude that $s_{i,j}^{(2)}$ can condense on the $(e_i e_j, m_i m_j, m_k)$-boundary, while $s_{j,k}^{(2)}$ and $s_{i,k}^{(2)}$ cannot. 

From the $\mathbb{Z}_2^3$ gauge theory perspective, the condensation of the $s_{i,j}^{(2)}$ domain wall on the $(e_i e_j, m_i m_j, m_k)$-boundary and the $(e_i e_j, m_i m_j, e_k)$-boundary can be understood as follows: Within the bulk of the $i^{\text{th}}$ and $j^{\text{th}}$ copies of the ungauged model (3D Ising model), a $\mathbb{Z}_2 \times \mathbb{Z}_2$ symmetry exists. However, this symmetry undergoes spontaneous breaking to the $\mathbb{Z}_2^{diag}$ diagonal symmetry at the folded boundary. This occurs because the two copies result from folding a single copy of the 3D Ising model; hence, the real symmetry of these copies is the $\mathbb{Z}_2^{diag}$ symmetry. This is just the global symmetry of the unfolded 3D Ising model. The $\mathbb{Z}_2 \times \mathbb{Z}_2$ symmetry associated with the $s_{i,j}^{(2)}$ domain wall thus undergoes spontaneous breaking. Consequently, the domain wall becomes trivialized and hence condenses on the boundary.

On the other hand, the reason why $s_{j,k}^{(2)}$ and $s_{i,k}^{(2)}$ domain wall cannot condense on the $(e_i e_j, m_i m_j, m_k)$-boundary can be understood as follows: Let's consider $s_{j,k}^{(2)}$ for example. As we discussed before, the actual symmetry on the $j^{\text{th}}$ copy is the $\mathbb{Z}_2^{diag}$ symmetry rather than the $\mathbb{Z}_2$ symmetry that only acts on that copy. For the $k^{th}$ copy, the boundary is an $m$-boundary, thereby preserving the $\mathbb{Z}_2$ symmetry. In the case of the $j^{th}$ copy, the $\mathbb{Z}_2^{diag}$ symmetry remains unbroken since the domain wall effectively remains within the bulk of the unfolded $i^{th}$ and $j^{th}$ copies of the 3D toric code. Since the $\mathbb{Z}_2^{diag} \times \mathbb{Z}_2$ symmetry of the $j^{th}$ and $k^{th}$ copies is preserved, the $s_{j,k}^{(2)}$ thus cannot condense on it; instead, it attaches to it.

For the condensation of $s_{j,k}^{(2)}$ and $s_{i,k}^{(2)}$ domain walls on the $(e_i e_j, m_i m_j, e_k)$-boundary, the reason can be understood as follows: The existence of the $e$-boundary of the $k^{\text{th}}$ copy spontaneously breaks the $\mathbb{Z}_2^{diag} \times \mathbb{Z}_2$ symmetry of the $s_{j,k}^{(2)}$ and $s_{i,k}^{(2)}$ domain walls. Therefore, they are trivialized on this boundary and hence can condense on it.

Furthermore, we investigate the condensation properties of $s_{i,j,k}^{(3)}$ on both types of the folded boundaries, which are shown as
\begin{equation}
    \begin{aligned}
        s_{i,j,k}^{(3)}: &(e_i e_j, m_i m_j, m_k) \equiv (e_i e_j, m_i m_j, m_k, s_{i,j}^{(2)}) \\
        &\to (e_i e_j, m_i m_j s_{j,k}^{(2)} s_{i,k}^{(2)}, m_k s_{i,j}^{(2)},s_{i,j}^{(2)}) \\
        &\equiv (e_i e_j, m_i m_j s_{j,k}^{(2)} s_{i,k}^{(2)}, m_k), \\
        s_{i,j,k}^{(3)}: &(e_i e_j, m_i m_j, e_k)
        \equiv (e_i e_j, m_i m_j, e_k, s_{j,k}^{(2)}, s_{i,k}^{(2)})\\
        &\to (e_i e_j, m_i m_j s_{j,k}^{(2)} s_{i,k}^{(2)}, e_k, s_{j,k}^{(2)}, s_{i,k}^{(2)}) \\
        &\equiv (e_i e_j, m_i m_j, e_k)
    \end{aligned} \label{eq:T_fold}
\end{equation}
For now, it seems like the $s^{(3)}_{i, j, k}$ cannot condense on the $(e_i e_j, m_i m_j, m_k)$-boundary since it leaves a $m_i m_j s_{j,k}^{(2)} s_{i,k}^{(2)}$ condensation that cannot condense on the initial boundary. However, we haven't checked the condensation property of the $s_{j,k}^{(2)} s_{i,k}^{(2)}$-domain wall yet. We have
\begin{equation}
    \begin{aligned}
        s_{j,k}^{(2)} s_{i,k}^{(2)}: &(e_i e_j, m_i m_j, m_k)
        \to s_{j,k}^{(2)}: (e_i e_j, m_i m_j e_k, m_k e_i) \\
        \to &(e_i e_j, m_i m_j, m_k e_i e_j) \equiv (e_i e_j, m_i m_j, m_k).\label{eq:double_s_1}
    \end{aligned}
\end{equation}
Here, we first apply the $s_{i,k}^{(2)}$ defect on the boundary, followed by the $s_{j,k}^{(2)}$ defect. The result shows that the composite of domain walls, $s_{j,k}^{(2)} s_{i,k}^{(2)}$, can condense on the boundary. Therefore, we have
\begin{equation}
    \begin{aligned}
        s_{i,j,k}^{(3)}: (e_i e_j, m_i m_j, m_k) &\to (e_i e_j, m_i m_j s_{j,k}^{(2)} s_{i,k}^{(2)}, m_k) \\
     &\equiv (e_i e_j, m_i m_j, m_k).
    \end{aligned} \label{eq:fold_m_condense}
\end{equation}
The $s_{i,j,k}^{(3)}$ domain wall can condense on both types of the folded boundaries.

From the TQFT perspective, the reason why the  $s_{i,j,k}^{(3)}$ domain wall can condense on the $(e_i e_j, m_i m_j, e_k)$-boundary is that the presence of the $e$-boundary spontaneously breaks the global symmetry. Therefore, the domain wall can condense on it. Regarding the $(e_i e_j, m_i m_j, m_k)$-boundary, the symmetry of the system encounter a symmetry breaking $\mathbb{Z}_2 \times \mathbb{Z}_2 \times \mathbb{Z}_2 \to \mathbb{Z}_2^{diag} \times \mathbb{Z}_2$. As a result, the $s_{i,j,k}^{(3)}$ domain wall, which is a $\mathbb{Z}_2 \times \mathbb{Z}_2 \times \mathbb{Z}_2$ gauged SPT defect, gets trivialized on this boundary, and hence can condense on it.

\subsubsection{The $m_1 m_2 m_3$- and $e_1 e_2 e_3$-boundaries}

The $m_1 m_2 m_3$-boundary is given by the condensation set $(e_1 e_2, e_2 e_3, m_1 m_2 m_3)$, and the $e_1 e_2 e_3$-boundary is given by the condensation set $(m_1 m_2, m_2 m_3, e_1 e_2 e_3)$. We investigate the condensation properties of the domain walls on these boundaries. Let's first consider the $s_{i,j}^{(2)}$ domain walls. We have
\begin{equation}
    \begin{aligned}
        s^{(2)}_{1,2}: &(e_1 e_2, e_2 e_3, m_1 m_2 m_3) \\
        \to &(e_1 e_2, e_2 e_3, m_1 m_2 m_3 e_1 e_2) \\
        \equiv &(e_1 e_2, e_2 e_3, m_1 m_2 m_3), \\
         s^{(2)}_{1,2}: &(m_1 m_2, m_2 m_3, e_1 e_2 e_3) \\
         \to &(m_1 m_2 e_1 e_2, m_2 m_3 e_1, e_1 e_2 e_3).
    \end{aligned} \label{eq:e1e2e3_s12}
\end{equation}
Since the condensation has permutation symmetries, the results hold for other $s_{i,j}^{(2)}$ domain walls. We thus conclude that the $s_{i,j}^{(2)}$ domain wall can condense on the $m_1 m_2 m_3$-boundary, while cannot condense on the $e_1 e_2 e_3$-boundary.

From the TQFT perspective, the $\mathbb{Z}_2 \times \mathbb{Z}_2 \times \mathbb{Z}_2$ symmetry in the bulk breaks into the $\mathbb{Z}_2^{diag}$ symmetry on the $(e_1 e_2, e_2 e_3, m_1 m_2 m_3)$-boundary, which corresponds to the charge conservation symmetry on the entire system. Since the $s_{i,j}^{(2)}$ domain walls are $\mathbb{Z}_2 \times \mathbb{Z}_2$ gauged SPT defects, while the boundary only have one $\mathbb{Z}_2^{diag}$ symmetry, it get trivialized and hence can condense on the boundary.

In the case of the $e_1 e_2 e_3$-boundary, the bulk symmetry breaks into a $\mathbb{Z}_2^{(110)} \times \mathbb{Z}_2^{(011)}$ symmetry. $\mathbb{Z}_2^{(110)}$ means the product of $\mathbb{Z}_2$ symmetries on the first and the second copies. Each $\mathbb{Z}_2$ symmetry here corresponds to the charge conservation symmetry of two copies of 3D toric code. Since the $s_{i,j}^{(2)}$ domain walls are $\mathbb{Z}_2 \times \mathbb{Z}_2$ gauged SPT defects, the symmetry is preserved on the boundary. Therefore, they cannot condense on it.

However, as we see in Eq.~\eqref{eq:double_s_1}, even though one $s_{i,j}^{(2)}$ domain wall cannot condense, the composite of them may condense. We have
\begin{equation}
    \begin{aligned}
        &\ \quad s^{(2)}_{2,3}s^{(2)}_{1,2}: (m_1 m_2, m_2 m_3, e_1 e_2 e_3) \\
        &\to s^{(2)}_{2,3}: (m_1 m_2 e_1 e_2, m_2 m_3 e_1, e_1 e_2 e_3) \\
        &\to (m_1 m_2 e_1 e_2 e_3, m_2 m_3 e_1 e_2 e_3, e_1 e_2 e_3) \\
        &\equiv (m_1 m_2, m_2 m_3, e_1 e_2 e_3),
    \end{aligned} \label{eq:composite_condensation}
\end{equation}
which shows that $s^{(2)}_{2,3}s^{(2)}_{1,2}$ can condense on this boundary. The other two composite domain wall can also condense because of the permutation symmetry.

From the TQFT perspective, the composite of domain walls, $s^{(2)}_{2,3}s^{(2)}_{1,2}$, corresponds to two $\mathbb{Z}_2 \times \mathbb{Z}_2$ symmetries, while on the boundary the symmetry is $\mathbb{Z}_2^{(110)} \times \mathbb{Z}_2^{(011)}$. Each $\mathbb{Z}_2$ symmetry on the boundary corresponds to the symmetry breaking of $\mathbb{Z}_2 \times \mathbb{Z}_2$ symmetry in the bulk. As a result, the composite domain walls get trivialized, and hence can condense on the boundary.

Finally let's check whether $s_{i,j,k}^{(3)}$ condenses on those boundaries. We have
\begin{equation}
    \begin{aligned}
        s_{i,j,k}^{(3)}: &(e_1 e_2, e_2 e_3, m_1 m_2 m_3) \\
        \equiv &(e_1 e_2, e_2 e_3, m_1 m_2 m_3, s^{(2)}_{1,2}, s^{(2)}_{2,3}, s^{(2)}_{3,1}) \\
        \to &(e_1 e_2, e_2 e_3, m_1 m_2 m_3s^{(2)}_{1,2}s^{(2)}_{2,3}s^{(2)}_{2,3} , s^{(2)}_{1,2}, s^{(2)}_{2,3}, s^{(2)}_{3,1}) \\
        \equiv &(e_1 e_2, e_2 e_3, m_1 m_2 m_3).
    \end{aligned} \label{eq:mmm_1}
\end{equation}
Therefore, $s_{i,j,k}^{(3)}$ domain wall can condense on the $m_1 m_2 m_3$-boundary.

For the $e_1 e_2 e_3$-boundary, we have
\begin{equation}
    \begin{aligned}
        s_{i,j,k}^{(3)}: &(m_1 m_2, m_2 m_3, e_1 e_2 e_3) \\
        \to &(m_1 m_2 s^{(2)}_{2,3} s^{(2)}_{3,1}, m_2 m_3 s^{(2)}_{3,1} s^{(2)}_{1,2}, e_1 e_2 e_3).
    \end{aligned} \label{eq:eee_1}
\end{equation}
We know from Eq.~\eqref{eq:composite_condensation} that, $s_{i,j}^{(2)}s_{j,k}^{(2)}$ can condense on this boundary. As a result, we have
\begin{equation}
    \begin{aligned}
        &(m_1 m_2 s^{(2)}_{2,3} s^{(2)}_{3,1}, m_2 m_3 s^{(2)}_{3,1} s^{(2)}_{1,2}, e_1 e_2 e_3) \\
        \equiv &(m_1 m_2, m_2 m_3, e_1 e_2 e_3).
    \end{aligned}\label{eq:eee_2}
\end{equation}
$s_{i,j,k}^{(3)}$ domain wall can condense on the $e_1 e_2 e_3$-boundary.

In the perspective of TQFT, the symmetry on the boundary spontaneously breaks into $\mathbb{Z}_2^{diag}$ symmetry for the $m_1m_2m_3$-boundary, and $\mathbb{Z}_2^{(110)} \times \mathbb{Z}_2^{(011)}$ symmetry for the $e_1 e_2 e_3$-boundary. In both case the original $\mathbb{Z}_2^3$ symmetry in the bulk is broken. Therefore, the $s_{i,j,k}^{(3)}$ domain wall gets trivialized, and hence can condense on these boundaries.

\subsubsection{The $(m_1, m_2, m_3)$-boundary}

In Ref.~\cite{zhu2022topological}, it has been  shown that exotic boundaries can be generated by attaching a domain wall which does not condense on the normal boundaries discussed above. Namely, consider the condensations on the boundary of three copies of 3D toric codes to be $(m_1, m_2, m_3)$, sweeping the CCZ-domain wall $ s_{1,2,3}^{(3)}$ on the entire system and CZ-domain wall $s_{i,j}^{(2)}$  on a codimension-1 submanifold can generate two different types of exotic boundaries respectively. By sweeping the  $ s_{1,2,3}^{(3)}$ domain wall, the condensations of the all-smooth boundaries become
\begin{align}
    s_{1,2,3}^{(3)}: (m_1, m_2, m_3) \to (m_1 s_{2,3}^{(2)}, m_2 s_{3,1}^{(2)}, m_3 s_{1,2}^{(2)}). \label{eq:sijk}
\end{align}
By sweeping the $s_{i,j}^{(2)}$ domain wall,  we have the following map on the condensations of the boundary:
\begin{align}
    s_{i,j}^{(2)}: (m_i, m_j, m_k) \to (m_i e_j, m_j e_i, m_k), \label{eq:sij}
\end{align}
in which $i,j,k \in \{1,2,3\}$ and $i \neq j \neq k$.

As we see from Eq.~\eqref{eq:sij}, the boundary generated by attaching  the $s_{i,j}^{(2)}$ domain wall can only condense a charge-flux composite $m_i e_j$ or $m_j e_i$ but not the pure flux $m_i$ or $m_j$. The generated boundary labeled by $(m_i e_j, m_j e_i, m_k)$ is different from the original boundary $(m_i, m_j, m_k)\equiv (m_1, m_2, m_3)$ since there is no way to rewrite the boundary with a basis change like previous cases.  This also shows that the $s_{i,j}^{(2)}$ domain wall  cannot condense on the $(m_1, m_2, m_3)$-boundary since the map does not keep the boundary invariant according to Eq.~\eqref{eq:code_space_boundary}.  

Instead, the $s_{i,j}^{(2)}$ domain wall is attached to the  $(m_1, m_2, m_3)$-boundary and produced an exotic nested boundary where a codimension-2 (1D) boundary $(m_i e_j, m_j e_i, m_k)$ is generated at the interface of the two original codimension-1 (2D) $(m_1, m_2, m_3)$-boundaries, as illustrated in Fig.~\ref{fig:s_domain_wall}(a) in the later subsection.

An even more exotic example is the one we show in Eq.~\eqref{eq:sijk}, where the produced boundary $(m_1 s_{2,3}^{(2)}, m_2 s_{3,1}^{(2)}, m_3 s_{1,2}^{(2)})$ is also different from the original $(m_1, m_2, m_3)$-boundary since the CZ-domain wall $s_{i,j}^{(2)}$ itself cannot condense on the $(m_1, m_2, m_3)$-boundary. In this case, neither the $e_i$ nor $m_i$ excitations can condense on this new boundary, which means the new gapped boundary is beyond the classification based on the Lagrangian subgroup.  Instead, only the composite of the $m$-flux and the CZ-domain wall $m_i s^{(2)}_{j,k}$ ($i \neq j \neq k$) can condense on this new \textit{magic boundary}, as illustrated in Fig.~\ref{fig:magic_boundary} in the later subsection.

The reason that neither $s_{i,j}^{(2)}$ nor $s_{1,2,3}^{(3)}$ condense on the  $(m_1, m_2, m_3)$-boundary is because that the $m$-boundaries preserve the $\mathbb{Z}_2$ symmetry in the unguaged model (Ising model) such that the SPT defects cannot be trivialized on the boundary and hence condense.

\subsubsection{Other boundaries classified by the Lagrangian subgroup}

There are also other boundaries that can be classified by the Lagrangian subgroup. These boundaries contains condensations that are composites of $e$-particles and $m$-flux loops. We listed them here for the readers' reference. The domain-wall condensation properties are summarized in Fig.~\ref{fig:web_1}.

Namely, those boundary condensations are 
\begin{equation}
    \begin{aligned}
        (m_i e_j, m_j e_i, m_k),\ &(m_i e_j e_k, m_j e_i, m_k e_i), \\
        (m_1 e_2 e_3, m_2 e_1 e_3, m_3 e_1 e_2),\ &(e_i, m_j e_k, m_k e_j), \\
        (m_1 m_2 e_1 e_2, m_2 m_3 e_1, e_1 e_2 e_3),\ &(e_i e_j, m_i m_j e_k, m_k e_j).
    \end{aligned}
\end{equation}
Those boundaries appeared as nested boundaries in our previous discussion, in which case they are 1D lines nested on the 2D boundaries. However, the boundaries can also has those condensations. In this case, the condensations on the nested boundaries are just the boundary condensations we discussed before, as illustrated in Fig.~\ref{fig:web_1}.
\\

We know from Ref.~\cite{kubica2015unfolding} that, there is a correspondence between 3D color code and three copies of 3D toric codes or equivalently the $\mathbb{Z}_2^3$ gauge theory. Then a natural question arises: what does the magic boundary and the other boundaries look like in the corresponding color code? We answer this question by explicitly constructing such kinds of boundaries in the 3D color-code lattice model in the following subsections. Although we have already derived the existence of these boundaries using the abstract TQFT language in this subsection, an explicit lattice construction helps to verify the TQFT prediction and understand the microscopic details of these boundaries such that we can also use these models for fault-tolerant quantum computing. In addition, we will also give a classification of the 3D color-code boundaries in the following subsections according to our current understanding of the defects and boundaries in the  $\mathbb{Z}_2^3$ gauge theory.

\subsection{Transversal gates on the 3D color code} \label{sec:transversal}

Following Ref.~\cite{kubica2015unfolding}'s prescription, we divide the lattice of 3D color code into two sublattices and guarantee each vertex of one sublattice is connected to the vertices of the other sublattice. When we apply a gate $\widetilde{U}$ transversally, we apply  $U$ on the vertices of one sublattice and $U^{\dagger}$ on the vertices of the other sublattice. This definition of transversal gates can avoid the appearance of overall phase factors. 

It is demonstrated in Ref.~\cite{kubica2015unfolding} that applying transversal-$C^{d-1}Z$ gates on $d$ copies of toric codes in $d$-dimension is equivalent to applying transversal-$R_d$ gate on its corresponding color code. In 3D, we are able to apply $R_2 = S$ and $R_3 = T$ transversally.

Transversal-$S$ gate can be applied on the $\mathbf{c}_i \mathbf{c}_j$-membrane. We denote this operator as $\widetilde{S}_{\mathbf{c}_i \mathbf{c}_j}$. We have the following equation.
\begin{equation}
    \begin{aligned}
        &P_C \widetilde{S}_{\mathbf{c}_i \mathbf{c}_j}\mathcal{A}_{\mathbf{c}_k}^{bulk} (X) \widetilde{S}_{\mathbf{c}_i \mathbf{c}_j}^{\dagger} P_C \\
    =& P_C \mathcal{A}_{\mathbf{c}_k}^{bulk} (X) \prod_{v \in \mathcal{M}} Z(v) P_C \\
    =& P_C \mathcal{A}_{\mathbf{c}_k}^{bulk} (X) P_C P_C \prod_{v \in \mathcal{M}} Z(v) P_C \\
    =& P_C \mathcal{A}_{\mathbf{c}_k}^{bulk} (X) P_C,
    \end{aligned} \label{eq:transversal_s_bulk}
\end{equation}
in which
\begin{align}
    \mathcal{M} = \mathrm{supp}(\widetilde{S}_{\mathbf{c}_i \mathbf{c}_j}) \cap \mathrm{supp} (\mathcal{A}_{\mathbf{c}_k}^{bulk}),
\end{align}
and the $\mathrm{supp}$ of an operator $\mathcal{O}$ means the set of qubits where $\mathcal{O}$ acts non-trivially.

From the third line to the fourth line of Eq.~\eqref{eq:transversal_s_bulk} we use the fact that $\prod_{v \in \mathcal{M}} Z(v)$ can be always expressed as the product of $Z$-stabilizers, which we discuss in detail in the next subsection. Therefore, this new $Z$-term remains a stabilizer, making it a logical identity within the code space. In addition, since $S_{\mathbf{c}_i \mathbf{c}_j}$ are composed of diagonal gates, it commutes with all the $Z$-stabilizers $\mathcal{B}_{\mathbf{c}_i \mathbf{c}_j}^{bulk} (Z)$, $\mathcal{B}_{\mathbf{c}_j \mathbf{c}_k}^{bulk} (Z)$ and $\mathcal{B}_{\mathbf{c}_i \mathbf{c}_k}^{bulk} (Z)$.  Therefore, we have  $P_C \widetilde{S}_{\mathbf{c}_i \mathbf{c}_j} H^{bulk} \widetilde{S}_{\mathbf{c}_i \mathbf{c}_j}^{\dagger} P_C = P_C H^{bulk} P_C$. The transversal-$S_{\mathbf{c}_i \mathbf{c}_j}$ gate is thus an emergent 1-form symmetry of the code space, according to the definition  using Eq.~\eqref{eq:bulk_invariant}.

When $\widetilde{S}_{\mathbf{c}_i \mathbf{c}_j}$ is applied on a region $\mathcal{R}$ with boundaries. Codimension-2 (1D) domain walls are created on these boundaries and we denote them as the $S_{\mathbf{c}_i \mathbf{c}_j}$-domain wall. When an electric excitation passes the domain wall, it remains unchanged. However, when a magnetic flux loop passes the domain wall, it gets attached to electric charges on the intersections. We have the following map:
\begin{equation}
    \begin{aligned}
        \widetilde{S}_{\mathbf{c}_i \mathbf{c}_j}: (\mathbf{c}_i \mathbf{c}_k)_{\mathbf{x}} &\to (\mathbf{c}_i)_{\mathbf{z}} (\mathbf{c}_i \mathbf{c}_k)_{\mathbf{x}}, \ (\mathbf{c}_i)_{\mathbf{z}} \to (\mathbf{c}_i)_{\mathbf{z}},  \\
    (\mathbf{c}_j \mathbf{c}_k)_{\mathbf{x}} &\to (\mathbf{c}_j)_{\mathbf{z}} (\mathbf{c}_j \mathbf{c}_k)_{\mathbf{x}}, \ (\mathbf{c}_j)_{\mathbf{z}} \to (\mathbf{c}_j)_{\mathbf{z}},
    \end{aligned}
\end{equation}
which is illustrated in Fig.~\ref{fig:two_domain_walls}(a) and corresponds to the mapping Eq.~\eqref{eq:s2_transform} in the TQFT language.

Transversal-$T$ gate can be applied on the entire system. After conjugation, the X-stabilizers on the purple cell $\mathcal{A}_{p}^{bulk} (X)$ becomes
\begin{align}  \adjincludegraphics[width=6cm,valign=c]{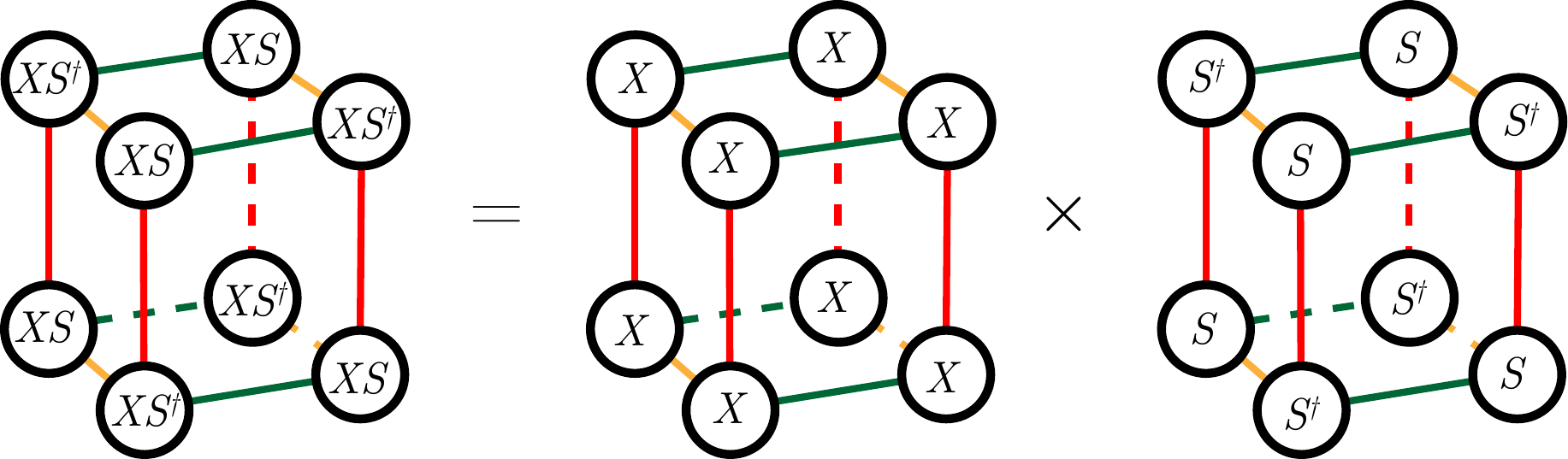}. \label{eq:gate_00}
\end{align}
Similar conjugation can be applied to the red, green, and yellow cells.

One can notice that the new $XS$-operator can be viewed as a contractible $S$-domain wall attached on the original $X$-stabilizer, which can  itself be considered as the smallest contractible $X$-membrane operator. The contractible $S$-domain wall  is a logical identity. To see this, we notice the following identity,
\begin{align}
    \adjincludegraphics[width=6cm,valign=c]{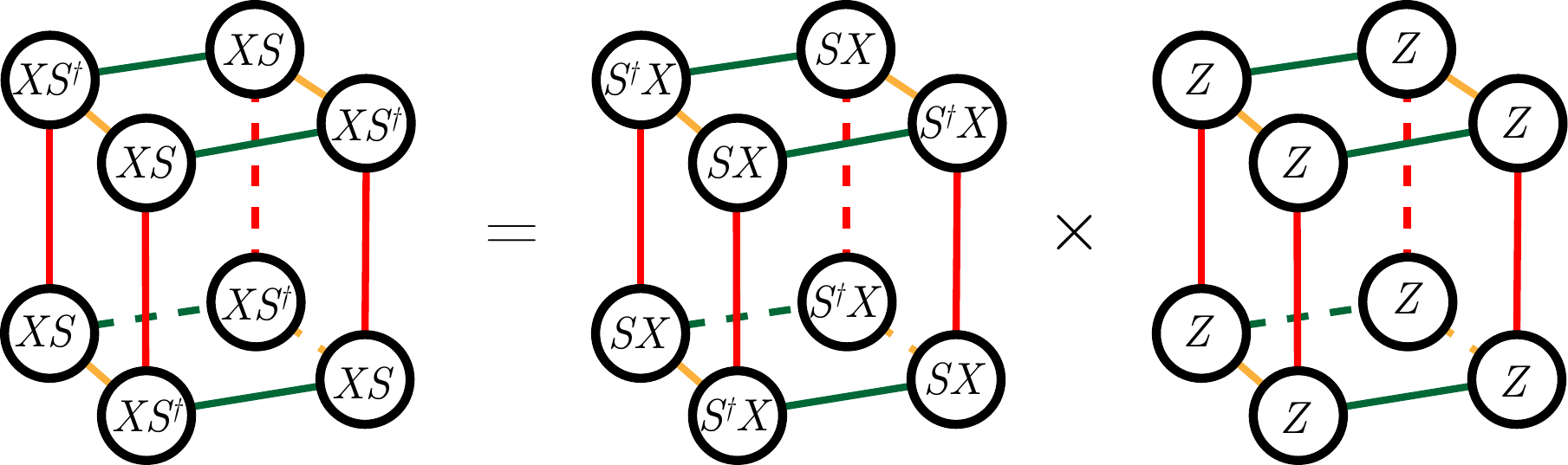}. \label{eq:gate_0}
\end{align}
The second operator on the right hand side is a stabilizer of the code, which is itself a logical identity. That is to say, the contractible $S$-domain wall and the $X$-stabilizer commute with each other within the code space, which means the contractible $S$-domain wall is also a logical identity, namely
\begin{align}
    P_C \mathcal{A}_{\mathbf{c}_i}^{bulk} (S) P_C = \mathbb{I}.
\end{align}
We hence obtain
\begin{equation}
    \begin{aligned}
        &P_C \widetilde{T} \mathcal{A}_{\mathbf{c}_i}^{bulk} (X) \widetilde{T}^{\dagger} P_C \\
    =& P_C \mathcal{A}_{\mathbf{c}_i}^{bulk} (XS)P_C \\
    =& P_C \mathcal{A}_{\mathbf{c}_i}^{bulk} (S) P_C P_C \mathcal{A}_{\mathbf{c}_i}^{bulk} (X) P_C \\
    =& P_C \mathcal{A}_{\mathbf{c}_i}^{bulk} (X) P_C,
    \end{aligned}
\end{equation}
which leads to $P_C \widetilde{T} H^{bulk} \widetilde{T}^\dag P_C$$=$$P_C H^{bulk} (X) P_C$ and satisfing the condition in Eq.~\eqref{eq:bulk_invariant}. From the symmetry perspective, this is equivalent to say that transversal-$T$ gate corresponds to an emergent 0-form symmetry of the code space in the bulk.

When the transversal-$T$ gate is applied on a region $\mathcal{R}$ with boundaries. Codimension-1 (2D) domain walls are created on these boundaries and we denote them as the $T$-domain wall. When an electric excitation passes the domain wall, it remains unchanged. However, when a magnetic flux loop passes the domain wall, it get attached to an $S$-domain wall. We have the following map:
\begin{equation}
    \begin{aligned}
        \widetilde{T}: pg_{\mathbf{x}} \to pg_{\mathbf{xs}}, \ py_{\mathbf{x}} \to py_{\mathbf{xs}},\ 
        yg_{\mathbf{x}} \to yg_{\mathbf{xs}},
    \end{aligned}
\end{equation}
which is illustrated in Fig.~\ref{fig:two_domain_walls}(b) and corresponds to the mapping Eq.~\eqref{eq:s3_transform} in the TQFT language.

In contrast, the $T$-symmetry does not hold when the system has the $X$-boundaries.  Similarly, the 1-form $S$ symmetry does not hold when its support touches the $X$-boundary.  This is due to the fact that the $T$ and $S$ domain walls do not condense on the $X$-boundary. By conjugating transversal-$T$ gate on the corresponding color code with $X$-boundaries, one can obtain a new type of boundaries of the 3D color code.

The $Z$-stabilizers on the boundary remain unchanged after applying the transversal-$T$ gate, while  the $X$-stabilizers on the boundary are turned into non-Pauli stabilizers which are dubbed as as $XS$-stabilizers \cite{ni2015non, webster2022xp}.\footnote{More general cases are discussed in Ref.~\cite{ni2015non, webster2022xp}. Following the notation in Ref.~\cite{webster2022xp}, a single operator can be represented by
\begin{align}
    XP_N(a|b|c) = \mathrm{exp}\left(\frac{a}{2N} 2\pi i\right) X^{b} P^c, \quad P^N = \mathbb{I}.
\end{align}
Therefore, the $XS$ and $XS^{\dagger}$ operators can be rewritten as
\begin{align}
    XS \equiv XP_4(0 | 1 | 1), \quad XS^{\dagger} \equiv XP_4(0 | 1 | 3)
\end{align}}
The condensation set of the new boundary is $\{pg_{\mathbf{xs}}, py_{\mathbf{xs}}, yg_{\mathbf{xs}}\}$, representing the $\{pg_{\mathbf{x}}, py_{\mathbf{x}}, yg_{\mathbf{x}}\}$ flux strings attached with $S$-domain walls, as illustrated in Fig.~\ref{fig:magic_boundary}(c) and (d).

In the subsequent subsections, we demonstrate explicitly that this boundary corresponds to the {\it magic boundary} previously outlined in the context of TQFT, as illustrated in Section~\ref{sec:TQFT}. Therefore, we also refer to this new boundary as the magic boundary. The origin of this name is due to the fact that the Hamiltonian of this boundary is beyond the Pauli stabilizers formalism, and hence possess `\textit{magic}' \cite{chen2024magic}. The magic boundary is also an important ingredient for performing fault-tolerant non-Clifford logical gates in fractal codes which can replace the usual magic state distillation protocol as pointed out in Ref.~\cite{zhu2022topological}. This magic boundary can also be viewed as a composite boundary, comprising the $X$-stabilizers on the boundaries attached with $S$-domain walls (equivalent to $s^{(3)}_{1,2,3}$), as depicted in Fig.~\ref{fig:magic_boundary}(a). 

\subsection{Logical operators after applying the transversal $T$-gate.}
\label{sec:logical}

\begin{figure*}
    \centering
    \includegraphics[width = 16cm]{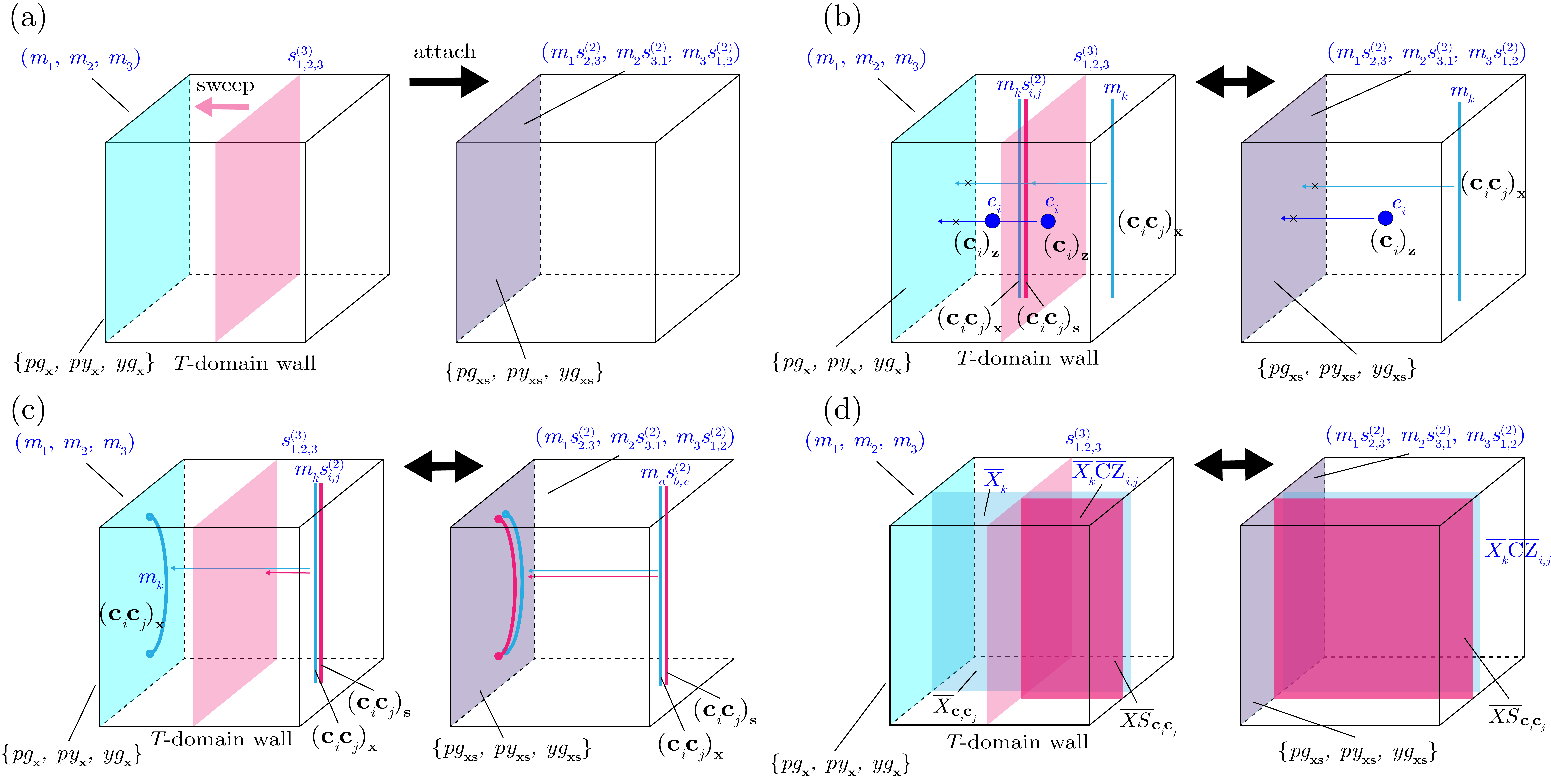}
    \caption{An illustration of the $T$ ($\mathrm{CCZ}$)-domain wall sweeping and its effect on the $\{pg_{\mathbf{x}}, py_\mathbf{x}, yg_{\mathbf{x}}\}$-boundary, topological excitations and logical operations. We use two equivalent languages, the 3D color-code language and the 3+1D TQFT/ toric-code language, to illustrate these phenomena, and use black and blue colored fonts, respectively. (a) The magic boundary can be obtained by sweeping the $T$-domain wall across the entire system and attaching it to the $\{pg_{\mathbf{x}}, py_\mathbf{x}, yg_{\mathbf{x}}\}$ boundary. (b) When traversing the $T$-domain wall, the electric charges $(\mathbf{c}_i)_{\mathbf{z}}$ remain unchanged. However, magnetic flux loops become attached to $S$-domain walls and cannot condense on the $\{pg_{\mathbf{x}}, py_\mathbf{x}, yg_{\mathbf{x}}\}$-boundary. This implies that neither electric charges nor magnetic flux loops can condense on the magic boundary. (c) Upon traversing the $T$-domain wall, the $S$-domain walls attached to the magnetic flux loops get re-absorbed into the $T$-domain wall. Consequently, the remaining $(\mathbf{c}_i \mathbf{c}_j)_{\mathbf{x}}$ excitation can condense on the $\{pg_{\mathbf{x}}, py_\mathbf{x}, yg_{\mathbf{x}}\}$-boundary. Equivalently, this means that magnetic flux loops with $S$-domain walls attached can condense on the magic boundary. (d) Upon sweeping the $T$-domain wall, the logical $X$-operator becomes attached to an $S$-membrane, forming the $\overline{XS}_{\mathbf{c}_i\mathbf{c}_j}$ operator. This membrane operator generates a magnetic flux loop with an $S$-domain wall attached to it. As the domain wall sweeps across the entire system, it can condense on the magic boundary.}
    \label{fig:magic_boundary}
\end{figure*}

Consider the $X$-membranes we introduce in Fig.~\ref{fig:gate_10}. These three membrane operators create $pg_{\mathbf{x}}, yg_{\mathbf{x}}$, and $yg_{\mathbf{x}}$ on their boundaries respectively. After conjugating the transversal-$T$ gate on the 3D color code, these $X$-membrane operators do not commute with the Hamiltonian terms on cells ($XS$-stabilizers). To better study the interplay between the membrane operators and the $XS$-stabilizers, we define the following group commutator,
\begin{align}
    K (T,V) := V^{-1} T^{-1} V T, \label{eq:group_commutator}
\end{align}
where $T$ and $V$ are both operators. In the case $T$ and $V$ commute with each other, we have $K = \mathbb{I}$. In the case $T$ and $V$ anti-commute with each other, we have $K = -\mathbb{I}$. In general cases, $K$ can be a non-trivial matrix. As suggested in Refs.~\cite{barkeshli2023codimension, barkeshli2022higher}, Eq.~\eqref{eq:group_commutator} is an example of higher group structure when $T$, $V$ and $K$ are symmetry operators acting on the Hilbert space.

We have the following algebraic relation between $XS$, $X S^{\dagger}$ and $X$-operators,
\begin{equation}
    \begin{aligned}
    K(X, XS) = i Z,\quad
    K(X, XS^{\dagger}) = - i Z.
\end{aligned} \label{eq:commutation}
\end{equation}

It suggests that $\widetilde{X}_{\mathbf{c}_i\mathbf{c}_j}$ and $\mathcal{A}^{bulk}_{\mathbf{c}_k}(XS)$ do not generally commute. However, when these operators act on the code space of the 3D color code, the resulting $Z$-terms are equivalent to product of $Z$-stabilizers for the $XS$-stabilizers in the bulk. Consequently, they commute with each other within the code space. This relationship can be illustrated by the following equation:
\begin{equation}
    \begin{aligned}
    &P_{C} \widetilde{X}_{\mathbf{c}_i\mathbf{c}_j} \mathcal{A}^{bulk}_{\mathbf{c}_k}(XS) P_{C}\\
    =& P_C \mathcal{A}^{bulk}_{\mathbf{c}_k}(XS) \widetilde{X}_{\mathbf{c}_i\mathbf{c}_j} K\left(\mathcal{A}^{bulk}_{\mathbf{c}_k}(XS), \widetilde{X}_{\mathbf{c}_i\mathbf{c}_j}\right) P_C \\
    =& P_C \mathcal{A}^{bulk}_{\mathbf{c}_k}(XS) \widetilde{X}_{\mathbf{c}_i\mathbf{c}_j} P_C, \label{eq:commutator}
\end{aligned}
\end{equation}
in which $P_C$ is the projector onto the code space $\mathcal{H}_C$, $\widetilde{X}_{\mathbf{c}_i\mathbf{c}_j}$ is the $X$-membrane operator with color $\mathbf{c}_i\mathbf{c}_j$ and $\mathcal{A}^{bulk}_{\mathbf{c}_k}(XS)$ is the $XS$-stabilizers on a $\mathbf{c}_k$-colored cell in the bulk. 

In the remainder of this subsection, we demonstrate that this equation is true by showing $K\left(\mathcal{A}^{bulk}_{\mathbf{c}_k}(XS), \widetilde{X}_{\mathbf{c}_i\mathbf{c}_j}\right)$ is equivalent to product of $Z$-stabilizers, which acts trivially on the code space.\footnote{In the following discussions, when we analyze the commutator between two operators, we consistently consider the scenario that they have a non-trivial intersection of supports. When the intersection of supports is trivial (non-overlapping), the commutator equals to $\mathbb{I}$. Eq.~\eqref{eq:commutator} still holds.} Although the illustration is based on the 3D color code lattice depicted in Fig.~\ref{fig:3DCC}, it holds for other cellulations and colorizations as well.

We first discuss the $X$-membrane operator on the purple-yellow membrane as shown in Fig.~\ref{fig:gate_10}(b), namely $\widetilde{X}_{py}$. The intersection between the support of this membrane operator and a yellow cell in the bulk is shown below.
\begin{align}
    \adjincludegraphics[width=6.5cm,valign=c]{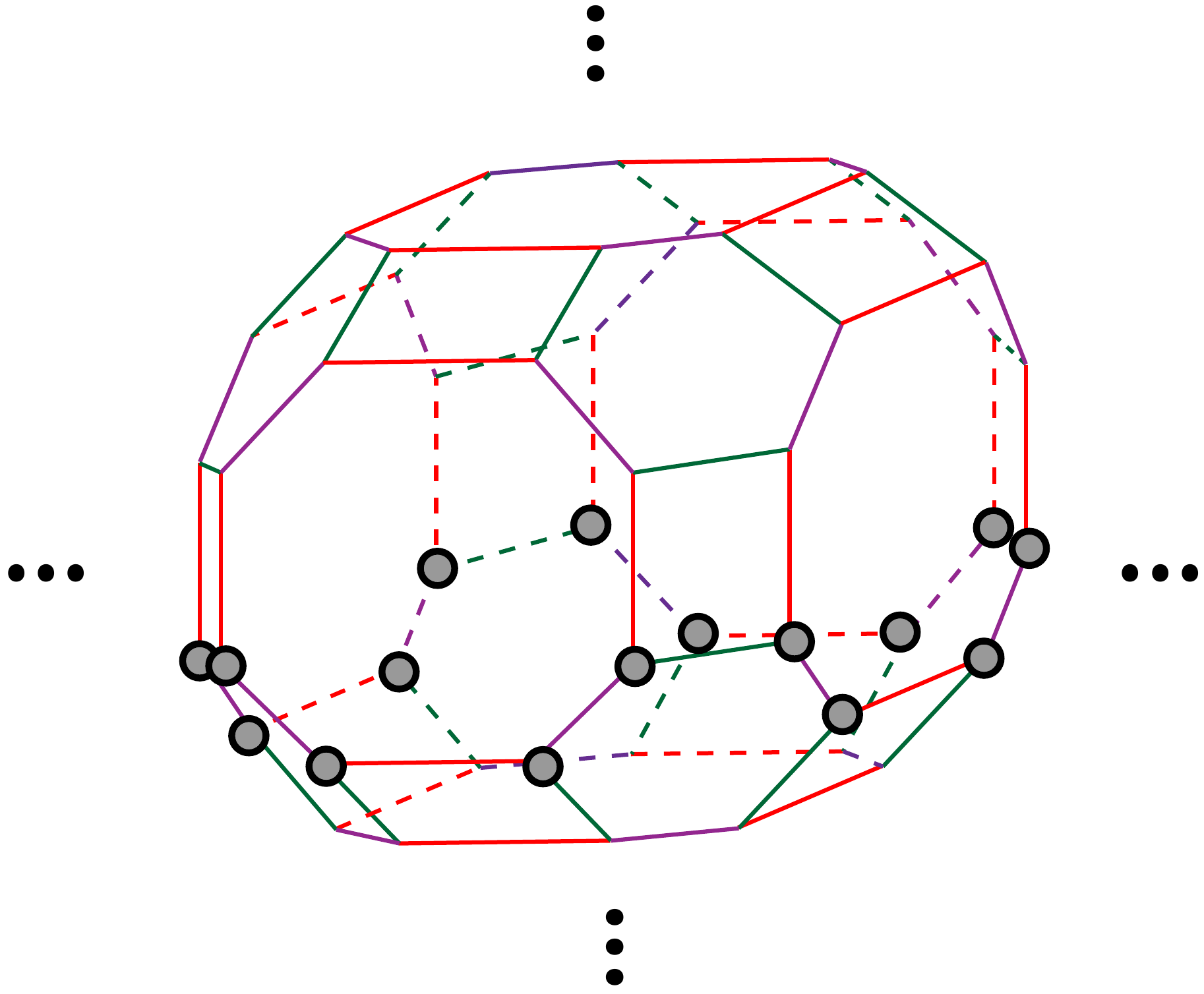}. \label{eq:absorb1}
\end{align}
Here, we utilize gray circles to denote the intersecting vertices between $\widetilde{X}_{py}$ and $\mathcal{A}^{bulk}_{y}(XS)$. As shown by Eq.~\eqref{eq:commutation}, upon exchanging the order of $\widetilde{X}_{py}$ and $\mathcal{A}^{bulk}_{y}(XS)$, an additional $Z$-operator, accompanied by a phase factor of $i$ or $-i$, emerges at each gray circle. These phase factors cancel themselves in our convention.

The commutator of $\widetilde{X}_{py}$ and $\mathcal{A}^{bulk}_{y}(XS)$ is given by the following $Z$-membrane,
\begin{align}
    \adjincludegraphics[width=6.5cm,valign=c]{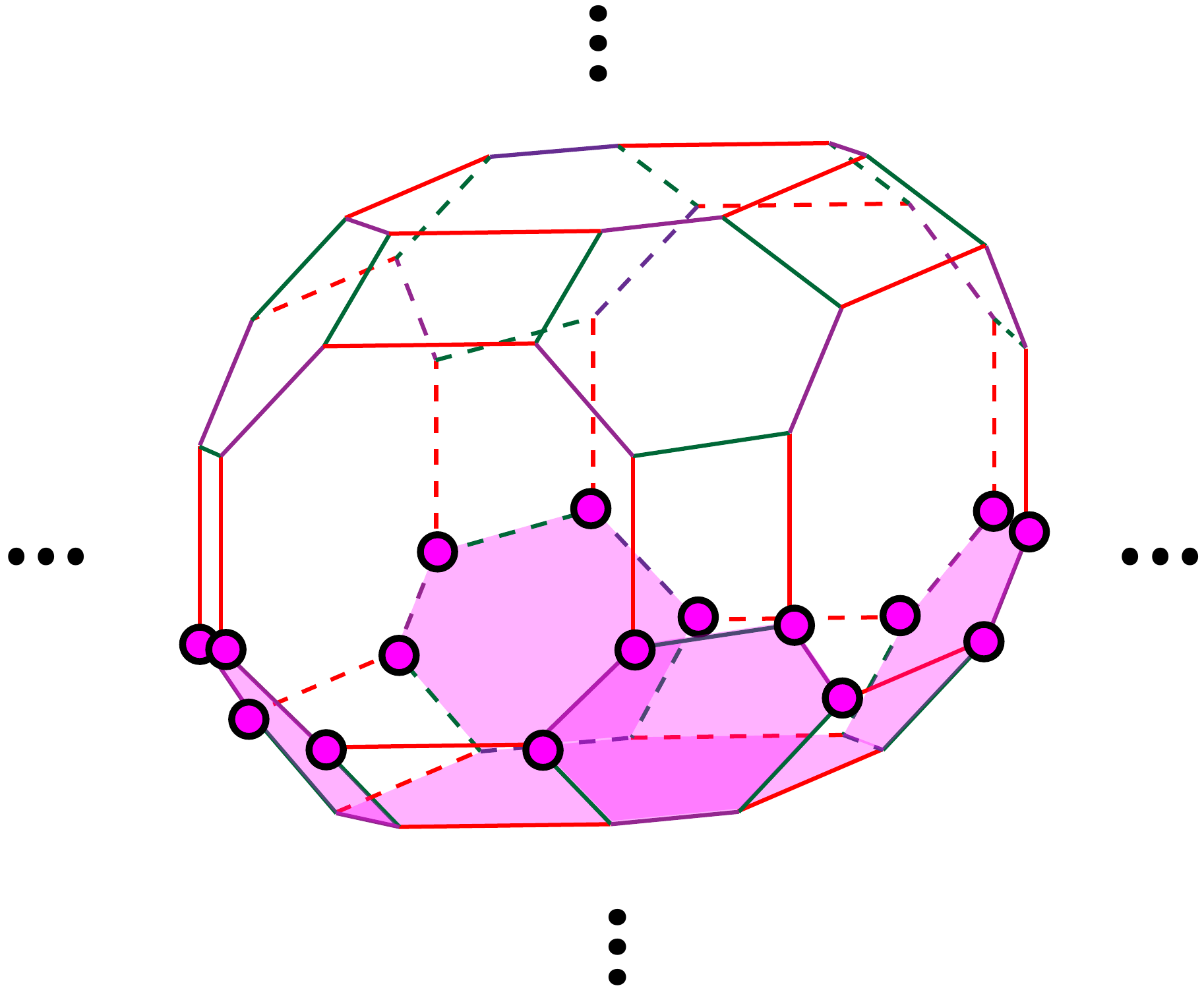}. \label{eq:absorb2}
\end{align}
This operator is equivalent to the product of four $Z$-stabilizers on the hexagons and one on the bottom octagon. Consequently, it acts trivially on the code space. As a result, $\widetilde{X}_{py}$ commutes with the $\mathcal{A}^{bulk}_{y}(XS)$ within the code space.

The commutator of $\widetilde{X}_{py}$ and $\mathcal{A}^{bulk}_{g}(XS)$ in the bulk is the following.
\begin{align}
    \adjincludegraphics[width=6.5cm,valign=c]{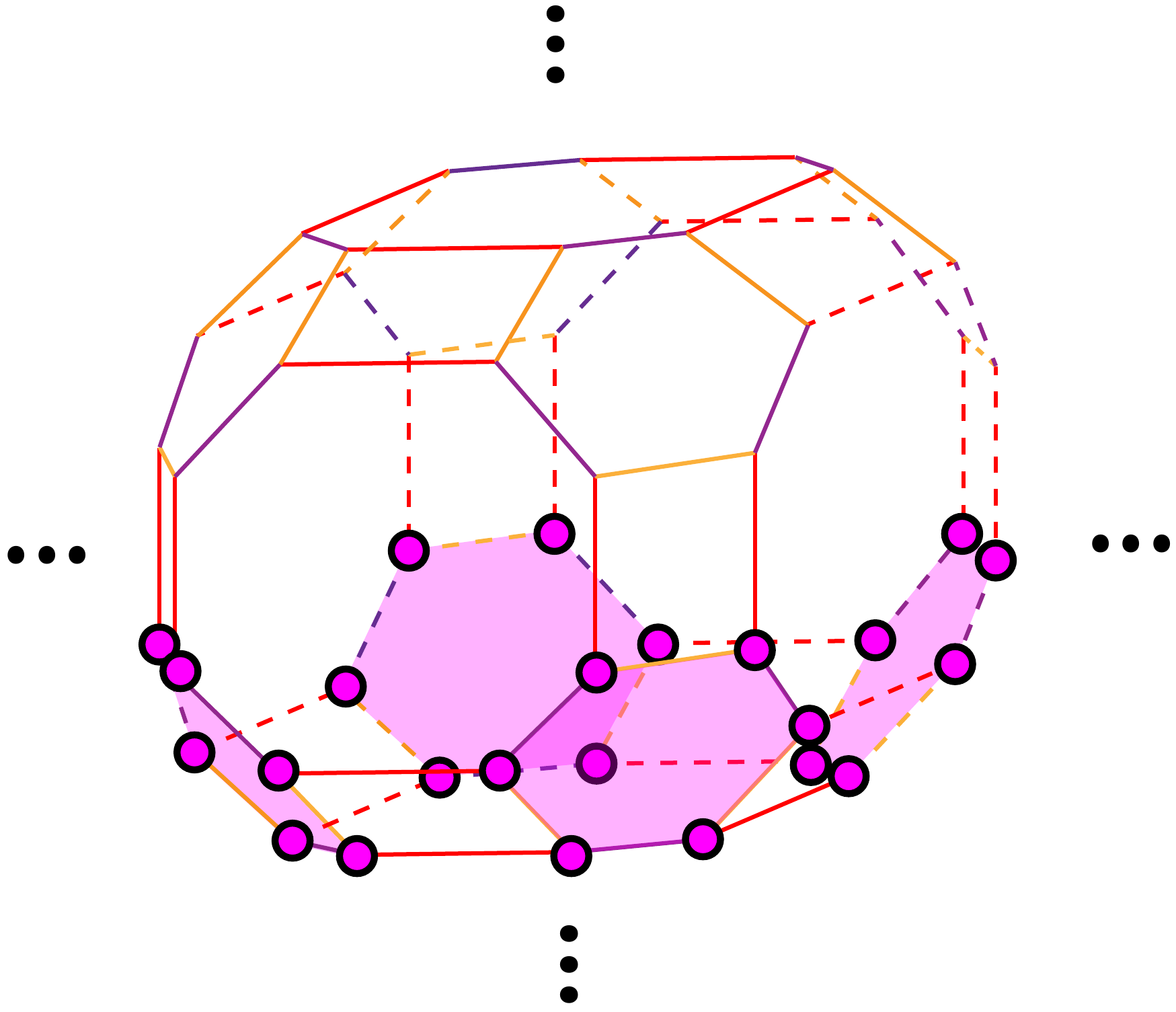}.
\end{align}
Similarly, this operator is equivalent to the product of four $Z$-stabilizers on the bottom hexagons. As a result, $\widetilde{X}_{py}$ commutes with $\mathcal{A}^{bulk}_{g}(XS)$ in the bulk within the code space.

The commutator of $\widetilde{X}_{py}$ and $\mathcal{A}^{bulk}_{r}(XS)$ in the bulk is the following.
\begin{align}
    \adjincludegraphics[width=4.5cm,valign=c]{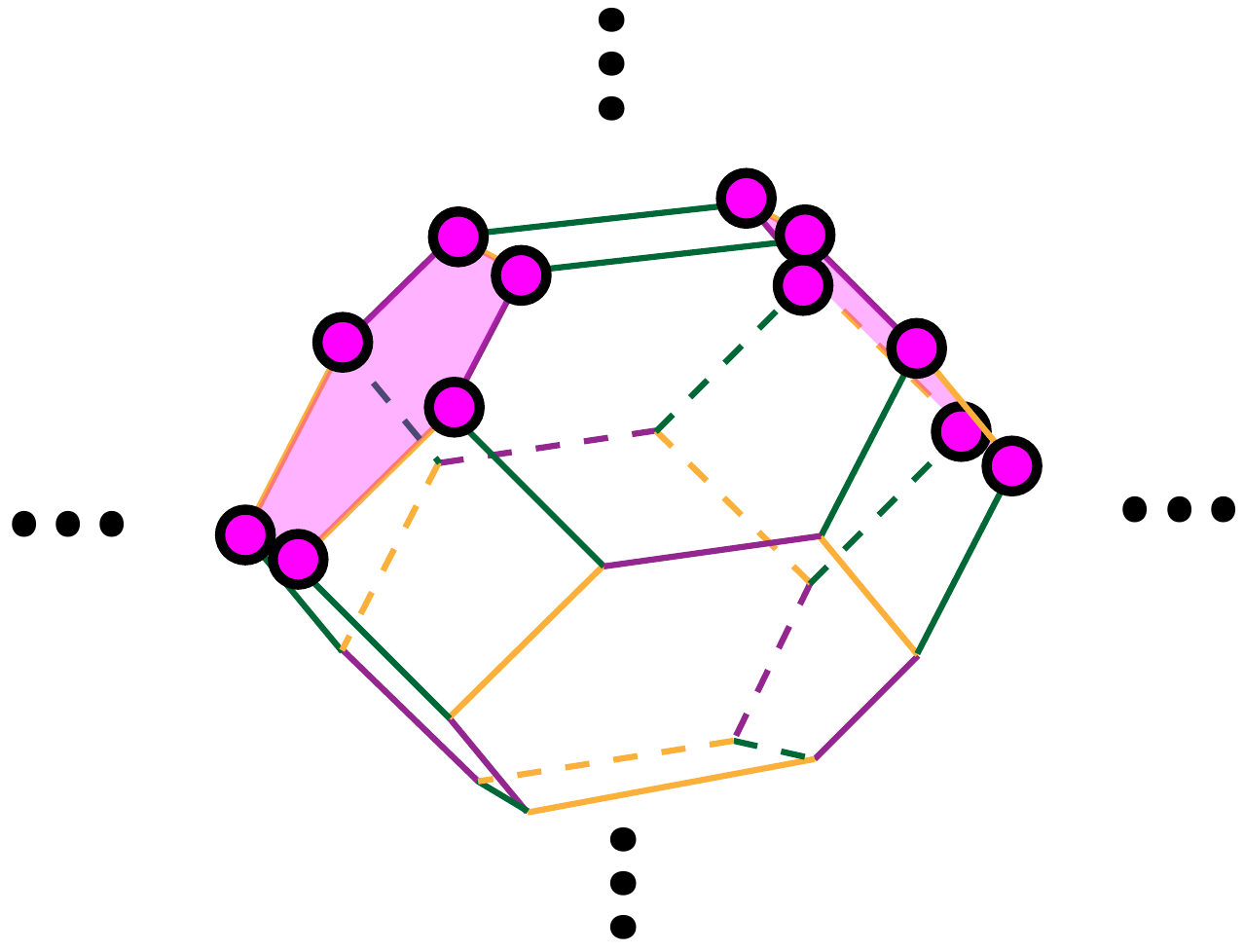}.
\end{align}
This operator is equivalent to the product of two $Z$-stabilizers on hexagons. Therefore, $\widetilde{X}_{py}$ commutes with $\mathcal{A}^{bulk}_{r}(XS)$ in the bulk within the code space.

Finally, the commutator of $\widetilde{X}_{py}$ and $\mathcal{A}^{bulk}_{p}(XS)$ in the bulk is the following.
\begin{align}
    \adjincludegraphics[width=3cm,valign=c]{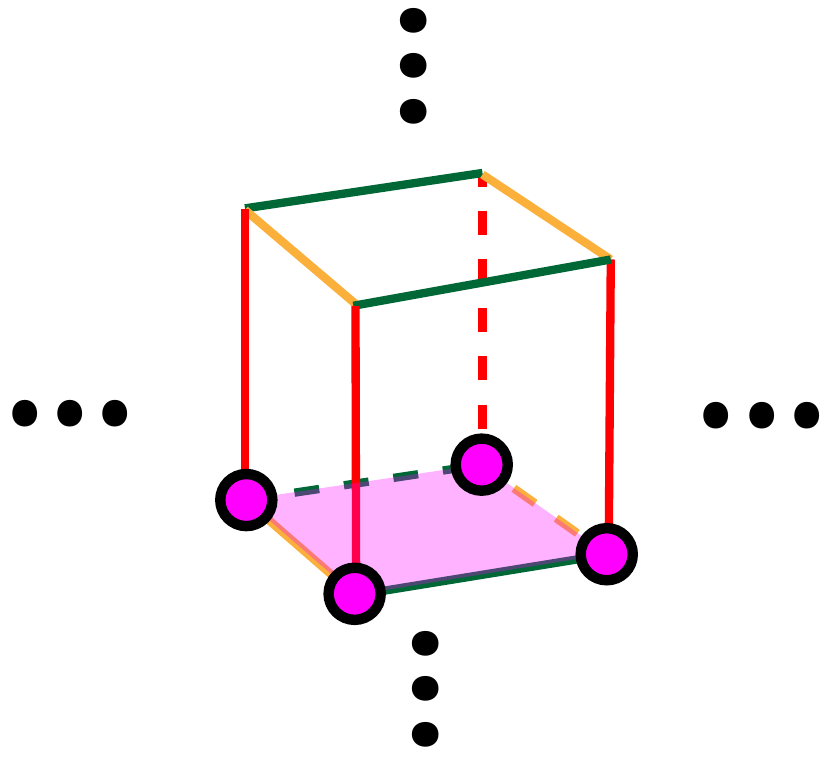}.
\end{align}
This operator is equivalent to a $Z$-stabilizer on the square. 

In conclusion, we demonstrate that $\widetilde{X}_{py}$ commutes with all the stabilizers in the bulk within the code space.

A similar argument can be applied to $\widetilde{X}_{pg}$ as shown in Fig.~\ref{fig:gate_10} (a) and $\mathcal{A}^{bulk}_{\mathbf{c}_i}(XS)$ due to the symmetry of the lattice.

Then we discuss the commutator between $\widetilde{X}_{yg}$, as shown in Fig.~\ref{fig:gate_10} (c), and $\mathcal{A}^{bulk}_{\mathbf{c}_i}(XS)$. The commutator of $\widetilde{X}_{yg}$ and $\mathcal{A}^{bulk}_{y}(XS)$ in the bulk is the following.
\begin{align}
    \adjincludegraphics[width=6.5cm,valign=c]{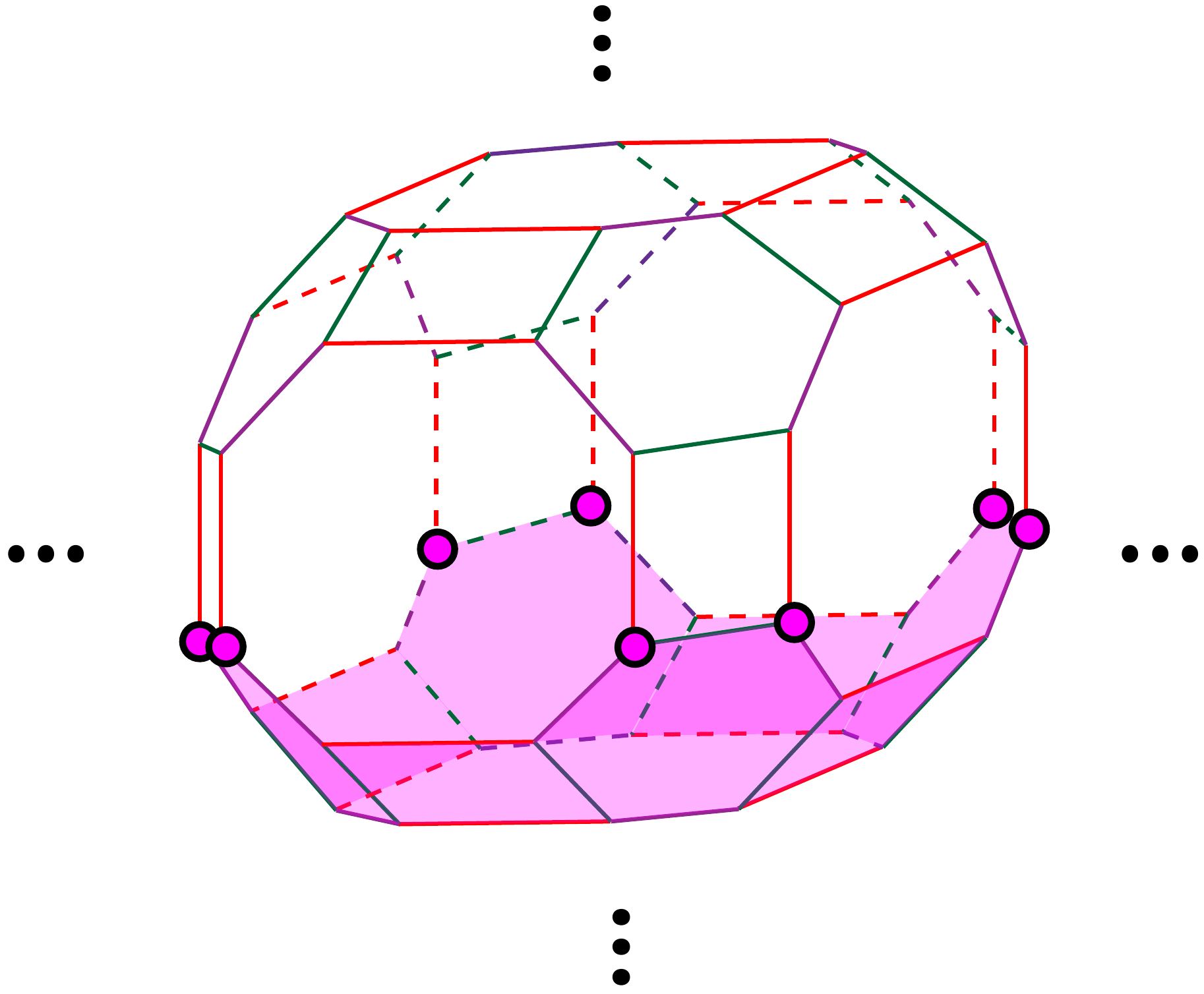}. \label{eq:gate_6}
\end{align}
This operator is equivalent to the product of four $Z$-stabilizers on hexagons and four on squares.

The commutator between $\widetilde{X}_{yg}$ and $\mathcal{A}^{bulk}_{g}(XS)$ is the same as shown in Eq.~\eqref{eq:gate_6} due to the symmetry of the lattice and the membrane operator. Therefore, they also commute with each other within the code space.

The commutator of $\widetilde{X}_{yg}$ and $\mathcal{A}^{bulk}_{r}(XS)$ is the following.
\begin{align}
    \adjincludegraphics[width=4.5cm,valign=c]{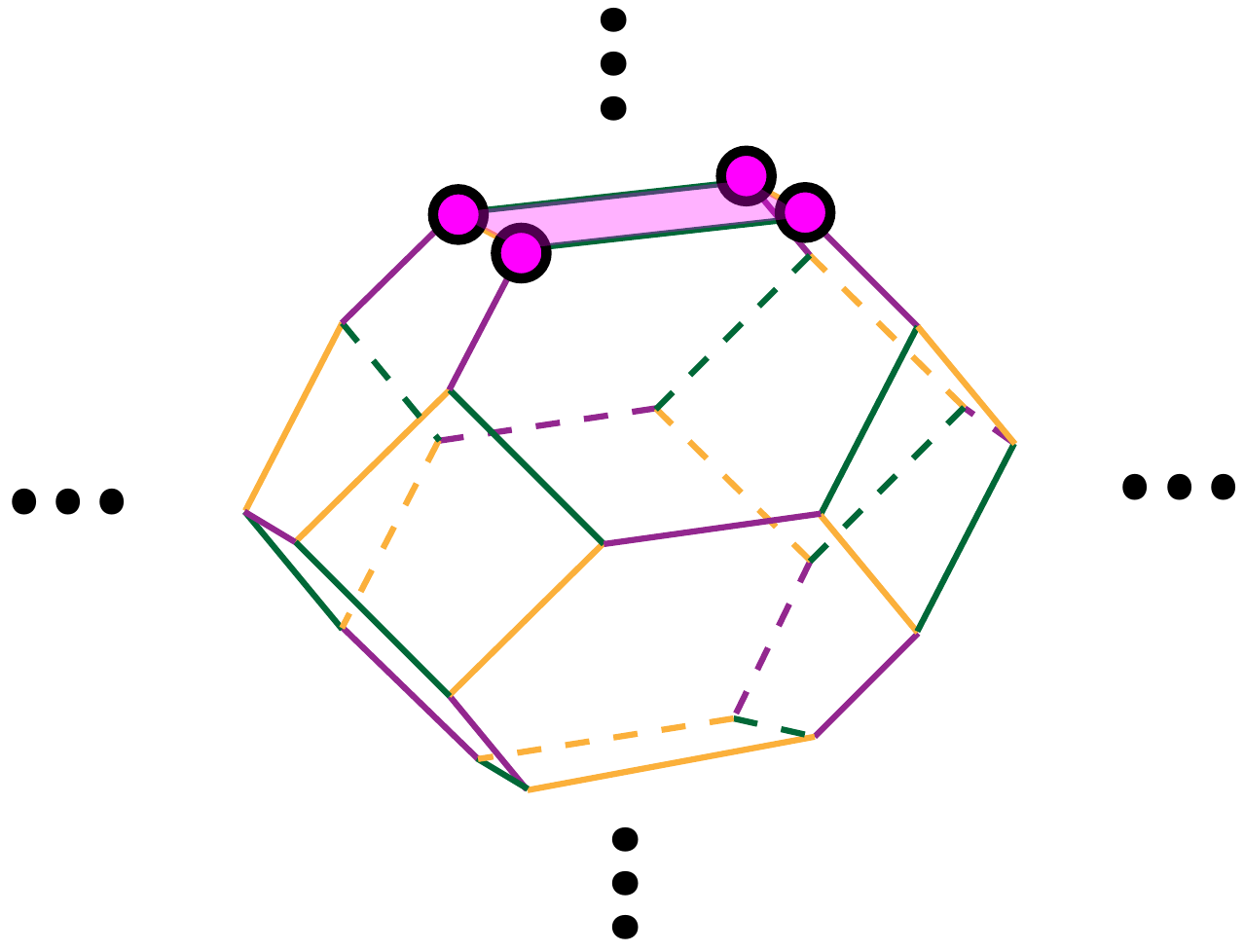}.
\end{align}
And the commutator of $\widetilde{X}_{yg}$ and $\mathcal{A}^{bulk}_{p}(XS)$ is given by
\begin{align}
    \adjincludegraphics[width=3cm,valign=c]{gate_5.pdf}.
\end{align}
Either term is equivalent to a $Z$-stabilizer on a square.

Up to now, we have demonstrated that $\widetilde{X}_{py}$, $\widetilde{X}_{pg}$, and $\widetilde{X}_{yg}$ commute with the transversal-$T$ gate conjugated Hamiltonian in the bulk within the code space. 

\subsection{The magic boundary} \label{sec:magicboundary}

In the presence of boundaries, the translation symmetry is broken, and the stabilizers on the boundaries may not be identical to those in the bulk. Consequently, after conjugating transversal-$T$ gate on the entire system, $\widetilde{X}_{\mathbf{c}_i\mathbf{c}_j}$ cannot terminate on the magic boundary, namely, the $\{pg_{\mathbf{xs}}, py_{\mathbf{xs}}, yg_{\mathbf{xs}}\}$-boundary.

We follow the same procedure to check the commutation relations on the magic boundary. Define $\mathcal{A}^{boundary}_{\mathbf{c}_i}(XS)$ to be the $XS$-stabilizers on the truncated $\mathbf{c}_i$-cells on the boundary, which we introduce in Section~\ref{sec:paulixboundary}. We first consider the commutator between $\widetilde{X}_{py}$, as shown in Fig.~\ref{fig:gate_10}(b), and $\mathcal{A}^{boundary}_{y}(XS)$, which is given by
\begin{align}
    \adjincludegraphics[width=5.5cm,valign=c]{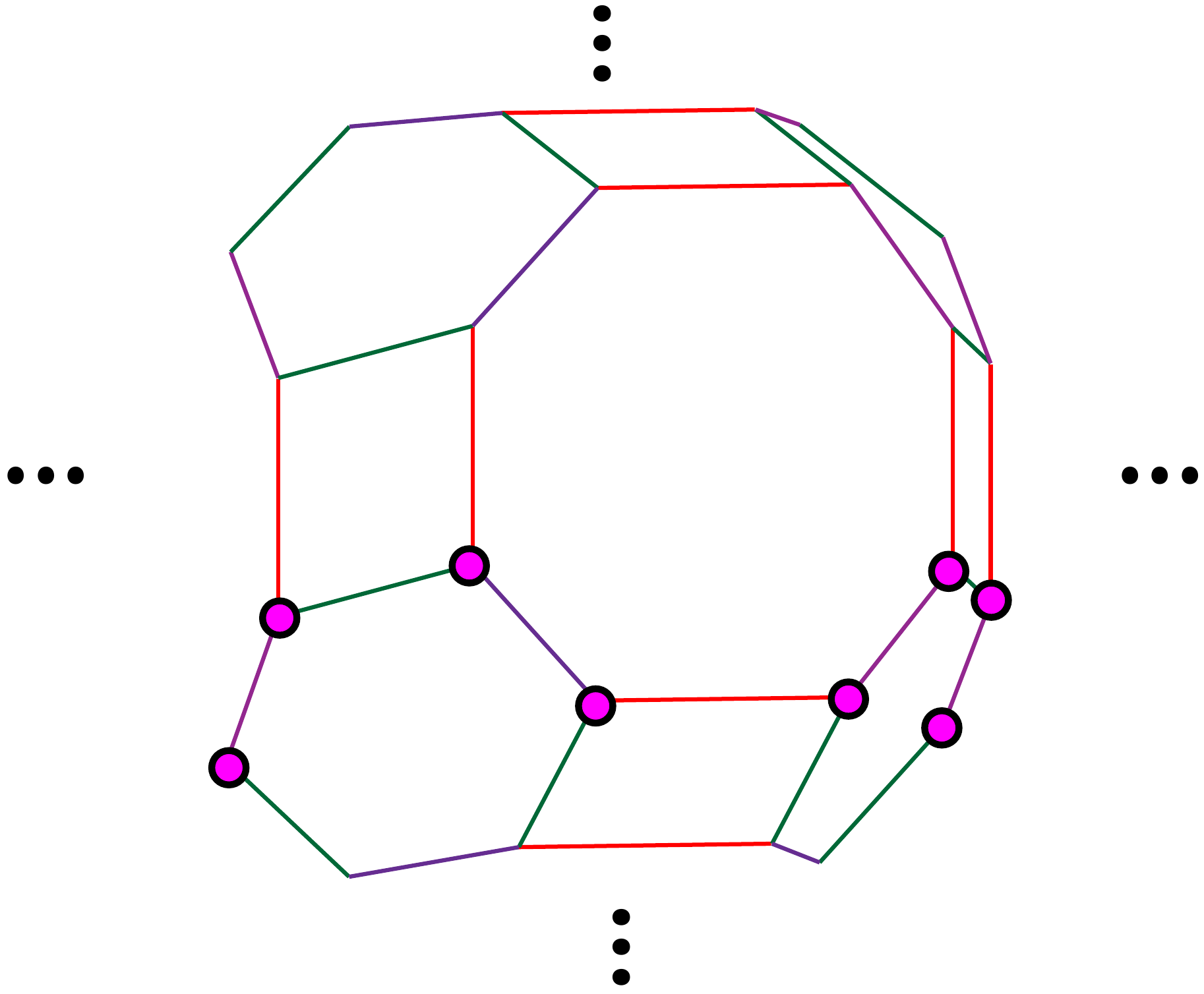}.
\end{align}

From Def.~\ref{def:Paulixboundary}, we know that this operator is not equivalent to the product of $Z$-stabilizers on this boundary. Therefore, $\widetilde{X}_{py}$ and $\mathcal{A}^{boundary}_{y}(XS)$ do not commute with each other within the code space.

Then, we consider the same membrane operator and $\mathcal{A}^{boundary}_{g}(XS)$. The commutator is the following.
\begin{align}
    \adjincludegraphics[width=5.5cm,valign=c]{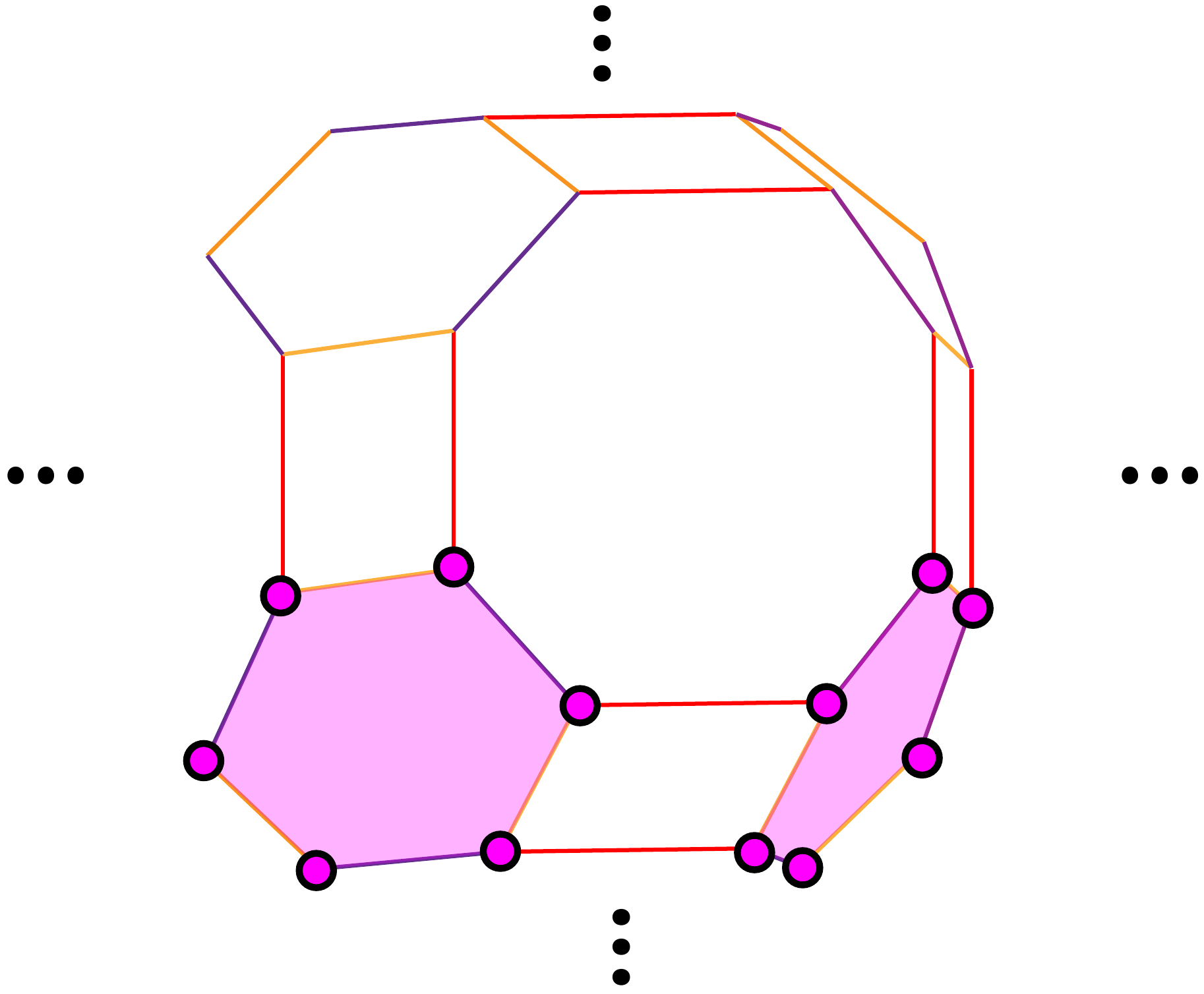}.
\end{align}

We see the resultant operator is a product of two $Z$-stabilizers on hexagons. Therefore, $\widetilde{X}_{py}$ and $\mathcal{A}^{boundary}_{g}(XS)$ commute with each other within the code space.

Finally, we consider the same membrane operator and $\mathcal{A}^{boundary}_{p}(XS)$. The commutator is the following.
\begin{align}
    \adjincludegraphics[width=3cm,valign=c]{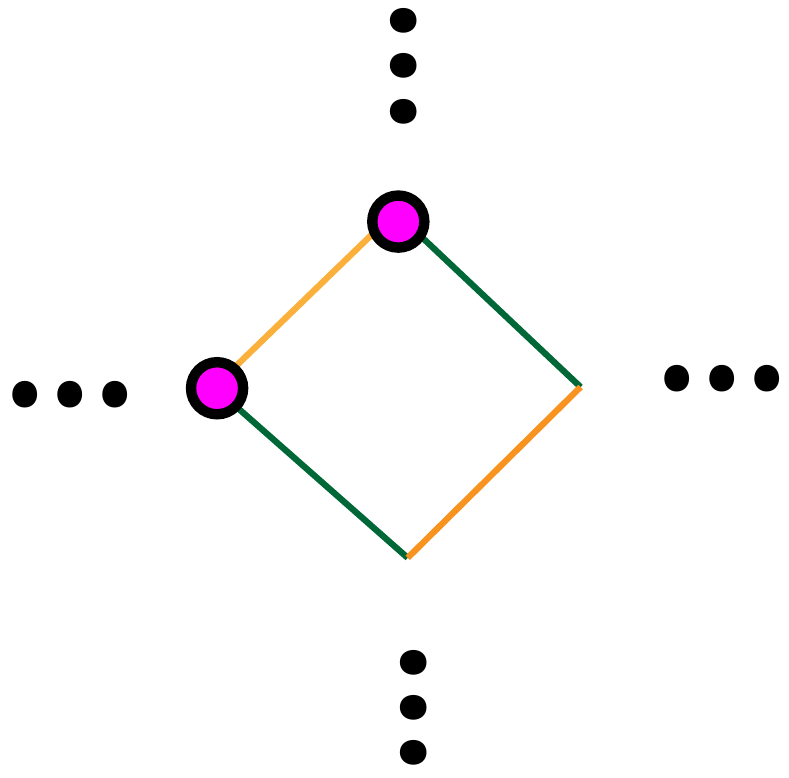}.
\end{align}
Following the same argument, one can demonstrate that the resulting $Z$-operator is not equivalent to the product of stabilizers within the code space. Consequently, $\widetilde{X}_{py}$ does not commute with $\mathcal{A}^{boundary}_{p}(XS)$. As $\widetilde{X}_{py}$ generates a $py_{\mathbf{x}}$ flux loop on its boundary, the above argument is equivalent to say that $py_{\mathbf{x}}$ cannot condense on the magic boundary.

In general, one can show that $\widetilde{X}_{\mathbf{c}_i\mathbf{c}_j}$ does not commute with $\mathcal{A}^{boundary}_{\mathbf{c}_i}(XS)$ and $\mathcal{A}^{boundary}_{\mathbf{c}_j}(XS)$ on the magic boundary, while commutes with $\mathcal{A}^{boundary}_{\mathbf{c}_k}(XS)$, in which $\mathbf{c}_i \in \mathbf{C}$ is the color of the 3D color code, $i,j,k \in \{1,2,3\}$, and $i \neq j \neq k$. 

Furthermore, since the $p_{\mathbf{z}}, y_{\mathbf{z}}$, and $g_{\mathbf{z}}$ excitations do not condense on the $X$-boundary, as we argued in Section~\ref{sec:paulixboundary}, it is straightforward to check they do not condense on the magic boundary as well.

\begin{figure}
    \centering
    \includegraphics[width = 1 \columnwidth]{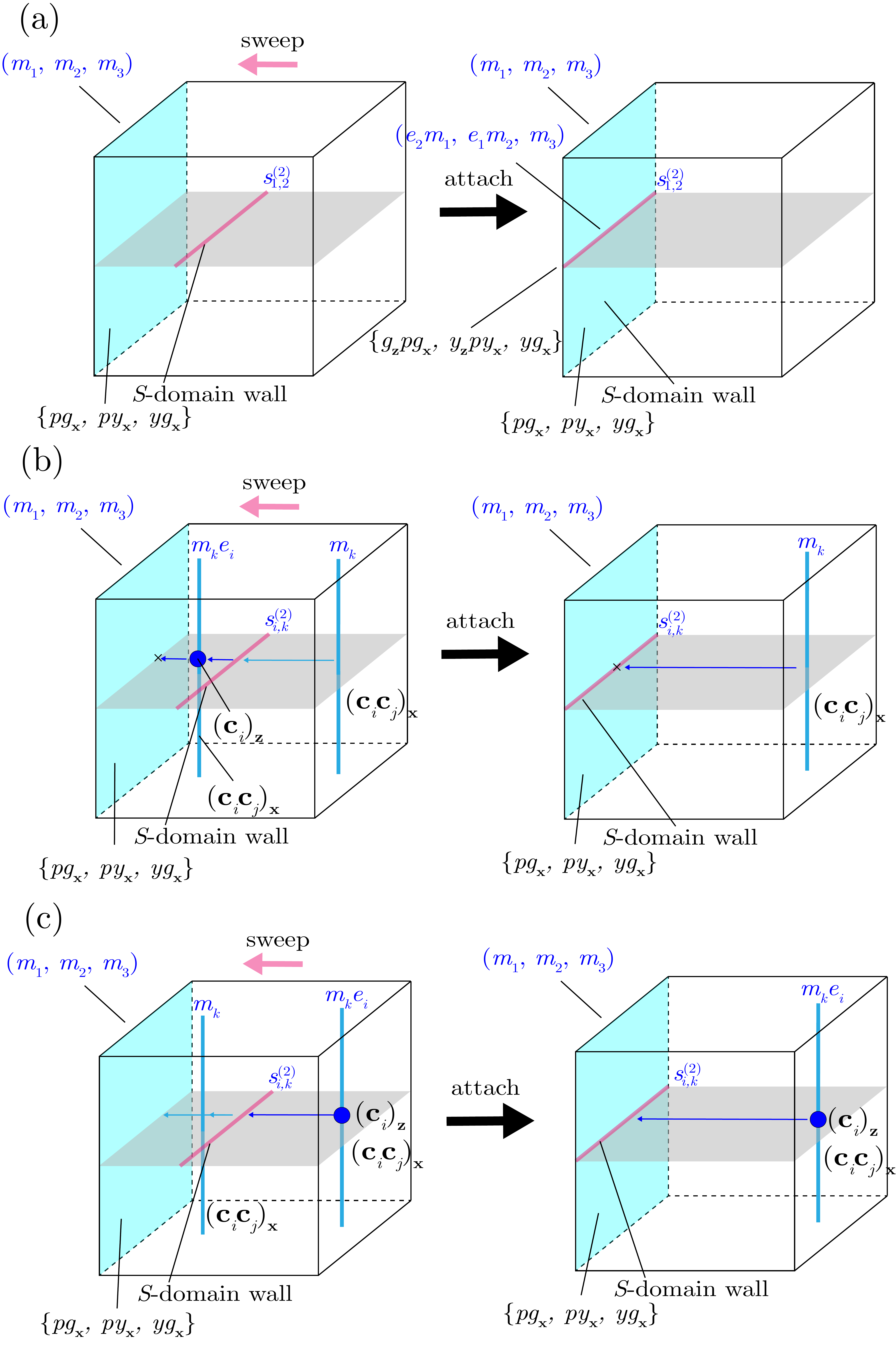}
    \caption{An illustration of the $S$ $(\mathrm{CZ})$-domain wall sweeping across a codimension-1 submanifold in the presence of the $X$-boundary. (a) After sweeping the $S$-domain wall across a codimension-1 submanifold, it attaches to the $X$-boundary and cannot condense on it. We refer to this type of boundaries as the nested boundaries. (b) When traversing the $S$-domain wall, the magnetic flux loops become attached to electric charges. As a result, they cannot condense on the $\{pg_{\mathbf{x}}, py_{\mathbf{x}}, yg_{\mathbf{x}}\}$-boundary. (c) When traversing the $S$-domain wall, the electric charges attached to the magnetic flux loops get re-absorbed into the domain wall. As a result, they can condense on the $\{pg_{\mathbf{x}}, py_{\mathbf{x}}, yg_{\mathbf{x}}\}$-boundary. It is equivalent to say that the magnetic flux loops with electric charges attached to them can condense on the $\{pg_{\mathbf{x}}, py_{\mathbf{x}}, yg_{\mathbf{x}}\}$-boundary with an $S$-domain wall attached.}
    \label{fig:s_domain_wall}
\end{figure}

In summary, we show that $p_{\mathbf{z}}, y_{\mathbf{z}}, g_{\mathbf{z}}, pg_{\mathbf{x}},$ $ py_{\mathbf{x}}$ and $yg_{\mathbf{x}}$ cannot condense on the magic boundary, as illustrated in Fig.~\ref{fig:magic_boundary}(b), and this boundary is beyond the classification of gapped boundaries based on the Lagrangian subgroup.

\subsection{The other elementary boundaries}
\label{sec:elementary}

In this subsection, we study the other elementary boundaries of 3D color code, and their domain-wall condensation properties. We summarize the properties of all the boundaries we study in Table \ref{table:I} and \ref{table:II}. 

\subsubsection{The $\{(\mathbf{c}_i)_{\mathbf{z}}, (\mathbf{c}_i\mathbf{c}_k)_{\mathbf{x}}, (\mathbf{c}_i\mathbf{c}_j)_{\mathbf{x}}\}$-boundaries}

When the boundary can condense at least one $Z$-type excitation, for example, $(\mathbf{c}_i)_{\mathbf{z}}$, then after conjugation of transversal $T$-gate, $\overline{X}_{\mathbf{c}_i \mathbf{c}_j}$ and $\overline{X}_{\mathbf{c}_i \mathbf{c}_k}$ can still terminate on it, for $i \neq j \neq k$.

Without loss of generality, let us consider the boundary that condenses $\{y_{\mathbf{z}}, py_{\mathbf{x}}, yg_{\mathbf{x}}\}$, as we introduce in Section~\ref{sec:zxx}.

After conjugation of transversal-$T$ gate, the $X$-stabilizers become $XS$-stabilizers and the $Z$-stabilizers remain invariant. The condensations thus become $\{y_{\mathbf{z}}, py_{\mathbf{xs}}, yg_{\mathbf{xs}}\}$.

We follow the same procedure in Section~\ref{sec:magicboundary}. The commutator between $\overline{X}_{py}$ and $\mathcal{A}^{boundary}_{p}(XS)$ is the following
\begin{align}
    \adjincludegraphics[width=3cm,valign=c]{gate_25.pdf}.
\end{align}
It is precisely the $Z$-stabilizers on the boundary. Therefore, $\overline{X}_{py}$ and $\mathcal{A}^{boundary}_{p}(XS)$ commute with each other on the boundary within the code space.

The commutator between $\overline{X}_{py}$ and $\mathcal{A}^{boundary}_{g}(XS)$ is the following.
\begin{align}
    \adjincludegraphics[width=5.5cm,valign=c]{gate_21.pdf}. \label{eq:gate_21}
\end{align}
This operator is the product of two $Z$-stabilizers thus $\overline{X}_{py}$ and $\mathcal{A}^{boundary}_{g}(XS)$ commute with each other on the boundary within the code space.

Finally, we examine the commutator between $\overline{X}_{py}$ and $\mathcal{B}^{boundary}_{pg}(Z)$, and find that they also commute with each other. A similar argument can be applied to the case of $\overline{X}_{yg}$. 

Therefore, we conclude that $yg_{\mathbf{x}}$, $py_{\mathbf{x}}$ and $y_{\mathbf{z}}$ can condense on the $\{y_{\mathbf{z}}, py_{\mathbf{xs}}, yg_{\mathbf{xs}}\}$-boundary. This is equivalent to say
\begin{align}
    \{y_{\mathbf{z}}, py_{\mathbf{xs}}, yg_{\mathbf{xs}}\} \equiv \{y_{\mathbf{z}}, py_{\mathbf{x}}, yg_{\mathbf{x}}\},
\end{align}
which is equivalent to Eq.~\eqref{eq:condensation_equivalence} in the TQFT description in Section~\ref{sec:TQFT}.
Although the 3D color code has color permutation symmetry, by purely permuting the color labels, the boundaries we get exhibit identical physical properties, differing only by relabeling of the lattice and excitations. Therefore, we conclude that there are 3 distinct types of boundaries which can condense one $Z$-type excitations on them. For the 3D color code with these type of boundaries, the transversal-$T$ gate corresponds to an emergent 0-form symmetry of the entire code space. 

In general, the relation can be written in the following form:
\begin{align}
    \{(\mathbf{c}_i)_{\mathbf{z}}, (\mathbf{c}_i\mathbf{c}_k)_{\mathbf{xs}}, (\mathbf{c}_i\mathbf{c}_j)_{\mathbf{xs}}\} \equiv \{(\mathbf{c}_i)_{\mathbf{z}}, (\mathbf{c}_i\mathbf{c}_k)_{\mathbf{x}}, (\mathbf{c}_i\mathbf{c}_j)_{\mathbf{x}}\},
\end{align}
in which $\mathbf{c}_i, \mathbf{c}_j, \mathbf{c}_k \in \mathbf{C}$ and $i \neq j \neq k$. This general relation also corresponds to Eq.~\eqref{eq:condensation_equivalence}.

According to the case we discuss in Eq.~\eqref{eq:sij_2} in the TQFT description, there are two types of $S$-domain wall that can condense on this boundary, $S_{\mathbf{c}_i\mathbf{c}_j}$- and $S_{\mathbf{c}_i\mathbf{c}_k}$-domain wall, while $S_{\mathbf{c}_j\mathbf{c}_k}$-domain wall cannot condense on this boundary. We leave the discussion of this case to the next subsection.

\subsubsection{The $\{(\mathbf{c}_i)_{\mathbf{z}}, (\mathbf{c}_j)_{\mathbf{z}}, (\mathbf{c}_i\mathbf{c}_j)_{\mathbf{x}}\}$-boundary}

Without loss of generality, we consider the $\{y_{\mathbf{z}}, g_{\mathbf{z}}, yg_{\mathbf{x}}\}$-boundary to illustrate the properties of this type of boundary. The boundary stabilizers can be found in Section~\ref{sec:zzx}.

After conjugation of the transversal-$T$ gate, the condensations on the boundary become $\{y_{\mathbf{z}}, g_{\mathbf{z}}, yg_{\mathbf{xs}}\}$. The commutator between $\overline{X}_{py}$ and $\mathcal{A}^{boundary}_{g}(XS)$ equals to identity. Therefore, $\mathcal{A}^{boundary}_{g}(XS)$ and $\overline{X}_{py}$ commute with each other within the code space and $yg_{\mathbf{x}}$ can condense on the boundary, which is equivalent to say
\begin{align}
    \{y_{\mathbf{z}}, g_{\mathbf{z}}, yg_{\mathbf{xs}}\} \equiv \{y_{\mathbf{z}}, g_{\mathbf{z}}, yg_{\mathbf{x}}\}.
\end{align}

Same as the TQFT description in Eq.~\eqref{eq:eem}, all three types of $S$-domain walls can condense on this boundary. Therefore, the total number of distinct types of the $\{(\mathbf{c}_i)_{\mathbf{z}}, (\mathbf{c}_j)_{\mathbf{z}}, (\mathbf{c}_i\mathbf{c}_j)_{\mathbf{x}}\}$-boundaries is 3.

\subsubsection{The $\{y_{\mathbf{z}}, g_{\mathbf{z}}, p_{\mathbf{z}}\}$-boundary}

We also discuss the $\{y_{\mathbf{z}}, g_{\mathbf{z}}, p_{\mathbf{z}}\}$-boundary ($Z$-boundary or the $(e_1,e_2,e_3)$-boundary in the TQFT description). The boundary condensations remain invariant under the conjugation of transversal-$T$ gate or under the conjugation of transversal-$S$ gate on codimension-1 submanifold. Therefore, the total number of distinct types of the $Z$-boundaries is 1.

\subsubsection{The folded boundaries}

In this subsection, we investigate the condensation properties of the folded boundaries introduced in Section~\ref{sec:folded_boundary}.

Consider the $\{y_{\mathbf{z}} g_{\mathbf{z}}, pg_{\mathbf{x}} py_{\mathbf{x}}, yg_{\mathbf{x}}\}$-boundary as an example for the fold$|m$-boundaries. After conjugation of transversal-$T$ gate, the boundary condensation set becomes $\{y_{\mathbf{z}} g_{\mathbf{z}}, pg_{\mathbf{xs}} py_{\mathbf{xs}}, yg_{\mathbf{x}}\}$. One can check the commutator between $\overline{X_{pg} X_{py}}$ and the boundary $XS$-stabilizer $\mathcal{A}_{y}^{boundary}(XS) \mathcal{A}_{g}^{boundary}(XS)$ is equivalent to the product of $Z$-stabilizers on the boundary. Therefore, we have the following relation:
\begin{align}
    \{y_{\mathbf{z}} g_{\mathbf{z}}, pg_{\mathbf{xs}} py_{\mathbf{xs}}, yg_{\mathbf{x}}\} \equiv \{y_{\mathbf{z}} g_{\mathbf{z}}, pg_{\mathbf{x}} py_{\mathbf{x}}, yg_{\mathbf{x}}\}
\end{align}

In the TQFT description, this corresponds to Eq.~\eqref{eq:S_fold_m}, Eq.~\eqref{eq:T_fold} and Eq.~\eqref{eq:fold_m_condense}. The total number of distinct types of fold$|m$-boundary is 3.

For the fold$|e$-boundaries, let's consider the $\{y_{\mathbf{z}} g_{\mathbf{z}}, pg_{\mathbf{x}} py_{\mathbf{x}}, p_{\mathbf{z}}\}$-boundary as an example. After conjugation of the transversal-$T$ gate, the condensation set becomes $\{y_{\mathbf{z}} g_{\mathbf{z}}, pg_{\mathbf{xs}} py_{\mathbf{xs}}, p_{\mathbf{z}}\}$. Similarly, one can check the commutator between $\overline{X_{pg} X_{py}}$ and the boundary $XS$-stabilizer $\mathcal{A}_{y}^{boundary}(XS) \mathcal{A}_{g}^{boundary}(XS)$ equals identity within the code space. Therefore, $pg_{\mathbf{xs}} py_{\mathbf{xs}}$ can condense on the bounday, which is equivalent to say
\begin{align}
    \{y_{\mathbf{z}} g_{\mathbf{z}}, pg_{\mathbf{xs}} py_{\mathbf{xs}}, p_{\mathbf{z}}\} \equiv \{y_{\mathbf{z}} g_{\mathbf{z}}, pg_{\mathbf{x}} py_{\mathbf{x}}, p_{\mathbf{z}}\}
\end{align}
This relation corresponds to Eq.~\eqref{eq:T_fold}. Furthermore, one can check all the three types of $S$-domain walls can condense on this boundary, which corresponds to Eq.~\eqref{eq:S_fold_e} in the perspective of TQFT. The total number of distinct types of the fold$|e$-boundary is 3.

\subsubsection{The $pg_{\mathbf{x}} py_{\mathbf{x}} yg_{\mathbf{x}}$- and the $y_{\mathbf{z}} g_{\mathbf{z}} p_{\mathbf{z}}$-boundaries}

Consider the lattice model we introduce in Section~\ref{sec:m1m2m3} and \ref{sec:e1e2e3}. After the conjugation of transversal-$T$ gate, the boundary condensation set becomes $\{y_{\mathbf{z}} g_{\mathbf{z}}, g_{\mathbf{z}} p_{\mathbf{z}}, pg_{\mathbf{xs}} py_{\mathbf{xs}} yg_{\mathbf{xs}}\}$ and $\{pg_{\mathbf{xs}} py_{\mathbf{xs}}, py_{\mathbf{xs}} yg_{\mathbf{xs}}, y_{\mathbf{z}} g_{\mathbf{z}} p_{\mathbf{z}}\}$, respectively.

For the former case, one can check that the commutator between the $\overline{X_{pg} X_{py} X_{yg}}$ operator and the $XS$-stabilizers on the boundary can be expressed as product of $Z$-stabilizers (product of short $Z$-strings corresponding to two different copies). Therefore, we have the following relation:
\begin{equation}
    \begin{aligned}
        &\{y_{\mathbf{z}} g_{\mathbf{z}}, g_{\mathbf{z}} p_{\mathbf{z}}, pg_{\mathbf{xs}} py_{\mathbf{xs}} yg_{\mathbf{xs}}\} \notag \\&\equiv \{y_{\mathbf{z}} g_{\mathbf{z}}, g_{\mathbf{z}} p_{\mathbf{z}}, pg_{\mathbf{x}} py_{\mathbf{x}} yg_{\mathbf{x}}\}.
    \end{aligned}
\end{equation}

For the latter case, one can check that the commutator between $\overline{X_{pg} X_{py}}$ and the boundary $XS$-stabilizers can be expressed as product of $Z$-stabilizers (product of short $Z$-strings corresponding to three different copies). Therefore, we have the following relation:
\begin{equation}
    \begin{aligned}
        &\{pg_{\mathbf{xs}} py_{\mathbf{xs}}, py_{\mathbf{xs}} yg_{\mathbf{xs}}, y_{\mathbf{z}} g_{\mathbf{z}} p_{\mathbf{z}}\} \\&\equiv \{pg_{\mathbf{x}} py_{\mathbf{x}}, py_{\mathbf{x}} yg_{\mathbf{x}}, y_{\mathbf{z}} g_{\mathbf{z}} p_{\mathbf{z}}\}.
    \end{aligned}
\end{equation}
In the perspective of TQFT, these relations correspond to Eq.~\eqref{eq:mmm_1}, ~\eqref{eq:eee_1} and ~\eqref{eq:eee_2}.

\subsection{The nested boundaries} \label{sec:nested}

\begin{figure}
    \centering
    \includegraphics[width = 1 \columnwidth]{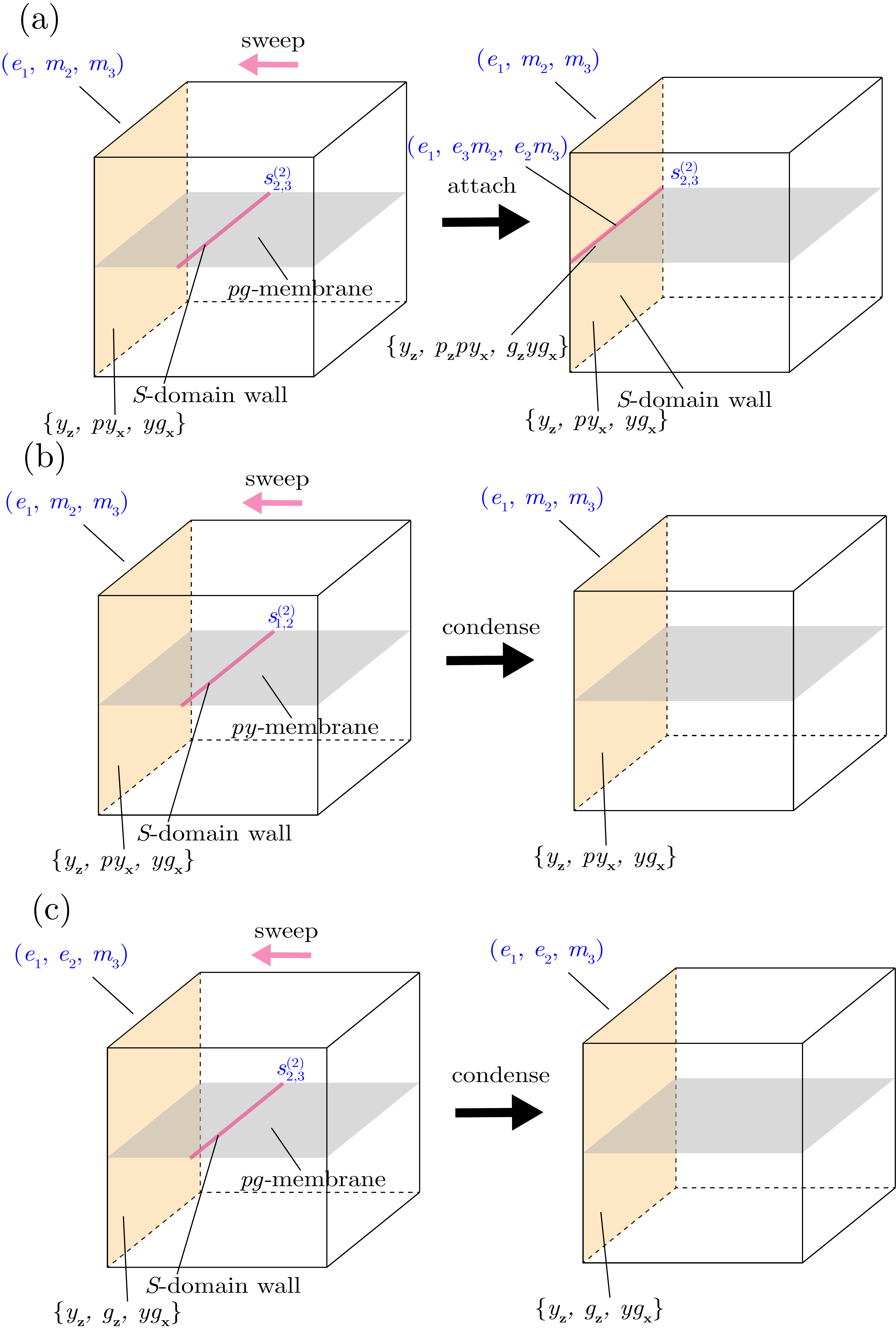}
    \caption{An illustration of the $S$ $(\mathrm{CZ}$)-domain wall sweeping across a codimension-1 submanifold in the presence of an boundary with at least one $Z$-type condensation. (a) When the boundary has one $Z$-type condensations, for example $y_{\mathbf{z}}$, the $S$-domain wall on the $pg$-membrane cannot condense on it. (b) However, the other two $S$-domain walls on the $py$- and $yg$-membranes can condense on this boundary. (c) When the boundary has two or more $Z$-type condensations, all of the three types of $S$-domain walls can condense on it.}
    \label{fig:s_domain_wall_2}
\end{figure}

As suggested in Section~\ref{sec:TQFT}, there are four types of boundaries that some of the $S$-domain walls cannot condense on, i.e., $(e_i, m_j, m_k)$, $(m_1, m_2, m_3)$, $(e_i e_j, m_i m_j, m_k)$ and $(m_1 m_2, m_2 m_3, e_1 e_2 e_3)$ boundaries in the TQFT descriptions. The corresponding color-code boundaries are the $\mathbf{c}_i$-color boundaries, the $X$-boundary, the fold$|m$-boundaries, and the  $y_{\mathbf{z}} g_{\mathbf{z}} p_{\mathbf{z}}$-boundary, respectively. We study the properties of these boundaries in this subsection, and summarize them in Fig.~\ref{fig:web_1}.

\subsubsection{The $\mathbf{c}_i$-color boundary} \label{sec:nested_zxx}

Consider the $\{(\mathbf{c}_i)_{\mathbf{z}}, (\mathbf{c}_i\mathbf{c}_j)_{\mathbf{x}},(\mathbf{c}_i\mathbf{c}_k)_{\mathbf{x}}\}$-boundary, the $S_{\mathbf{c}_i\mathbf{c}_j}$- or $S_{\mathbf{c}_i\mathbf{c}_k}$-domain wall can condense on the boundary. However, the $S_{\mathbf{c}_j\mathbf{c}_k}$-domain wall can not condense on it. It attaches to the boundary and generates a codimension-2 (1D) domain-wall defect. The condensation set on that defect becomes $\{(\mathbf{c}_i)_{\mathbf{z}}, (\mathbf{c}_k)_{\mathbf{z}} (\mathbf{c}_i\mathbf{c}_k)_{\mathbf{x}}, (\mathbf{c}_j)_{\mathbf{z}} (\mathbf{c}_i\mathbf{c}_j)_{\mathbf{x}}\}$, while other part of the boundary is unchanged. The corresponding condensation in the TQFT description is $(e_i, m_j e_k, m_k e_j)$,  which has been introduced in Section~\ref{sec:TQFT}. The total number of distinct types of the $\mathbf{c}_i$-color boundaries with non-condensing $S$-domain wall is 3. We denote this type of boundary as the {\it nested boundary}, as illustrated in Fig.~\ref{fig:s_domain_wall}, Fig.~\ref{fig:s_domain_wall_2}(a), and Fig.~\ref{fig:web_1}(d).

\subsubsection{The $X$-boundary} \label{sec:nested_xxx}

Similarly, in the case the boundary condensation is $\{pg_\mathbf{x}, py_\mathbf{x}, yg_\mathbf{x}\}$ (X-boundary), $S$-domain walls on all the three-colored membranes cannot condense on this boundary, and introduce a codimension-2 defect on it, as depicted in Fig.~\ref{fig:s_domain_wall}. For example, consider the $S_{yg}$-domain wall attaches to the $X$-boundary. The condensation set of this codimension-2 domain-wall defect becomes $\{g_\mathbf{z} pg_\mathbf{x}, y_\mathbf{z} py_\mathbf{x}, yg_\mathbf{x}\}$, while the other part of the boundary is unchanged. In the TQFT description, this nested defect is denoted by $(e_2 m_1, e_1 m_2, m_3)$. In the general cases, one can attach composites of $S$-domain walls on the $X$-boundary. The total number of distinct $X$-boundaries with nested defects is 7, as illustrated in Fig.~\ref{fig:web_1}(a).

\subsubsection{The fold$|m$-boundary}\label{sec:nested_fold}

Consider the $\{y_{\mathbf{z}} g_{\mathbf{z}}, pg_{\mathbf{x}} py_{\mathbf{x}}, yg_{\mathbf{x}}\}$-boundary as an example. the $S_{yg}$ domain wall can condense on it. However, the $S_{pg}$ and $S_{yp}$ domain wall cannot. At the intersection of this boundary and the $S_{pg}$-domain wall, the condensation set becomes  $\{y_{\mathbf{z}} g_{\mathbf{z}}, p_{\mathbf{z}} pg_{\mathbf{x}} py_{\mathbf{x}}, g_{\mathbf{z}} yg_{\mathbf{x}}\}$, while the other part of the boundary remain unchanged. In the TQFT description, this boundary is denoted as $(e_1 e_2, e_3 m_1 m_2, e_2 m_3)$. The total number of distinct types of fold$|m$-boundaries with nested boundaries is 3, as illustrated in Fig.~\ref{fig:web_1}(c).

\begin{figure*}
    \centering
    \includegraphics[width = 14cm]{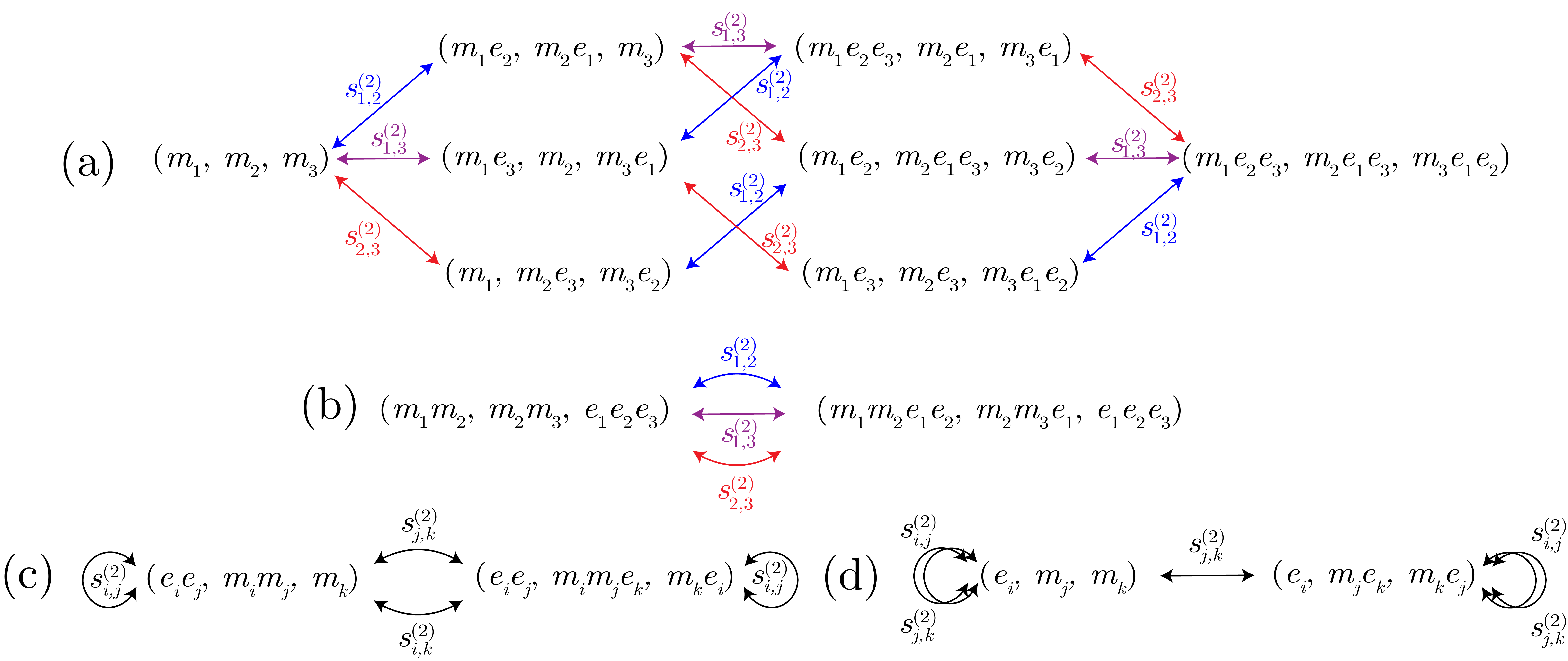}
    \caption{This figure shows the web of domain-wall attachment. Each pair of condensations can be viewed as a pair of the boundary condensation and the nested condensation. For example, consider the pair $(m_1, m_2, m_3)$ and $(m_1 e_2, m_2 e_1 e_3, m_3 e_2)$, If $(m_1, m_2, m_3)$ is the condensation on the boundary, them by attaching the composite domain wall $s_{1,2}^{(2)} s_{2,3}^{(2)}$, one could get a nested boundary with $(m_1 e_2, m_2 e_1 e_3, m_3 e_2)$ condensation on the 1D intersection and $(m_1, m_2, m_3)$ condensation on other part. The inverse process is also possible, as indicated by the arrows in the figure. One can start with the $(m_1 e_2, m_2 e_1 e_3, m_3 e_2)$-boundary, and attach $s_{1,2}^{(2)} s_{2,3}^{(2)}$ domain wall on it. The condensation set at the intersection becomes $(m_1, m_2, m_3)$. The total number of distinct types of the nested boundaries is 70, in which 56 are associated with (a), 2 are associated with (b), 6 are associated with (c), and 6 are associated with (d).}
    \label{fig:web_1}
\end{figure*}

\subsubsection{The $y_{\mathbf{z}} g_{\mathbf{z}} p_{\mathbf{z}}$-boundary}

The condensation set on the $y_{\mathbf{z}} g_{\mathbf{z}} p_{\mathbf{z}}$-boundary is $\{pg_{\mathbf{x}} py_{\mathbf{x}}, py_{\mathbf{x}} yg_{\mathbf{x}}, y_{\mathbf{z}} g_{\mathbf{z}} p_{\mathbf{z}}\}$. The lattice construction is given in Section~\ref{sec:e1e2e3}. One can show that by sweeping the $S$-domain wall across the codimension-1 submanifold, at the intersection with the boundary, the $S$-domain wall cannot terminate. In the TQFT description, this relation corresponds to Eq.~\eqref{eq:e1e2e3_s12}. The total number of distinct types of nested boundaries associated with this boundary is 1, as illustrated in Fig.~\ref{fig:web_1}(b).

\subsubsection{The other boundaries}

We summarize the the cases for the other boundaries that some $S$-domain walls cannot condense in Fig.~\ref{fig:web_1}.

\subsection{Summary of boundary types}
In summary, we identify 101 distinct types of boundaries for the 3D color code, with one being the magic boundary and 70 being nested boundaries. While a trivial color permutation ($y \to g$, $g \to p$, $p \to r$, and $r \to y$) can generate additional boundary types, the new boundaries created by this permutation exhibit identical physical properties, differing only by relabeling of the lattice and excitations. Therefore, we exclude these from our count. We summarize the properties of these boundaries in Table~\ref{table:I}, Table~\ref{table:II}, and Fig.~\ref{fig:web_1}.

\begin{table*}
\centering
\resizebox{1.9\columnwidth}{!}{
\begin{tabular}{||c | c | c | c | c ||} 
 \hline
 Condensations on boundaries & TQFT/toric code & Distinct & Types & Sections\\

  & counterpart & types &   &  \\[0.5ex] 
 \hline\hline
 $\{y_{\mathbf{z}}, g_{\mathbf{z}}, p_{\mathbf{z}}\}$ & $(e_1, e_2, e_3)$ & 1 & Elementary& \ref{sec:paulizboundary}  \\
 \hline
 $\{pg_\mathbf{x}, py_\mathbf{x}, yg_\mathbf{x}\}$ & $(m_1, m_2, m_3)$ & 1 & Elementary& \ref{sec:paulixboundary}\\
 \hline
 $\{(\mathbf{c}_j \mathbf{c}_k)_{\mathbf{x}} (\mathbf{c}_j)_{\mathbf{z}}, (\mathbf{c}_i \mathbf{c}_k)_{\mathbf{x}} (\mathbf{c}_i)_{\mathbf{z}}, (\mathbf{c}_i \mathbf{c}_j)_{\mathbf{x}}\}$ & $(m_i e_j, m_j e_i, m_k)$ & 3 & Elementary  & \ref{sec:other_boundary_color}\\
  \hline
 $\{(\mathbf{c}_j \mathbf{c}_k)_{\mathbf{x}} (\mathbf{c}_j)_{\mathbf{z}} (\mathbf{c}_k)_{\mathbf{z}}, (\mathbf{c}_i \mathbf{c}_k)_{\mathbf{x}} (\mathbf{c}_i)_{\mathbf{z}}, (\mathbf{c}_i \mathbf{c}_j)_{\mathbf{x}}(\mathbf{c}_i)_{\mathbf{z}}\}$ & $(m_i e_j e_k, m_j e_i, m_k e_i)$ & 3 & Elementary  & \ref{sec:other_boundary_color}\\
 \hline
 $\{pg_{\mathbf{x}} g_{\mathbf{z}} p_{\mathbf{z}}, py_{\mathbf{x}} y_{\mathbf{z}} p_{\mathbf{z}}, yg_{\mathbf{x}} y_{\mathbf{z}} g_{\mathbf{z}}\}$& $(m_1 e_2 e_3, m_2 e_1 e_3, m_3 e_1 e_2)$ & 1 & Elementary  & \ref{sec:other_boundary_color}\\
 \hline
 $\{(\mathbf{c}_i)_{\mathbf{z}}, (\mathbf{c}_i\mathbf{c}_k)_{\mathbf{x}}, (\mathbf{c}_i\mathbf{c}_j)_{\mathbf{x}}\}$ & $(e_i, m_j, m_k)$ & 3 & Elementary & \ref{sec:zxx}\\
 \hline
 $\{(\mathbf{c}_i)_{\mathbf{z}}, (\mathbf{c}_i \mathbf{c}_k)_{\mathbf{x}} (\mathbf{c}_k)_{\mathbf{z}}, (\mathbf{c}_i \mathbf{c}_j)_{\mathbf{x}}(\mathbf{c}_j)_{\mathbf{z}}\}$& $(e_i, m_j e_k, m_k e_j)$ & 3 &  Elementary & \ref{sec:other_boundary_color}\\
 \hline
 $\{(\mathbf{c}_i)_{\mathbf{z}}, (\mathbf{c}_j)_{\mathbf{z}}, (\mathbf{c}_i\mathbf{c}_j)_{\mathbf{x}}\}$ & $(e_i, e_j, m_k)$ & 3 & Elementary&\ref{sec:zzx}\\
 \hline
 $\{(\mathbf{c}_i)_{\mathbf{z}} (\mathbf{c}_j)_{\mathbf{z}}, (\mathbf{c}_j \mathbf{c}_k)_{\mathbf{x}} (\mathbf{c}_i \mathbf{c}_k)_{\mathbf{x}}, (\mathbf{c}_k)_{\mathbf{z}}\}$ & $(e_i e_j, m_i m_j, e_k)$ & 3  & Elementary& \ref{sec:folded_boundary}\\
 \hline
 $\{(\mathbf{c}_i)_{\mathbf{z}} (\mathbf{c}_j)_{\mathbf{z}}, (\mathbf{c}_j \mathbf{c}_k)_{\mathbf{x}} (\mathbf{c}_i \mathbf{c}_k)_{\mathbf{x}}, (\mathbf{c}_i \mathbf{c}_j)_{\mathbf{x}}\}$ & $(e_i e_j, m_i m_j, m_k)$ & 3 & Elementary& \ref{sec:folded_boundary}\\
 \hline
 $\{(\mathbf{c}_i)_{\mathbf{z}} (\mathbf{c}_j)_{\mathbf{z}}, (\mathbf{c}_j \mathbf{c}_k)_{\mathbf{x}}(\mathbf{c}_i \mathbf{c}_k)_{\mathbf{x}}(\mathbf{c}_k)_{\mathbf{z}}, (\mathbf{c}_i \mathbf{c}_j)_{\mathbf{x}}(\mathbf{c}_j)_{\mathbf{z}}\}$& $(e_i e_j, m_i m_j e_k, m_k e_j)$ & 3 &  Elementary & \ref{sec:other_boundary_color}\\
 \hline
  $\{y_{\mathbf{z}} g_{\mathbf{z}}, g_{\mathbf{z}} p_{\mathbf{z}}, pg_{\mathbf{x}} py_{\mathbf{x}} yg_{\mathbf{x}}\}$ & $(e_1 e_2, e_2 e_3, m_1 m_2 m_3)$ & 1 & Elementary  & \ref{sec:m1m2m3} \\
 \hline
 $\{pg_{\mathbf{x}} py_{\mathbf{x}}, py_{\mathbf{x}} yg_{\mathbf{x}}, y_{\mathbf{z}} g_{\mathbf{z}} p_{\mathbf{z}}\}$& $(m_1 m_2, m_2 m_3, e_1 e_2 e_3)$ & 1 & Elementary  & \ref{sec:e1e2e3}\\
 \hline
 $\{pg_{\mathbf{x}} py_{\mathbf{x}} y_{\mathbf{z}} g_{\mathbf{z}}, py_{\mathbf{x}} yg_{\mathbf{x}} y_{\mathbf{z}}, y_{\mathbf{z}} g_{\mathbf{z}} p_{\mathbf{z}}$\} & $(m_1 m_2 e_1 e_2, m_2 m_3 e_1, e_1 e_2 e_3)$ & 1 & Elementary  & \ref{sec:other_boundary_color}\\
 \hline
 $\{pg_\mathbf{xs}, py_\mathbf{xs}, yg_\mathbf{xs}\}$ & $(m_1 s_{2,3}^{(2)}, m_2 s_{3,1}^{(2)}, m_3 s_{1,2}^{(2)})$ & 1 & Magic& \ref{sec:magicboundary}\\
 [1ex] 
  \hline
\end{tabular}}
\caption{This table summarizes all the types of elementary boundaries as well as the magic boundary.}
\label{table:I}
\end{table*}

\begin{table*}
\centering
\resizebox{1.9\columnwidth}{!}{
\begin{tabular}{||c | c | c | c | c | c | c | c | c ||} 
 \hline
 Condensations on  & TQFT/toric code  & $s^{(2)}_{i,j}$  & $s^{(2)}_{j, k}$  & $s^{(2)}_{i, k}$  & $s^{(2)}_{i,j}s^{(2)}_{j, k}$ & $s^{(2)}_{j,k}s^{(2)}_{i, k}$& $s^{(2)}_{i,j}s^{(2)}_{i, k}$&  $s^{(3)}_{i, j, k}$  \\ 
 boundaries & counterpart & ($S_{\mathbf{c}_i\mathbf{c}_j}$) & ($S_{\mathbf{c}_j\mathbf{c}_k}$) & ($S_{\mathbf{c}_i\mathbf{c}_k}$) & ($S_{\mathbf{c}_i\mathbf{c}_j}S_{\mathbf{c}_j\mathbf{c}_k}$)& ($S_{\mathbf{c}_j\mathbf{c}_k}S_{\mathbf{c}_i\mathbf{c}_k}$) & ($S_{\mathbf{c}_i\mathbf{c}_j}S_{\mathbf{c}_i\mathbf{c}_k}$) &($T$)\\[0.5ex] 
 \hline\hline
 $\{y_{\mathbf{z}}, g_{\mathbf{z}}, p_{\mathbf{z}}\}$ & $(e_1, e_2, e_3)$ & Y & Y & Y & Y & Y & Y & Y\\ 
 \hline
 $\{pg_\mathbf{x}, py_\mathbf{x}, yg_\mathbf{x}\}$ & $(m_1, m_2, m_3)$ & N & N & N & N & N & N & N\\
 \hline
 $\{(\mathbf{c}_i)_{\mathbf{z}}, (\mathbf{c}_i\mathbf{c}_k)_{\mathbf{x}}, (\mathbf{c}_i\mathbf{c}_j)_{\mathbf{x}}\}$ & $(e_i, m_j, m_k)$ & Y & N & Y & N & N & Y & Y \\
 \hline
 $\{(\mathbf{c}_i)_{\mathbf{z}}, (\mathbf{c}_j)_{\mathbf{z}}, (\mathbf{c}_i\mathbf{c}_j)_{\mathbf{x}}\}$ & $(e_i, e_j, m_k)$ & Y & Y & Y & Y & Y & Y & Y\\
 \hline
 $\{(\mathbf{c}_i)_{\mathbf{z}} (\mathbf{c}_j)_{\mathbf{z}}, $ & $(e_i e_j, m_i m_j, e_k)$ & Y & Y & Y & Y & Y & Y & Y\\
 $ (\mathbf{c}_j \mathbf{c}_k)_{\mathbf{x}} (\mathbf{c}_i \mathbf{c}_k)_{\mathbf{x}}, (\mathbf{c}_k)_{\mathbf{z}}\}$ &  &  &  &  &  &  &  & \\
 \hline
 $\{(\mathbf{c}_i)_{\mathbf{z}} (\mathbf{c}_j)_{\mathbf{z}}, (\mathbf{c}_j \mathbf{c}_k)_{\mathbf{x}} (\mathbf{c}_i \mathbf{c}_k)_{\mathbf{x}},$ & $(e_i e_j, m_i m_j, m_k)$ & Y & N & N & N & Y & N & Y \\
  $(\mathbf{c}_i \mathbf{c}_j)_{\mathbf{x}}\}$ &  &  &  &  &  &  &  &  \\
 \hline
 $\{y_{\mathbf{z}} g_{\mathbf{z}}, g_{\mathbf{z}} p_{\mathbf{z}}, pg_{\mathbf{x}} py_{\mathbf{x}} yg_{\mathbf{x}}\}$ & $(e_1 e_2, e_2 e_3, m_1 m_2 m_3)$ & Y & Y & Y & Y & Y & Y & Y \\
 \hline
 $\{pg_{\mathbf{x}} py_{\mathbf{x}}, py_{\mathbf{x}} yg_{\mathbf{x}}, y_{\mathbf{z}} g_{\mathbf{z}} p_{\mathbf{z}}\}$& $(m_1 m_2, m_2 m_3, e_1 e_2 e_3)$ & N & N & N & Y & Y & Y & Y\\
 \hline
\end{tabular}}
\caption{This table summarizes the condensation properties of the $S$ and $T$ domain walls. ``Y" indicates that a certain type of domain wall can condense on the boundary, while ``N" indicates it cannot.}
\label{table:II}
\end{table*}

\section{Summary and discussion} \label{sec:discussion}

In this paper, we systematically classify the codimension-1 (2D) boundaries and codimension-2 (1D) nested boundaries of the 3D color code.

We first explicitly construct the $X$-boundaries and the $Z$-boundaries of 3D color code. By applying the unfolding unitaries, we further demonstrate that they are equivalent to the all-smooth and all-rough boundaries of three copies of 3D toric codes, respectively.

We demonstrate that, following the conjugation of the transversal-$T$ gate or the sweeping of the $T$-domain wall across the entire system, the 3D color-code stabilizers located on the boundary transform into non-Pauli stabilizers. Instead of condensing on the boundaries, the $T$-domain wall attaches to the $X$-boundary, and generates an exotic boundary. This is also equivalent to say that the transversal-$T$ gate is no longer a logical gate of the 3D color code when the $X$-boundary is present. We refer to this new type of boundary as the magic boundary, since it goes beyond the Pauli stabilizer formalism and can be potentially used for implementing fault-tolerant non-Clifford logical gate such as in the case of fractal topological codes \cite{zhu2022topological, dua2023quantum}.  It also goes beyond the classification of Lagrangian subgroup, which means the condensations on the boundary are no longer a subset of the electric and magnetic excitations. We explicitly verify with the lattice model that neither the $X$-type (magnetic flux), nor the $Z$-type (electric charge) excitations can condense on the magic boundary, while the combination of the $X$-type excitation (magnetic flux) and the $S$-domain wall (electric charges condensation defect) can condense on it. Moreover, we illustrate that if a boundary exhibits the spontaneous breaking of the $\mathbb{Z}_2 \times \mathbb{Z}_2 \times \mathbb{Z}_2$ symmetry, e.g. $e$-boundary for at least one copy, the $T$-domain wall can condense on it, making the resulting boundary equivalent to the initial one.

To complete the classification of boundaries, we also discuss the codimension-2 domain wall generated by transversal-$S$ gate applied on the codimension-1 submanifold.  Similar to the case with the $T$-domain wall, $S$-domain walls cannot condense on the $X$-boundary, thereby leading to the creation of new nested boundaries. Additionally, for those boundaries that preserve a certain $\mathbb{Z}_2 \times \mathbb{Z}_2$ symmetry, the corresponding $S$-domain wall cannot condense on them either. However, in cases where the boundary exhibits spontaneous breaking of the $\mathbb{Z}_2 \times \mathbb{Z}_2$ symmetry, the corresponding $S$-domain wall can condense on it. The resulting boundary is then equivalent to the initial one.

In conclusion, we give a classification of the gapped boundaries of 3D color code based on the current understanding of emergent symmetries, symmetry defects, and gapped boundaries in the corresponding $\mathbb{Z}_2^3$ gauge theory. The total number of distinct types of boundaries we study is 101, including 1 type of magic boundary and 70 types of nested boundaries.

In principle, it is possible to include additional types of nested boundaries, specifically, those with nested $S$-domain walls where the walls intersect at multiple points, e.g., forming a crossing. As detailed in Ref.~\cite{barkeshli2023codimension}, at these intersections, $e$-particles emerge and attach to the boundary. However, these boundaries can be derived from the nested boundaries we explore in this paper. Consequently, we have chosen not to include them in our count.

Other new boundaries could also be obtained by attaching any (2+1)D topological order, e.g., those realized by (2+1)D (twisted) quantum double models or (extended) string-net models \cite{kitaev2003fault, levin2005string, hu2013twisted, schotte2022quantum, lin2021generalized},  to the existing boundary of the 3D color codes. In these scenarios, the potential number of boundaries could be infinite. A perhaps more natural classification is to mod out the boundaries which are attached to the bulk via a trivial tensor product and can hence be disentangled (decoupled) by a constant-depth local unitary circuit acted only on the boundary. For example, the magic boundary considered in our paper cannot be disentangled/detached by such a circuit and is hence more non-trivial. In this paper, we mainly focus on the codimension-1 boundaries corresponding to the type-III cocycle of the gauged $\mathbb{Z}_2^3$ SPT defect.  We have not included the dicussion of attaching the gauged SPT defect corresponding to type-I and type-II cocycle, where the attachment of the defect of type-I cocycle (Levin-Gu state) has been considered in the case of a single copy of 3D toric code ($\mathbb{Z}_2$ gauge theory) \cite{Zhao:2022String, ji2023boundary, luo2023gapped} forming the so-called \textit{twisted smooth boundary}. The classification of such twisted boundaries in the context of 3D color code ($\mathbb{Z}_2^3$ gauge theory) will be consider in future works.

Future directions along this line include the application of the magic boundaries to fault-tolerant non-Clifford logical gates, where the initial exploration has occurred in Ref.~\cite{zhu2022topological, dua2023quantum} in the context of fractcal topological codes. Besides, a natural theoretical framework describing such magic boundaries which go beyond the Lagrangian subgroup formalism is the \textit{Lagrangian (condensable) algebra}, which has been used to describe gapped boundaries and domain walls in non-Abelian topological orders in two spatial dimensions~\cite{Kaidi:2021Higher,kong2014anyon,levin2013protected,cong2017hamiltonian,zhang2023anomalies}. In three and higher spatial dimensions, the Lagrangian (condensable) algebra framework can be generalized using higher category theory.~\cite{kong2024higher} Recent studies have explored examples of boundaries and domain walls, including the twisted smooth boundary in the 3D toric code~\cite{Zhao:2022String,luo2023gapped,ji2023boundary}, as well as type-I, type-II, and type-III domain walls between 3D toric codes~\cite{li2024domain}. Since the $T$-domain wall $s^{(3)}_{1,2,3}$ and $S$-domain wall $s^{(2)}_{i,j}$ can be understood as condensation defects corresponding to condensation of $e$ particles (i.e., summing over the electric Wilson lines) according to Ref.~\cite{barkeshli2023codimension}, it is expected that the magic boundary can be described by the Lagrangian algebra and the framework of higher category theory similar to the study in Ref.~\cite{Zhao:2022String}.

\noindent{\it Acknowledgements} --- We thank Arpit Dua and Tomas Jochym-O'Connor  for previous collaboration on magic boundaries in the 3D toric codes.  We appreciate the discussion with Massaim Barkeshli, Ryohei Kobayashi, and Nathanan Tantivasadakarn on the symmetry preserving mechanism of the magic boundary. Z.S. thanks Isaac Kim and Yabo Li for helpful discussions. G.Z. is supported by the U.S. Department of Energy, Office of Science, National Quantum Information Science Research Centers, Co-design Center for Quantum Advantage (C2QA) under contract number DE-SC0012704. Z.S. is supported by funds from the UC Multicampus Research Programs and Initiatives of the University of California, Grant Number M23PL5936.

\bibliographystyle{unsrturl}
\bibliography{reference.bib}

\appendix
\section{Detailed description of the unfolding unitary transformation} \label{sec:unfolding}

This section provides an detailed explanation of the unfolding unitary transformation introduced in Section~\ref{sec:unfold_without_boundary} and Ref.~\cite{kubica2015unfolding}.

Without loss of generality, we select the lattice displayed in Fig.~\ref{fig:3DCC} to illustrate the unfolding unitaries explicitly. In general, one can choose arbitrary 4-colorable and 4-valent cellulations and the recipe still works up to local unitary transformation.

First, let us consider the unfolding unitaries associated with the purple cells. We introduce one ancilla qubit on each purple edge. The lattice near a purple cell is shown below.
\begin{align}
     \adjincludegraphics[width=3cm,valign=c]{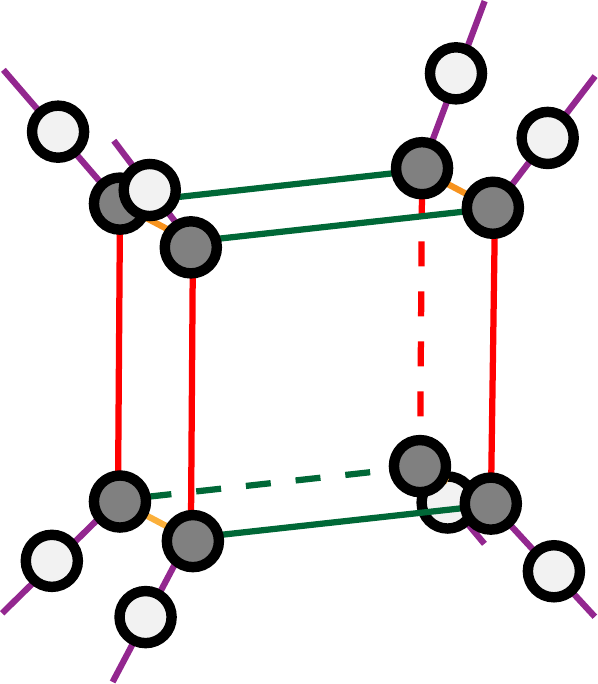}. \label{eq:apdx_1}
\end{align}
The gray circles represent the qubits of the original 3D color code and the white circles represent the ancilla qubits introduced on the purple edges. We define the following unitary operators,
\begin{align}
     \adjincludegraphics[width=4cm,valign=c]{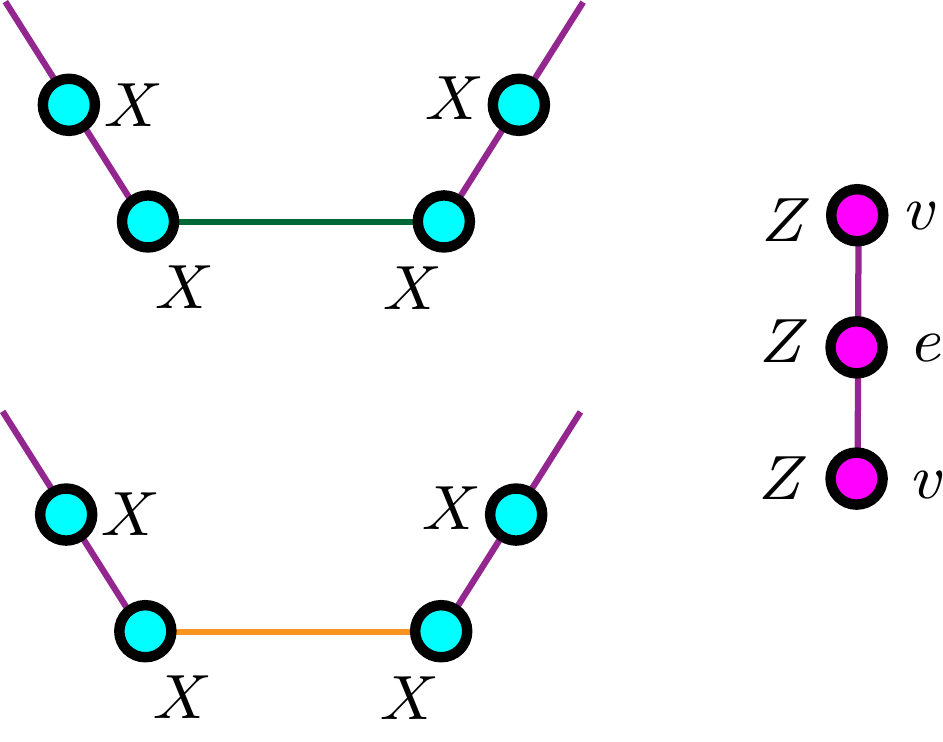}. \label{eq:apdx_2}
\end{align}
These two operators couple qubits on vertices and edges together, and the second operator only apply on the purple edges.

After applying the above operators to the stabilizers, the $X$-stabilizer on the purple cell and the $Z$-stabilizer on the plaquettes with purple edges become the following.
\begin{align}
     \adjincludegraphics[width=6cm,valign=c]{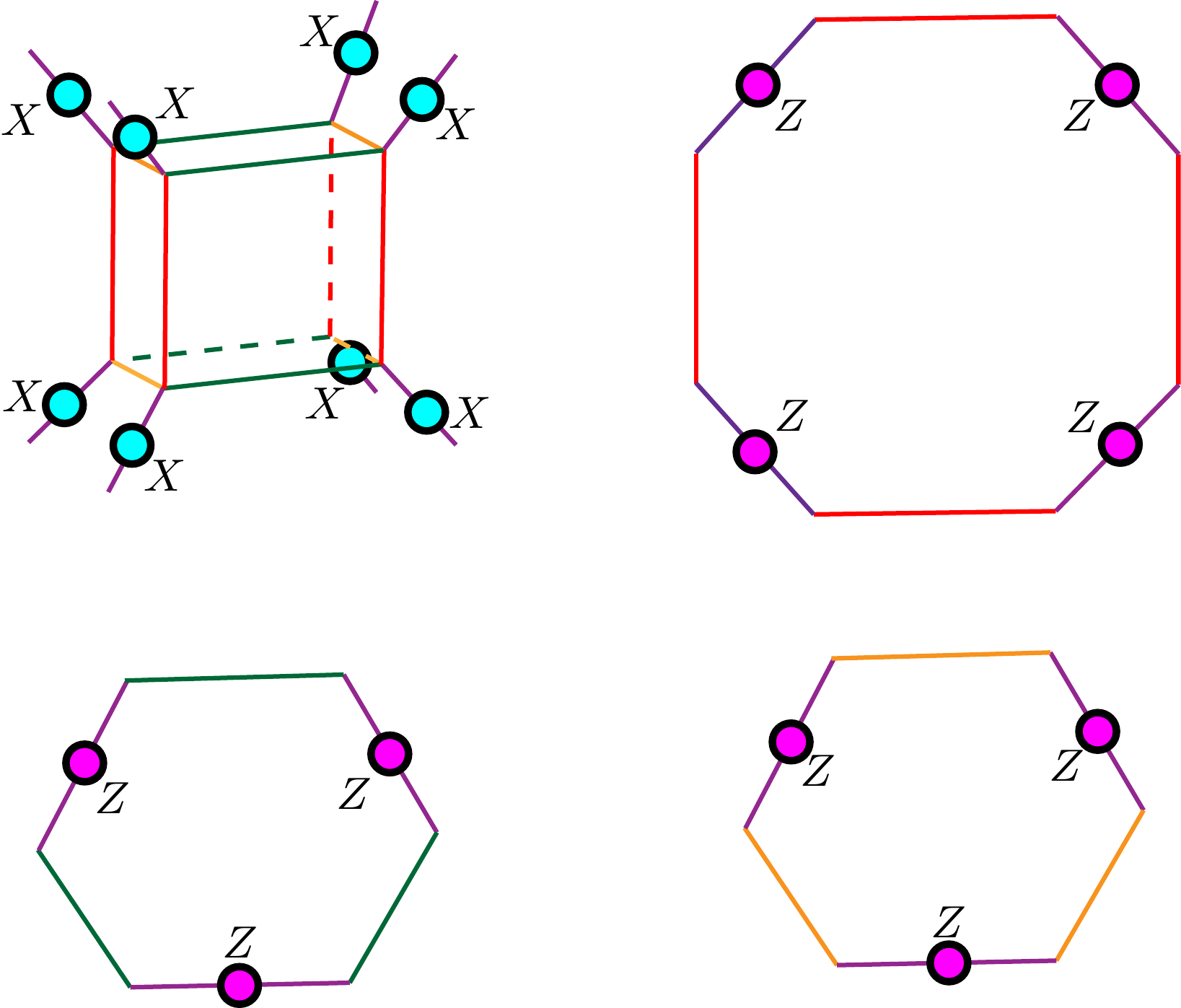}. \label{eq:apdx_3}
\end{align}
One can check that the other $X$-stabilizers and $Z$-stabilizers map to identity.

One can further shrink the purple cells. The stabilizers in Eq.~\eqref{eq:apdx_3} become the stabilizers shown in Fig.~\ref{fig:mapping_1}, and the model becomes a 3D toric code model on purple lattice, as introduced in Section~\ref{sec:unfold_without_boundary}.

Then let us consider the unfolding unitaries associated with the yellow cells. We introduce one ancilla qubit on each yellow edge, and applying the following unitaries to the stabilizers.
\begin{align}
     \adjincludegraphics[width=4cm,valign=c]{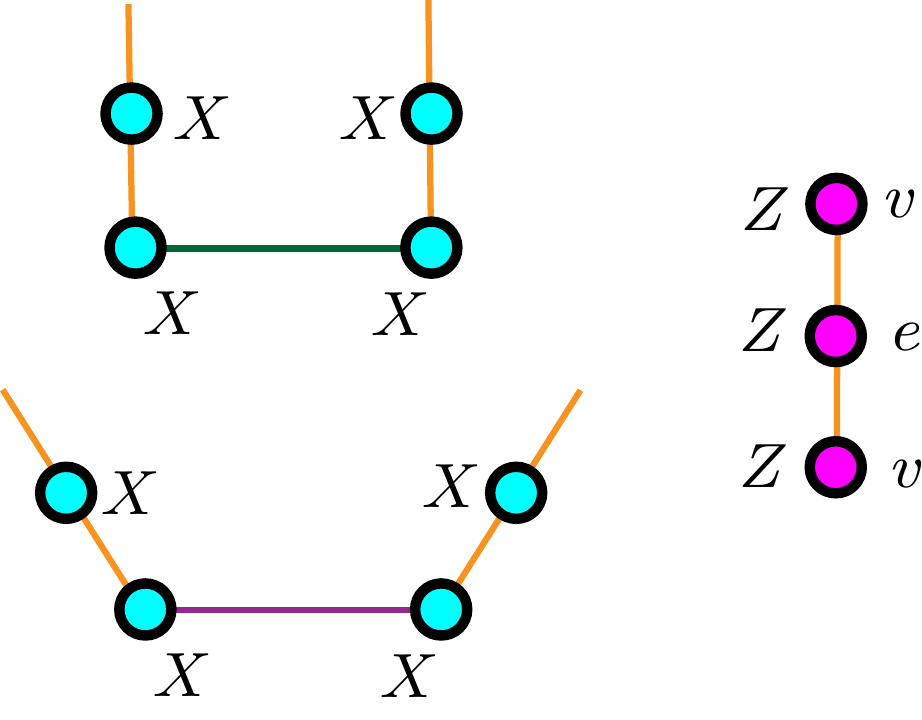}. \label{eq:apdx_4}
\end{align}

After applying the above operators to the stabilizers, the $X$-stabilizer on the yellow cell and the $Z$-stabilizer on the plaquettes with yellow edges become the following.
\begin{align}
     \adjincludegraphics[width=6cm,valign=c]{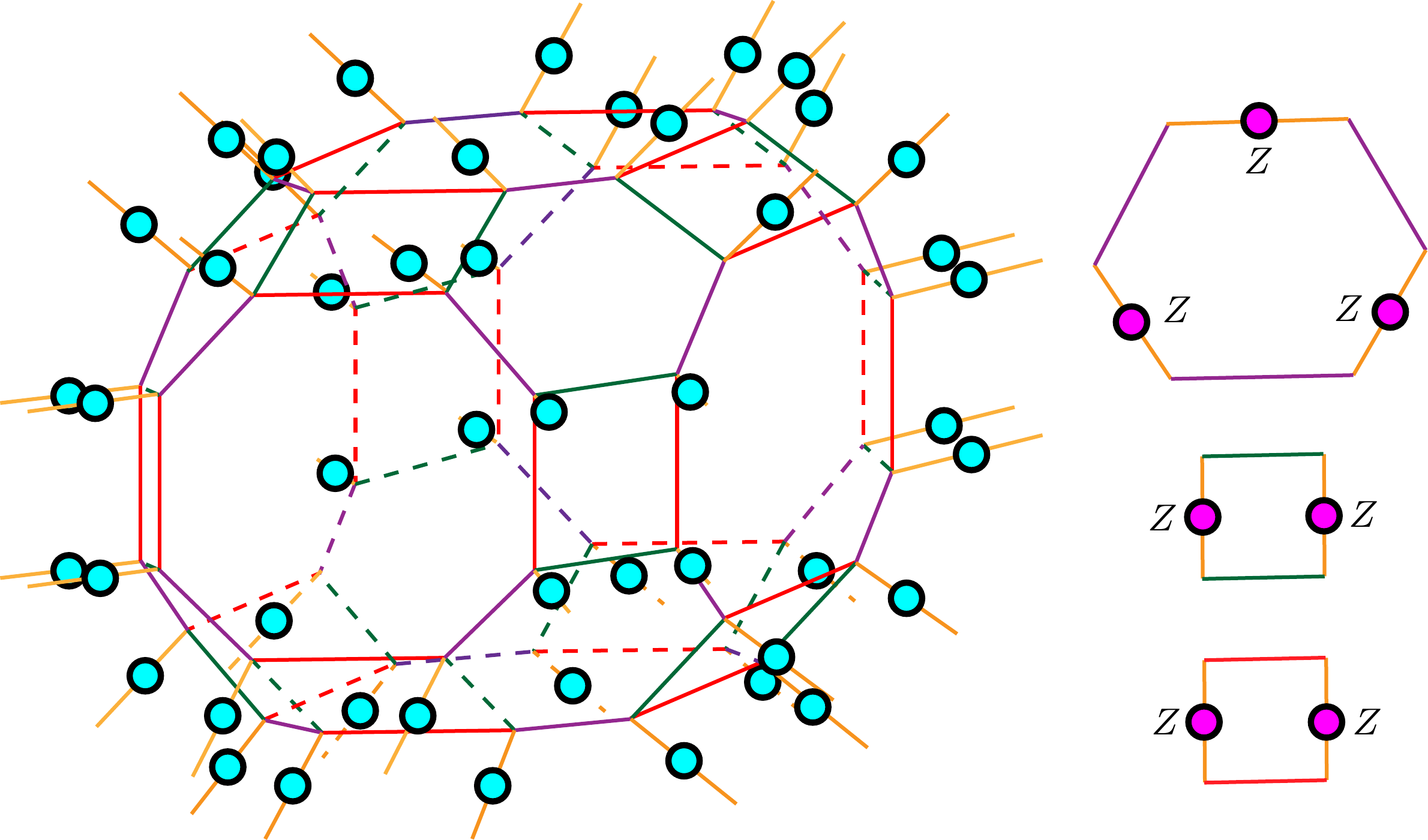}. \label{eq:apdx_5}
\end{align}
Similar to the case for purple cells, the other $X$-stabilizers and $Z$-stabilizers map to identity. We can then disentangle and discard the qubits on the vertices. The remaining stabilizers are generated by the stabilizers in Eq.~\eqref{eq:apdx_5}.

However, after shrinking the yellow cells, there are four qubits lying on each edge. They are coupled with four neighbouring $Z$-stabilizers respectively, as shown below.
\begin{align}
     \adjincludegraphics[width=6cm,valign=c]{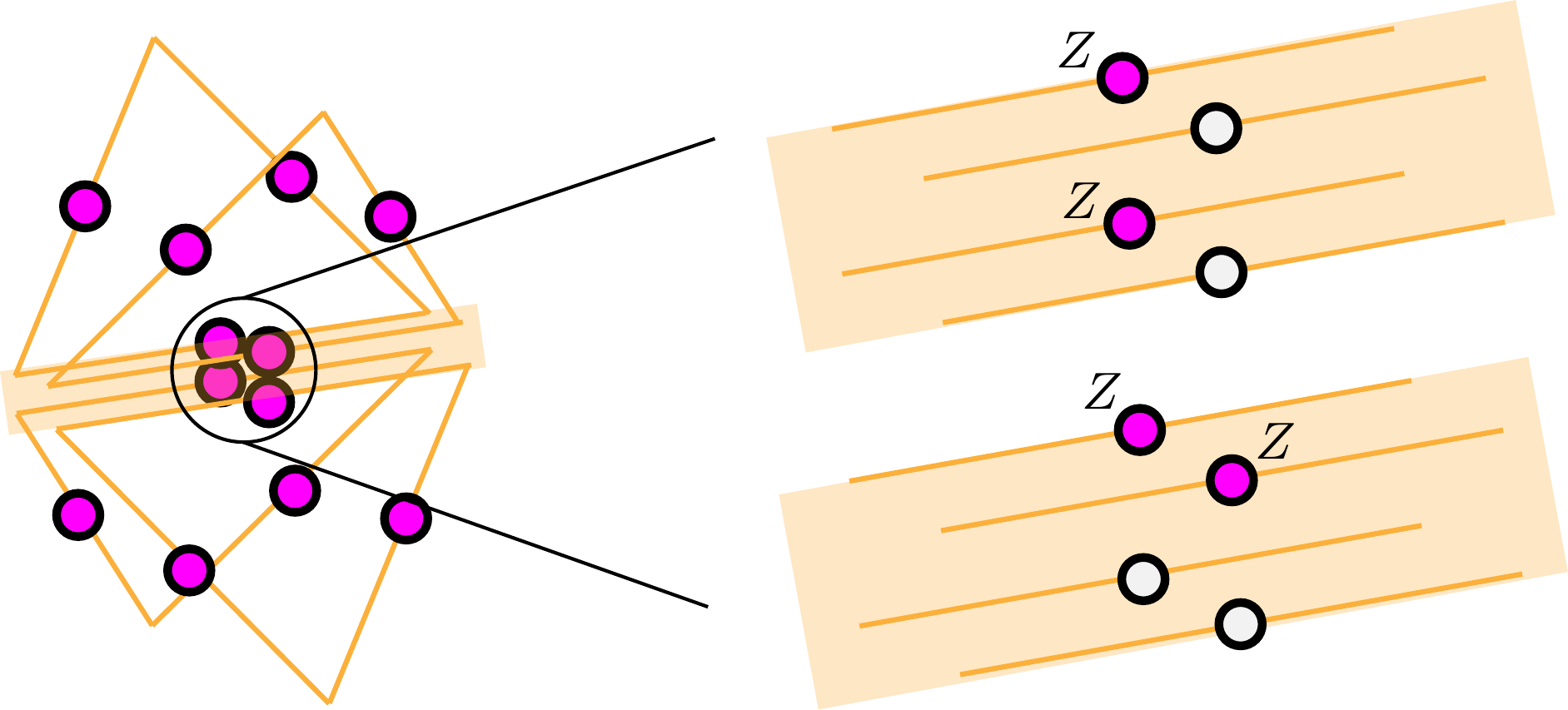}. \label{eq:apdx_6}
\end{align}
The shaded area on the left represents the fine structure of an yellow edge after we apply the lattice deformation. There are four qubits correspond to one edge and each qubit are coupled with a $Z$-stabilizer, which are displayed as triangles. On the right hand side, we zoom into this area, and there are two more types of $Z$-stabilizers. Both are two body stabilizers and couple qubits on two different directions. To get rid of these redundant qubits, one could first change the basis of stabilizers by product some of them together to make a single qubit on this edge couple to all the four different plaquettes, and then disentangle and discard the other three qubits. The remaining model is a 3D toric code on the yellow lattice as introduced in Section~\ref{sec:unfold_without_boundary}.

The unfolding process on the green lattice is similar to the two cases we discussed above. Therefore, we leave the details to the readers who are interested.

\section{3D toric code from gauging 3D Ising model} \label{sec:gauging}

Gauging is a bijective, isometric duality map from wave functions with global symmetries to wave functions with gauge (local) symmetries~\cite{williamson2016fractal,kubica2018ungauging,song2023topological}. In this appendix, we discuss the gauging process for the 3D Ising model to obtain the 3D toric code. And further show how the global symmetry is spontaneously broken in the presence of certain boundaries.

\subsection{Gauging without boundary}

First let us discuss the case without boundary. Consider a $3$-dimensional square lattice $\mathcal{L}(\mathcal{V}, \mathcal{E}, \mathcal{F}, \mathcal{C})$. $\mathcal{V}$ is the set of vertices $v \in \mathcal{V}$, $\mathcal{E}$ is the set of edges $e \in \mathcal{E}$, $\mathcal{F}$ is the set of faces $f \in \mathcal{F}$, and $\mathcal{C}$ is the set of cells $c \in \mathcal{C}$. The Hamiltonian of the 3D Ising model is given by
\begin{align}
    H_{Ising} = - J \sum_{v} X(v) - \sum_{e} \prod_{v \in \partial e} Z(v),
\end{align}
in which $J$ is the coupling constant. Pictorially, the above two terms can be shown as follows.
\begin{align}
     \adjincludegraphics[width=5cm,valign=c]{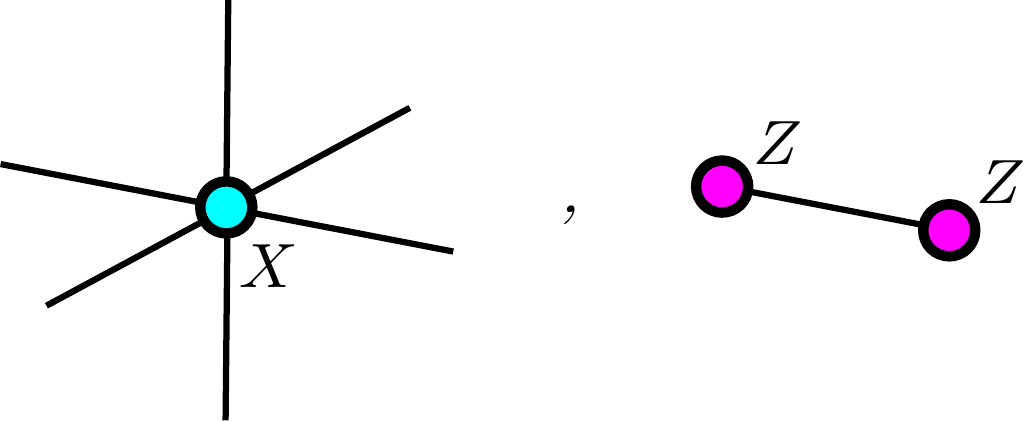}. \label{eq:apdx_10}
\end{align}
The $ZZ$-terms can be defined on edges along all directions.

The Hamiltonian is invariant under $\mathbb{Z}_2$ global symmetry, $\prod_{v}X(v)$, transformation. We have
\begin{align}
    \left(\prod_{v}X(v)\right)^{\dagger} H_{Ising}  \left(\prod_{v}X(v)\right) = H_{Ising}. \label{eq:globalz2}
\end{align}

To gauge this model, we first couple the constraint terms ($ZZ$-terms) with gauge qubits. We couple one gauge qubit to each $ZZ$-term. The map is given by the following.
\begin{align}
     \adjincludegraphics[width=5cm,valign=c]{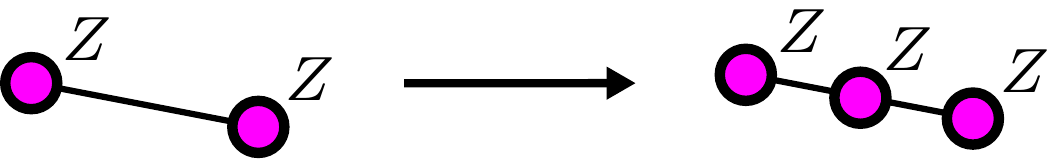}. \label{eq:apdx_11}
\end{align}
The new Ising Hamiltonian becomes
\begin{align}
    H_{Ising}' = - J \sum_{v} X(v) - \sum_{e} Z(e) \prod_{v \in \partial e} Z(v).
\end{align}
Similar to Eq.~\eqref{eq:globalz2}, now we can find a local (gauge) $\mathbb{Z}_2$ symmetry operator that keeps $H_{Ising}'$ invariant under gauge transformation. The gauge symmetry operator is the following.
\begin{align}
    X(v)\prod_{v \in \partial e} X(e) = \adjincludegraphics[width=2.3cm,valign=c]{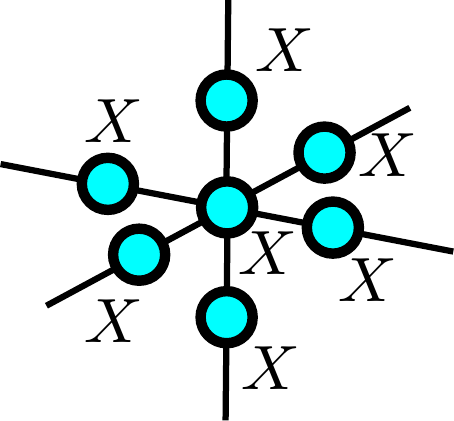}
\end{align}
And we obtained the following gauged Hamiltonian:
\begin{equation}
    \begin{aligned}
        H_{gauged} = &- J \sum_{v} X(v) - \sum_{e} Z(e) \prod_{v \in \partial e} Z(v) \\ &-  X(v)\prod_{v \in \partial e} X(e).
    \end{aligned}
\end{equation}
By taking the strong coupling limit $J \to \infty$, one can obtain a model that purely defined on the gauge qubits, which is also called the 3D toric code or pure $\mathbb{Z}_2$ gauge theory. The Hamiltonian is given in Eq.~\eqref{eq:toric_code}.

%\zj{Conceptually speaking, gauging $\mathbb{Z}_2$ 3D Ising model is equivalent to the condensation of codimension-1 domain walls, which are the symmetry charges of the $\mathbb{Z}_2$ 1-form symmetry. The condensation breaks the $\mathbb{Z}_2$ 1-form symmetry spontaneously. Therefore, gauging is a map from the $\mathbb{Z}_2$ $0$-form symmetric phase of the 3D Ising model to the $\mathbb{Z}_2$ $1$-form symmetry-breaking phase of the 3D toric code.~\cite{Gaiotto_2015,lake2018higher,zhao2021higher}.}

\subsection{Gauging in the presence of Boundaries}

\subsubsection{The smooth boundary} \label{sec:gauging_smooth}

First, let us consider the smooth boundary. The `smooth' here means all the terms on the boundary are intact. For example, we have the following terms on the boundary.
\begin{align}
     \adjincludegraphics[width=4cm,valign=c]{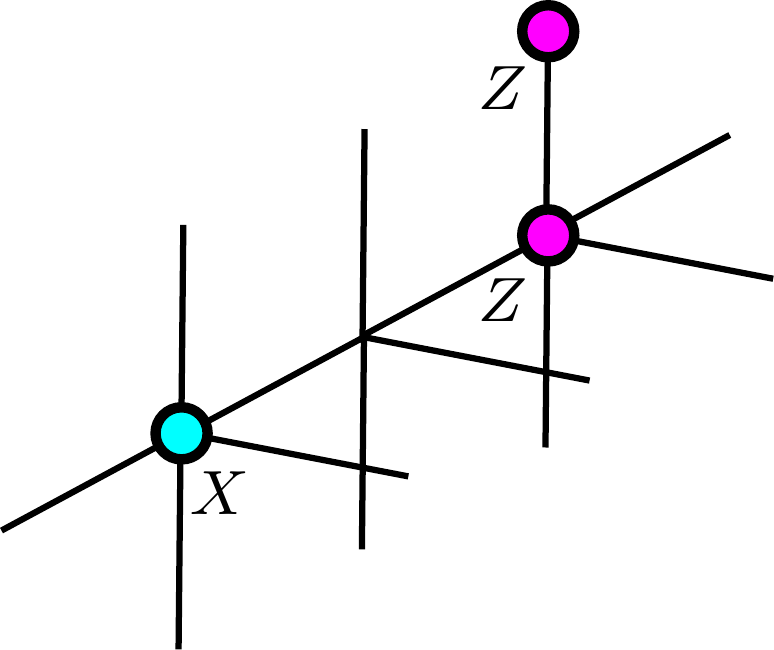}, \label{eq:apdx_13}
\end{align}
which are single qubit $X$ terms and two-body $Z$ terms that are entirely supported on the boundary qubits.

After gauging, the Hamiltonian terms on the smooth boundary becomes.
\begin{align}
     \adjincludegraphics[width=5cm,valign=c]{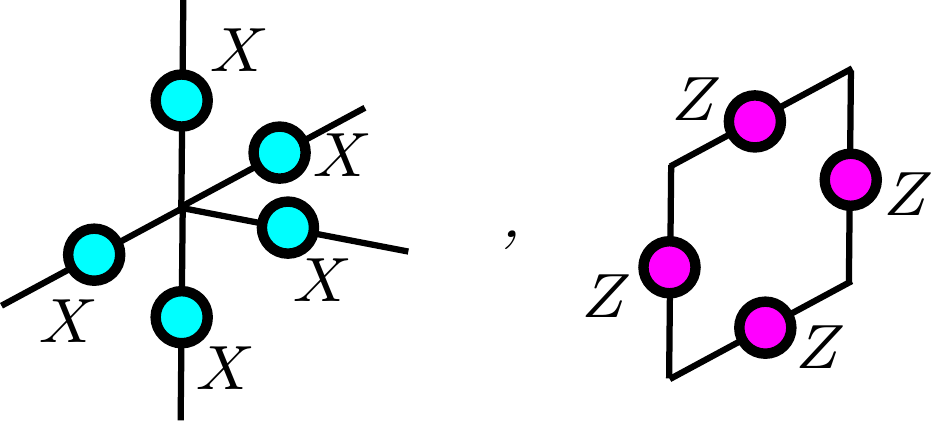}. \label{eq:apdx_14}
\end{align}
This boundary is exatly the smooth ($m$) boundary we discussed in Sec~\ref{sec:review}. By checking the commutation relation between the $\mathbb{Z}_2$ global symmetry operator and the boundary Hamiltonian terms, one can also conclude that the ungauged Ising-model boundary corresponding to the $m$-boundary preserves the $\mathbb{Z}_2$ global symmetry.

\subsubsection{The rough boundary} \label{sec:gauging_rough}

In contrast, one can also consider the `rough' boundary of 3D Ising model. In this case, there are truncated $ZZ$-terms existing. The Hamiltonian terms on this boundary is shown below.
\begin{align}
     \adjincludegraphics[width=6cm,valign=c]{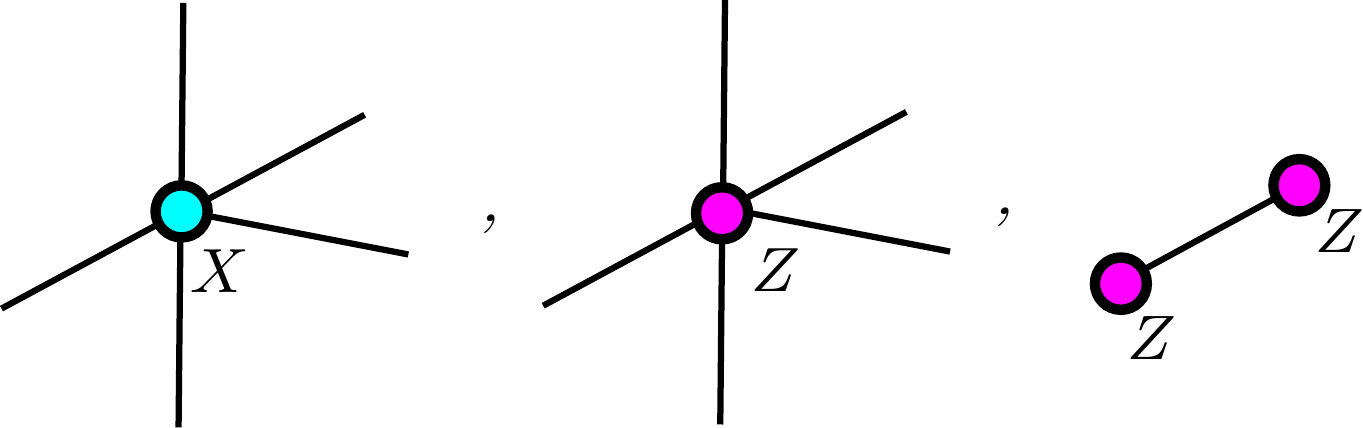}. \label{eq:apdx_15}
\end{align}
The single $Z$-term (truncated $ZZ$-term) corresponds to the $ZZ$-term on the missing leg pointing outward to the boundary. After coupling the gauge qubits, the Hamiltonian terms become
\begin{align}
     \adjincludegraphics[width=6cm,valign=c]{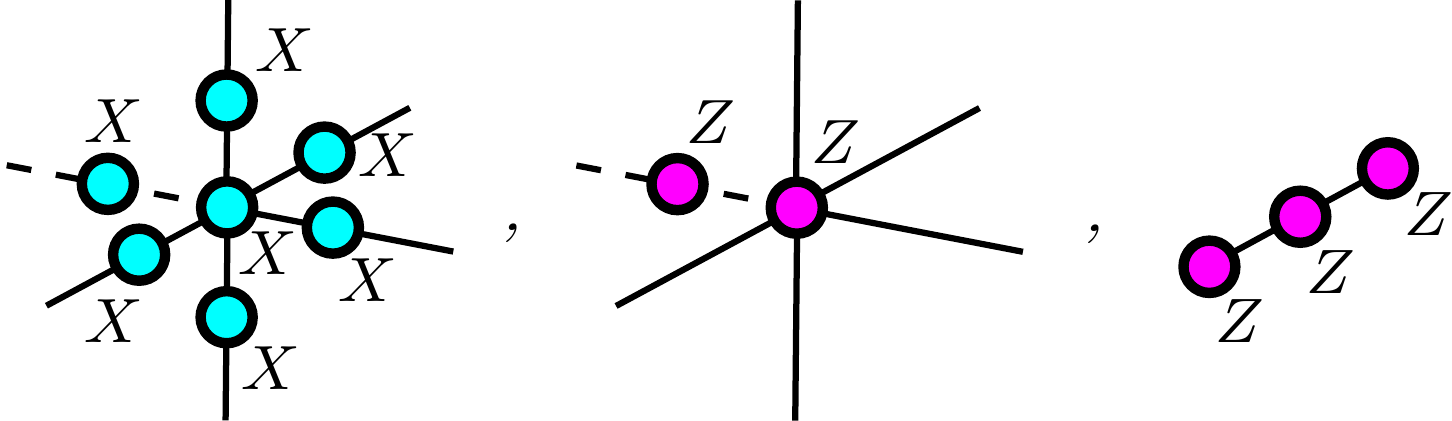}. \label{eq:apdx_16}
\end{align}
The dashed line here represents the leg outside the boundary. Then we can take the strong coupling limit and get a pure gauge model. The Hamiltonian terms on this boundary become
\begin{align}
     \adjincludegraphics[width=6cm,valign=c]{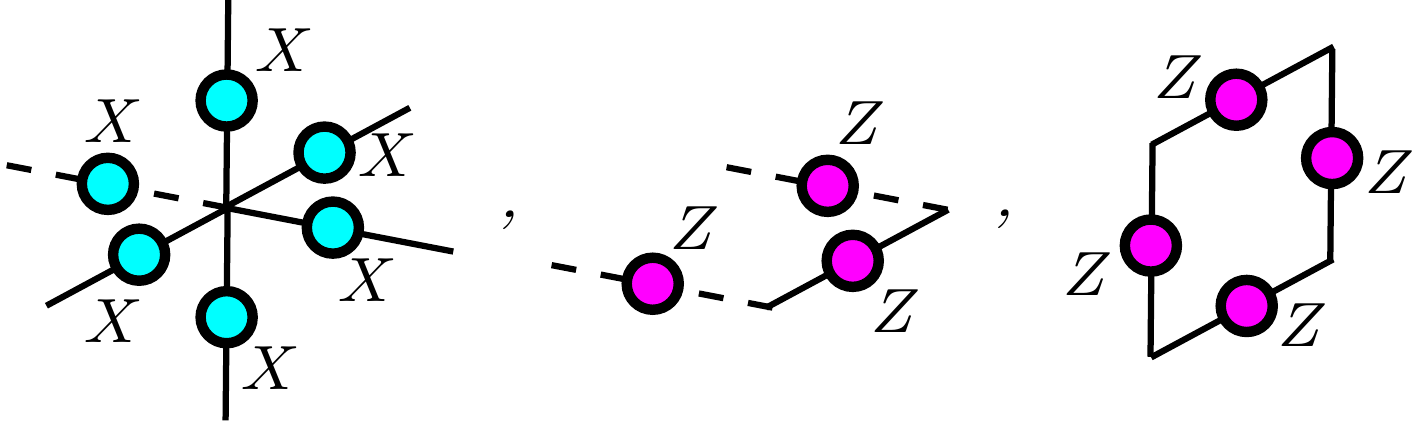}, \label{eq:apdx_17}
\end{align}
which is exactly the stabilizers for smooth ($e$) boundary of 3D toric code. One can check that, the Hamiltonian terms in Eq.~\eqref{eq:apdx_15} violate the $\mathbb{Z}_2$ global symmetry. Therefore, we conclude that the $e$-boundary in the ungauged model spontaneously breaks the $\mathbb{Z}_2$ global symmetry.

The $0$-form symmetry breaking can also be undertood in the context of $\mathbb{Z}_2$ gauge theory. Condensation of $e$-particles (which corresponds to the symmetry charges of the $0$-form symmetry) violates the charge conservation symmetry. Therefore, the corresponding $\mathbb{Z}_2$ global symmetry is spontaneously broken on this boundary.

\subsubsection{The folded boundary} \label{sec:gauging_folded}

The folded boundary can be obtained by folding the 3D Ising model lattice along a codimension-1 submanifold. An example is shown below.
\begin{align}
     \adjincludegraphics[width=6cm,valign=c]{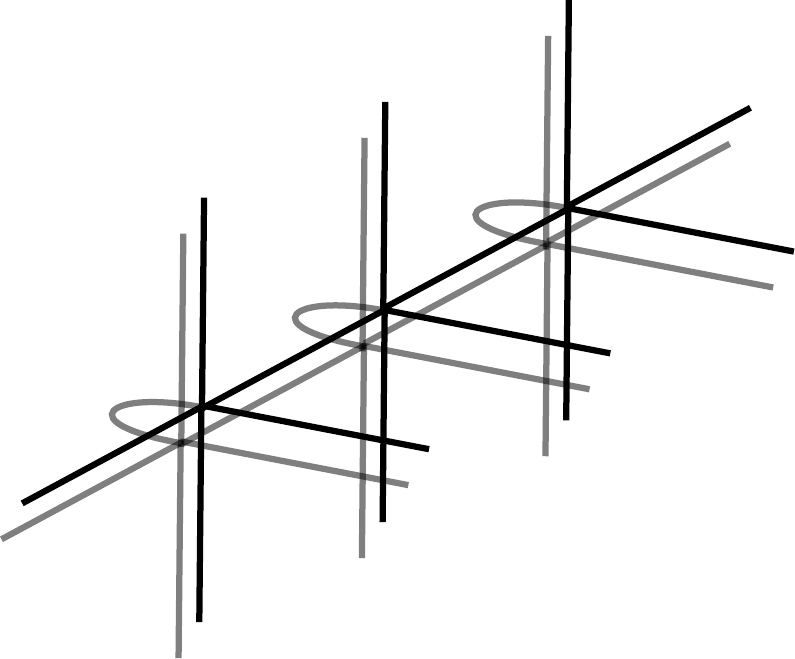}. \label{eq:apdx_18}
\end{align}
The black and the gray lattice correspond to the two copies of lattices after folding, and they are connected at the boundary by the curved lines. The Hamiltonian terms are the following.
\begin{align}
     \adjincludegraphics[width=6cm,valign=c]{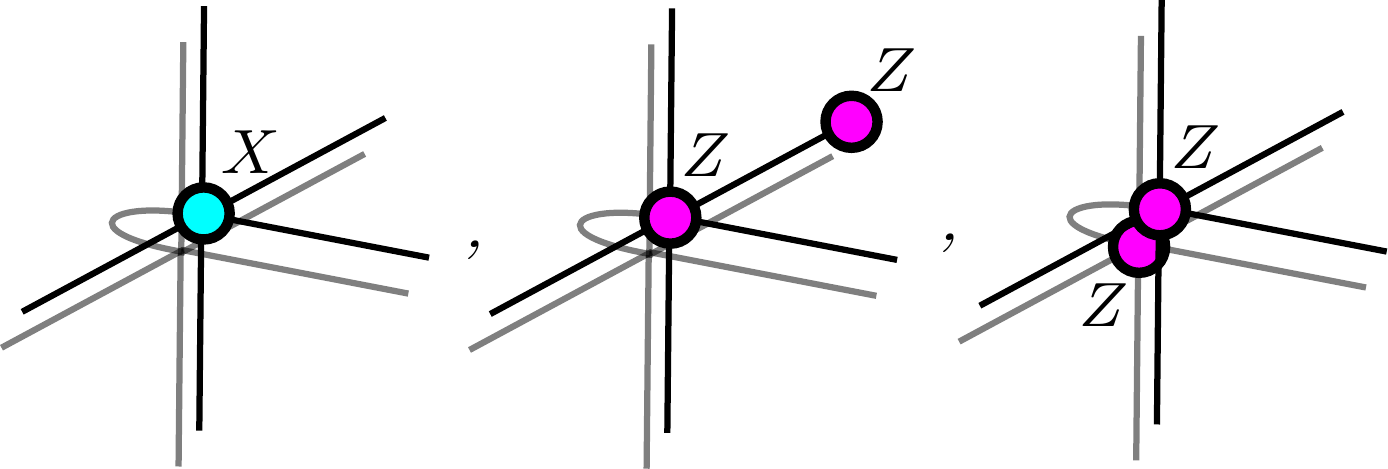}. \label{eq:apdx_19}
\end{align}
Besides the typical 3D Ising model terms on each copy, there are also $ZZ$ terms that couple two copies together on the boundary.

There is a $\mathbb{Z}_2 \times \mathbb{Z}_2$ global symmetry in the bulk of the two copies of 3D Ising model. However, this $\mathbb{Z}_2 \times \mathbb{Z}_2$ global symmetry breaks into a $\mathbb{Z}_2^{diag}$ symmetry in the presence of the folded boundary. The symmetry generator of the $\mathbb{Z}_2$ diagonal symmetry is the product of Pauli X operators on all the vertices in both copies. The reason for this symmetry breaking is as follows: The folded system comes from folding a single copy of 3D Ising model, in which there is only one $\mathbb{Z}_2$ global symmetry. Therefore, all though the symmetry in the bulk looks like a $\mathbb{Z}_2 \times \mathbb{Z}_2$ global symmetry, the actual symmetry of the system is $\mathbb{Z}_2^{diag}$.

Following the gauging process we described before, the gauged Hamiltonian terms on the boundary are the following.
\begin{align}     \adjincludegraphics[width=6cm,valign=c]{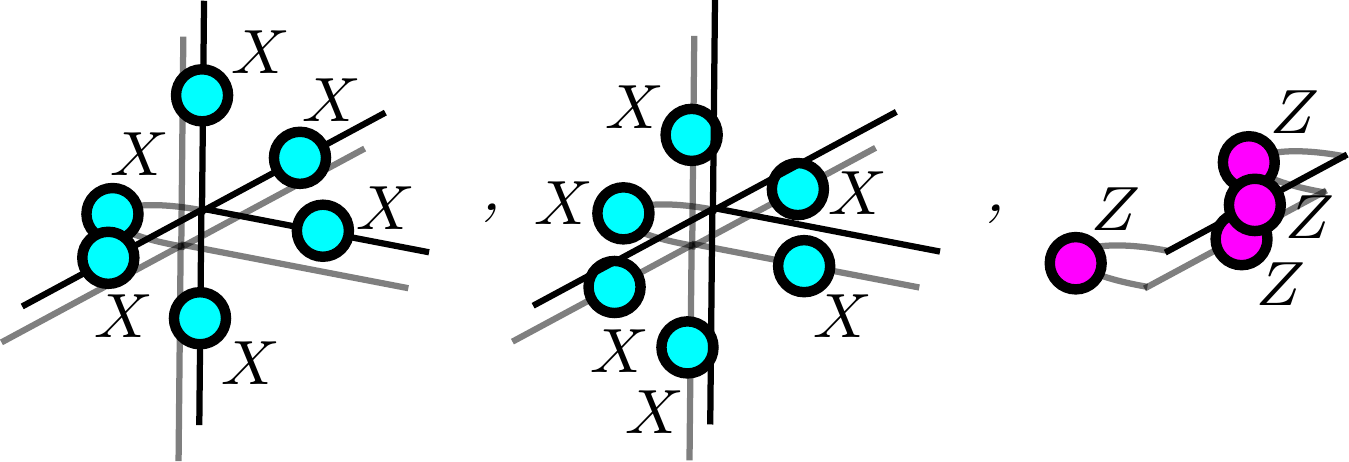}. \label{eq:apdx_20}
\end{align}
One can check this boundary corresponds to the $(e_1 e_2, m_1 m_2)$-boundary discussed in Se.~\ref{sec:folded_boundary}.

Similar to the case of rough boundary, the $\mathbb{Z}_2 \times \mathbb{Z}_2 \to \mathbb{Z}_2^{diag}$ symmetry breaking picture can also be understood in the context of $\mathbb{Z}_2 \times \mathbb{Z}_2$ gauge theory. Consider we move one $e_1$ particle to the folded boundary, it get annihilated and created another $e_2$ particle on the other copy. However, if we only focus on one copy, the charge conservation symmetry is violated. Therefore, the $\mathbb{Z}_2$ global symmetry on each individual copy is spontaneously broken.

Then the readers may ask: what symmetry is broken corresponding to the fact that $e_1 e_2$ can condense on this boundary? The answer is again the $\mathbb{Z}_2$ global symmetry on each individual copy. Since the $e_1 e_2$ particles are always created and annihilated in pairs on this boundary, the $\mathbb{Z}_2$ diagonal symmetry is always preserved on this boundary.

\subsubsection{The $m_1 m_2 m_3$-boundary} \label{sec:gauging_mmm}

The lattice on the boundary of the ungauged model is given by
\begin{align}
\adjincludegraphics[width=5cm,valign=c]{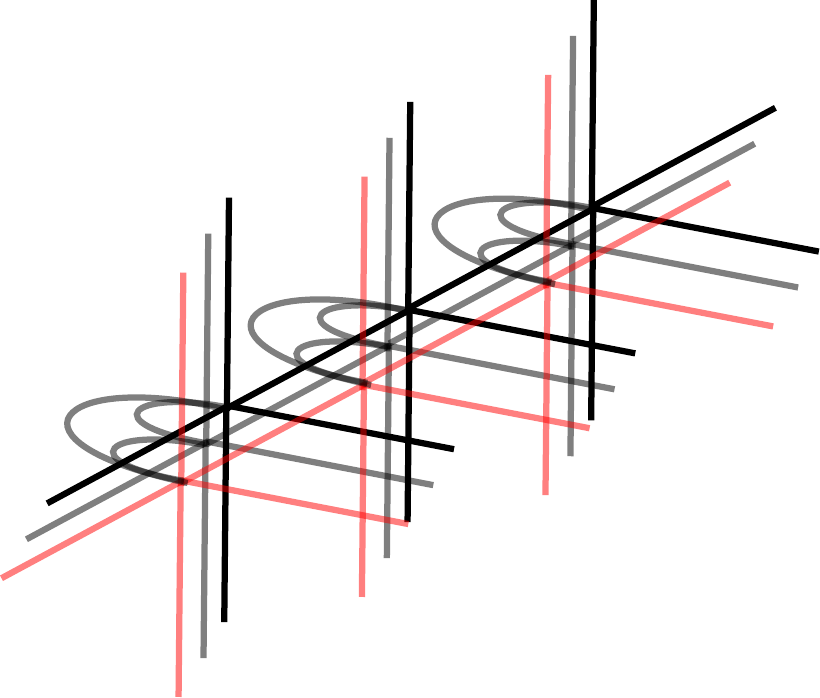}. \label{eq:apdx_21}
\end{align}
We use three different color to represent three copies of 3D Ising models. On the boundary, each pair of copies are connected. The Hamiltonian terms are shown below.
\begin{align}
\adjincludegraphics[width=6cm,valign=c]{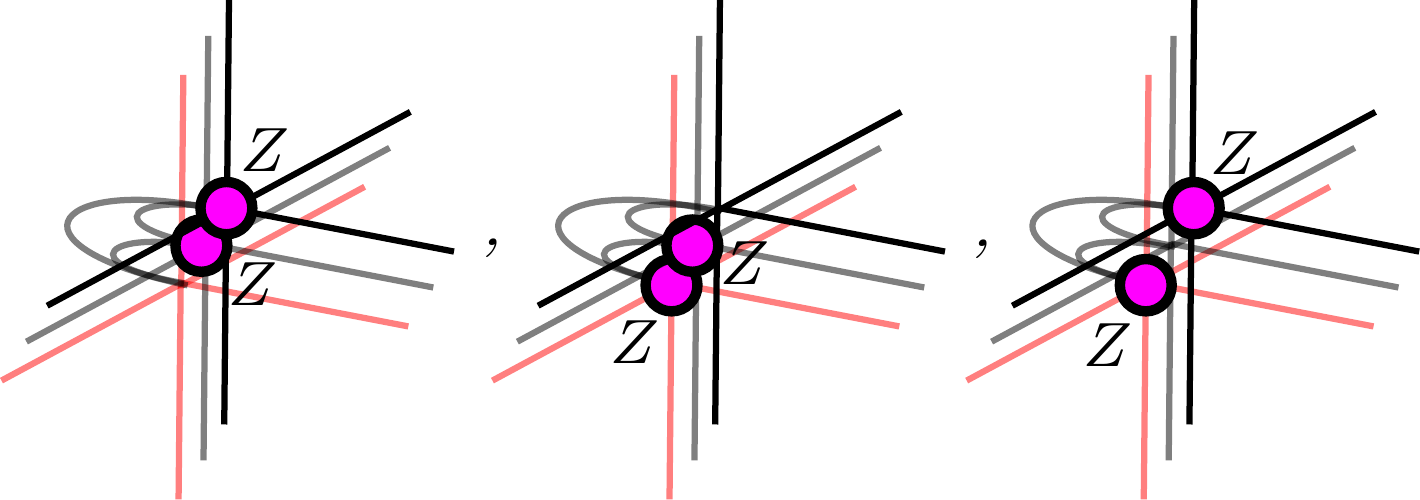}. \label{eq:apdx_22}
\end{align}
There are additional $ZZ$-terms connecting different copies of 3D Ising models. We omit the $X$-terms here since they are just single qubit terms on every vertices as displyed in Eq.~\eqref{eq:apdx_19}.

After gauging, the boundary terms become
\begin{align}
\adjincludegraphics[width=5cm,valign=c]{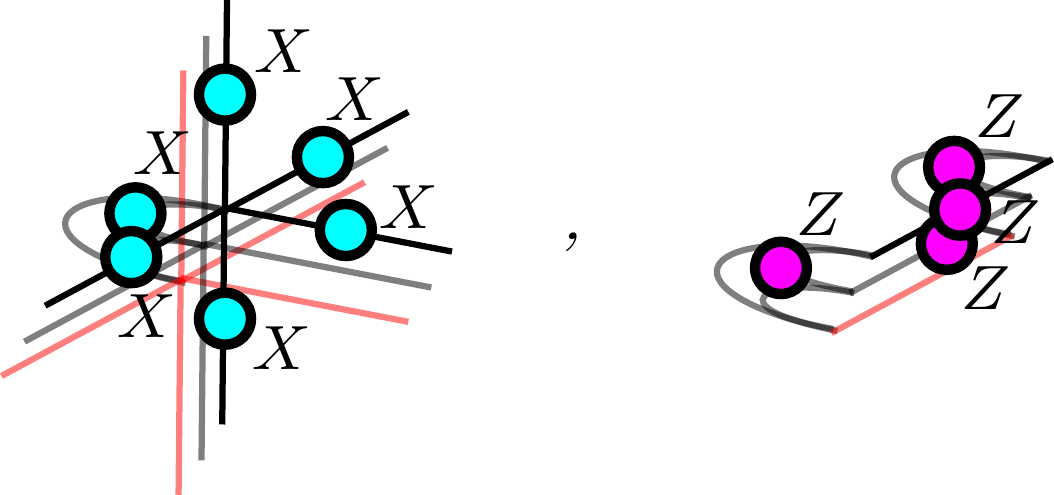}. \label{eq:apdx_23}
\end{align}
Here we show one example of $X$-stabilizer and one example of $Z$-stabilizer. Other stabilizers can be obtained by permuting the copy labels. We assign three additional gauge qubits for each $ZZ$-terms. Therefore, we have six boundary $X$-stabilizers corresponding to each boundary site (group of three boundary vertices), and three boundary $Z$-stabilizers corresponding to each pair of truncated plaquettes.

In the 3D Ising-model perspective, the existence of $ZZ$-terms violates the $\mathbb{Z}_2$ symmetry corresponding to each copy. The only residue symmetry is the $\mathbb{Z}_2^{diag}$ symmetry. Therefore, after gauging, this $\mathbb{Z}_2^{diag}$ symmetry corresponds to the charge conservation symmetry of the entire system. Since this symmetry is unbroken, it's equivalent to say a single $e$-particle cannot condense on this boundary.

\subsubsection{The $e_1 e_2 e_3$-boundary} \label{sec:gauging_eee}

The lattice on the boundary of the ungauged 3D Ising model is given by
\begin{align}
\adjincludegraphics[width=5cm,valign=c]{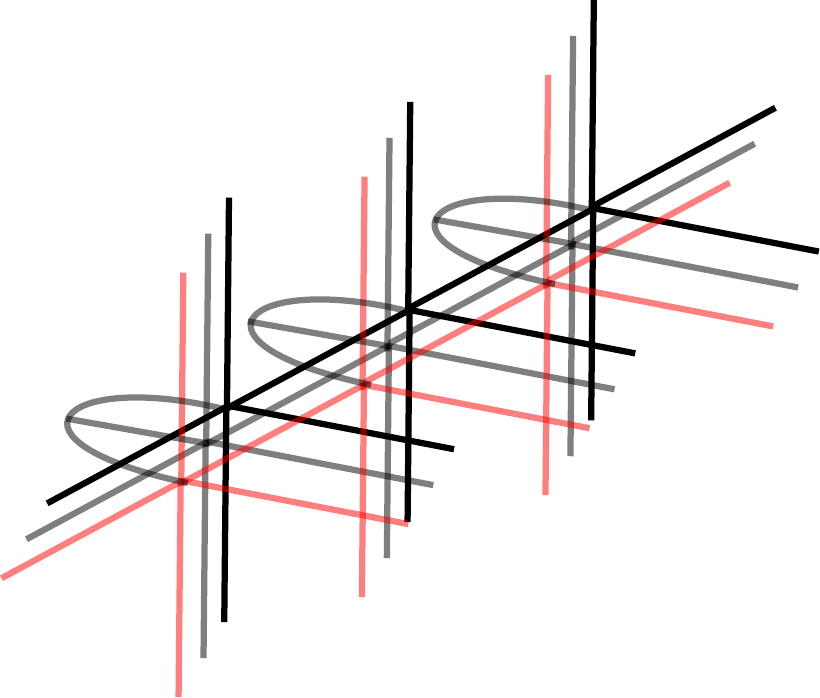}. \label{eq:apdx_24}
\end{align}

On the boundary there are $ZZZ$-terms that coupling three copies together.
\begin{align}
\adjincludegraphics[width=2cm,valign=c]{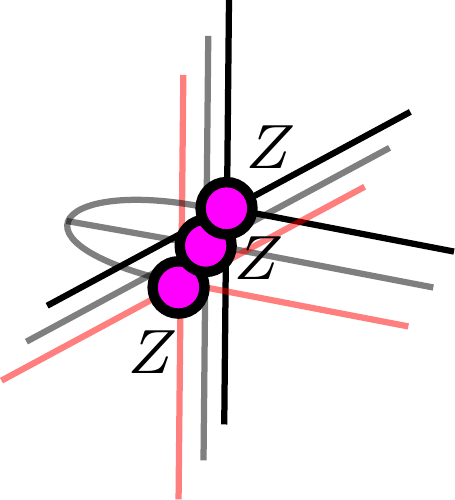}. \label{eq:apdx_25}
\end{align}

After gauging the boundary terms become
\begin{align}
\adjincludegraphics[width=5cm,valign=c]{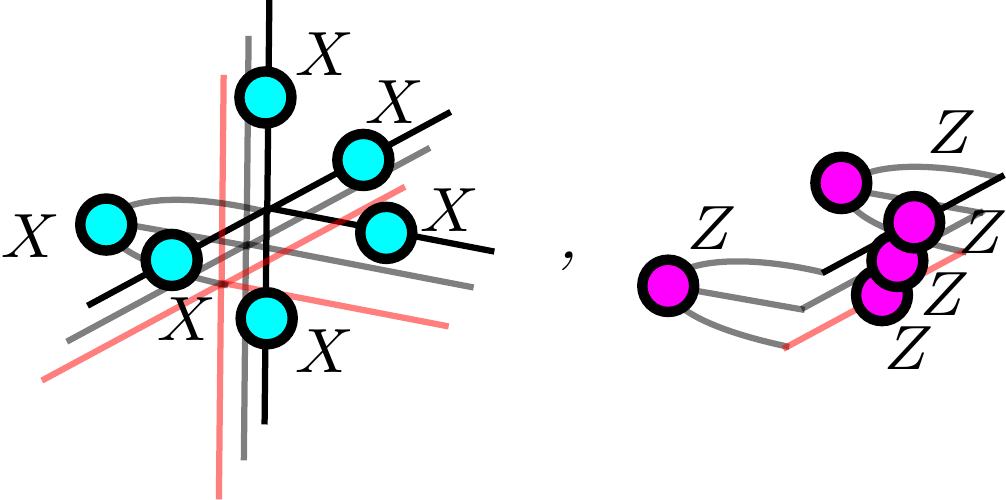}. \label{eq:apdx_26}
\end{align}

In the symmetry perspective, the existence of the $ZZZ$-terms on the boundary breaks the bulk symmetry into a $\mathbb{Z}_2^{(110)} \times \mathbb{Z}_2^{(011)}$ symmetry, which is also equivalent to say, in the gauged theory, the charge conservation symmetry on each pair of copies is preserved.

\section{A Minimal Model Example} \label{sec:minimal}

In this appendix, we present the construction of a minimal model for the 3D color code, featuring $\{pg_{\mathbf{x}}, py_{\mathbf{x}}, yg_{\mathbf{x}}\}$-boundaries and encoding three logical qubits. The lattice structure of this minimal model is depicted in Fig.~\ref{fig:minimal_model}. It includes two $\{y_{\mathbf{z}}, g_{\mathbf{z}}, p_{\mathbf{z}}\}$-boundaries at the top and bottom, while the remaining four boundaries are $\{pg_{\mathbf{x}}, py_{\mathbf{x}}, yg_{\mathbf{x}}\}$-boundaries. Below, we outline the steps for constructing this minimal model.

\begin{figure}
    \centering
    \includegraphics[width = 7cm]{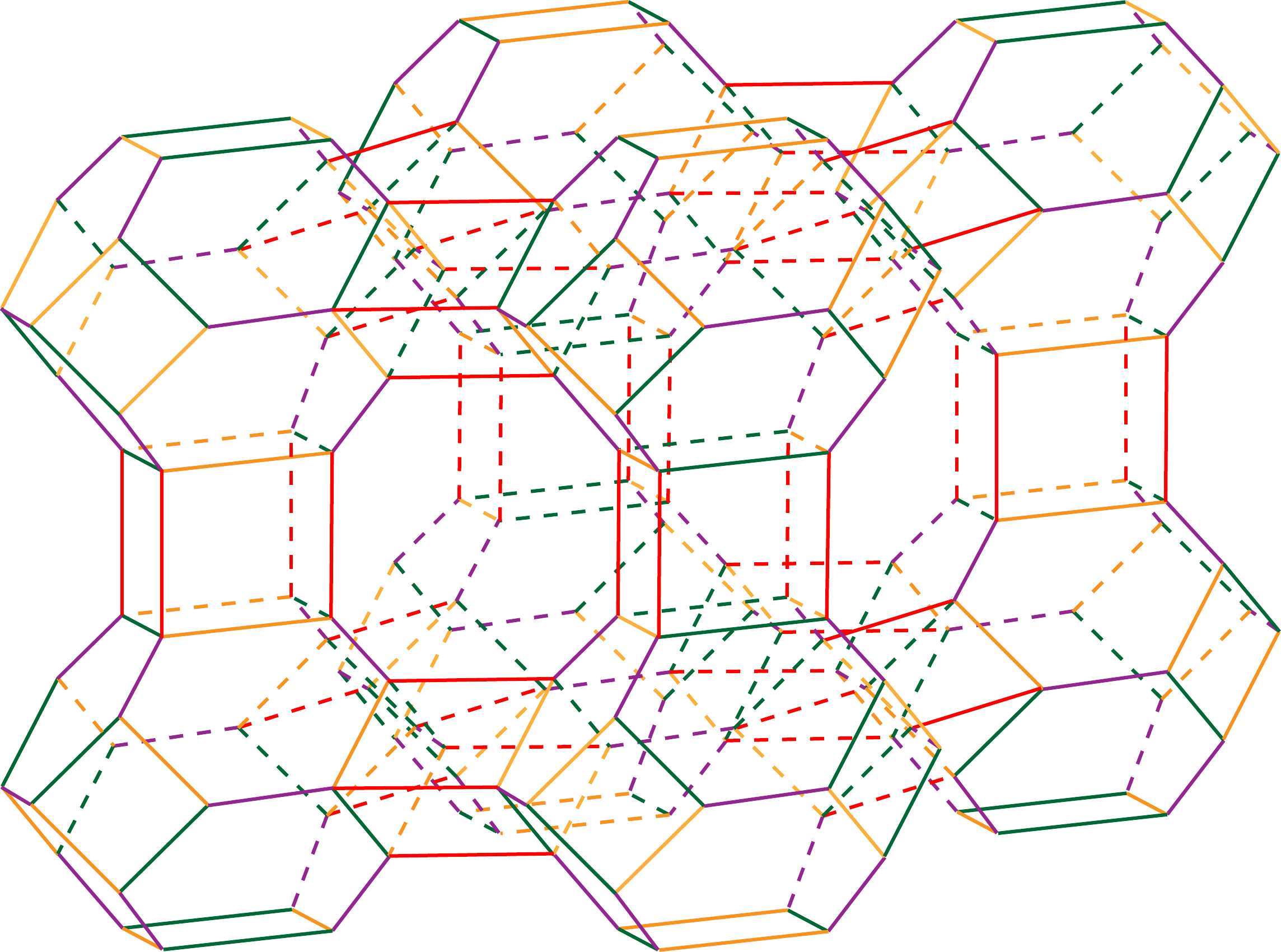}
    \caption{The lattice of the minimal model of 3D color code that can realize the $\{pg_{\mathbf{x}},py_{\mathbf{x}},yg_{\mathbf{x}}\}$-boundary. Qubits are supported on the vertices of the lattice. The lattice is constructed by one truncated cuboctahedron (yellow cell), eight truncated octahedrons (red cells), and twelve cubes (purple cells).}
    \label{fig:minimal_model}    
\end{figure}

First, we begin with a truncated cuboctahedron lattice (yellow cell).
\begin{align}
     \adjincludegraphics[width=4cm,valign=c]{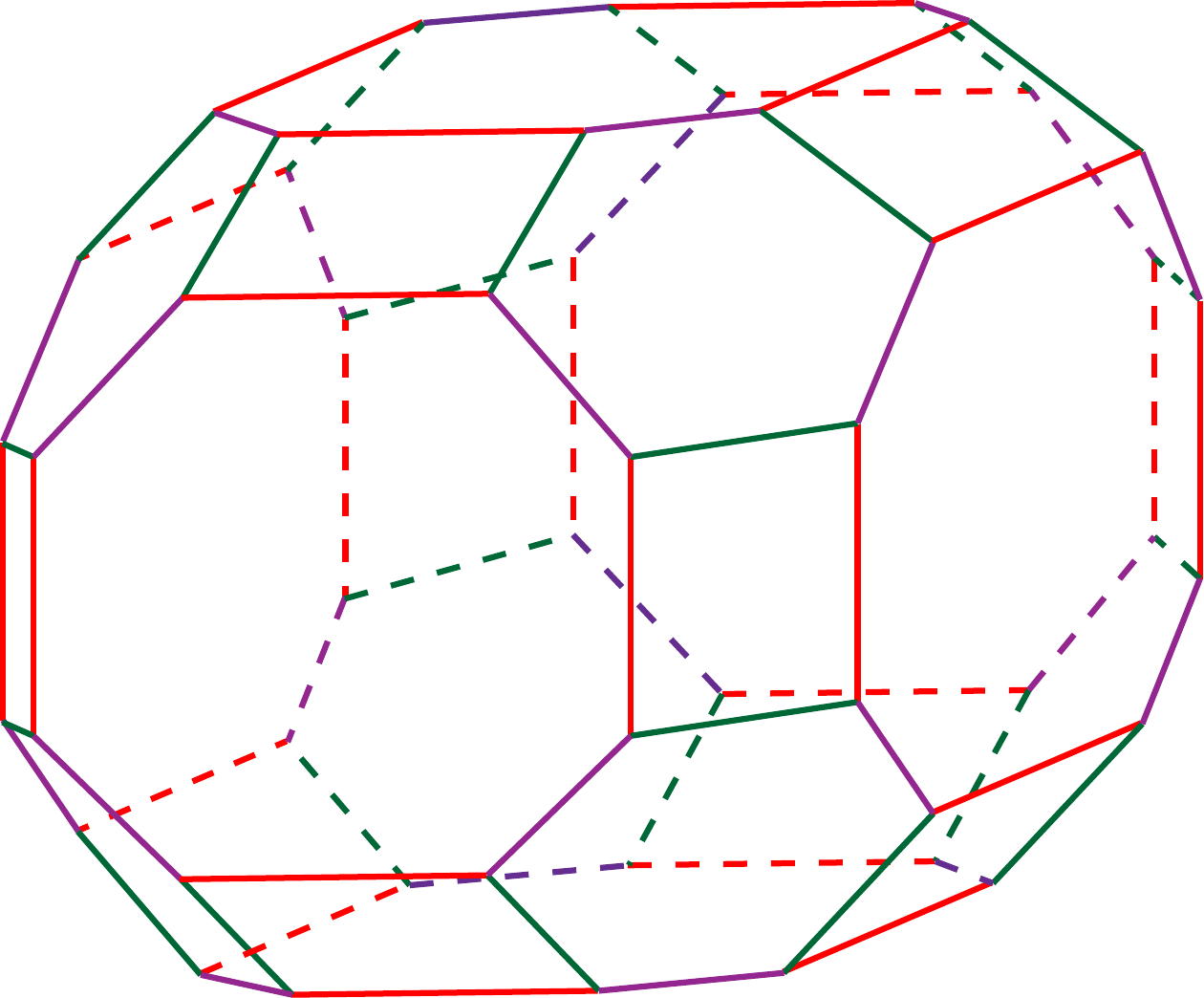}. \label{eq:minimal_1}
\end{align}
We define an $X$-stabilizer for the cell and $Z$-stabilizers for all the plaquettes. 

Second, we attach a truncated octahedron lattice (red cell) to each hexagonal plaquette of the yellow cell. An example of the truncated octahedron lattice is displayed below.
\begin{align}
     \adjincludegraphics[width=2.5cm,valign=c]{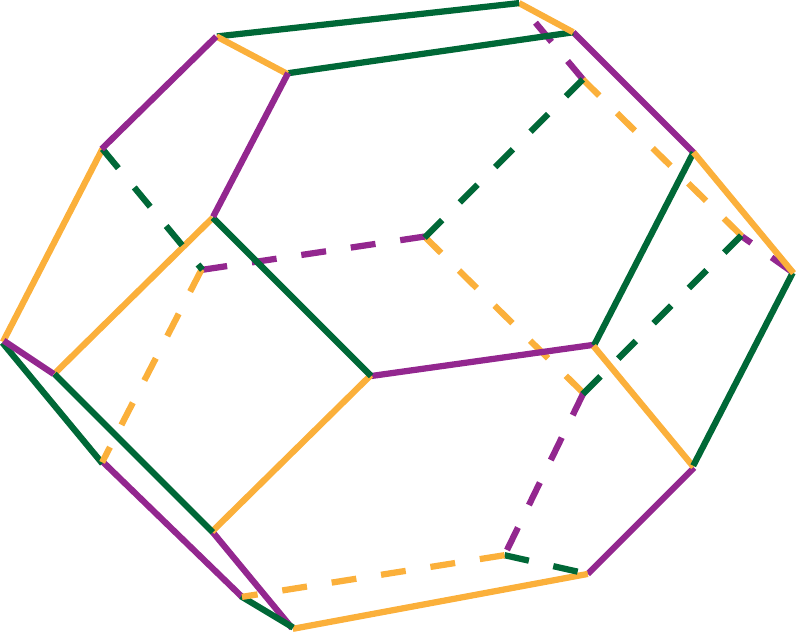}
\end{align}
On each truncated octahedron, we define one $X$-stabilizer on the cell and $Z$-stabilizers on all the plaquettes.

Finally, 12 cubes are attached to the squares in Eq.~\eqref{eq:minimal_1}. The resulting lattice is depicted in Fig.~\ref{fig:minimal_model}.

We select the top and bottom boundaries to be the $Z$-boundaries and the remaining four boundaries to be the $X$-boundaries, as we defined in Section~\ref{sec:boundaries}.

At the intersection of two $X$-boundaries, a yellow cell is truncated from two directions. Consequently, we introduce the $\mathcal{A}_{y}^{corner}(X)$ stabilizers, illustrated as follows:
\begin{align*}
     \adjincludegraphics[width=5cm,valign=c]{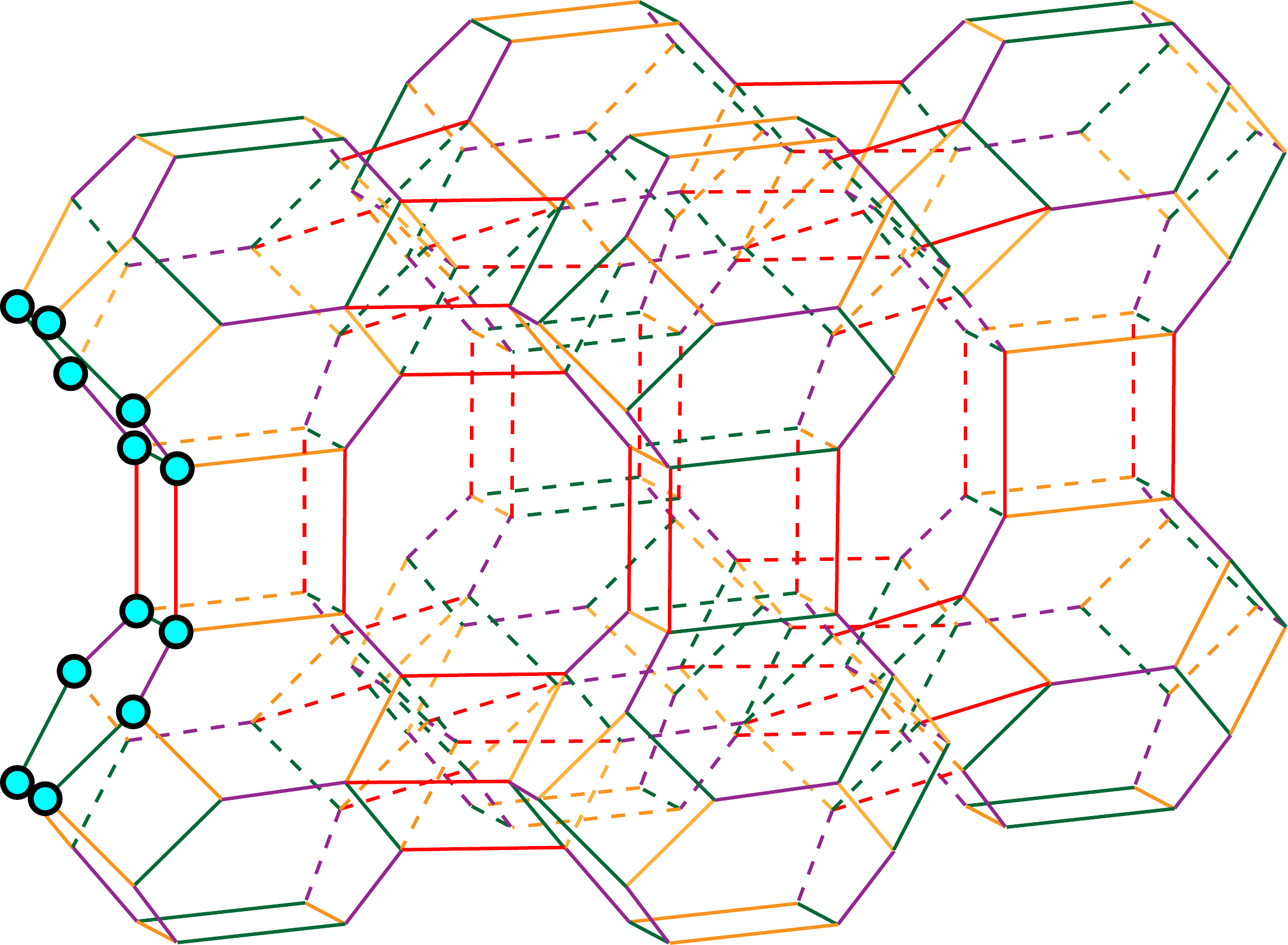}
\end{align*}

At the intersection between the $X$-boundary and the $Z$-boundary, no additional stabilizers are introduced beyond those discussed in Section~\ref{sec:paulizboundary}.

One can further verify that after conjugation of the transversal-$T$ gate, the $X$-boundaries map to the magic boundaries, which then cannot condense either the $X$-type or the $Z$-type excitations, as we discuss in Section~\ref{sec:magicboundary}.

To further demonstrate that this model is equivalent to the one proposed in Fig.~\ref{fig:magic_boundary_2}(b), we initially demonstrate that, through dimension counting, this model can encode three logical qubits. Subsequently, by applying unfolding unitaries to this model, one can obtain three copies of the toric codes, each having the same boundary alignment.

The minimal model we obtain has 192 qubits on the vertices, 45 $X$-stabilizers (8 red, 4 green, 5 yellow, and 28 purple), 206 $Z$-stabilizers (166 on plaquettes and 40 on truncated plaquettes), and 62 independent identity relations (8 on red cells, 24 on purple cells, 28 on the top and bottom boundaries and 2 on yellow cells). Therefore, the total degree of freedom of this system is given by
\begin{align}
    DOF = 2^{192-45-206+62} = 2^3.
\end{align}
Three logical qubits can be encoded in this system. The logical operators are displayed in Fig.~\ref{fig:logical_operators}.
\begin{figure*}
    \centering
    \includegraphics[width = 0.9\linewidth]{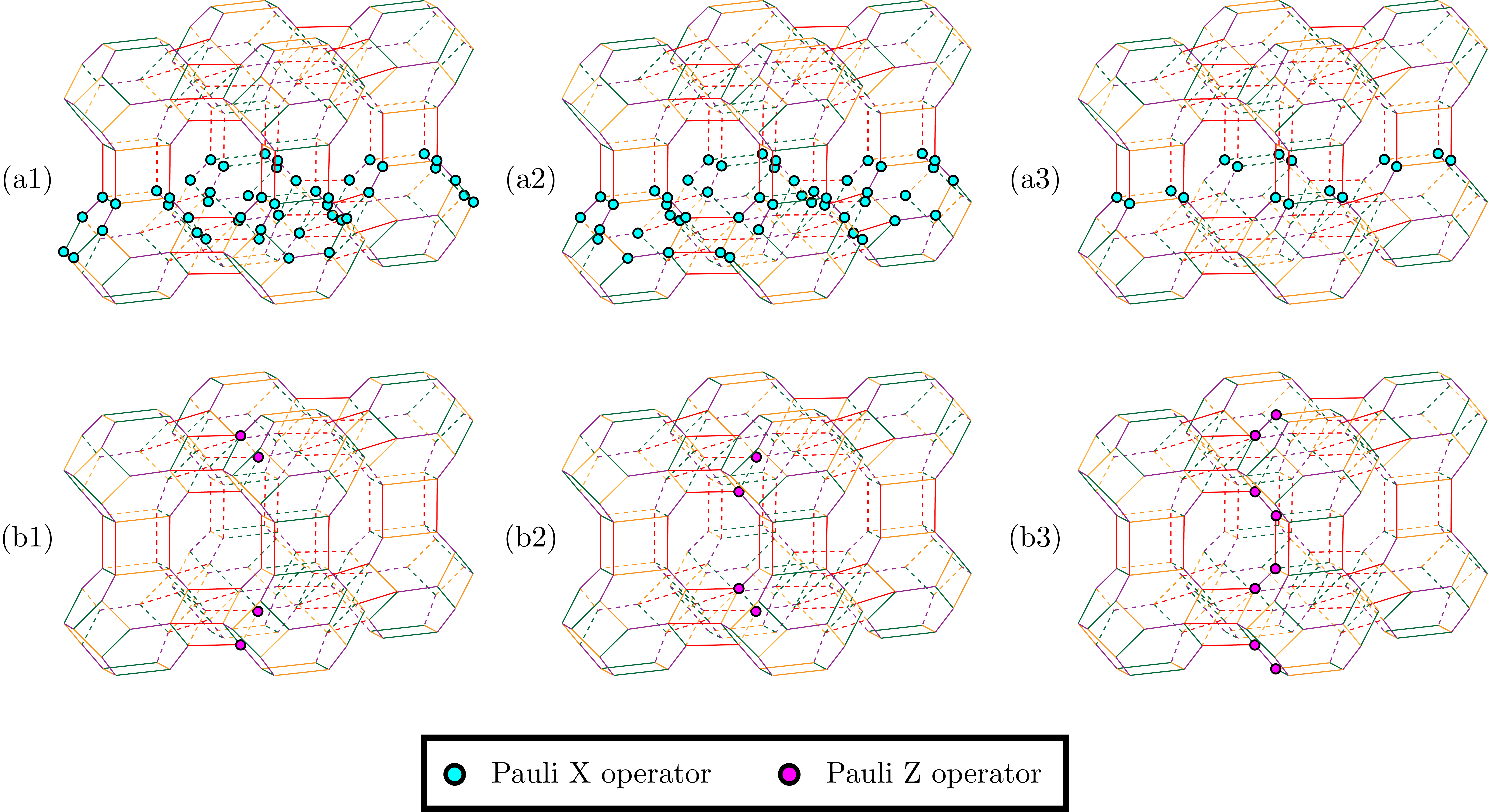}
    \caption{(a1), (a2) and (a3) display the logical $\overline{X}_{pg}$-, $\overline{X}_{py}$- and $\overline{X}_{yg}$-operators, respectively. (b1), (b2) and (b3) display the corresponding logical $\overline{Z}_y$-, $\overline{Z}_g$- and 
 $\overline{Z}_p$-operators, respectively.}
    \label{fig:logical_operators} 
\end{figure*}

\begin{figure*}
    \centering
    \includegraphics[width=0.9\linewidth]{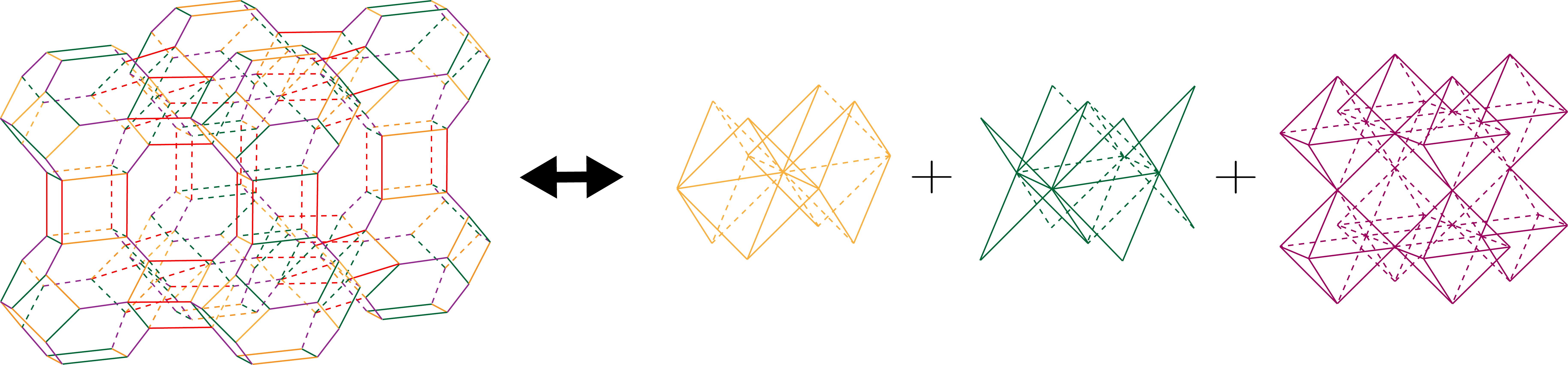}
    \caption{Unfolding the 3D color code to three copies of the 3D toric codes, featuring $(e_1, e_2, e_3)$-boundaries at the top and the bottom, and $(m_1, m_2, m_3)$-boundaries on the four remaining sides.}
    \label{fig:minimal_equivalence}
\end{figure*}

After the unfolding unitaries, the green lattice we get is
\begin{align*}
     \adjincludegraphics[width=5cm,valign=c]{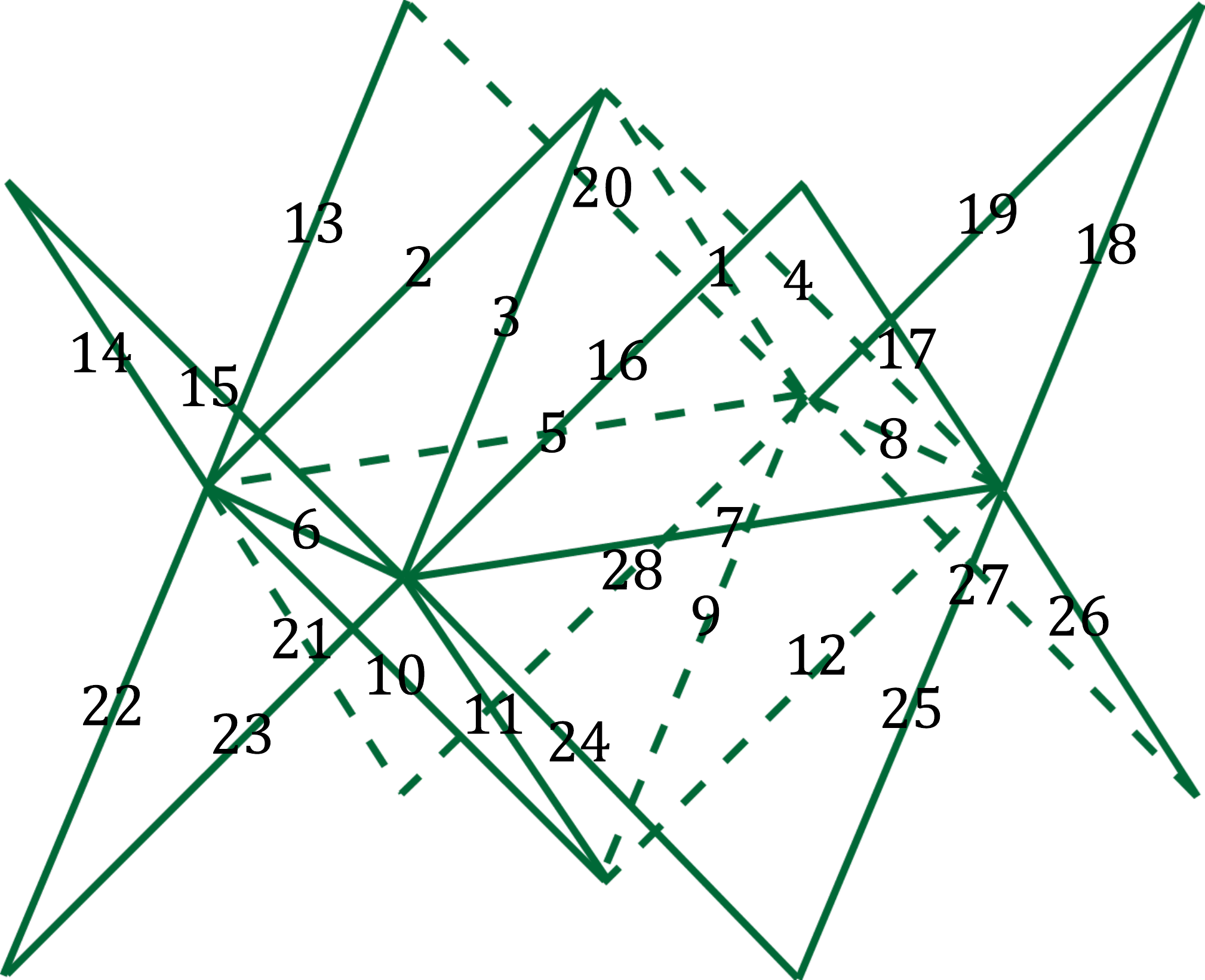}.
\end{align*}
The $X$-stabilizers defined on this lattice are given by
\begin{align*}
    A_1 &= X_1 X_5 X_8 X_9,\ A_2 = X_2 X_5 X_6 X_{10}, \\
    A_3 &= X_3 X_6 X_7 X_{11},\ A_{4} = X_4 X_7 X_8 X_{12}.
\end{align*}
The $Z$-stabilizers are given by
\begin{align*}
    B_1 &= Z_1 Z_2 Z_5,\ B_2 = Z_2 Z_3 Z_6,\ B_3 = Z_3 Z_4 Z_7, \\
    B_4 &= Z_1 Z_4 Z_8,\ B_5 = Z_5 Z_7 Z_{10},\ B_6 = Z_{6} Z_{10} Z_{11},\\
    B_7 &= Z_7 Z_{11} Z_{12},\ B_{8} = Z_{8} Z_9 Z_{12},\ B_9 = Z_{1} Z_{20},\\
    B_{10} &= Z_{2} Z_{13},\ B_{11} = Z_{2} Z_{14},\ B_{12} = Z_{3} Z_{15}, \\
    B_{13} &= Z_3 Z_{16},\ B_{14} = Z_{4} Z_{17},\ B_{15} = Z_{4} Z_{18},\\
    B_{16} &= Z_1 Z_{19},\ B_{17} = Z_9 Z_{27},\ B_{18} = Z_9 Z_{28},\\
    B_{19} &= Z_{10} Z_{21},\ B_{20} = Z_{10} Z_{22},\ B_{21} = Z_{11} Z_{23},\\
    B_{22} &= Z_{11} Z_{24},\ B_{23} = Z_{12} Z_{25},\ B_{24} = Z_{12} Z_{26},\\
    B_{25} &= Z_{5} Z_{13} Z_{20},\ B_{26} = Z_{6} Z_{14} Z_{15},\ B_{27} = Z_{7} Z_{16} Z_{17},\\
    B_{28} &= Z_{8} Z_{18} Z_{19},\ B_{29} = Z_{5} Z_{21} Z_{28},\ B_{30} = Z_{6} Z_{22} Z_{23},\\
    B_{31} &= Z_{7} Z_{24} Z_{25},\ B_{32} = Z_{8} Z_{26} Z_{27}
\end{align*}
The stabilizers are not independent; they are interconnected by nine relations. Specifically,
\begin{align}
\prod_{i=1}^{8} B_i = 1, \quad B_{2} B_{11} B_{12} B_{26} = 1,
\end{align}
with the remaining seven relations being analogous to the second one, but applied to rotationally equivalent positions. Therefore, the degree of freedom of this model is given by
\begin{align}
    DOF_{green} = 2^{28-4-32+9} = 2.
\end{align}
This model encodes one logical qubit. The logical operators are given by
\begin{align}
    \overline{X} = \prod_{i=9}^{12} X_{i} \prod_{j = 21}^{28} X_{j},\ \overline{Z} = Z_{3} Z_{11}.
\end{align}
One can also confirm that the boundaries at the top and bottom are rough boundaries, while the other four boundaries are smooth boundaries.

The yellow lattice we get is
\begin{align*}
     \adjincludegraphics[width=5cm,valign=c]{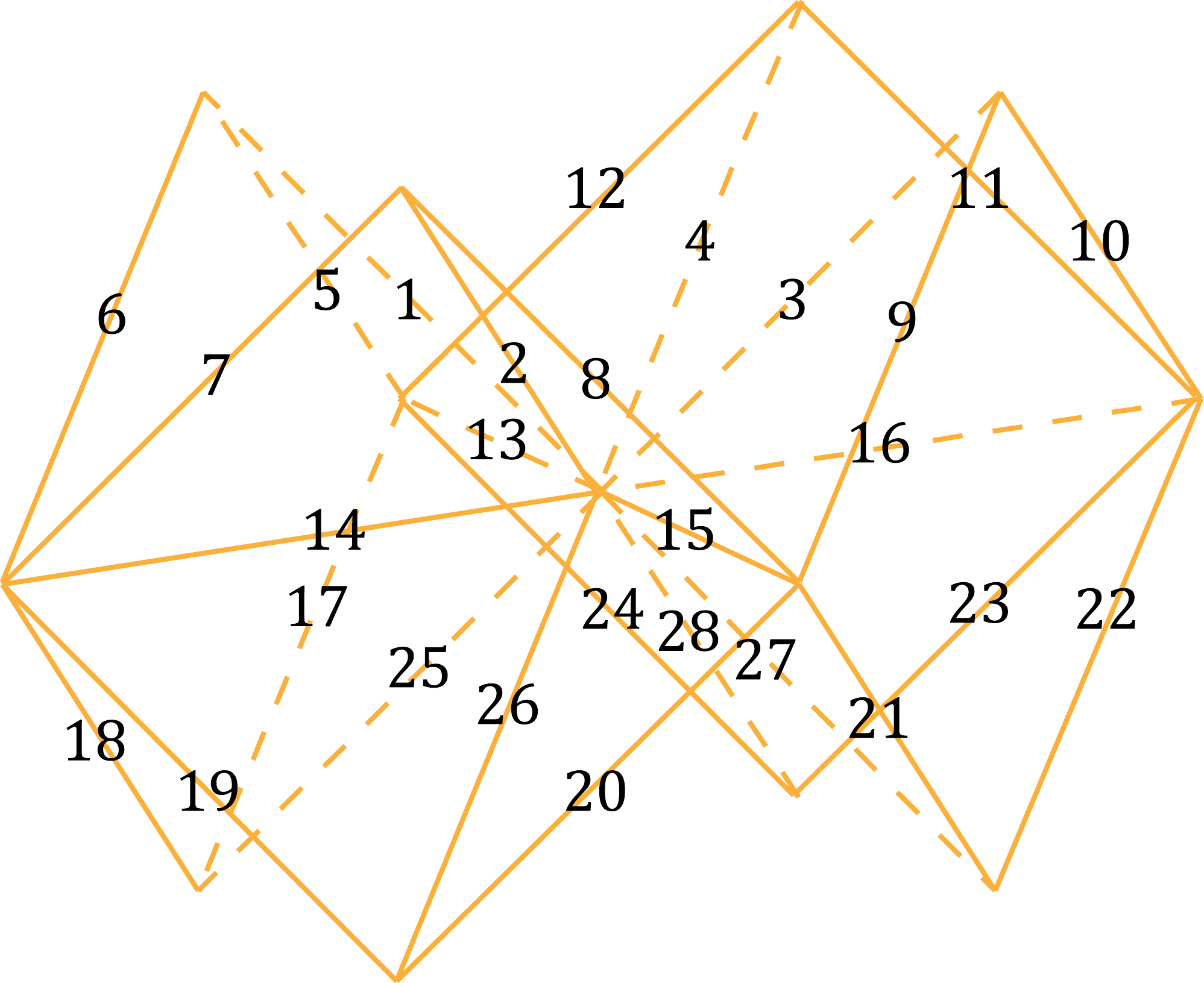}.
\end{align*}
The $X$-stabilizers are given by
\begin{align*}
    A_1 &= X_6 X_7 X_{18} X_{19},\ A_2 = X_5 X_{12} X_{17} X_{24} \\
    A_3 &= X_{10} X_{11} X_{22} X_{23},\ A_{4} = X_{8} X_{9} X_{20} X_{21} \\
    A_5 &= X_1 X_2 X_3 X_4 X_{13} X_{14} X_{15} X_{16} X_{25} X_{26} X_{27} X_{28}
\end{align*}
The $Z$-stabilizers are given by
\begin{align*}
    B_1 &= Z_1 Z_{5} Z_{13},\ B_{2} = Z_{1} Z_{6} Z_{14},\ B_{3} = Z_{2} Z_{7} Z_{14} \\
    B_4 &= Z_2 Z_8 Z_{15},\ B_5 = Z_3 Z_9 Z_{15},\ B_6 = Z_{3} Z_{10} Z_{16} \\
    B_7 &= Z_4 Z_{11} Z_{16},\ B_8 = Z_4 Z_{12} Z_{13},\ B_9 = Z_{13} Z_{17} Z_{25} \\
    B_{10} &= Z_{14} Z_{18} Z_{25},\ B_{11} = Z_{14} Z_{19} Z_{26},\\
    B_{12} &= Z_{15} Z_{20} Z_{26},\ B_{13} = Z_{15} Z_{21} Z_{27},\\
    B_{14} &= Z_{16} Z_{22} Z_{27},\ B_{15} = Z_{16} Z_{23} Z_{28} \\
    B_{16} &= Z_{13} Z_{24} Z_{28},\ B_{17} = Z_1 Z_2,\ B_{18} = Z_2 Z_3,\\
    B_{19} &= Z_3 Z_4,\ B_{20} = Z_1 Z_4,\ B_{21} = Z_{25} Z_{28},\\
    B_{22} &= Z_{25} Z_{26},\ B_{23} = Z_{26} Z_{27},\ B_{24} = Z_{27} Z_{28},\\
    B_{25} &= Z_6 Z_7,\ B_{26} = Z_8 Z_9,\ B_{27} = Z_{10} Z_{11},\\
    B_{28} &= Z_{5} Z_{12},\ B_{29} = Z_{18} Z_{19},\ B_{30} = Z_{20} Z_{21},\\
    B_{31} &= Z_{22} Z_{23},\ B_{32} = Z_{17} Z_{24}
\end{align*}
The stabilizers are not independent. They are interconnected by ten relations. Specifically
\begin{align*}
    B_{17} B_{18} B_{19} B_{20} &= 1,\\
    B_{21} B_{22} B_{23} B_{24} &= 1,\\
    B_2 B_3 B_{17} B_{25} &= 1,
\end{align*}
with the remaining seven relations being analogous to the second one, but applied to rotationally equivalent positions. So the remaining degree of freedom of this model is given by
\begin{align}
    DOF_{yellow} = 2^{28-5-32+10} = 2.
\end{align}
This model encodes one logical qubit. The logical operators are given by
\begin{align}
    \overline{X} = \prod_{i=1}^{12} X_{i},\ \overline{Z} = Z_{2} Z_{26}.
\end{align}

The purple lattice we get is
\begin{align*}
     \adjincludegraphics[width=5cm,valign=c]{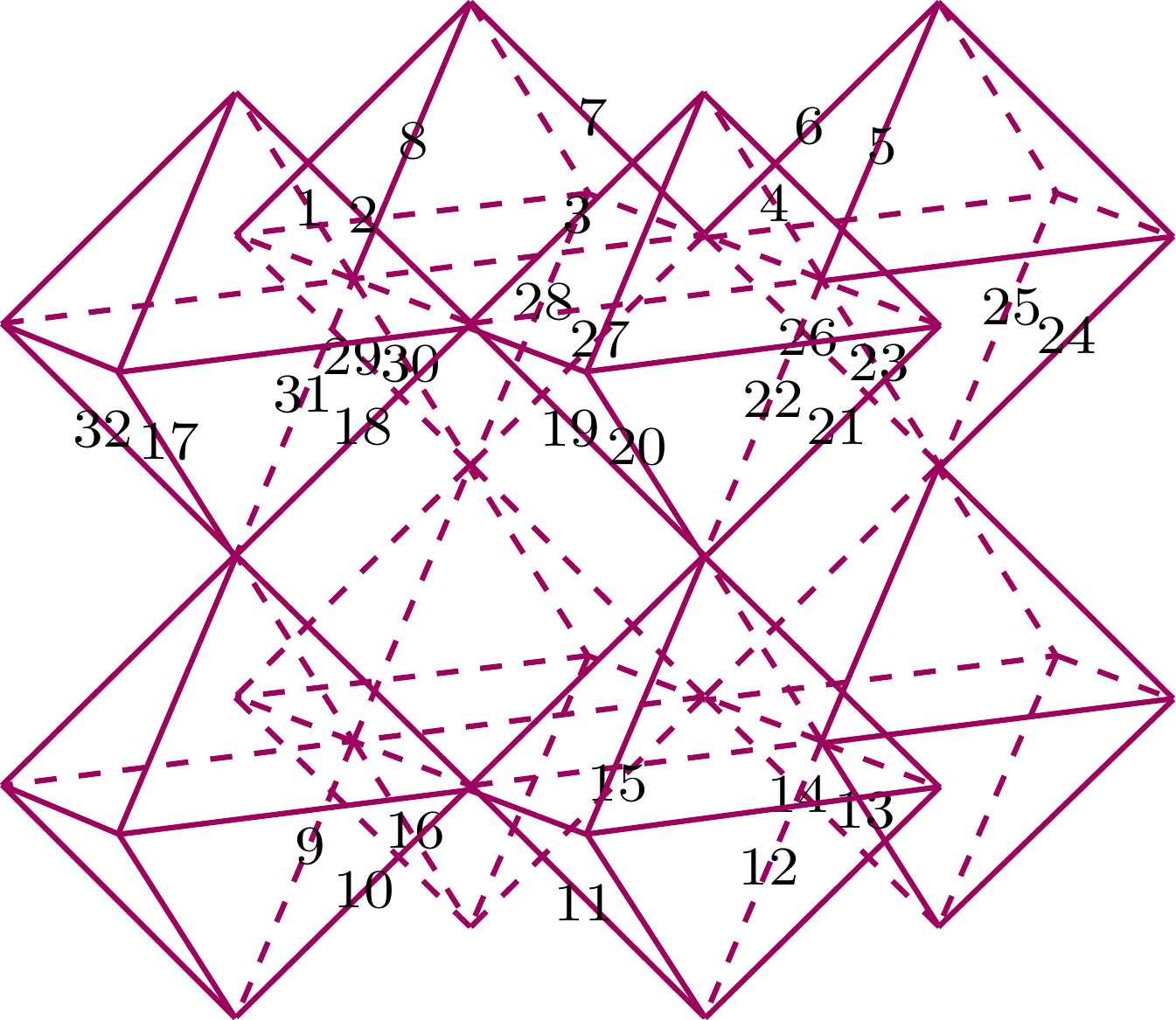}.
\end{align*}
The labels are ordered in the following way. Qubits 1 to 8 are at the same height, qubits 17 to 32 are at the same height, and qubits 9 to 16 are at the same height. 

Since there are 96 qubits on this lattice, writing down all the stabilizers is tedious. For better understanding, we give a brief summary here. According to the construction before, there are 28 $X$-stabilizers on the vertices (In total there are 36 vertices. However, the top four and the bottom four vertices don't have $X$-stabilizers on them.). There are $Z$-stabilizers on every plaquette, and 8 $Z$-stailizers on the top and bottom, which are
\begin{align*}
    &Z_2 Z_3, Z_4 Z_5, Z_6 Z_7, Z_8 Z_1,\\ 
    &Z_{10} Z_{11}, Z_{12} Z_{13}, Z_{14} Z_{15}, Z_{16} Z_{9}.
\end{align*}
In addition, there are 11 independent identity relations. Therefore, the total degree of freedom of this model is given by
\begin{align}
    DOF_{purple} = 2^{96-28-78+11} = 2.
\end{align}
This model encodes one logical qubit. The logical operators are given by
\begin{align}
    \overline{X} = \prod_{i=17}^{32} X_{i},\ \overline{Z} = Z_{2} Z_{10} Z_{18} Z_{33}.
\end{align}

In conclusion, the 3D color code obtained is equivalent to three copies of the toric codes, featuring $(e_1, e_2, e_3)$-boundaries at the top and bottom, and $(m_1, m_2, m_3)$-boundaries on the four remaining sides, up to local unitaries and ancillary qubits. Analogous to Fig.~\ref{fig:magic_boundary_2}(b), we show the relation in Fig.~\ref{fig:minimal_equivalence}.

\end{document}